\pdfoutput=1

\documentclass[acmsmall,screen,authorversion,noacm]{acmart}
\usepackage[utf8]{inputenc}
\usepackage[T1]{fontenc}
\usepackage[american]{babel}
\usepackage{mathtools}
\usepackage{xcolor}
\usepackage{paralist}
\usepackage{enumitem}
\usepackage{hyperref}
\usepackage{xspace}
\usepackage{xifthen}
\usepackage{stackengine}
\usepackage{wrapfig}
\usepackage{tabularx}
\usepackage{wasysym}
\usepackage{empheq}
\usepackage{scalerel}

\def\techreportversion{}
\usepackage{trimspaces}
\makeatletter
\newcommand*{\trimSpaces}[1]{\begingroup\edef\my@tmp@trimmed{\trim@spaces@noexp{#1}}\expandafter\endgroup\my@tmp@trimmed}
\makeatother

\newcommand{\moreless}[2]{\ifdefined\techreportversion\trimSpaces{#1}\else\trimSpaces{#2}\fi}
\newcommand{\techreport}[1]{\moreless{#1}{}}
\newcommand{\nontechreport}[1]{\moreless{}{#1}}

\usepackage{listings}
\lstdefinelanguage{SPL}{
  morekeywords={acc, struct,if,else,returns,procedure,requires,ensures,:=,var,
    new,old,free,implicit,modifies,call,locals,assume,assert,choose,havoc,ghost,
    predicate,function,invariant,while, return,atomic, split, type, field, result,
    define, datatype, axiom, val, for, restart},
  deletekeywords={union,int},
  numbers=left,
  xleftmargin=2em,
  escapeinside={@}{@},
  numberstyle=\tiny,
  basicstyle=\footnotesize\ttfamily,
  columns=flexible,
  morecomment=[s][\color{green!60!black}]{/*}{*/},
  morecomment=[l][\color{green!60!black}]{//},
  moredelim=**[is][\color{purple}]{|<}{>|},
  mathescape=true,
}

\lstdefinestyle{proof}{
  language=Java,
  basicstyle=\linespread{1.15}\footnotesize\ttfamily,
  commentstyle=\color{gray}\ttfamily,
  tabsize=2,
  numberbychapter=false,
  keepspaces=true,
  numbers=left,
  numberstyle=\scriptsize,
  numbersep=7pt,
  firstnumber=last,
  mathescape=true,
  keywords={},
  xleftmargin=3pt,
}
\lstdefinestyle{inline}{
  basicstyle=\relscale{.9}\ttfamily,
  keywords={},
}
\usepackage{relsize}
\newcommand{\code}[2][]{\lstinline[style=inline,#1]!#2!}

\usepackage{cleveref}
\crefformat{section}{\S#2#1#3}
\crefmultiformat{section}{\S\S#2#1#3}{ and~#2#1#3}{, #2#1#3}{, and~#2#1#3}
\crefname{line}{Line}{Lines}
\Crefname{line}{Line}{Lines}

\usepackage{tikz}
\usetikzlibrary{matrix,arrows,positioning,calc,fit,backgrounds,shapes}

\tikzset{%
  array/.style={matrix of nodes,nodes={draw, minimum size=5mm, anchor=center},column sep=-\pgflinewidth, row sep=-\pgflinewidth, nodes in empty cells,anchor=center},
  ptr/.style={*->, shorten <=-(1.8pt+1.4\pgflinewidth)},
  edge/.style={->, thick},
  dedge/.style={<->, dashed},
  fedge/.style={->, dashed},
  unode/.style={circle, draw=black, thick, minimum size=8mm, inner sep=0},
  mnode/.style={circle, draw=black, thick, fill=gray!20, minimum size=8mm, inner sep=0, font=\scriptsize},
  stackVar/.style={circle, fill=none, inner sep=0pt, minimum size=8mm, font=\normalsize, outer sep=-4pt},
  gnode/.style={circle, draw=black, thick, minimum size=8mm},
  pnode/.style={circle, draw=black, thick, minimum size=8mm},
  rnode/.style={draw=black, thick, minimum size=8mm},
  lbl/.style={circle, fill=none, inner sep=0pt, minimum size=8mm},
  dnode/.style={circle, draw=black, thick, dotted, minimum size=8mm},
  inflow/.style={circle, fill=none, inner sep=0pt, minimum size=5mm, font=\normalsize},
  phantomNode/.style={circle, fill=none, inner sep=0pt, minimum
    size=0pt}
}

\usepackage{pftools}

\definecolor{colorMyGreen}{RGB}{80,170,0}
\definecolor{colorMyRed}{RGB}{180,20,20}
\definecolor{colorMyBlue}{RGB}{40,40,220}
\definecolor{colorMyPink}{RGB}{220,40,220}
\definecolor{colorMyFutCol}{RGB}{20,20,180}
\definecolor{colorMyHistCol}{RGB}{180,20,20}

\newcommand{\makeTeal}[1]{\textcolor{teal}{#1}}
\newcommand{\makeFutCol}[1]{\textcolor{colorMyFutCol}{#1}}
\newcommand{\makeHistCol}[1]{\textcolor{colorMyHistCol}{#1}}
\newcommand{\setHistCol}[1]{\textcolor{colorMyHistCol}{#1}}

\usepackage{mathpartir}
\makeatletter
\newcommand{\rulelabel}[2]{%
   \protected@write \@auxout {}{\string \newlabel {#1}{{#2}{\thepage}{#2}{#1}{}} }%
   \hypertarget{#1}{}
}
\makeatother
\newcommand{\mkrulelabel}[1]{\rulelabel{rule:#1}{\textsc{(#1)}}}
\newcommand{\infrule}[3][]{\ifthenelse{\equal{#1}{}}{\inferrule{#2}{#3}}{\inferrule[\textsc{(#1)}\mkrulelabel{#1}]{#2}{#3}}}

\usepackage{stackengine}
\stackMath
\newcommand\xxrightarrow[1]{\raisebox{-.85pt}{\ensuremath{\smash{\mathrel{%
  \setbox2=\hbox{\stackon{\scriptstyle#1}{\scriptstyle#1}}%
  \stackon[-4.5pt]{%
    \xrightarrow{\makebox[\dimexpr\wd2\relax]{}}%
  }{%
   \scriptstyle#1\,%
  }%
}}}}}

\newcommand{\smartparagraph}[1]{\subsubsection*{#1}}

\newcommand{\mkmathop}[1]{\operatorname{#1}}
\newcommand{\mymathtt}[1]{\text{\relscale{.9}\ttfamily#1}}

\def\prallspacing{\mskip 2mu plus 2mu minus 3mu}
\newcommand{\prall}[1]{{\prallspacing#1\prallspacing}}

\newcommand{\bool}{\mathbb{B}}
\newcommand{\true}{\mathit{true}}
\newcommand{\false}{\mathit{false}}
\renewcommand{\emptyset}{\varnothing}
\newcommand{\setcompact}[1]{\{#1\}}
\newcommand{\set}[1]{\{\,#1\,\}}
\newcommand{\setnd}[1]{\{#1\}}
\newcommand{\setcond}[2]{\set{#1\;\mid\;#2}}
\newcommand{\cardof}[1]{|#1|}
\newcommand{\powerset}[1]{\mathbb{P}(#1)}
\newcommand{\nat}{\mathbb{N}}
\newcommand{\ZZ}{\mathbb{Z}}
\newcommand{\domof}[1]{\mkmathop{dom}(#1)}

\newcommand{\freeVariablesOf}[1]{\mkmathop{fv}(#1)}
\newcommand{\ite}[3]{#1\:?\;#2\::\:#3}

\newcommand{\bnf}{\;\mid\;}
\newcommand{\defeq}{\triangleq} % :=
\newcommand{\defebnf}{\Coloneqq}

\newcommand{\invNode}{\varphi}

\newcommand{\acom}{\mymathtt{com}} % tt vs it
\newcommand{\astmt}{\mymathtt{st}} % tt vs it
\newcommand{\afun}{\mymathtt{fun}} % tt vs it
\newcommand{\cskip}{\mymathtt{skip}}
\newcommand{\seqof}[2]{#1;#2}
\newcommand{\choiceof}[2]{#1+#2}
\newcommand{\loopof}[1]{{#1}^{*}}
\newcommand{\setstmt}{\mathtt{ST}}
\newcommand{\setcom}{\mathtt{COM}}

\newcommand{\setconfig}{\mathsf{CF}}
\newcommand{\setstates}{\Sigma}
\newcommand{\localindex}{\mathsf{L}}
\newcommand{\sharedindex}{\mathsf{G}}

\newcommand{\setshared}{\Sigma_{\sharedindex}}
\newcommand{\setlocal}{\Sigma_{\localindex}}
\newcommand{\sharedemp}{\emp_{\sharedindex}}
\newcommand{\localemp}{\emp_{\localindex}}

\newcommand{\acfg}{\mathsf{cf}} % sf vs it
\newcommand{\aconfig}{\acfg} % delete
\newcommand{\apc}{\mathsf{pc}}
\newcommand{\alocal}{\mathsf{l}}
\newcommand{\ashared}{\mathsf{g}}
\newcommand{\alocalseq}{\lambda}
\newcommand{\asharedseq}{\gamma}

\newcommand{\lastOf}[1]{\mathsf{last}(#1)}
\newcommand{\astate}{\mathsf{s}}
\newcommand{\astatep}{\mathsf{t}}

\newcommand{\semCom}[1]{\llbracket#1\rrbracket}
\newcommand{\semComOf}[2]{\semCom{#1}(#2)}

\newcommand{\pcStepOf}[3]{#1\,\xxrightarrow{\vphantom{pt}#2}\,#3}
\newcommand{\progStepRel}{\rightarrow}

\newcommand{\anop}{\oplus}

\newcommand{\setaddrs}{\mathit{Addr}}
\newcommand{\setsvars}{\mathit{Var}}
\newcommand{\setvalues}{\mathit{Val}}
\newcommand{\stackindex}{{\mathsf{S}}}
\newcommand{\setstacks}{\setstates_\stackindex}
\newcommand{\astack}{\iota}
\newcommand{\heapindex}{{\mathsf{H}}}
\newcommand{\setheaps}{\setstates_\heapindex}
\newcommand{\heapmult}{\mathop{{\mstar}_\heapindex}}
\newcommand{\heapmultdef}{\mathop{\#_\heapindex}}
\newcommand{\stackheapmult}{\mathop{{\mstar}_{\stackindex\heapindex}}}
\newcommand{\heapemp}{\emp_\heapindex}
\newcommand{\aheap}{h}

\newcommand{\discup}{\uplus}
\newcommand{\myemph}[1]{{\bfseries #1}}
\newcommand{\modelsp}{\models'}

\newcommand{\abort}{\textsf{abort}}
\newcommand{\stateunit}{1}
\newcommand{\statemult}{\mathop{*}}
\newcommand{\sharedstates}{\setstates_{\sharedindex}}
\newcommand{\sharedmult}{\mathop{{\mstar}_{\sharedindex}}}
\newcommand{\localmult}{\mathop{{\mstar}_{\localindex}}}
\newcommand{\localstates}{\setstates_{\localindex}}

\newcommand{\statemultdef}{\mathop{\#}}
\newcommand{\apred}{\mathit{p}}
\newcommand{\apredp}{\mathit{q}}
\newcommand{\apredpp}{\mathit{o}}
\newcommand{\apredppp}{\mathit{r}}
\newcommand{\apredpppp}{\mathit{n}}

\newcommand{\semOf}[1]{\llbracket#1\rrbracket}
\newcommand{\csemOf}[1]{\llbracket\mkern-7mu\semOf{#1}\mkern-7mu\rrbracket}

\newcommand{\anipred}{\mathit{i}}
\newcommand{\acpred}{\mathit{a}}
\newcommand{\acpredp}{\mathit{b}}
\newcommand{\acpredpp}{\mathit{c}}
\newcommand{\acpredppp}{\mathit{d}}
\newcommand{\acpredpppp}{\mathit{e}}

\newcommand{\PREDleft}{\langle}
\newcommand{\PREDright}{\rangle}

\newcommand{\OBLname}{\mathsf{Obl}}
\newcommand{\FULname}{\mathsf{Ful}}

\newcommand{\annotPRED}[1]{\left\PREDleft\,#1\,\right\PREDright}

\newcommand{\FUT}[3]{\annotPRED{#1}#2\annotPRED{#3}}
\newcommand{\FUTp}[3]{\left\PREDleft\!\!\begin{array}{l}#1\end{array}\!\!\middle\PREDright\:#2\:\middle\PREDleft\!\!\begin{array}{l}#3\end{array}\!\!\right\PREDright}
\newcommand{\OBL}[1]{\OBLname(\mymathtt{#1})}
\newcommand{\FUL}[2]{\FULname(\mymathtt{#1},#2)}
\newcommand{\mstar}{\mathop{*}}
\DeclareMathOperator*{\bigmstar}{\scalerel*{\ast}{\sum}}
\newcommand{\sepimp}{\mathrel{-\!\!*}}
\newcommand{\emp}{\mathsf{emp}}
\newcommand{\hoareOf}[3]{\set{#1}\:#2\:\set{#3}}
\newcommand{\subModels}{\models}
\newcommand{\semCalc}{\Vdash}

\newcommand{\past}{\Diamonddot}
\newcommand{\pastOf}[1]{\past#1}
\newcommand{\nowOf}[1]{\_{#1}}

\newcommand{\initset}[2]{\mathsf{Init}_{#1, #2}}
\newcommand{\acceptset}[1]{\mathsf{Acc}_{#1}}
\newcommand{\reachset}[1]{\mathsf{Reach}(#1)}

\newcommand{\wpre}{\mathit{wp}}
\newcommand{\wpreOf}[2]{\wpre(#1, #2)}
\newcommand{\commentl}[1]{\text{\small{(\;#1\;)\;\;\;}}}

\newcommand{\theInterference}{\mathbb{I}}

\newcommand{\thePredicates}{\mathbb{P}}
\newcommand{\locsafe}[6]{\mathsf{safe}^{#3}_{#1, #2}(#4, #5, #6)}
\newcommand{\locsafedef}[4]{\locsafe{\thePredicates}{\theInterference}{#1}{#2}{#3}{#4}}
\newcommand{\locsafek}[3]{\locsafedef{k}{#1}{#2}{#3}}
\newcommand{\cfsafedef}[3]{\mathsf{cfsafe^{#1}(#2, #3)}}
\newcommand{\cfsafek}[2]{\cfsafedef{k}{#1}{#2}}
\newcommand{\inter}[2]{\mathsf{inter}(#1, #2)}
\newcommand{\accept}[3]{\mathsf{acc}(#1, #2, #3)}
\newcommand{\effectful}[2]{\mathsf{eff}(#1, #2)}
\newcommand{\isInterferenceFreeOf}[2][\theInterference]{\boxast_{#1}\,#2}

\newcommand{\nextsel}[1][]{\mymathtt{next}_{#1}}
\newcommand{\keysel}{\mymathtt{key}}

\newcommand{\selof}[2]{#1\mkern+2mu{\rightarrow}\mkern+2mu#2}

\newcommand{\malloc}{\mymathtt{malloc}}
\newcommand{\assume}{\mymathtt{assume}}
\newcommand{\assignsign}{:=}
\newcommand{\assignof}[2]{#1\mkern-1mu\assignsign\mkern-1mu#2}
\newcommand{\mallocof}[1]{\assignof{#1}{\malloc}}
\newcommand{\assumeof}[1]{\assume\;#1}
\newcommand{\assertof}[1]{\mymathtt{assert}\;#1}

\newcommand{\astateseq}{\sigma}
\newcommand{\astateseqp}{\tau}

\newcommand{\bluetext}[1]{\textcolor{blue}{#1}}

\newcommand{\monfun}[1]{\mkmathop{Mon}(#1)}
\newcommand{\join}{\sqcup}

\newcommand{\contentsof}[1]{\mathsf{C}(#1)}
\newcommand{\keysetof}[1]{\mathsf{KS}(#1)}
\newcommand{\insetof}[1]{\mathsf{IS}(#1)}
\newcommand{\newinsetof}[1]{\mathsf{IS}'(#1)}

\newcommand{\annot}[2][]{\makeTeal{\ifthenelse{\equal{#1}{}}{}{\givename{#1}}\left\{\,\begin{aligned}#2\end{aligned}\,\right\}}}
\newcommand{\annotml}[2][]{\makeTeal{\ifthenelse{\equal{#1}{}}{}{\givename{#1}}\left\{\,\begin{aligned}#2\end{aligned}\,\right\}}}

\newcommand{\amonoid}{\mathbb{M}}
\newcommand{\monadd}{+}
\newcommand{\monunit}{0}
\newcommand{\amonval}{\mathit{m}}
\newcommand{\amonvalp}{\mathit{n}}
\newcommand{\amonvalpp}{\mathit{o}}

\newcommand{\setnodes}{\mathit{X}}
\newcommand{\anode}{\mathit{x}}
\newcommand{\anodep}{\mathit{y}}

\newcommand{\pfun}{\rightharpoondown}
\newcommand{\edges}{\mathit{E}}

\newcommand{\rhs}{\mathit{rhs}}

\newcommand{\rhsatof}[2]{\rhs_{#1}(#2)}
\newcommand{\aflowconstraint}{\mathit{c}}
\newcommand{\setflowconstraints}{\mathit{FC}}
\newcommand{\fval}{\mathit{flow}}
\newcommand{\fvalof}[1]{\fval(#1)}
\newcommand{\inflow}{\mathit{in}}
\newcommand{\inflowof}[1]{\inflow(#1)}
\newcommand{\outflow}{\mathit{out}}
\newcommand{\outflowof}[1]{\outflow(#1)}

\newcommand{\updfun}{\mathit{up}}
\newcommand{\applyupdfun}{[\updfun]}

\newcommand{\transformerof}[1]{\mathit{tf}(#1)}

\newcommand{\lfpof}[1]{\mathit{lfp}(#1)}
\newcommand{\pairingof}[2]{\langle #1, #2\rangle}

\newcommand{\myid}{\mathit{id}}
\newcommand{\myidof}[1]{\myid_{#1}}

\newcommand{\cval}{\mathit{cval}}

\newcommand{\setsel}{\mathit{Sel}}
\newcommand{\setdsel}{\mathit{DSel}}
\newcommand{\setpsel}{\mathit{PSel}}
\newcommand{\asel}{\mathsf{sel}}

\newcommand{\apsel}{\mathsf{psel}}

\newcommand{\aheapgraph}{\mathit{h}}
\newcommand{\aflowgraph}{\mathit{fg}}
\newcommand{\setheapgraphs}{\mathit{HG}}
\newcommand{\setflowgraphs}{\mathit{FG}}

\newcommand{\sval}{\mathit{sval}}
\newcommand{\dval}{\mathit{dval}}
\newcommand{\pval}{\mathit{pval}}
\newcommand{\svalof}[1]{\sval(#1)}
\newcommand{\dvalof}[1]{\dval(#1)}
\newcommand{\pvalof}[1]{\pval(#1)}

\newcommand{\genrhs}{\mathit{gen}}
\newcommand{\genrhsof}[1]{\genrhs(#1)}

\newcommand{\anaval}{\mathit{a}}
\newcommand{\anavalp}{\mathit{b}}
\newcommand{\initfgof}[1]{\mathit{initfg}(#1)}

\newcommand{\setavars}{\mathit{ALVar}}
\newcommand{\setfvars}{\mathit{FLVar}}
\newcommand{\anavar}{\mathit{ax}}
\newcommand{\anfvar}{\mathit{fx}}
\newcommand{\anfval}{\amonval}

\newcommand{\setgpvars}{\mathit{GPVar}}
\newcommand{\setlpvars}{\mathit{LPVar}}
\newcommand{\setpvars}{\mathit{PVar}}
\newcommand{\setvars}{\mathit{Y}}

\newcommand{\apvar}{\mathit{px}}
\newcommand{\apvarp}{\mathit{qx}}
\newcommand{\fsel}{\mathsf{flow}}
\newcommand{\insel}{\mathsf{in}}

\newcommand{\setaops}{\mathit{AOP}}
\newcommand{\anaop}{\mathsf{aop}}
\newcommand{\anaopof}[1]{\anaop(#1)}

\newcommand{\setfops}{\mathit{FOP}}
\newcommand{\anfop}{\mathsf{fop}}
\newcommand{\anfopof}[1]{\anfop(#1)}

\newcommand{\setafpreds}{\mathit{Pred}}
\newcommand{\anafpred}{\mathsf{pred}}
\newcommand{\anafpredof}[1]{\anafpred(#1)}

\newcommand{\setvalsof}[1]{\mathit{Val}(#1)}
\newcommand{\aval}{\mathit{val}}
\newcommand{\avalof}[1]{\aval(#1)}

\newcommand{\fterm}{\mathit{fterm}}
\newcommand{\aterm}{\mathit{aterm}}
\newcommand{\afterm}{\mathit{term}}

\newcommand{\setlogval}{\mathit{LogVal}}

\newcommand{\setproglogval}{\mathit{ProgLogVal}}
\newcommand{\anlvval}{\mathit{lv}}

\newcommand{\apvlvval}{\mathit{plv}}
\newcommand{\apvlvvalof}[1]{\apvlvval(#1)}

\newcommand{\atomicpred}{\mathit{apred}}
\newcommand{\statepred}{\mathit{A}}

\newcommand{\comppred}{\mathit{C}}
\newcommand{\pointsto}[1]{\mapsto_{#1}}
\newcommand{\spointsto}[2]{#1 \mapsto \langle #2 \rangle}

\newcommand{\splitfg}{\mathit{split}}
\newcommand{\splitOf}[1]{\splitfg(#1)}

\newcommand{\extfg}{\mathit{ext}}
\newcommand{\extOf}[2]{\extfg(#1, #2)}

\newcommand{\updOf}[2]{[#1\mapsto#2]}
\newcommand{\natmon}{\nat\discup\amonoid}
\newcommand{\adom}{\mathbb{D}}
\newcommand{\condval}[2]{#1\triangleright#2}
\newcommand{\wellsplitof}[1]{\mathit{wellsplit}(#1)}

\newcommand{\assignfromop}{\mathit{assignfromop}}
\newcommand{\assignfromsel}{\mathit{assignfromsel}}
\newcommand{\assigntosel}{\mathit{assigntosel}}

\newcommand{\nextfld}{\mymathtt{next}}
\newcommand{\keyfld}{\mymathtt{key}}
\newcommand{\markfld}{\mymathtt{mark}}
\newcommand{\flowfld}{\mymathtt{flow}}
\newcommand{\inflowfld}{\mymathtt{in}}
\newcommand{\head}{\mymathtt{head}}
\newcommand{\vsucc}{\mathit{succ}}
\newcommand{\vres}{\mathit{res}}
\newcommand{\assign}{\mathop{\mathtt{:=}}}

\newcommand{\keyvarof}[1]{\mathit{key}(#1)}
\newcommand{\flowvarof}[1]{\mathit{flow}(#1)}
\newcommand{\inflowvarof}[1]{\inflow(#1)}
\newcommand{\markvarof}[1]{\mathit{mark}(#1)}
\newcommand{\nextvarof}[1]{\mathit{next}(#1)}
\newcommand{\lnext}{\mathit{ln}}
\newcommand{\lmark}{\mathit{lmark}}
\newcommand{\tnext}{\mathit{tn}}
\newcommand{\tmark}{\mathit{tmark}}
\newcommand{\hnext}{\mathit{hn}}
\newcommand{\hmark}{\mathit{hmark}}

\newcommand{\inv}{\mathsf{Inv}}
\newcommand{\tinv}{\mathsf{TInv}}
\newcommand{\travinv}{\mathsf{TravInv}}
\newcommand{\pastharris}{\mathsf{Past}}
\newcommand{\futharris}{\mathsf{Fut}}
\newcommand{\futharrispre}{\mathsf{P}}
\newcommand{\futharrisprep}{\hat{\mathsf{P}}}
\newcommand{\futharrispost}{\mathsf{Q}}
\newcommand{\futharrispostp}{\skew{-1.5}\hat{\mathsf{Q}}}
\newcommand{\old}[1]{{}^\backprime\mkern-1mu#1}

\newcommand{\nodeof}[1]{\mathsf{Node}(#1)}
\newcommand{\htinv}{\mathsf{HD}}
\newcommand{\allinv}{\inv}
\newcommand{\allinvof}[1]{\allinv(#1)}
\newcommand{\somenodes}{\mathit{N}}
\newcommand{\somenodesp}{\mathit{N}'}

\newcommand{\hist}{\mathsf{Hist}}

\newcommand{\atoolname}[1]{\code{#1}\xspace}
\newcommand{\cpp}{\atoolname{C++}}
\newcommand{\plankton}{\atoolname{plankton}}

\usepackage{pifont}
\newcommand{\rawsymbolYes}{{\color[RGB]{0,155,85}\smash{\ding{51}}}}%\ding{52}
\newcommand{\symbolYes}{\rawsymbolYes\xspace}

\newcommand{\anobl}[1]{\OBL{\ensuremath{#1}}}
\newcommand{\aful}[2]{\FUL{\ensuremath{#1}}{#2}}
\newcommand{\asspec}{\Psi}

\newcommand{\acss}{\mathsf{CSS}}
\newcommand{\acssup}{\mathsf{UP}}
\newcommand{\abscontent}{\mathcal{C}}
\newcommand{\abscontentp}{\mathcal{C}'}
\newcommand{\absop}{\mathit{op}}
\newcommand{\semcalclin}{\semCalc_\mathit{lin}}

\newcommand{\varr}{\mathit{arr}}
\newcommand{\vlen}{\mathit{S}}

\newcommand{\cc}{\cdot}
\newcommand{\ccc}{\cdots}

\makeatletter
\techreport{
    \let\@authorsaddresses\@empty
}
\makeatother

\title{A Concurrent Program Logic with a Future and History}
\subtitle{Extended Version}

\setcopyright{rightsretained}
\acmPrice{}
\acmDOI{10.1145/3563337}
\acmYear{2022}
\copyrightyear{2022}
\acmSubmissionID{oopslab22main-p465-p}
\acmJournal{PACMPL}
\acmVolume{6}
\acmNumber{OOPSLA2}
\acmArticle{174}
\acmMonth{10}

\author{Roland Meyer}
\orcid{0000-0001-8495-671X}
\affiliation{%
  \institution{TU Braunschweig}
  \country{Germany}
}
\email{roland.meyer@tu-bs.de}

\author{Thomas Wies}
\orcid{0000-0003-4051-5968}
\affiliation{%
  \institution{New York University}
  \country{USA}
}
\email{wies@cs.nyu.edu}

\author{Sebastian Wolff}
\orcid{0000-0002-3974-7713}
\affiliation{%
  \institution{New York University}
  \country{USA}
}
\email{sebastian.wolff@cs.nyu.edu}

\ccsdesc[500]{Theory of computation~Separation logic}
\ccsdesc[500]{Theory of computation~Hoare logic}
\ccsdesc[500]{Theory of computation~Automated reasoning}
\ccsdesc[300]{Theory of computation~Program verification}
\ccsdesc[100]{Theory of computation~Programming logic}

\nontechreport{
    \keywords{Linearizability, Non-blocking Data Structures, Harris Set} % mandatory in camera-ready submission
}

\begin{document}
	%!TEX root = ../main.tex

\begin{abstract}
Verifying fine-grained optimistic concurrent programs remains an open problem. Modern program logics provide abstraction mechanisms and compositional reasoning principles to deal with the inherent complexity. However, their use is mostly confined to pencil-and-paper or mechanized proofs. We devise a new separation logic geared towards the lacking automation. While local reasoning is known to be crucial for automation, we are the first to show how to retain this locality for (i) reasoning about inductive properties without the need for ghost code, and (ii) reasoning about computation histories in hindsight. We implemented our new logic in a tool and used it to automatically verify challenging concurrent search structures that require inductive properties and hindsight reasoning, such as the Harris set.
\end{abstract}

	\maketitle
	
	%!TEX root = ../main.tex

%!TEX root = ../main.tex

\section{Introduction}
\label{sec:motivation}

Concurrency comes at a cost, at least, in terms of increased effort when verifying program correctness. 
There has been a proliferation of concurrent program logics that provide an  arsenal of reasoning techniques to address this challenge~\cite{DBLP:conf/concur/FuLFSZ10,DBLP:conf/esop/GotsmanRY13,DBLP:conf/esop/SergeyNB15,DBLP:conf/ecoop/DelbiancoSNB17,DBLP:conf/sas/BellAW10,DBLP:conf/popl/ParkinsonBO07,DBLP:conf/wdag/HemedRV15,DBLP:conf/pldi/LiangF13,DBLP:books/daglib/0080029,DBLP:conf/concur/VafeiadisP07,DBLP:journals/jfp/JungKJBBD18,DBLP:conf/pldi/GuSKWKS0CR18,DBLP:conf/tacas/ElmasQSST10}.
In addition, a number of general approaches have been developed to help structure the high-level proof argument~\cite{DBLP:conf/podc/OHearnRVYY10,DBLP:conf/wdag/FeldmanE0RS18,DBLP:journals/pacmpl/FeldmanKE0NRS20,DBLP:journals/tods/ShashaG88,DBLP:conf/cav/KraglQH20}.
However, the use of these techniques has been mostly confined to manual proofs done on paper, or mechanized proofs constructed in interactive proof assistants.
We distill from these works a concurrent separation logic suitable for automating the construction of local correctness proofs for highly concurrent data structures.
We focus on concurrent search structures (sets and maps indexed by keys), but the developed techniques apply more broadly.
Our guiding principle is to perform all inductive reasoning, both in time and space, in lock-step with the program execution. 
The reasoning about inductive properties of graph structures and computation histories is relegated to the meta-theory of the logic by choosing appropriate semantic models.

%-------------------------------------------------------------------------------------------------%
%-------------------------------------------------------------------------------------------------%
%-------------------------------------------------------------------------------------------------%
\smartparagraph{Running Example}
We motivate our work using the Harris non-blocking set data structure~\cite{DBLP:conf/wdag/Harris01}, which we will also use as a running example throughout the paper. 

We assume a garbage-collected programming language, supporting (first-order) recursive functions, product types, and mutable heap-allocated structs. 
The language further provides a \emph{compare-and-set} operation, \code{CAS($x$.f,$o$,$n$)}, that atomically sets field \code{f} of $x$ to $n$ and returns \code{true} if \code{f}'s current value is $o$, or otherwise  returns \code{false} leaving \code{$x$.f} unchanged.

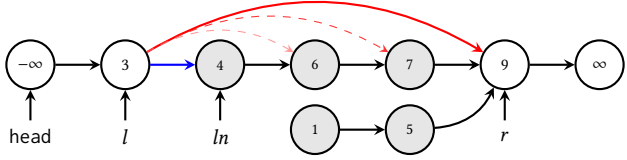
\begin{figure}[t]
\begin{minipage}[t]{.49\linewidth}
\begin{lstlisting}[language=SPL,mathescape=true,escapechar=@,belowskip=3.5pt]
struct N = { val key: K; var next: N }
\end{lstlisting}
\begin{lstlisting}[language=SPL,mathescape=true,escapechar=@,belowskip=3.5pt]
val tail = new N { key = $\infty$; next = tail }
val head = new N { key = $-\infty$; next = tail }
\end{lstlisting}
\begin{lstlisting}[language=SPL,mathescape=true,escapechar=@,belowskip=0pt]
procedure traverse($k$:$\,$K,$\,l$:$\,$N,$\,\lnext$:$\,$N,$\,t$:$\,$N):$\,$(N$\,$*$\,$N$\,$*$\,$N)$\,${
  val $\tnext$ = $t$.next  @\label[line]{line:traverse-read-next}@  
  val $\tmark$ = is_marked($\tnext$) @\label[line]{line:traverse-read-mark}@
  if ($\tmark$) return traverse($k$, $l$, $\lnext$, $\tnext$) @\label[line]{line:traverse-check-mark}@
  else if ($t$.key < $k$) return traverse($k$,$\,t$,$\,\tnext$,$\,\tnext$)  @\label[line]{line:traverse-check-key}@
  else return ($l$, $\lnext$, $t$)
}
\end{lstlisting}
\end{minipage}\hfill
\begin{minipage}[t]{.45\linewidth}
\begin{lstlisting}[language=SPL,mathescape=true,escapechar=@,belowskip=0pt]
procedure find($k$: K): N * N {
  val $\hnext$ = head.next @\label[line]{line:traverse-read-head-next}@
  val $l$, $\lnext$, $r$ := traverse($k$,$\,$head,$\,\hnext$,$\,\hnext$)
  if (($\lnext$ == $r$ || CAS($l$.next, $\lnext$, $r$)) @\label[line]{line:search-cas}@
      && !is_marked($r$.next)) return ($l$,$\,r$) @\label[line]{line:search-mark-check}@
  else return find($k$)
}
\end{lstlisting}
\begin{lstlisting}[language=SPL,mathescape=true,escapechar=@,belowskip=0pt]
procedure search($k$: K) : Bool {
  val _, $r$ = find($k$) @\label[line]{line:find-search}@
  return $r$.key == $k$ @\label[line]{line:find-return}@
}
\end{lstlisting}
\end{minipage}
\centering
  \begin{tikzpicture}[>=stealth, font=\footnotesize, scale=0.8, every node/.style={scale=0.8}]
    % Nodes
    \def\xsep{6}%6
    \def\ysep{-1.6}

    \node[unode] (inf) {$-\infty$};
    \node[unode, right=\xsep mm of inf] (n3) {$3$};
    \node[mnode, right=\xsep mm of n3] (n4) {$4$};
    \node[mnode, right=\xsep mm of n4] (n6) {$6$};
    \node[mnode, right=\xsep mm of n6] (n7) {$7$};
    \node[mnode] (n11) at ($(n6) + (0, \ysep/1.5)$) {$1$};
    \node[mnode] (n5) at ($(n7) + (0, \ysep/1.5)$) {$5$};
    \node[unode, right=\xsep mm of n7] (n9) {$9$};
    \node[unode, right=\xsep mm of n9] (Inf) {$\infty$};

    \node[stackVar, below=.4cm of inf] (hd) {$\mathtt{head}$};
    \node[stackVar, below=.4cm of n3] (l) {$l$};
    \node[stackVar, below=.4cm of n4] (ln) {$\lnext$};
    \node[stackVar, below=.4cm of n9] (r) {$r$};

    % Edges
    \draw[edge] (hd) to (inf);
    \draw[edge] (inf) to (n3);
    \draw[edge,blue] (n3) to (n4);
    \draw[edge] (n4) to (n6);
    \draw[edge] (n6) to (n7);
    \draw[edge] (n7) to (n9);
    \draw[edge] (n9) to (Inf);
    \draw[edge] (n11) to (n5);
    \draw[edge, bend right=30] (n5) to (n9);

    \draw[edge, red] (n3) to[bend left=30] (n9);
    \draw[fedge,red!40] (n3) to[bend left=30] (n6);
    \draw[fedge,red!80] (n3) to[bend left=30] (n7);
    \draw[edge] (l) to (n3);
    \draw[edge] (ln) to (n4);
    \draw[edge] (r) to (n9);
  \end{tikzpicture}
  \caption{%
    The \citet{DBLP:conf/wdag/Harris01}\label{fig:harris-list-algo} set algorithm. The lower half shows a state of the Harris set containing keys $\set{3,9}$. 
    Nodes are labeled with the value of their \code{key} field. Edges indicate \code{next} pointers. Marked nodes are shaded gray. 
    The blue edge between the nodes $l$ and $\lnext$ represents the state of $l$ before the \code{CAS} and the red edge between $l$ and $r$ the state after the \code{CAS} from \cref{line:search-cas}.
    Dashed edges represent the hypothetical update chunks that inductively capture the effect of the \code{CAS}.%
    \label{fig:harris-list-state}%
  }
\end{figure}

Harris' algorithm implements a set data structure that takes elements from a totally ordered type \code{K} of keys and provides operations for concurrently searching, inserting, and deleting a given \emph{operation key}~\code{k}. 
We focus on the \code{search} operation shown in \Cref{fig:harris-list-algo}. 
The data structure is represented as a linked list consisting of nodes implemented by the struct type~\code{N}. 
Each node stores a \code{key} and a \code{next} pointer to the successor node in the list. 
A potential state of the data structure is illustrated in \Cref{fig:harris-list-state}. 
The algorithm maintains several important invariants. 
First, the list is strictly sorted by the keys in increasing order and has a sentinel head node, pointed to by the immutable shared pointer \code{head}. 
The key of the head node is $-\infty$. 
Likewise, there is a sentinel tail node with key~$\infty$.
We assume $-\infty < k < \infty$ for all operation keys~$k$. 
To support concurrent insertions and deletions without lock-based synchronization, a node that is to be removed from the list is first \emph{marked} to indicate that it has been logically deleted before it is physically unlinked. 
Node marking is implemented by \emph{bit-stealing} on the \code{next} pointers. 
We abstract from the involved low-level bit-masking using the function \code{is_marked}. 
A call \code{is_marked($p$)} returns \code{true} iff the mark bit of pointer \code{$p$} is set. 
We say that a node $x$ is marked if \code{is_marked($x$.next)} returns \code{true}. 
The sentinel nodes are never marked.

The workhorse of the algorithm is the function \code{find}. 
It takes an operation key $k$ and returns a pair of nodes \code{($l$, $r$)} such that $l.\keyfld < k \leq r.\keyfld$, and the following held true at a single point in time during \code{find}'s execution:
$r$ was unmarked and $l$'s direct successor, and $l$ was reachable from \code{head}. All client-facing functions such as \code{search} then use \code{find}.

%-------------------------------------------------------------------------------------------------%
%-------------------------------------------------------------------------------------------------%
%-------------------------------------------------------------------------------------------------%
\smartparagraph{Contributions}
Recall that in order to prove linearizability of a concurrent data structure, one has to show that each of the data structure's operations takes effect instantaneously at some time point between its invocation and return, the \emph{linearization point}, and behaves according to its sequential specification~\cite{DBLP:journals/toplas/HerlihyW90}.
The Harris set exhibits two key challenges in automating the linearizability proof that are common to non-blocking data structures and that we aim to address in our work.

The first challenge is that linearization points may not be statically fixed, but instead depend on the interference of concurrent operations performed by other threads.
We discuss this issue using \code{search}, whose sequential specification says that the return value is \code{true} iff $k$ is present in the data structure.
Consider a thread $T$ executing a call \code{search(8)}.
\Cref{fig:harris-list-state} shows a possible intermediate state of the list observed right after the successful execution of the \code{CAS} on \cref{line:search-cas} during $T$'s execution of \code{search}.
This is also the linearization point of \code{search} for $T$, because the \code{CAS} guarantees that if value $8$ is present in the data structure, it must be the key of node $r$. Hence, if the check on $r$'s mark bit on \cref{line:search-mark-check} still succeeds, then the test on \cref{line:find-return} will produce the correct return value for the \code{search} relative to the abstract state of the data structure at the linearization point.
However, in an alternative execution, a concurrently executing \code{delete(9)} may mark $r$ before $T$ executes \cref{line:search-mark-check}, but after $T$ executes the \code{CAS}.
In this situation, $T$ will abort and restart causing its linearization point to be delayed.
Thus, the correct linearization point is only known in hindsight~\cite{DBLP:conf/podc/OHearnRVYY10} at a later point in the execution when all interferences have been observed.
As a consequence, the proof must track information about earlier states in the execution history to enable reasoning about linearization points that already happened in the past.

Our first contribution is a lightweight embedding of computation histories into separation logic~\cite{DBLP:conf/csl/OHearnRY01} that supports local proofs using hindsight arguments, but without having to perform explicit induction over computation histories.

The second challenge is that the proof needs to reason about maintenance operations that affect an unbounded heap region.
The procedure \code{traverse} guarantees that all nodes between $l$ and $r$ are marked.
The \code{CAS} on \cref{line:search-cas} then unlinks the segment of marked nodes between $l$ and $r$ from the structure making $r$ the direct successor of $l$.
This step is depicted in \Cref{fig:harris-list-state}.
The blue edge between $l$ and $\lnext$ refers to the pre state of the \code{CAS} and the red edge between $l$ and $r$ to the post state.
A traditional automated analysis needs to infer the precise inductive shape invariant of the traversed region (e.g. a recursive predicate stating that it is a list segment of marked nodes).
Then, at the point where the segment is unlinked, it has to infer that the global data structure invariant is maintained.
This involves an inductive proof argument, and the analysis needs to rediscover how this induction relates to the invariant of the traversal.

Our second contribution is a mechanism for reasoning about such updates with non-local effects.
The idea is to compose these updates out of \emph{ghost update chunks}.
This is illustrated in \Cref{fig:harris-list-state}, where the effect of the \code{CAS} is composed out of simpler updates that move the edge from $l$ towards $r$ one node at a time (indicated by the dashed red edges).
One only needs to reason about four nodes to prove that the edge can be moved forward by one node.
We refer to a correctness statement of such a ghost update chunk as a \emph{future}.
The crux is to construct these futures during the traversal of the marked segment, i.e., in lock-step with the program execution.
This avoids the need for explicit inductive reasoning at the point where the \code{CAS} takes effect.
When proving the future for unlinking a single traversed node $t$, we directly apply interference-free facts learned during the traversal, e.g., that $t$ must be marked and can therefore be unlinked.
We call this mechanism accounting.
The final future can then be invoked on \cref{line:search-cas} to prove the correctness of the \code{CAS}.

The focus of the paper is on the development of the new program logic rather than algorithmic details on efficient automatic proof search.
However, we have implemented a prototype tool called \plankton \cite{plankton} that uses the logic to verify linearizability of highly concurrent search structures automatically, provided an appropriate structural invariant is given by the user.
The tool's implementation follows a standard abstract interpretation approach~\cite{DBLP:conf/popl/CousotC77}.
We have successfully applied the tool to verify several fine-grained non-blocking and lock-based concurrent set implementations, including: Harris set~\cite{DBLP:conf/wdag/Harris01}, Michael set~\mbox{\cite{DBLP:conf/spaa/Michael02}}, Vechev and Yahav CAS set~\cite[Figure~2]{DBLP:conf/pldi/VechevY08}, ORVYY set~\cite{DBLP:conf/podc/OHearnRVYY10}, and the Lazy set~\cite{DBLP:conf/opodis/HellerHLMSS05}.
All these benchmarks require hindsight reasoning.
To our knowledge, \plankton is the first tool that automates hindsight reasoning for such a variety of benchmarks.
The Harris set additionally requires futures that are also automated in plankton.
With this, \plankton is the first tool that can automatically verify the Harris set algorithm.

\moreless{%
  % This technical report originally appeared as \cite{oopsla}.
}{
  A companion technical report containing additional details is available as \cite{DBLP:journals/corr/abs-2207-02355}.
}

%!TEX root = ../main.tex

\section{Overview}
\label{sec:overview}

We aim for a proof strategy that is compatible with local reasoning principles and agnostic to the detailed invariants of the specific data structure under consideration.
In particular, we want to avoid proof arguments that devolve into explicit reasoning about heap reachability or other inductive heap properties, which tend to be difficult to automate.

Our strategy builds on the \emph{keyset framework}~\cite{DBLP:journals/tods/ShashaG88,DBLP:conf/pldi/KrishnaPSW20} for designing and verifying concurrent search structures. In this framework, the data structure's heap graph is abstracted by a mathematical graph $(N,E)$ where each node $\anode \in N$ is labeled by its local contents, a set of keys $\contentsof{\anode}$.
The abstract state of the data structure $\contentsof{N}$ is then given by the union of all node-local contents.
Moreover, each node $\anode$ has an associated set of keys $\keysetof{\anode}$ called its \emph{keyset}.
The keysets are defined inductively over the graph such that the following \emph{keyset invariants} are maintained:
\begin{inparaenum}[(1)]
  \item the keysets of all nodes partition the set of all keys, and
  \item for all nodes $\anode$, $\contentsof{\anode} \subseteq \keysetof{\anode}$.
\end{inparaenum}
For the Harris set, we define $\contentsof{\anode}\defeq \ite{{\markvarof{\anode}}}{{\emptyset}\!}{\!{\set{\keyvarof{\anode}}}}$ and let $\keysetof{\anode}$ be the empty set if $\anode$ is not reachable from $\head$ and otherwise the interval $(\keyvarof{\anodep},\keyvarof{\anode}]$ where $\anodep$ is the predecessor of $\anode$ in the list (cf. \Cref{fig:harris-list-hindsight}).
Here, we denote by $\keyvarof{\anode}$ the value of field $\keysel$ in a given state, and similarly for $\markvarof{\anode}$.
Throughout the rest of the paper, we will follow this convention of naming variables in a way that reflects the field dereferencing mechanism.

The keyset invariants ensure that for any node $\anode \in N$ and key $k$ we have
\[
  k \in \keysetof{\anode} \;\Rightarrow\; \bigl(k \in \contentsof{N} \Leftrightarrow k \in \contentsof{\anode}\bigr)
  \enspace.
\]
This property allows us to reduce the correctness of an insertion, deletion, and search for $k$ from the global abstract state $\contentsof{N}$ to $\anode$'s local contents $\contentsof{\anode}$, provided we can show $k \in \keysetof{\anode}$.
For example, suppose that a concurrent invocation of \code{search($k$)} returns \code{false}.
To prove that this invocation is linearizable, it suffices to show that there exists a node $\anode$ such that both $k \in \keysetof{\anode}$ and $k \notin \contentsof{\anode}$ were true at the same point in time during \code{search($k$)}'s execution.
We refer to $\anode$ as the decisive node of the operation.
The point in time where the two facts about $\anode$ hold is the linearization point. 

The ingredients for the linearizability proofs are thus
\begin{inparaenum}[(i)]% order matters
  \item defining the keysets for the data structure at hand,
  \item proving that the keyset invariants are maintained by the data structure's operations, and
  \item identifying the linearization point and decisive node $\anode$ for an operation on key $k$ by establishing the relevant facts about $k$'s membership in the keyset and contents of $\anode$. 
\end{inparaenum}

Our contributions focus on the automation of (ii) and (iii).
While we do not automate (i), the keyset definitions follow general principles and can be reused across many data structures~\cite{DBLP:journals/tods/ShashaG88} (e.g., we use the same definition for all the list-based set implementations considered in our evaluation, cf.~\Cref{Section:Evaluation}).
In the remainder of this section, we provide a high-level overview of the reasoning principles that underlie our new program logic and enable proof automation.

%-------------------------------------------------------------------------------------------------%
%-------------------------------------------------------------------------------------------------%
%-------------------------------------------------------------------------------------------------%
\smartparagraph{Automating History Reasoning}
We start with the linearizability argument.
Consider a concurrent execution of \code{search($k$)} that returns value $b$.
The decisive node of \code{search} is always the node $r$ returned by the call to \code{find}.
The proof thus needs to establish that at some point during the execution, $k \in \keysetof{r}$ and $b \Leftrightarrow (k \in \contentsof{r})$ were true.
The issue is that when we reach the corresponding point in the execution during proof construction, we may not be able to linearize the operation right away because the decisive node and linearization point depend on the modifications that may still be done by other threads before \code{search} returns.
We can thus only linearize the execution in hindsight, once all the relevant interferences have been observed.

\begin{figure}[t]
  \centering
  \begin{tikzpicture}[>=stealth, font=\footnotesize, scale=0.8, every node/.style={scale=0.8}]
    % Nodes
    \def\xsep{1.4}
    \def\ysep{-1.8}

    \begin{scope}[local bounding box=g1] % lower left
    \node[unode,label={$\{\bluetext{-\infty}\}$}] (n1) {$n_1$};
    \node[unode, right=6mm of n1,label={$(-\infty,\bluetext{2}]$}] (n2) {$n_2$};
    \node[unode, right=6mm of n2,label={$(2,\bluetext{5}]$}] (n3) {$n_3$};
    \node[unode, right=6mm of n3,label={$(5,\bluetext{\infty}]$}] (n4) {$n_4$};

    \node[stackVar, below=.3cm of n1,label={[label distance=-15pt]\normalsize$\mathtt{head},l$}] (hd) {};
    \node[stackVar, below=.3cm of n2] (l) {$\lnext,t$};
    \node[stackVar, below=.3cm of n3] (ln) {$\tnext$};

    % Edges
    \draw[edge] (n1) to (n2);
    \draw[edge] (n2) to (n3);
    \draw[edge] (n3) to (n4);
 
    \draw[edge] (hd) to (n1);
    \draw[edge] (l) to (n2);
    \draw[edge] (ln) to (n3);

    \draw ($(n1.north west)+(-.5,.7)$) rectangle ($(n4.south east)+(.5,-1.2)$);
    \end{scope}
    \node[stackVar, below=-.3cm of g1.south] (s1) {\large$s_1$ (\cref{line:traverse-read-next})};
    \begin{scope}[local bounding box=g2, shift={($(g1.north west)+(4,2.2)$)}] % upper left

    \node[unode,label={$\{\bluetext{-\infty}\}$}] (n1) {$n_1$};
    \node[unode, right=6mm of n1,label={$(-\infty,\bluetext{2}]$}] (n2) {$n_2$};
    \node[mnode, right=6mm of n2,label={$(2,\bluetext{5}]$}] (n3) {$n_3$};
    \node[unode, right=6mm of n3,label={$(5,\bluetext{\infty}]$}] (n4) {$n_4$};

    \node[stackVar, below=.3cm of n1] (hd) {$\mathtt{head}$};
    \node[stackVar, below=.3cm of n2] (l) {$l$};
    \node[stackVar, below=.3cm of n3] (ln) {$\lnext,r$};

    % Edges
    \draw[edge] (n1) to (n2);
    \draw[edge] (n2) to (n3);
    \draw[edge] (n3) to (n4);
 
    \draw[edge] (hd) to (n1);
    \draw[edge] (l) to (n2);
    \draw[edge] (ln) to (n3);

    \draw ($(n1.north west)+(-.5,.7)$) rectangle ($(n4.south east)+(.5,-1.2)$);
    \end{scope}
    \node[stackVar, below=-.35cm of g2.south] (s2p) {\large$s_2'$ (\cref{line:search-mark-check})};
    
    \draw[edge,dotted] ($(g1.north)+(-.5,0)$) to[bend left=20] node[left] {\large\code{delete}($5$) $\rightsquigarrow$} ($(g2.west)+(0,0)$);

    \begin{scope}[local bounding box=g3, shift={($(g2.north east)+(2.3,-1.0)$)}] % upper right
    \node[unode,label={$\{\bluetext{-\infty}\}$}] (n1) {$n_1$};
    \node[unode, right=6mm of n1,label={$(-\infty,\bluetext{2}]$}] (n2) {$n_2$};
    \node[mnode, right=6mm of n2,label={$\emptyset$}] (n3) {$n_3$};
    \node[unode, right=6mm of n3,label={$(2,\bluetext{\infty}]$}] (n4) {$n_4$};

    \node[stackVar, below=.3cm of n1] (hd) {$\mathtt{head}$};
    \node[stackVar, below=.3cm of n2] (l) {$l$};
    \node[stackVar, below=.3cm of n3] (ln) {$\lnext$};
    \node[stackVar, below=.3cm of n4] (r) {$r$};

    % Edges
    \draw[edge] (n1) to (n2);
    \draw[edge] (n2) to[bend left=45] (n4);
    \draw[edge] (n3) to (n4);
 
    \draw[edge] (hd) to (n1);
    \draw[edge] (l) to (n2);
    \draw[edge] (ln) to (n3);
    \draw[edge] (r) to (n4);

    \draw ($(n1.north west)+(-.5,.7)$) rectangle ($(n4.south east)+(.5,-1.2)$);
    \end{scope}
    \node[stackVar, below=-.35cm of g3.south] (s3p) {\large$s_3'$ (\cref{line:search-mark-check})};

    \draw[edge,dotted] ($(g2.east)+(0,0)$) to ($(g3.west)+(0,0)$);

    \begin{scope}[local bounding box=g4, shift={($(g3.south west)+(.8,-1.7)$)}] % lower right
    \node[unode,label={$\{\bluetext{-\infty}\}$}] (n1) {$n_1$};
    \node[unode, right=6mm of n1,label={$(-\infty,\bluetext{2}]$}] (n2) {$n_2$};
    \node[unode, right=6mm of n2,label={$(2,\bluetext{5}]$}] (n3) {$n_3$};
    \node[unode, right=6mm of n3,label={$(5,\bluetext{\infty}]$}] (n4) {$n_4$};

    \node[stackVar, below=.3cm of n1] (hd) {$\mathtt{head}$};
    \node[stackVar, below=.3cm of n2] (l) {$l$};
    \node[stackVar, below=.3cm of n3] (ln) {$\lnext,r$};

    % Edges
    \draw[edge] (n1) to (n2);
    \draw[edge] (n2) to (n3);
    \draw[edge] (n3) to (n4);
 
    \draw[edge] (hd) to (n1);
    \draw[edge] (l) to (n2);
    \draw[edge] (ln) to (n3);

    \draw ($(n1.north west)+(-.5,.7)$) rectangle ($(n4.south east)+(.5,-1.2)$);
    \end{scope}
    \node[stackVar, below=-.35cm of g4.south] (s2) {\large$s_2$ (\cref{line:search-mark-check})};

    \draw[edge,dotted] ($(g1.east)+(0,0)$) to ($(g4.west)+(0,0)$);
  \end{tikzpicture}
  \vspace*{-1cm}
  \caption{%
    Some states observed during two executions of \code{search($5$)} on a Harris set that initially contains the keys $2$ and $5$ and no marked nodes.
    Each node in a state is labeled above by its keyset.
    If the node is reachable from $\head$, the right bound of the keyset interval indicates the key stored in the node (highlighted in blue).
    \label{fig:harris-list-hindsight}
  }
\end{figure}
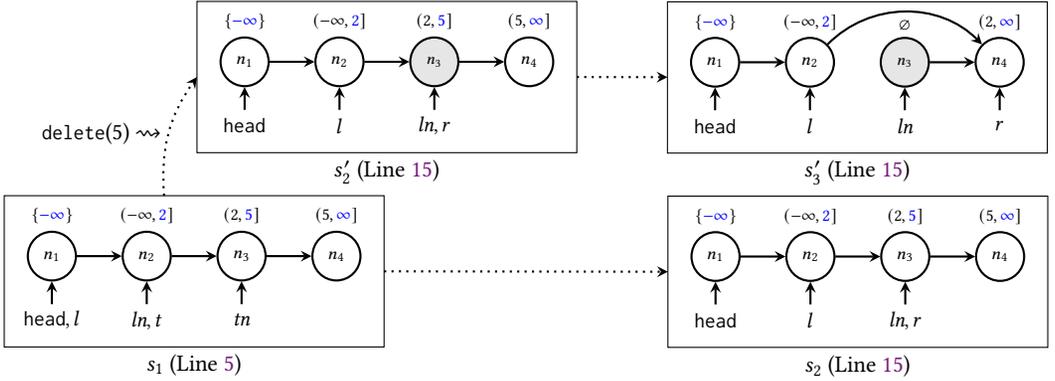

To illustrate this point, consider the scenario depicted in \Cref{fig:harris-list-hindsight}.
It shows intermediate states of two executions of \code{search($5$)} on a Harris set that initially contains the keys $\{2,5\}$.
The two executions agree up to the point when state $s_1$ is reached after execution of \cref{line:traverse-read-next} in the first call to \code{traverse} (i.e., after $n_2$'s $\nextsel$ pointer and mark bit have been read).
The execution depicted on the bottom proceeds without interference to state $s_2$ at the beginning of \cref{line:search-mark-check} and will return $b=\true$.
Here, the decisive node is $n_3$ and the linearization point is $s_1$. Note that both $5 \in \keysetof{n_3}$ and $5 \in \contentsof{n_3}$ hold in $s_1$.
On the other hand, the execution depicted on the top of \Cref{fig:harris-list-hindsight} is interleaved with a concurrently executing \code{delete}($5$).
The \code{delete} thread marks $n_3$ before the \code{search} thread reaches the beginning of \cref{line:search-mark-check}, yielding state $s_2'$ at this point.
The test whether $r$ is marked will now fail, causing the \code{search} thread to restart.
After traversing the list again, the \code{search} thread will unlink $n_3$ from the list with the \code{CAS} on \cref{line:search-cas}.
This yields state $s_3'$ when the \code{search} reaches \cref{line:search-mark-check} again.
The \code{search} thread will then proceed to compute the return value $b=\false$.
For this execution, the decisive node is $n_4$ and the linearization point is $s_3'$.

Our program logic provides two ingredients for dealing with the resulting complexity in the linearizability proof. We discuss these formally in \Cref{Section:History}. The first ingredient is the \emph{past predicate} $\past \apred$, which asserts that the current thread owned the resource $\apred$ at some point in the past. The second ingredient is a set of proof rules for introducing and manipulating past predicates. In particular, we will use the following three rules in our proof:
\begin{mathpar}
  \inferHtop{h-intro}{}{\vdash \hoareOf{p}{\cskip}{p \mstar \past p}}
  \and
  \inferH{h-hindsight}{p \text{ pure}}{
    p \mstar \past q \vdash \past (p \mstar q)
  }
  \and
  \inferH{h-infer}{p \vdash q}{
    \past p \vdash \past q
  }
  \enspace .
\end{mathpar}
Rule \ruleref{h-intro} states the validity of the Hoare triple $\hoareOf{p}{\cskip}{p \mstar \past p}$ which introduces a past predicate $\past p$ using a \emph{stuttering step}.
Here, $\mstar$ is separating conjunction.
The rule expresses that if the thread owns $p$ now, it trivially owned $p$ at some past point up until now. 
Rule \ruleref{h-hindsight} captures the essence of hindsight reasoning:
ownership of $p$ can be transferred from the now into the past, if $p$ is a purely logical fact that is independent of the program state.
Rule \ruleref{h-infer} states the monotonicity of the past operator with respect to logical weakening.

We demonstrate the use of past predicates and their associated rules by sketching the linearizability proof for executions of \code{search($k$)} that follow the same code path as the one in the bottom half of \Cref{fig:harris-list-hindsight}.
The proof relies on a predicate $\nodeof{\anode}$, which expresses the existence of the physical representation of $\anode$ in the heap and binds the \emph{logical variables} $\keyvarof{\anode}$ and $\markvarof{\anode}$ to the values stored in the relevant fields of $\anode$---it is a \emph{points-to} predicate ${\anode \mapsto \keyvarof{\anode},\markvarof{\anode},\dots}$ for all fields of $\anode$.
The predicate also expresses important properties needed for maintaining the keyset invariants.
\Cref{Section:Keysets} discusses the definition of predicate $\nodeof{\anode}$ in detail.

The proof proceeds by symbolic execution of the considered path. The goal is to infer
\[
  \past \bigl(
    \old{\nodeof{r}} \;\mstar\; k \in \old{\keysetof{r}} \land (b \Leftrightarrow k \in \old{\contentsof{r}}
  \bigr)
\]
as postcondition where $b$ is the return value of \code{search}($k$).
This implies the existence of a linearization point as discussed earlier.
Here, we write $\old{e}$ for the expression obtained from expression $e$ by replacing all logical variables like $\markvarof{r}$ by fresh variables $\old{\markvarof{r}}$.
That is, $\old{e}$ can be thought of as the expression $e$ evaluated with respect to the old state of $r$ captured by $\old{\nodeof{r}}$ inside the past predicate, rather than the current state.\footnote{%
  Observe that $\nodeof{r}\mstar\past\nodeof{r}$ would implicitly state that the fields of $r$ are the same in the past and the present state.
  Renaming the past state, $\nodeof{r}\mstar\past\old{\nodeof{r}}$, allows the fields of $r$ to have changed between the past and present.
}

The symbolic execution starts from a global invariant that maintains $\nodeof{\anode}$ for all nodes $\anode$ in the heap.
When it reaches \cref{line:traverse-read-next} in the proof, we can establish $\nodeof{t}$ and $\nodeof{\tnext}$ using the derived invariant of the traversal.
The two predicates imply that if $t$ is unmarked, then its keyset must be non-empty.
In turn, this implies $\keysetof{\tnext} = (\keyvarof{t},\keyvarof{\tnext}]$, because $\tnext$ is the successor of $t$.
Together, we deduce (a) $\neg \markvarof{t} \land \keyvarof{t} < k \land k \leq \keyvarof{\tnext} \Rightarrow k {\:\in\:} \keysetof{\tnext}$.
Moreover, the definition of $\contentsof{\tnext}$ gives us (b) $k {\:\in\:} \contentsof{\tnext} \Leftrightarrow \neg \markvarof{\tnext} \land \keyvarof{\tnext} = k$.
We let $H(\tnext,t)$ be the conjunction of (a), (b), and $\nodeof{\tnext}$.

Next, we use rule~\ruleref{h-intro} to transfer $H(\tnext,t)$ into a past predicate, yielding $H(\tnext,t) \mstar \past(H(\tnext,t))$.
For our proof to be valid, we need to make sure that all intermediate assertions are stable under interferences by other threads.
Unfortunately, this is not the case for the assertion $H(\tnext,t) \mstar \past(H(\tnext,t))$.
Notably, this assertion implies that the current value of $\tnext$'s mark bit is the same as the value of its mark bit in the past state referred to by the past predicate.

To make the assertion stable under interference, we first introduce fresh logical variables $\old{\keyvarof{\tnext}}$ and $\old{\markvarof{t}}$ which we substitute for $\keyvarof{\tnext}$ and $\markvarof{\tnext}$ under the past operator.
This yields the equivalent intermediate assertion:
\[
  H(\tnext,t) \;\mstar\; \past(\old{H(\tnext,t)}) \;\mstar\; (\old{\markvarof{\tnext}} = \markvarof{\tnext} \land \old{\keyvarof}{\tnext} = \keyvarof{\tnext})
  \enspace.
\]
Next, we observe that other threads executing \code{search}, \code{insert}, and \code{delete} operations can only interfere by marking node $\tnext$ in case it is not yet marked.
Such interference invalidates the equality $\markvarof{\tnext} = \old{\markvarof{\tnext}}$.
So we weaken it to $\old{\markvarof{\tnext}} \Rightarrow \markvarof{\tnext}$.
In \Cref{Section:OG} we introduce a general Owicki-Gries-style separation logic framework that formalizes this form of interference reasoning.
In addition, we only keep $\nodeof{\tnext}$ from $H(\tnext,t)$, leaving us with the interference-free assertion
\[
  P(\tnext,t) \defeq \nodeof{\tnext} \;\mstar\; \past(\old{H(\tnext,t)}) \;\mstar\; (\old{\markvarof{\tnext}} \Rightarrow \markvarof{\tnext}) \land \old{\keyvarof}{\tnext} = \keyvarof{\tnext}
  \enspace.
\]
We then propagate this assertion forward along the considered execution path of \code{search}($k$), obtaining $P(\lnext,t)$ when \cref{line:search-mark-check} is reached in the proof.
During the propagation, we accumulate the facts $\neg \markvarof{t}$, $\keyvarof{t} {\,<\,} k$, $k {\:\leq\:} \keyvarof{r}$, and $\lnext=r$ according to the branches of the conditional expressions taken along the path.
As these facts are all pure, we use rule \ruleref{h-hindsight} to transfer them, together with the equality $\old{\keyvarof{\lnext}} = \keyvarof{\lnext}$, into the past predicate $\past(\old{H(\lnext,t)})$.

Using rule \ruleref{h-infer} we can then simplify the resulting past predicate as follows:
\[
  \past \bigl(
    \old{\nodeof{r}} \; \mstar \; k \in \old{\keysetof{r}} \;\land\; (k \in \old{\contentsof{r}} \Leftrightarrow \neg \old{\markvarof{r}} \land \old{\keyvarof{r}} = k)
  \bigr)
  \enspace.
\]
As we propagate the overall assertion further to the return point of \code{search}($k$), we accumulate the additional pure facts $\neg \markvarof{r}$ and $b \Leftrightarrow \keyvarof{r} = k$.
We again invoke rule \ruleref{h-hindsight} to transfer these into the past predicate, together with $\old{\markvarof{r}} \Rightarrow \markvarof{r}$ and $\old{\keyvarof{r}} = \keyvarof{r}$.
The resulting past predicate can then be simplified with \ruleref{h-infer} to finally obtain the desired:
\[
  \past \bigl(
    \old{\nodeof{r}} \;\mstar\; k \in \old{\keysetof{r}} \;\land\; (b \Leftrightarrow k \in \old{\contentsof{r}}
  \bigr)
  \enspace.
\]

While this proof is non-trivial, it is easy to automate.
The analysis performs symbolic execution of the code.
After each atomic step, it applies the rules \ruleref{h-intro} and \ruleref{h-hindsight} eagerly.
The resulting past predicates are then simplified and weakened with respect to interferences by other threads.
This way, the analysis maintains the strongest interference-free information about the history of the computation.
These steps are integrated into a classical fixpoint computation to infer loop invariants, applying standard widening techniques to enforce convergence~\mbox{\cite{DBLP:conf/popl/CousotC79}}.

It is worth pointing out that the above reasoning could also be done with prophecies~\cite{DBLP:conf/pldi/LiangF13,DBLP:journals/pacmpl/JungLPRTDJ20}.
However, prophecies are not amenable to automation in the same way as past predicates.
The main reason for this is the hindsight rule: it works relative to facts that have already been discovered during symbolic execution.
Prophecies would require to \emph{guess} the same facts prior to being discovered.
This guessing step is notoriously hard to automate~\mbox{\cite{DBLP:conf/cav/BouajjaniEEM17}}.

%-------------------------------------------------------------------------------------------------%
%-------------------------------------------------------------------------------------------------%
%-------------------------------------------------------------------------------------------------%
\smartparagraph{Automating Future Reasoning}
Our second major contribution is the idea of futures and their governing reasoning principles.
We motivate futures with the problem of automatically proving that the Harris set maintains the keyset invariants.
As noted earlier, this is challenging because the \code{CAS} performed by \code{find} may unlink unboundedly many marked nodes.
Unlinking a node changes its keyset to the empty set.
Hence, the \code{CAS} may affect an unbounded heap region.
Showing that the \code{CAS} maintains the keyset invariants therefore inevitably involves an inductive proof argument.

\begin{figure}[t]
  \centering
  \begin{tikzpicture}[>=stealth, font=\footnotesize, scale=0.8, every node/.style={scale=0.8}]
    % Nodes
    \def\xsep{1.2}
    \def\ysep{-1.8}

    \begin{scope}[local bounding box=g1]
    \node[unode] (n1) {$\{-\infty\}$};
    \node[unode, ellipse, right=4mm of n1,inner sep=-2pt] (n2) {$(\mathtt{-}\infty,\bluetext{2}]$};
    \node[mnode, right=4mm of n2] (n3) {$(2,\bluetext{5}]$};
    \node[mnode, right=4mm of n3] (n4) {$(5,\bluetext{6}]$};
    \node[unode, right=4mm of n4] (n5) {$(6,\bluetext{\infty}]$};
    %\node[unode, right=4mm of n5] (n6) {$(8,\bluetext{\infty}]$};

    \node[stackVar, below=.3cm of n1,label={[label distance=-14pt]\normalsize$\mathtt{head}$}] (hd) {};
    \node[stackVar, below=.3cm of n2, outer sep=-4pt] (l) {$l$};
    \node[stackVar, below=.3cm of n3] (ln) {$\lnext$};
    \node[stackVar, below=.3cm of n4] (t) {$t$};
    \node[stackVar, below=.3cm of n5] (tn) {$\tnext$};

    % Edges
    \draw[edge] (n1) to (n2);
    \draw[edge] (n2) to (n3);
    \draw[edge] (n3) to (n4);
    \draw[edge] (n4) to (n5);
    %\draw[edge] (n5) to (n6);
 
    \draw[edge] (hd) to (n1);
    \draw[edge] (l) to (n2);
    \draw[edge] (ln) to (n3);
    \draw[edge] (t) to (n4);
    \draw[edge] (tn) to (n5);

    \draw ($(n1.north west)+(-.4,.5)$) rectangle ($(n5.south east)+(.4,-1.2)$);
    \end{scope}

    \begin{scope}[local bounding box=g2, shift={($(g1.east)+(4,.4)$)}]
    \node[unode] (n1) {$\{-\infty\}$};
    \node[unode, ellipse, right=4mm of n1,inner sep=-2pt] (n2) {$(\mathtt{-}\infty,\bluetext{2}]$};
    \node[mnode, right=4mm of n2] (n3) {$\emptyset$};
    \node[mnode, right=4mm of n3] (n4) {$\emptyset$};
    \node[unode, right=4mm of n4] (n5) {$(2,\bluetext{\infty}]$};
    %\node[unode, right=4mm of n5] (n6) {$(8,\bluetext{\infty}]$};

    \node[stackVar, below=.3cm of n1,label={[label distance=-14pt]\normalsize$\mathtt{head}$}] (hd) {};
    \node[stackVar, below=.3cm of n2, outer sep=-4pt] (l) {$l$};
    \node[stackVar, below=.3cm of n3] (ln) {$\lnext$};
    \node[stackVar, below=.3cm of n4] (t) {$t$};
    \node[stackVar, below=.3cm of n5] (tn) {$\tnext$};

    % Edges
    \draw[edge] (n1) to (n2);
    \draw[edge] (n2) to[bend left=25] (n5);
    \draw[edge] (n3) to (n4);
    \draw[edge] (n4) to (n5);
    %\draw[edge] (n5) to (n6);
 
    \draw[edge] (hd) to (n1);
    \draw[edge] (l) to (n2);
    \draw[edge] (ln) to (n3);
    \draw[edge] (t) to (n4);
    \draw[edge] (tn) to (n5);

    \draw ($(n1.north west)+(-.4,.5)$) rectangle ($(n5.south east)+(.4,-1.2)$);
    \end{scope}

    \begin{scope}[local bounding box=g3, shift={($(g1.south)+(2,-1)$)}]
    \node[unode] (n1) {$\{-\infty\}$};
    \node[unode, ellipse, right=4mm of n1,inner sep=-2pt] (n2) {$(\mathtt{-}\infty,\bluetext{2}]$};
    \node[mnode, right=4mm of n2] (n3) {$\emptyset$};
    \node[mnode, right=4mm of n3] (n4) {$(2,\bluetext{6}]$};
    \node[unode, right=4mm of n4] (n5) {$(6,\bluetext{\infty}]$};
    %\node[unode, right=4mm of n5] (n6) {$(8,\bluetext{\infty}]$};

    \node[stackVar, below=.3cm of n1,label={[label distance=-14pt]\normalsize$\mathtt{head}$}] (hd) {};
    \node[stackVar, below=.3cm of n2, outer sep=-4pt] (l) {$l$};
    \node[stackVar, below=.3cm of n3] (ln) {$\lnext$};
    \node[stackVar, below=.3cm of n4] (t) {$t$};
    \node[stackVar, below=.3cm of n5] (tn) {$\tnext$};

    % Edges
    \draw[edge] (n1) to (n2);
    \draw[edge] (n2) to[bend left=30] (n4);
    \draw[edge] (n3) to (n4);
    \draw[edge] (n4) to (n5);
    %\draw[edge] (n5) to (n6);
 
    \draw[edge] (hd) to (n1);
    \draw[edge] (l) to (n2);
    \draw[edge] (ln) to (n3);
    \draw[edge] (t) to (n4);
    \draw[edge] (tn) to (n5);

    \draw ($(n1.north west)+(-.4,.5)$) rectangle ($(n5.south east)+(.4,-1.2)$);
    \end{scope}

    \draw[edge,dotted] ($(g1.east)+(-0,0)$) to node[above] {\Large\code{$l$.next $\assign \tnext$}} ($(g2.west)+(0,0)$);
    \draw[edge,dotted] ($(g1.south)+(0,0)$) to[bend right=45] node[left] {\Large\code{$l$.next $\assign t$}} ($(g3.west)+(0,0)$);
    \draw[edge,dotted] ($(g3.east)+(0,0)$) to[bend right=45] node[right] {\Large\code{$\;l$.next $\assign \tnext$}} ($(g2.south)+(0,0)$);
  \end{tikzpicture}
  \caption{%
    A compound update (top) and its decomposition into update chunks (bottom).
    Compound updates may affect an unbounded number of nodes, e.g., changing their keysets.
    Update chunks localize this effect of an update to a bounded (and small) number of nodes.
    \label{fig:harris-list-futures}
  }
\end{figure}
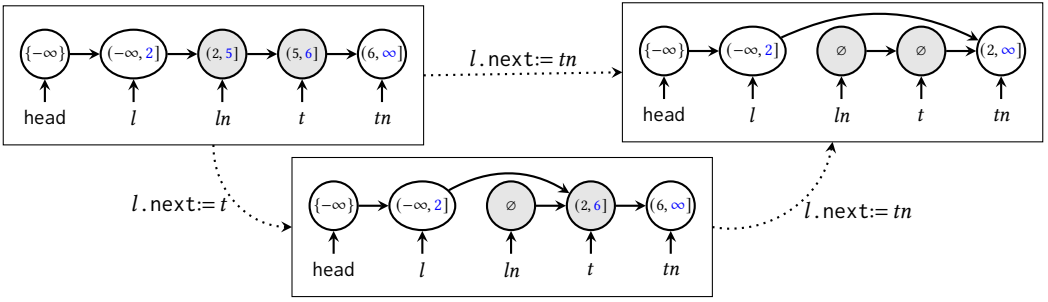

Our observation is that one can reason about the effect of the \code{CAS} that unlinks the marked segment by composing it out of a sequence of \emph{update chunks} that unlink the nodes in the segment one by one as indicated in \Cref{fig:harris-list-state}.
This yields an inductive proof argument where we reason about simpler updates that only affect a bounded number of nodes at a time.
The correctness of an update chunk $\acom$ is represented by a future $\FUT{P}{\acom}{Q}$.
A future can be thought of as a Hoare triple with precondition $P$ and postcondition $Q$.
However, futures inhabit the assertion level of the logic, rather than the meta level.
The crux of our program logic is that it allows one to derive a future \emph{on the side} in a subproof, while proving the correctness of the thread's traversal of the marked segment.
The advantage of this approach is that one can reuse the loop invariant for the proof of \code{traverse} towards proving the correctness of $\FUT{P}{\acom}{Q}$. This aids proof automation: the analysis no longer needs to synthesize an induction hypothesis for reasoning about the affected heap region out of thin air at the point where the \code{CAS} is executed.
In a way, \code{traverse} moonlights as ghost code that aids the correctness proof of the \code{CAS}.

At the point where the \code{CAS} is executed, the proof then invokes the constructed future.
A future can thus be thought of as a subproof that is saved up to be applied at some future point.

This idea is illustrated in \Cref{fig:harris-list-futures}.
The top left of the figure shows a state that is reached during the traversal of the marked segment from $l$ to $\tnext$.
The transition at the top depicts the effect of the update $\assignof{l.\nextsel}{\tnext}$ applied to this state.
The update unlinks the marked nodes between $l$ and $\tnext$ in one step.
The correctness of this update is expressed by the future
\[
  F \defeq \FUT{P \mstar \nextvarof{l}=\lnext \mstar \markvarof{\lnext}}{~l.\nextfld \assign \tnext~}{P \mstar \nextvarof{l}=\tnext}
\]
where $P$ is an appropriate invariant holding the relevant physical resources of the involved nodes.
We derive this future by composing two futures for simpler update chunks as depicted at the bottom half of \Cref{fig:harris-list-futures}.
The left resp. right update chunk is described by the future $F_1$ resp. $F_2$ as follows:
\begin{align*}
  F_1 &\defeq \FUT{P \mstar \nextvarof{l}=\lnext \mstar \markvarof{\lnext}}{~l.\nextfld \assign t~}{P \mstar \nextvarof{l}=t}
  \\
  F_2 &\defeq \FUT{P \mstar \nextvarof{l}=t \mstar \markvarof{t} \mstar \nextvarof{t}=\tnext}{~l.\nextfld \assign \tnext~}{P \mstar \nextvarof{l}=\tnext}
  \enspace.
\end{align*}
The future $F_1$ is derived inductively during the traversal of the marked segment from $l$ to $t$ using the same process that we are about to describe for deriving $F$.
The future $F_2$ can be easily proved in isolation.
In particular, the precondition $\markvarof{t}$ implies $\contentsof{t}=\emptyset$.
Hence, the keyset invariant $\contentsof{t} \subseteq \keysetof{t}$ is maintained.
The condition $\nextvarof{t}=\tnext$ guarantees that the update does not affect keysets of other nodes beyond $t$ and $\tnext$.

We would now like to compose $F_1$ and $F_2$ using the standard sequential composition rule of Hoare logic to derive the future
\[
  F' \defeq \FUT{P \mstar \nextvarof{l}=\lnext \mstar \markvarof{\lnext}}{~l.\nextfld \assign t; l.\nextfld \assign \tnext~}{P \mstar \nextvarof{l}=\tnext}
  \enspace.
\]
Once we have $F'$, we get $F$ by replacing the command $l.\nextfld \assign t; l.\nextfld \assign \tnext $ in $F'$ with $l.\nextfld \assign \tnext $ using a simple subsumption argument.

However, sequential composition requires that the postcondition of $F_1$ implies the precondition of $F_2$.
Unfortunately, the precondition of $F_2$ makes the additional assumptions $\markvarof{t}$ and $\nextvarof{t}=\tnext$ that are not guaranteed by $F_1$.
Now, observe that both facts are readily available in the outer proof context of the traversal: $\nextvarof{t}=\tnext$ follows from \cref{line:traverse-read-next} of \code{traverse} and $\markvarof{t}$ is obtained from the condition on \cref{line:traverse-check-mark}.
We transfer the facts from the outer proof context into $F_2$.
This eliminates them from the precondition and enables the sequential composition to obtain $F'$.
We refer to this transfer of facts as \emph{accounting}.
In a concurrent setting, accounting is sound provided the accounted facts are interference-free.
This is the case here since the \code{next} fields of marked nodes are never changed and marked nodes are never unmarked.
We explain this reasoning in more detail in \Cref{Section:Futures}.

The idea of futures applies more broadly.
They are useful whenever complex updates are prepared in advance by a traversal (as e.g. in Bw trees~\cite{DBLP:conf/icde/LevandoskiLS13a} and skip lists~\cite{DBLP:phd/ethos/Fraser04}).

%!TEX root = ../main.tex

\section{Programming Model}
\label{Section:Preliminaries}

We develop our reasoning principles in the context of concurrency libraries, which offer code to client applications that may be executed by an arbitrary number of threads. 
With this definition, not only the previously discussed search structures but also storage structures like stacks or lists and mutual exclusion mechanisms form concurrency libraries, and our techniques will apply to them as well. 
In this section, we formalize concurrency libraries.
Our development is parametric in the set of states and the set of commands, following the approach of abstract separation logic~\mbox{\cite{DBLP:conf/lics/CalcagnoOY07,DBLP:conf/popl/Dinsdale-YoungBGPY13,DBLP:journals/jfp/JungKJBBD18}}.

%-------------------------------------------------------------------------------------------------%
%-------------------------------------------------------------------------------------------------%
%-------------------------------------------------------------------------------------------------%
\smartparagraph{States}\label{Subsection:States}
We assume that states form a \emph{separation algebra}, a partial commutative monoid $(\setstates, \statemult, \emp)$.
For the set of units $\emp \subseteq \setstates$ we require that
\begin{inparaenum}[(i)]
	\item for all $\astate \in \setstates$, there exists some $\stateunit_\astate \in \emp$ with $\astate \statemult \stateunit_\astate = \astate$ and
	\item for all distinct $\stateunit,\stateunit' \in \emp$, $\stateunit \statemult \stateunit'$ is undefined.
\end{inparaenum}
We use $\astate_1\statemultdef\astate_2$ to indicate definedness.

\emph{Predicates} are sets of states $\apred\in \powerset{\setstates}$. 
\moreless{
	In the appendix, we introduce an assertion language to denote predicates but to simplify the exposition we stay on the semantic level.
}{
	To simplify the exposition, we stay on the semantic level and do not introduce an assertion language.
}
Predicates form a Boolean algebra $(\powerset{\setstates}, \cup, \cap, \subseteq, \overline{\phantom{\bullet}}, \emptyset, \setstates)$.
\emph{Separating conjunction} $\apred\mstar\apredp$ lifts the composition from states to predicates.
This yields a commutative monoid with unit $\emp$.
\emph{Separating implication} $\apred\sepimp\apredp$ gives residuals:
\[
	\apred\statemult\apredp\;\defeq\;\setcond{\astate_1\statemult\astate_2}{\astate_1\in\apred\;\wedge\; \astate_2\in\apredp\;\wedge\; \astate_1 \statemultdef \astate_2}
	\quad~~\text{and}\quad~~
	\apred\sepimp\apredp\;\defeq\;\setcond{\astate}{\setcompact{\astate}\mstar \apred\; \subseteq\; \apredp}
	\enspace.
\]

In our development, $(\setstates, \statemult, \emp)$ is a product of two separation algebras $(\setshared, \sharedmult, \sharedemp)$ and $(\setlocal, \localmult, \localemp)$.
We require $\emp\defeq \sharedemp \times \localemp \subseteq \setstates \subseteq \setshared \times \setlocal$.
In addition, $\setstates$ must be closed under decomposition: if $(\ashared_1\sharedmult\ashared_2, \alocal_1\localmult\alocal_2) \in \setstates$ then $(\ashared_1, \alocal_1) \in \setstates$, for all $\ashared_1,\ashared_2 \in \setshared$ and $\alocal_1,\alocal_2 \in \setlocal$.

For $(\ashared, \alocal) \in \setstates$ we call $\ashared$ the \emph{global state} and $\alocal$ the \emph{local state}.
States are composed component-wise, $(\ashared_1, \alocal_1)\statemult(\ashared_2, \alocal_2)\defeq(\ashared_1\sharedmult\ashared_2, \alocal_1\localmult\alocal_2)$, and this composition is defined only if the product is again in~$\setstates$.

\begin{lemma}
  \label{lem-composed-state-algebra}
  $(\setstates, \statemult, \emp)$ is a separation algebra.
\end{lemma}

%-------------------------------------------------------------------------------------------------%
%-------------------------------------------------------------------------------------------------%
%-------------------------------------------------------------------------------------------------%
\smartparagraph{Commands}
The second parameter to our development is the set of commands $(\setcom, \semCom{-})$ that may be used to modify states.
The set may be infinite, which allows us to treat atomic blocks as single commands.
The effect of commands on states is defined by an interpretation.
It assigns to each command a non-deterministic state transformer $\semCom{\acom}$ that takes a state and returns the set of possible successor states, $\semCom{\acom}: \setstates \to \powerset{\setstates}$.
We lift the state transformer to predicates~\cite{DBLP:books/ph/Dijkstra76} in the expected way, $\semComOf{\acom}{\apred}\defeq\bigcup_{\astate\in\apred}\semComOf{\acom}{\astate}$. 
We assume to have a command $\cskip$ that is interpreted as the identity.
For the frame rule to be sound, we expect the following monotonicity to hold~\cite{DBLP:conf/lics/CalcagnoOY07,DBLP:conf/popl/Dinsdale-YoungBGPY13}, for all $\apred, \apredp, \apredpp$:
\begin{align*}
	\label{cond-loccom}
	\semComOf{\acom}{\apred}\subseteq\apredp\quad\text{implies}\quad\semComOf{\acom}{\apred\mstar\apredpp}\subseteq\apredp\mstar\apredpp
	\enspace.
	\tag{LocCom}
\end{align*}
We handle commands that may fail as in~\cite{DBLP:conf/lics/CalcagnoOY07} by letting them return $\abort$ which is added as a new top element to the powerset lattice $\powerset{\Sigma}$, so that \eqref{cond-loccom} is trivially satisfied.

%-------------------------------------------------------------------------------------------------%
%-------------------------------------------------------------------------------------------------%
%-------------------------------------------------------------------------------------------------%
\smartparagraph{Concurrency Libraries}
Having fixed the set of states and the set of commands, a concurrency library is defined by a single program that is executed in every thread. 
The assumption of a single program can be made without loss of generality. 
The program code is drawn from the standard while-language $\setstmt$ defined by:
\begin{align*}
	\astmt\;\defebnf\;
		\acom
		\bnf
		\choiceof{\astmt}{\astmt}
		\bnf
		\seqof{\astmt}{\astmt}
		\bnf
		\loopof{\astmt}
		\ . 
\end{align*}
The semantics of the library is defined in terms of unlabeled transitions among configurations.
A \emph{configuration} is a pair  $\aconfig=(\ashared, \apc)$ consisting of a global state $\ashared\in \setshared$ and a program counter $\apc:\nat\to\setlocal\times\setstmt$.
The program counter assigns to every thread, modeled as a natural number,  the current local state and the statement to be executed next.
We use $\setconfig$ to denote the set of all configurations.
A configuration $(\ashared, \apc)$ is \emph{initial} for predicate~$\apred$ and library code~$\astmt$, if the program counter of every thread yields a local state $(\alocal, \astmt)$ where the code is the given one and the state satisfies $(\ashared, \alocal)\in \apred$.
The configuration is \emph{accepting} for predicate $\apredp$, if every terminated thread  $(\alocal, \cskip)$ satisfies the predicate, $(\ashared, \alocal)\in\apredp$.
We write these configuration predicates as the following sets
\begin{align*}
	\initset{\apred}{\astmt}
		\mkern-1mu&\defeq\mkern-1mu
		\{(\ashared, \apc) \,{\mid}\,
			\forall i,\alocal,\mkern-1mu\widehat\astmt.\mkern+2mu
				\apc(i)\prall{=}(\alocal,\mkern-1mu \widehat\astmt) \prall{\Rightarrow}
				(\ashared,\mkern-1mu \alocal)\prall{\in}\apred \prall{\wedge} \widehat\astmt\prall{=}\astmt
			\}
	\\
	\acceptset{\apredp}
		\mkern-1mu&\defeq\mkern-1mu
		\{(\ashared, \apc) \,{\mid}\,
			\forall i,\alocal.~
			\apc(i)=(\alocal, \cskip)\Rightarrow (\ashared, \alocal)\in\apredp \,
		\}
	\ .
\end{align*}

The unlabeled transition relation among configurations is defined in \Cref{Figure:Relations}.
It relies on a labeled transition relation capturing the flow of control.
A command may change the global state and the local state of the executing thread.
It will not change the local state of other threads.
A computation of the library is a finite sequence of consecutive transitions.
A configuration is reachable if there is a computation that leads to it.
We write $\reachset{\aconfig}$ for the set of all configurations reachable from $\aconfig$ and lift the notation to sets where needed.

\begin{figure}%\small
	\begin{mathpar}
		\infrule{}{
			\pcStepOf{\acom}{\acom}{\cskip}
		}
		\and
		\infrule{}{
			\pcStepOf{\seqof{\cskip}{\astmt}}{\cskip}{\astmt}
		}
		\and
		\infrule{}{
	          \pcStepOf{\loopof{\astmt}}{\cskip}{\choiceof{\cskip}{\seqof{\astmt}{\loopof{\astmt}}}}
		}
	        \\
		\infrule{i \in \set{1,2}}{
			\pcStepOf{\choiceof{\astmt_1}{\astmt_2}}{\cskip}{\astmt_i}
		}
	        \and
		\infrule{
			\pcStepOf{\astmt_1}{\acom}{\astmt_1'}
		}{
			\pcStepOf{\seqof{\astmt_1}{\astmt_2}}{\acom}{\seqof{\astmt'_1}{\astmt_2}}
		}\and
		\infrule{
			\pcStepOf{\astmt_1}{\acom}{\astmt_2}
			\\
			(\ashared_2, \alocal_2)\in\semComOf{\acom}{\ashared_1, \alocal_1}
		}{
			(\ashared_1, \apc[i\mapsto(\alocal_1, \astmt_1)])
			\rightarrow	(\ashared_2, \apc[i\mapsto(\alocal_2, \astmt_2)])
		}
	\end{mathpar}
	\caption{%
		Transition relation $\progStepRel\;\subseteq\setconfig\times\setconfig$ based on the control-flow relation $\rightarrow{} \subseteq \setstmt\times\setcom\times\setstmt$.
		\label{Figure:Relations}
	}
\end{figure}

%!TEX root = ../main.tex

\section{Owicki-Gries for Concurrency Libraries}\label{Section:OG}
We formulate the correctness of concurrency libraries as the validity of Hoare triples $\hoareOf{\apred}{\astmt}{\apredp}$. 
A Hoare triple is valid if for every configuration $\aconfig$ that is initial wrt. $\apred$ and $\astmt$, every reachable configuration $\aconfig'$ is accepting wrt. $\apredp$. 
The definition refers to all threads executing the library code.

\begin{definition}
	$\subModels\mkern-2mu\hoareOf{\apred}{\astmt}{\apredp} \defeq \mkern+2mu\reachset{\initset{\apred}{\astmt}}\prall{\subseteq}\acceptset{\apredp}$.
\end{definition}

To establish this validity, we develop a thread-modular reasoning principle~\cite{DBLP:journals/acta/OwickiG76} that proceeds in two steps. 
First, we verify the library code as if it was run by an isolated thread using judgments of the form $\thePredicates, \theInterference\semCalc\hoareOf{\apred}{\astmt}{\apredp}$.

The Hoare triple of interest is augmented by two pieces of information.
The set $\thePredicates$ contains the intermediary predicates encountered during the proof of the isolated thread.
The set $\theInterference$ contains the interferences, the changes the isolated thread may perform on the shared state.
The notion of interference will be made precise in a moment.
Recording both sets during the proof is an idea we have taken from~\cite[Section 7.3]{DBLP:conf/popl/Dinsdale-YoungBGPY13}.

The second phase of the thread-modular reasoning is to check that the local proof still holds in the presence of other threads.
This is the famous interference-freedom check.
It takes as input the computed sets $\thePredicates$ and $\theInterference$ and verifies that no interference can invalidate a predicate,  denoted by $\isInterferenceFreeOf[\theInterference]{\thePredicates}$.

%-------------------------------------------------------------------------------------------------%
%-------------------------------------------------------------------------------------------------%
%-------------------------------------------------------------------------------------------------%
\smartparagraph{Interference} 
An \emph{interference} is a pair $(\apredpp, \acom)$ consisting of a predicate and a command.
It represents the fact that from states in $\apredpp$ environment threads may execute command $\acom$.
A state $(\ashared, \alocal)$ held by the isolated thread of interest will change under the interference to a state in
\[
	\semOf{(\apredpp, \acom)}(\ashared, \alocal)
	~\defeq~
	\setcond{(\ashared', \alocal)}{
		\exists \alocal_1, \alocal_2.\;\; (\ashared, \alocal_1)\in \apredpp ~\wedge~ (\ashared', \alocal_2)\in \semComOf{\acom}{\ashared, \alocal_1}}
	\enspace.
\]
We consider every state $(\ashared, \alocal_1)\in \apredpp$ that agrees with $(\ashared, \alocal)$ on the global component, compute the post, and combine the resulting global component with the local component $\alocal$.
The agreement of different threads on the global state is precisely what is used in program logics like RGSep~\cite{DBLP:conf/concur/VafeiadisP07,DBLP:phd/ethos/Vafeiadis08}.
We lift $\semOf{(\apredpp, \acom)}$ to predicates in the expected way.

We only record interferences that have an effect on the global state.
An interference $(\apredpp, \acom)$ is \emph{effectful}, denoted by $\effectful{\apredpp}{\acom}$, if it changes the shared state of an element in $\apredpp$:
\[
	\effectful{\apredpp}{\acom} \,\defeq\, \exists(\ashared_1,\alocal_1)\in \apredpp.\exists(\ashared_2, \alocal_2)\in\semOf{\acom}(\ashared_1, \alocal_1).\ashared_2\neq \ashared_1
	\enspace. 
\]

The thread-local proof computes a set of interferences.
For a predicate $\apredpp$ and a command $\acom$, the interference set is $\inter{\apredpp}{\acom}\defeq \setcompact{(\apredpp, \acom)}$ if $\effectful{\apredpp}{\acom}$ and $\inter{\apredpp}{\acom}\defeq \emptyset$ otherwise.
We consider interference sets up to the operation of joining predicates for the same command, $\setcompact{(\apred, \acom)}\cup \setcompact{(\apredp, \acom)}=\setcompact{(\apred\cup\apredp, \acom)}$.
Then $\inter{\apredpp}{\acom}\subseteq \theInterference$ means there is no interference to capture or there is an interference $(\apredppp, \acom)\in\theInterference$ with $\apredpp\subseteq \apredppp$.
We write $\theInterference\mstar\apredppp$ for the set of interferences $(\apredpp\mstar\apredppp, \acom)$ with $(\apredpp, \acom)\in\theInterference$.
Similarly, we write $\thePredicates\mstar\apredppp$ for the set of predicates $\apred\mstar\apredppp$ with $\apred\in\thePredicates$.

The \emph{interference-freedom check} takes as input a set of interferences $\theInterference$ and a set of predicates $\thePredicates$.
It checks that no interference can invalidate a predicate, $\semOf{(\apredpp, \acom)}(\apred)\subseteq \apred$ for all $(\apredpp, \acom)\in\theInterference$ and all $\apred\in\thePredicates$.
If this is the case, we write $\isInterferenceFreeOf[\theInterference]{\thePredicates}$ and say that the set of predicates $\thePredicates$ is interference-free wrt. $\theInterference$.

The interference-freedom check is non-compositional, and in manual/mechanized program verification this has been the reason to prefer rely-guarantee methods~\cite{DBLP:phd/ethos/Vafeiadis08,DBLP:conf/concur/VafeiadisP07,DBLP:conf/popl/Feng09}.
From the point of view of automated verification, the difference does not matter.
After all, there is no compositional way of computing the relies and guarantees.

%-------------------------------------------------------------------------------------------------%
%-------------------------------------------------------------------------------------------------%
%-------------------------------------------------------------------------------------------------%
\smartparagraph{Program Logic}
To verify library code as if it was run by an isolated thread, we derive (augmented) Hoare triples using the proof rules in \Cref{Figure:ProgramLogic}.
The rules are standard except that they work on the semantic level.
This is best seen in rule~\ruleref{com-sem}, which explicitly checks the postcondition for over-approximating the postimage.
The rule only adds the postcondition to the set of predicates to be checked for interference freedom. 
Similarly, the consequence rule \ruleref{infer-sem} neither adds the strengthened precondition nor the weakened postcondition.  
We can freely manipulate predicates as long as there is an interference-free predicate between every pair of consecutive statements, rule~\ruleref{seq}.  
To ensure this for loops which may be left without execution, rule~\ruleref{loop} adds $\apred$ to the set of predicates. 
The initial predicate of the overall Hoare triple is added to the set of predicates by the assumption of \Cref{Theorem:Soundness}.

\begin{figure*}[t]%\small
\begin{mathpar}
	\inferH{com-sem}{
		\semComOf{\acom}{\apred}\subseteq \apredp
	}{
		\setcompact{\apredp}, \inter{\apred}{\acom}\semCalc\hoareOf{\apred}{\acom}{\apredp}
	}
	\and
	\inferH{infer-sem}{
		\apred\subseteq \apred'
		\quad
		\thePredicates_1, \theInterference_1 \semCalc\hoareOf{\apred'}{\astmt}{\apredp'}
		\quad
		\apredp'\subseteq \apredp
	}{
		\thePredicates_1\cup\thePredicates_2, \theInterference_1\cup\theInterference_2\semCalc\hoareOf{\apred}{\astmt}{\apredp}
	}
	\and
	\inferH{frame}{
		\thePredicates, \theInterference\semCalc\hoareOf{\apred}{\astmt}{\apredp}
	}{
		\thePredicates\mstar\apredpp, \theInterference\mstar\apredpp\semCalc\hoareOf{\apred\mstar \apredpp}{\astmt}{\apredp\mstar\apredpp}
	}\and
	\inferH{seq}{
		\thePredicates_1, \theInterference_1\semCalc\hoareOf{\apred}{\astmt_1}{\apredp}\\
		\thePredicates_2, \theInterference_2\semCalc\hoareOf{\apredp}{\astmt_2}{\apredpp}
	}{
		\setcompact{\apredp}\cup\thePredicates_1\cup\thePredicates_2, \theInterference_1\cup\theInterference_2\semCalc\hoareOf{\apred}{\astmt_1;\astmt_2}{\apredpp}
	}\and
	\inferH{loop}{
		\thePredicates, \theInterference\semCalc\hoareOf{\apred}{\astmt}{\apred}
	}{
		\setcompact{\apred}\cup\thePredicates, \theInterference\semCalc\hoareOf{\apred}{\loopof{\astmt}}{\apred}
	}\and
	\inferH{choice}{
		\thePredicates_1, \theInterference_1\semCalc\hoareOf{\apred}{\astmt_1}{\apredp}\\
		\thePredicates_2, \theInterference_2\semCalc\hoareOf{\apred}{\astmt_2}{\apredp}
	}{
		\thePredicates_1\cup\thePredicates_2, \theInterference_1\cup\theInterference_2\semCalc\hoareOf{\apred}{\choiceof{\astmt_1}{\astmt_2}}{\apredp}
	}
\end{mathpar}
\caption{%
	Program logic.
	\label{Figure:ProgramLogic}}
\end{figure*}

\begin{theorem}[Soundness]
	\label{Theorem:Soundness}
	$\thePredicates, \theInterference\semCalc\hoareOf{\apred}{\astmt}{\apredp}$\; and \;$\isInterferenceFreeOf[\theInterference]{\thePredicates}$\: and\; $\apred\in\thePredicates$\; imply \;$\subModels\hoareOf{\apred}{\astmt}{\apredp}$.
\end{theorem}

\techreport{See Appendix~\ref{Section:SoundnessProgramLogic} for the detailed proof of Theorem~\ref{Theorem:Soundness}.}

%-------------------------------------------------------------------------------------------------%
%-------------------------------------------------------------------------------------------------%
%-------------------------------------------------------------------------------------------------%
\smartparagraph{Linearizability}
\looseness=-1
We extend the above program logic in order to utilize it for linearizability proofs.
Our extension draws on ideas from atomic triples~\cite{DBLP:conf/ecoop/PintoDG14}.
That is, we enforce that every operation linearizes exactly once and satisfies a given sequential specification when doing so.
In the context of concurrent search structures (CSS), sequential specifications $\asspec$ take the form: \[
	\asspec~=~\hoareOf{ \abscontent.~ \acss(\abscontent) }{ \absop(k) }{ v.~ \exists \abscontentp.~ \acss(\abscontentp) \mstar \acssup(\abscontent,\abscontentp,k,v) }
	\ .
\]
The sets $\abscontent$ and $\abscontentp$ are the logical contents of the structure before and after operation $\absop(k)$.
The predicate $\acss(\abscontent)$ connects the structure's physical state with the logical contents $\abscontent$.
The relation $\acssup(\abscontent,\abscontentp,k,v)$ encodes the admissible updates and the expected return value $v$ of $\absop(k)$.

To enforce exactly one linearization point, we use update tokens~\cite{DBLP:phd/ethos/Vafeiadis08} of the form $\anobl{\asspec}$ and $\aful{\asspec}{v}$.
The former states that an operation has not yet encountered its linearization point, it is still obliged to linearize.
The latter states that the linearization point has been encountered and that value $v$ must be returned to comply with $\asspec$.
Technically, the update tokes are thread-local ghost resources.
We assume the program semantics to simply ignore this ghost component.
Note that having thread-local update tokens does not allow for proving impure future-dependent linearization points as they require intricate helping protocols where threads exchange their update tokes through the global state so that other threads can resolve them.
A generalization is straight-forward, but we prefer the simpler setting to not distract from our contributions.

%!TEX root = ../main.tex

\begin{figure}
	\begin{mathpar}
		\inferH{lin-none}{
			\thePredicates, \theInterference
			\semCalc
			\hoareOf{\acpred}{\acom}{\acpredp}
			\\\\
			\acpred\subseteq\acss(\abscontent)
			\\
			\acpredp\subseteq\acss(\abscontent)
		}{
			\thePredicates, \theInterference
			\semcalclin
			\hoareOf{\acpred}{\acom}{\acpredp}
		}
		\and
		\setHistCol{
		\inferH{lin-past}{
			\acpred\subseteq\pastOf{\bigl(\acss(\abscontent)\cap\acssup(\abscontent,\abscontent,k,v)\bigr)}
		}{
			\thePredicates, \theInterference
			\semcalclin
			\hoareOf{\anobl{\asspec}\mstar\acpred}{\cskip}{\aful{\asspec}{v}\mstar\acpred}
		}}
		\and
		\inferH{lin-now}{
			\acpred\subseteq\acss(\abscontent)
			\\
			\thePredicates, \theInterference
			\semCalc
			\hoareOf{\acpred}{\acom}{\acpredp}
			\\
			\acpredp\subseteq\acss(\abscontentp)\cap\acssup(\abscontent,\abscontentp,k,v)
		}{
			\thePredicates, \theInterference
			\semcalclin
			\hoareOf{\anobl{\asspec}\mstar\acpred}{\acom}{\aful{\asspec}{v}\mstar\acpredp}
		}
	\end{mathpar}
	\caption{%
		Proof rules handling linearizability tokens for commands. \setHistCol{Rule~\ruleref{lin-past} is detailed in \Cref{Section:History}.}
		\label{fig:linearizability-rules}
    }
\end{figure}

To prove linearizability, we introduce a new proof system $\semcalclin$ that coincides with $\semCalc$ from \Cref{Figure:ProgramLogic} except that it replaces rule~\ruleref{com-sem} with the rules from \Cref{fig:linearizability-rules} (ignore rule~\ruleref{lin-past} for now, we explain it in \Cref{Section:History}).
The new rules handle the update tokens, leaving the task of establishing the validity of the actual Hoare triple to the base proof system $\semCalc$.
To do this, Rule~\ruleref{lin-none} ensures that the command does not change the logical contents of the structure, meaning that the update tokens are unaffected.
Rule~\ruleref{lin-now} converts the update token $\anobl{\asspec}$ into $\aful{\asspec}{v}$ if the command is a linearization point satisfying the sequential specification $\asspec$.
This conversion is applicable only once (because $\anobl{\asspec}$ is not duplicable), which makes sure there can be at most one linearization point.
Note that this rule handles both pure and impure linearization points.
Similar to atomic triples~\cite{DBLP:conf/ecoop/PintoDG14}, the rule requires threads to observe the very moment the linearization point occurs.
To ensure there is at least one linearization point, we require the operation's postcondition to contain the appropriate fulfillment token $\aful{\asspec}{v}$.
Formally, we seek to establish Hoare triples of the following form: \[
	\semcalclin\hoareOf{ \abscontent.~ \acss(\abscontent) \mstar \anobl{\asspec} }{ \absop(k) }{ v.~ \exists \abscontentp.~ \acss(\abscontentp) \mstar \aful{\asspec}{v} }
	\ .
\]
We take the derivability of such a Hoare triple as the ground truth for linearizability, trusting that a semantic result connecting the derivability to a statement about computation histories would be routine to derive~\cite{DBLP:conf/pldi/LiangF13}.

%!TEX root = ../main.tex

\section{Reasoning about Keysets using Flows}
\label{Section:Keysets}

Recall from \Cref{sec:overview} that we localize the reasoning about the abstract state $\contentsof{N}$ of the data structure to the contents $\contentsof{\anode}$ of a single node $\anode$ using its keyset $\keysetof{x}$.
In this section, we define the keysets as a derived quantity that we can reason about locally in a separation logic.
Then, we use this formalism to define the node-local invariant $\nodeof{\anode}$ of our running example.
This node-local invariant is used by our tool to fully automatically generate a proof of the Harris set (cf. \Cref{Section:Evaluation}).

%-------------------------------------------------------------------------------------------------%
%-------------------------------------------------------------------------------------------------%
%-------------------------------------------------------------------------------------------------%
\techreport{\smartparagraph{Flow Framework}}
We derive the keyset of a node $\anode$ from another quantity $\insetof{\anode}$, the node's \emph{inset}.
Intuitively, $k \in \insetof{\anode}$ if a thread searching for $k$ will traverse node $\anode$.
For the Harris set, we define $\insetof{\head}\prall{=}[-\infty,\infty]$ and for every other node we obtain $\insetof{\anode}$ as the solution of the following fixpoint equation:
\[
    \insetof{\anode}\:=\:{\textstyle \bigcup_{ (\anodep,\anode) \in E}}~~ \insetof{\anodep} \cap (\anodep.\keysel,\infty]
    \enspace.
\]
Here, the set of edges $E$ is induced by the next pointers in the heap.
If we remove those keys $k$ from $\insetof{\anode}$ for which a search leaves $\anode$ (i.e., if $k > x.\keysel$ in the Harris set), we obtain $\keysetof{\anode}$.
These definitions ensure for free that the keysets are disjoint, the first of our keyset invariants.
They also generalize to any search structure~\cite{DBLP:journals/tods/ShashaG88}.

To express keysets in separation logic, we use the flow framework~\cite{DBLP:conf/esop/KrishnaSW20,DBLP:journals/pacmpl/KrishnaSW18}.
In this framework, the heap is augmented by associating every node $\anode$ with a quantity $\flowvarof{\anode}$ that is defined as a solution of a fixpoint equation over the heap like the one defining the inset above.
Assertions describe disjoint fragments of the augmented global heap, similar to classical separation logic.
The augmented heap fragments are called \emph{flow graphs}.
In addition to tracking the flow of each node, a flow graph also has an associated \emph{interface} consisting of an \emph{inflow} and an \emph{outflow}.
The inflow $\inflowof{\anodep,\anode}$ captures the contribution to the flow of $\anode$ inside the heap fragment via an edge from a heap node $\anodep$ outside the fragment, and conversely for the outflow $\outflowof{\anode,\anodep}$.
Flow graphs $\aflowgraph$ and $\aflowgraph'$ compose if they are disjoint and their interfaces are compatible (i.e., the composed flow graph $\aflowgraph \mstar \aflowgraph'$ has the same flow as the components).
The framework then enables local reasoning about the effects of heap updates on flow graphs.
In essence, if a local update inside a region $\aflowgraph$ of a larger flow graph $\aflowgraph \mstar \aflowgraph'$ maintains $\aflowgraph$'s interface, then the flow in $\aflowgraph'$ does not change.
Hence, any property about $\aflowgraph'$, such as that each of its nodes $\anode$ satisfies the keyset invariant $\contentsof{\anode} \subseteq \keysetof{\anode}$, can be framed across the update.
Note that the flow and interfaces associated with the physical heap constitute ghost state.
Intuitively, they are recomputed after each update.

\techreport{Appendix~\ref{Section:Flows} provides the technical details of the flow framework adapted to the semantic setting of our program logic.}
For the remainder of the paper, it suffices to know that we instantiate the framework such that a node's inset can be obtained from its flow\techreport{{} as described above}.
To apply the framework to concurrent search structures and to reason locally about their sequential specifications, we define the $\acss$ predicate from \Cref{Section:OG} \moreless{as follows}{by}: \[
    \acss(\abscontent) ~\defeq~ \exists \somenodes.~~ \inv(\abscontent,\somenodes,\somenodes)
    \ ,
\]
where $\somenodes$ is the set of nodes the search structure is composed of and $\inv(\abscontent,\somenodesp,\somenodes)$ is a search structure specific invariant.
The invariant is parameterized in $\abscontent$ and $\somenodes$ as well as a subset $\somenodesp\subseteq\somenodes$ for which $\inv(\abscontent,\somenodesp,\somenodes)$ holds actual resources.
It must ensure that $\abscontent$ is the contents of the subregion $\somenodesp$, i.e. $\abscontent=\contentsof{\somenodesp}$.
It must also ensure the keyset invariant $\contentsof{\anode} \subseteq \keysetof{\anode}$ for all $\anode \in \somenodesp$.
\citet{DBLP:journals/pacmpl/KrishnaSW18,DBLP:conf/esop/KrishnaSW20} show how the two constraints allow us to split and merge the invariant for disjoint subregions of the structure for the purpose of framing.
It is this splitting/merging that localizes the reasoning.
The following definition makes the desired property formally precise.
\begin{definition}
    \label{def:decomposable-invariant}
    An invariant $\inv$ is \emph{decomposable} if it satisfies: \[
        \inv(\abscontent, \somenodesp_1 \uplus \somenodesp_2,\somenodes)
        ~\Longleftrightarrow~
        \exists\, \abscontent_1\, \abscontent_2.~\, \abscontent = \abscontent_1 \uplus \abscontent_2 \,\land\, \inv(\abscontent_1,\somenodesp_1,\somenodes) \,\mstar\, \inv(\abscontent_2,\somenodesp_2,\somenodes)
        \ .
    \]
\end{definition}

For linearizability, it thus suffices to identify a small subregion that contains the decisive node of the operation.
A search for key $k$, for instance, only requires the region $\inv(\abscontent_\anode,\{\anode\},\somenodes)$ with $k\in\keysetof{\anode}$ in order to linearize its return value $v$ because we obtain
\[
    \inv(\abscontent', \somenodes \setminus \{\anode\},\somenodes) \mstar \inv(\abscontent_\anode,\{\anode\},\somenodes) \land k \in \keysetof{\anode} \land v = (k \in \abscontent_\anode) \,\vdash\, \exists \abscontent. \; \acss(\abscontent) \land v = (k \in \abscontent)
    \enspace.
\]

As shown in \cite{DBLP:conf/pldi/KrishnaPSW20}, there is a generic construction for a decomposable $\inv$ that works for all search structures.
To avoid additional technical machinery, we next present a simplified version of this construction that is specific to the Harris set and similar list-based search structures.

%-------------------------------------------------------------------------------------------------%
%-------------------------------------------------------------------------------------------------%
%-------------------------------------------------------------------------------------------------%
\smartparagraph{The Harris Set Invariant}
We represent predicates $\apred \prall\subseteq \setstates$ syntactically using separation logic assertions that are for the most part standard.
In particular, we use boxed assertions \fbox{$\statepred$} that are inspired by RGSep~\cite{DBLP:conf/concur/VafeiadisP07,DBLP:phd/ethos/Vafeiadis08} to mean that $\statepred$ is interpreted in the global state.
Unboxed assertions are interpreted in the local state.
A \emph{points-to} predicate takes the form $\spointsto{x}{\overline{\asel_i:t_i},\flowfld:t_\fval,\inflowfld:t_\inflow}$ and describes a flow graph consisting of a single node $x$.
Here, each $\asel_i$ is a field selector and $t_i$ is a term denoting the field's associated value.
The \emph{ghost field} $\flowfld$ stores $x$'s flow and $\inflowfld$ stores its inflow.
The semantics of assertions is $\semOf{\statepred} \subseteq \setstates$.
\techreport{We defer the technical details to Appendix~\ref{Section:Instantiation} as they are standard.}

We next define the resources associated with a node $x$, its inset, and keyset.
In proofs, we will assume that assertions are existentially closed, and will omit the corresponding outer quantifiers.
Formulas like $\nodeof{x}$ defined in the following introduce logical variables like $\markvarof{x}$ that are visible beyond $\nodeof{x}$, for example, to define the keyset term $\keysetof{x}$.
We define:
\begin{gather*}
    \mkern+6mu\nodeof{x} {} \defeq
        \boxed{\,
            \spointsto{x}{
                \markfld\!:\markvarof{x},
                \nextfld\!:\nextvarof{x},
                \keyfld\!: \keyvarof{x},
                \flowfld\!:\flowvarof{x},
                \inflowfld\!:\inflowvarof{x}
            }
        \,}
    \\%[.15em]
    \begin{aligned}
    \insetof{x} &\defeq
        \ite{x = \head}{[-\infty,\infty]}{\flowvarof{x}}
    \qquad&\qquad
    \keysetof{x} &\defeq
        \insetof{x} \setminus (\keyvarof{x},\infty]
    \\%[.15em]
    \contentsof{\anode} &\defeq
        \ite{\markvarof{\anode}}{\emptyset}{\set{ \keyvarof{\anode} }}
    \qquad&\qquad
    \contentsof{\somenodes} &\defeq
        {\textstyle \bigcup_{\anode\in\somenodes}}\; \contentsof{\anode}
    \enspace.
    \end{aligned}
\end{gather*}
With these definitions in place, we define the invariant $\allinvof{\abscontent,\somenodesp,\somenodes}$ that is maintained by each subregion $\somenodesp \subseteq \somenodes$ of a Harris set structure consisting of nodes $\somenodes$:
\begin{align*}
    \allinvof{\abscontent,\somenodesp,\somenodes} & {} \defeq
        \abscontent=\contentsof{\somenodesp}
        \mstar
        \htinv(\somenodes)
        \mstar 
        \bigmstar x\in\somenodesp.\;
            \nodeof{x} \mstar
            \invNode^1(x) \mstar
            \invNode^2(x) \mstar
            \invNode^3(x, \somenodes) \mstar
            \invNode^4(x)
    \\%[.15em]
    \htinv(\somenodes) & {} \defeq
        \head\in\somenodes \,\mstar\,
        \keyvarof{\head} \prall{=} -\infty \,\mstar\,
        \neg\markvarof{\head}
    \\%[.15em]
    \invNode^1(x) & \defeq {}
        \neg\markvarof{x} \:\Rightarrow\: \insetof{x} \neq \emptyset
    \\%[.15em]
    \invNode^2(x) & \defeq {}
        \insetof{x} \neq \emptyset \:\Rightarrow\: [\keyvarof{x},\infty)\subseteq\insetof{x}
    \\%[.15em]
    \invNode^3(x, \somenodes) & {} \defeq
        \setcompact{x, \nextvarof{x}} \prall{\in} \somenodes \land (\keyvarof{x} \prall{=} \infty \:\Rightarrow\: \neg\markvarof{x})
    \\%[.15em]
    \invNode^4(x) & {} \defeq
        \forall y, z.\;
        \inflowof{x}(y,x) \prall{\neq} \emptyset \mstar \inflowof{x}(z,x) \prall{\neq} \emptyset  \:\Rightarrow\: y \prall{=} z
    \enspace .
\end{align*}
Formula $\invNode^1(x)$ captures that all unmarked nodes are reachable from $\head$.
Formula $\invNode^2(x)$ implies the second keyset invariant $C(x) \subseteq \keysetof{x}$ and will also allow us to establish $k \in \keysetof{x}$ at the appropriate points in the proof.
Formula $\invNode^3(x)$ ensures $\somenodesp \subseteq \somenodes$ and that $\somenodes$ is closed under traversal of next pointers.
It also implies that the tail node is unmarked.
Finally, $\invNode^4(x)$ implies that there exists at most one path from $\head$ to each $x$ that a traversal would actually follow.
This is needed to prove that unlinking marked nodes from the structure preserves the invariant.
It is worth noting that the invariant does not put many constraints on the data structure shape.
The sole purpose is to provide enough information to reason about the keysets.
If $\somenodes$ is clear, we abbreviate $\inv(\abscontent,\somenodesp,\somenodes)$ to $\inv(\abscontent,\somenodesp)$.
The invariant of the Harris set is decomposable as per \Cref{def:decomposable-invariant}.
\begin{lemma}
    \label{thm:invariant-locality}
    The Harris set invariant $\inv(\abscontent,\somenodesp,\somenodes)$ is decomposable.
\end{lemma}

%!TEX root = ../main.tex

\section{Futures}
\label{Section:Futures}

We now make our program logic future-proof.
We refer to the set of nodes affected by an update as the update's \emph{footprint}.
An update affects a node $\anode$ if it changes a field value of $\anode$ or the flow at $\anode$.
The footprint of an update on a flow graph is in general larger than the footprint of the same update on the underlying heap graph alone.
For instance, $\assignof{l.\nextsel}{r}$ does not abort as long as the location at $l$ is in the heap graph.
However, if the same command is executed on a flow graph, it will typically require other nodes such as $r$ and $l.\nextsel$ to be present in order for the command not to abort.
This is because the command may change the flow of these other nodes.
For instance, the footprint of the \code{CAS($l$.next,$\lnext$,$r$)} on \cref{line:search-cas} of \Cref{fig:harris-list-algo} comprises the entire marked segment between $l$ and $r$, because it changes the flow and hence the keysets of all nodes in the segment.
We introduce futures to reason about updates with such unbounded footprints.

As futures admit general reasoning principles, we study them in the abstract semantic setting of \cref{Section:OG} and then apply the developed principles to our concrete running example.

%-------------------------------------------------------------------------------------------------%
%-------------------------------------------------------------------------------------------------%
%-------------------------------------------------------------------------------------------------%
\smartparagraph{Reasoning about Futures}
Futures are expressed in terms of weakest preconditions.
We define the weakest precondition $\wpreOf{\acom}{\apredp}$ of a command $\acom$ and predicate $\apredp$ in the expected way: $\wpreOf{\acom}{\apredp}\defeq\setcond{\astate \in \setstates}{\semOf{\acom}(\astate)\subseteq\apredp}$.
The weakest precondition of the sequential composition $\seqof{\acom_1}{\acom_2}$ is also defined as usual: $\wpreOf{\acom_1;\acom_2}{\apredp}\defeq\wpreOf{\acom_1}{\wpreOf{\acom_2}{\apredp}}$.

\begin{definition}
	\label{Definition:Future}
  Futures are $\FUT{\apred}{\acom}{\apredp} \defeq \apred \sepimp \wpreOf{\acom}{\apredp}$.
\end{definition}

Readers familiar with Iris~\cite{DBLP:journals/jfp/JungKJBBD18} will note that our definition of futures resembles Iris' notion of Hoare triples.
Technically, Hoare triples in Iris are duplicable resources (they are guarded by a persistence modality) while our futures are not (they may carry resources and are therefore subject to interference).
However, one can directly encode our definition of futures in Iris via $\wpre$.
The key novelty of our approach is the way futures are used, in particular the reasoning technique of accounting and composition as showcased in \Cref{sec:overview}.

\begin{figure*}
	\newcommand{\IMP}{\;\,\subseteq\,\;}
	\begin{mathpar}
		\inferH{f-intro}{
			\apred\subseteq\wpreOf{\acom}{\apredp}
		}{
			\emp \IMP \FUT{\apred}{\acom}{\apredp}
		}
		\and
		\inferH{f-seq}{
			\wpreOf{\acom_1;\acom_2}{\apredpp}\subseteq\wpreOf{\acom}{\apredpp}
		}{
			\FUT{\apred}{\acom_1}{\apredp}\statemult\FUT{\apredp}{\acom_2}{\apredpp} \IMP \FUT{\apred}{\acom}{\apredpp}
		}
		\and 
		\inferH{f-infer}{
			\apred_2\subseteq\apred_1\\
			\apredp_1\subseteq\apredp_2}{
			\FUT{\apred_1}{\acom}{\apredp_1} \IMP \FUT{\apred_2}{\acom}{\apredp_2}
		}
		\and
		\inferH{f-frame}{}{
			\FUT{\apred}{\acom}{\apredp} \IMP \FUT{\apred\mstar\apredpp}{\acom}{\apredp\mstar\apredpp}
		}
		\and
		\inferH{f-account}{}
			{
			\apred\statemult\FUT{\apred\mstar\apredp}{\acom}{\apredpp} \IMP \FUT{\apredp}{\acom}{\apredpp}
		}
		\and
		\inferH{f-invoke}{}{
			\apred\statemult\FUT{\apred}{\acom}{\apredp} \IMP \wpreOf{\acom}{\apredp}
		}
	\end{mathpar}
	\caption{%
		Implications among futures.
		\label{Figure:FutureImplications}
	}
\end{figure*}

\Cref{Figure:FutureImplications} gives the implications we use for reasoning about futures.
Rule \ruleref{f-intro} turns an ordinary Hoare triple that proves the correctness of a command $\acom$ into a future.
Rules \ruleref{f-infer} and \ruleref{f-frame} correspond to rules \ruleref{infer-sem} and \ruleref{frame} for Hoare triples.
Rule \ruleref{f-invoke} allows us to invoke a future $\FUT{\apred}{\acom}{\apredp}$ at the point in the proof where the update chunk $\acom$ is actually executed.
That is, we can use this rule to discharge the premise of rule~\ruleref{com-sem}.

The composition of ghost update chunks is implemented by rule~\ruleref{f-seq}.
It is similar to the rule for sequential composition in Hoare logic with two important differences. 
First, it requires a separating conjunction of the futures for the update chunks $\acom_1$ and $\acom_2$.
The reason is that futures may carry resources.
Second, it replaces the composition of $\acom_1$ and $\acom_2$ by a new update chunk $\acom$ that is equivalent.
Unlike rule~\ruleref{seq}, rule~\ruleref{f-seq} does not take into account interferences on the intermediate assertion $\apred$.
This is correct, since update chunks represent ghost computation that takes effect instantaneously, meaning $\acom_1$ and $\acom_2$ are executed uninterruptedly at the moment when the new update chunk $\acom$ is invoked.

Finally, rule \ruleref{f-account} enables the partial invocation of a future $\FUT{\apred\mstar\apredp}{\acom}{\apredpp}$ by eliminating the premise $\apred$ if it is present in the current proof context.
% We refer to an application of this rule as \emph{accounting}.
We refer to this rule as \emph{accounting}.
We have already seen in \Cref{sec:overview} that accounting is useful to enable the composition of two futures using \ruleref{f-seq}.
Note that if $\apred$ is subject to inference, so is the future $\FUT{\apredp}{\acom}{\apredpp}$ obtained from rule~\ruleref{f-account}. 

% The following lemma states the soundness of the rules.
%
\begin{lemma}
	\label{Lemma:FutureSoundness}
	% Rules \ruleref{f-intro}, \ruleref{f-seq}, \ruleref{f-infer}, \ruleref{f-frame}, \ruleref{f-account}, and \ruleref{f-invoke} are valid implications.
	Rules \ruleref{f-intro}, \ruleref{f-seq}, \ruleref{f-infer}, \ruleref{f-frame}, \ruleref{f-account}, \ruleref{f-invoke} are sound (valid implications).
\end{lemma}

%-------------------------------------------------------------------------------------------------%
%-------------------------------------------------------------------------------------------------%
%-------------------------------------------------------------------------------------------------%
\smartparagraph{Proving the Harris Set Invariant}\label{Section:HarrisFuture}
We demonstrate the versatility of futures by using them to prove that the \code{CAS} on \cref{line:search-cas} of the Harris set preserves the invariant of the data structure.
\Cref{fig-harris-future-reasoning} shows the proof outline.
The code is equivalent to the one in \Cref{fig:harris-list-algo}, except that the mark bits have been made explicit.
We discuss the key aspects of the proof in more detail.

The precondition of \code{traverse} contains the predicate $\tinv(\abscontent,N,M,l,\lnext,\lmark,t)$.
It is the invariant of \code{traverse} and states that the data structure's invariant $\inv(\abscontent,N)$ is maintained.
Additionally, the precondition contains the future $\futharris(M,l,\lnext,t)$.
It captures the fact that the segment from $l$ to $t$ consisting of the nodes $M \setminus \set{l,t}$ can be safely unlinked via the update $\assignof{l.\nextfld}{t}$, provided ${l.\nextfld = \lnext}$ and $\neg \markvarof{l}$.
After the update, the future guarantees that we have ${(\keyvarof{l},\infty] \prall{\subseteq} \flowvarof{t}}$, a crucial fact we will later use in the linearizability proof~(cf. \Cref{Section:History}).
It also guarantees that the contents of the modified segment is not changed by the update.

To satisfy the precondition when invoking \code{traverse} from \code{find}, observe that the invocation is of the form \code{traverse($k,{}$head$,\hnext$, $\hnext$)} where $\hnext$ stems from $\hnext = \head.\nextsel$.
The invariant of \code{traverse}, $\tinv(\abscontent,N,\{\head,\hnext\},\head,\hnext,\hmark,\hnext)$, follows from the data structure invariant $\inv(\abscontent,N)$.
The future $\futharris(M,\head,\hnext,\hnext)$ can be obtained trivially via \ruleref{f-intro} because the update $\assignof{\head.\nextfld}{\hnext}$ has no effect if $\nextvarof{\head} = \hnext$.

The postcondition of \code{traverse} contains the invariant $\tinv(\abscontent,N,M,l,\lnext,\lmark,r)$ and the future $\futharris(M,l,\lnext,r)$.
By applying rule~\ruleref{f-invoke}, we can use the future to prove the correctness of the case where the \code{CAS($l$.next, $\lnext$, $r$)} at \cref{line:search-cas} (\Cref{fig:harris-list-algo}) succeeds.
The remaining facts of the postcondition state that \code{traverse} has found the part of the data structure that contains the search key $k$ if present.

%!TEX root = ../main.tex

\begin{figure}
  \newcommand\MSTAR{\,\mstar\,}
  \newcommand{\leFuture}[1]{\makeFutCol{\ensuremath{#1}}}
\begin{lstlisting}[language=SPL,escapechar=@, belowskip=.1em]
$\makeTeal{\tinv(\abscontent,N,M,l,\lnext,\lmark,t) \defeq \inv(\abscontent,N) \MSTAR \contentsof{M}\subseteq\contentsof{\set{l,t}} \MSTAR \neg \lmark \MSTAR \keyvarof{l} \prall{<} k \prall{<} \infty
 \MSTAR \set{l,\lnext,t} \prall{\subseteq} M \prall{\subseteq} N}$
$\makeTeal{\leFuture{\futharris(M, l, \lnext, t)} \defeq \FUT{\futharrispre(M,l,\lnext,t,\lnext)}{l.\nextfld \assign t}{\futharrispost(M,l, \lnext,t,t)}}$
$\makeTeal{\futharrispre(M, l, \lnext, t, u) \defeq \inv(\contentsof{\set{l,t}},M) \MSTAR \set{l,t} \subseteq M \MSTAR \nextvarof{l} = u \MSTAR \neg \markvarof{l}}$
$\makeTeal{\futharrispost(M, l, \lnext, t, u) \defeq \inv(\contentsof{\set{l,t}},M) \MSTAR \set{l,t} \subseteq M \MSTAR \nextvarof{l} = u \MSTAR \neg \markvarof{l} \MSTAR (\keyvarof{l},\infty] \subseteq \flowvarof{t}}$
\end{lstlisting}
\begin{lstlisting}[language=SPL,escapechar=@, belowskip=0pt]
$\annot{
  \exists N\,M\,\abscontent.\;
  \tinv(\abscontent,N,M,l,\lnext,\lmark,t) \MSTAR \leFuture{\futharris(M,l,\lnext,t)}
}$
procedure traverse($k$: K, $l$: N, $\lnext$: N, $\lmark$: Bool, $t$: N) {
  val $\tnext$, $\tmark$ = atomic {$t$.next, $t$.mark}
  $\annot{
    \tinv(\abscontent,N,M,l,\lnext,\lmark,t) \MSTAR \leFuture{\futharris(M,l,\lnext,t)} \MSTAR \tnext \in N \MSTAR (\tmark \Rightarrow \markvarof{t} = \tmark \MSTAR \nextvarof{t} = \tnext)
  }$
  if ($\tmark$) {
    $\annot{
      \tinv(\abscontent,N,M,l,\lnext,\lmark,t) \MSTAR \leFuture{\futharris(M,l,\lnext,t)} \MSTAR \tnext \in N \MSTAR \markvarof{t} \MSTAR {}
      \markvarof{t} = \tmark \MSTAR \nextvarof{t} = \tnext
    }$ @\label[line]{line-harris-pre-f-seq}@
    $\annot{
      \tinv(\abscontent,N,M,l,\lnext,\lmark,\tnext) \MSTAR \leFuture{\futharris(M,l,\lnext,\tnext)}
    }$ @\label[line]{line-harris-post-f-seq}@
    return traverse($k$, $l$, $\lnext$, $\tnext$)
  } else if ($t$.key < $k$) {
    $\annot{
      \tinv(\abscontent,N,M,l,\lnext,\lmark,\tnext) \MSTAR \leFuture{\futharris(M,l,\lnext,\tnext)} \MSTAR \keyvarof{t} < k
    }$
    return traverse($k$, $t$, $\tnext$, $\tmark$, $\tnext$)
  } else {
    $\annot{
      \tinv(\abscontent,N,M,l,\lnext,\lmark,t) \MSTAR \leFuture{\futharris(M,l,\lnext,t)} \MSTAR t\!\neq\!\head \MSTAR t\!\neq\!l \MSTAR \keyvarof{l}\!<\!k\!\leq\!\keyvarof{t}
    }$ @\label[line]{line-harris-traverse-return}@
    return ($l$, $\lnext$, $\lmark$, $t$) 
} }
$\annot{
  (l,\lnext,\lmark,r).\, \exists N\,M\,\abscontent.\;
  \tinv(\abscontent,N,M,l,\lnext,\lmark,r) \MSTAR \leFuture{\futharris(M,l,\lnext,r)} \MSTAR r\prall{\neq}\head \MSTAR {} r\prall{\neq}l \MSTAR \keyvarof{l}\prall{<}k\prall{\leq}\keyvarof{r}
}$
\end{lstlisting}
\caption{%
  Proof outline showing that Harris set \code{traverse} prepares the \code{CAS} from \code{search}.
  The preparation guarantees that the \code{CAS} \emph{will} maintain the invariant.
  We capture this with a \leFuture{\text{future}}.
  \label{fig-harris-future-reasoning}
}
\end{figure}

%%% Local Variables:
%%% mode: latex
%%% TeX-master: "../main"
%%% End:

The most interesting part of the proof is the transition between \cref{line-harris-pre-f-seq,line-harris-post-f-seq}, particularly the transition from $\futharris(M,l,\lnext,t)$ to $\futharris(M,l,\lnext,\tnext)$.
Here, we need to extend the marked segment $M$ by adding $\tnext$.
This step involves an application of \ruleref{f-seq} to compose the update chunk for $l.\nextfld \assign t$ with the one for $l.\nextfld \assign \tnext$. We elaborate this step in detail, extending the discussion in \Cref{sec:overview}.

We start from $\futharris(M,l,\lnext,t)$ and use rule \ruleref{f-frame} to extend both sides of this future with $\inv(\contentsof{\tnext},\tnext)$.
The resulting future can be rewritten into the form
\begin{gather*}
	\FUTp{\futharrisprep(l,t,\tnext,\lnext) ~\mstar~ \inv(\emptyset,M\setminus\set{l,t,\tnext})}
	{~l.\nextfld \assign t~}
	{\futharrispostp(l,t,\tnext,t) ~\mstar~ \inv(\emptyset,M\setminus\set{l,t,\tnext})}
	% \\\text{where}\quad
	\shortintertext{where}
	\begin{aligned}[t]
		\futharrisprep(l,t,\tnext,u) &\:\defeq\; \inv(\contentsof{\set{l,t,\tnext}},\set{l,t,\tnext}) ~\mstar~ \nextvarof{l} = u ~\mstar~ \neg \markvarof{l}
	\\
		\futharrispostp(l,t,\tnext,u) &\:\defeq\; \inv(\contentsof{\set{l,u,\tnext}},\set{l,t,\tnext}) ~\mstar~ \nextvarof{l} = u ~\mstar~ \neg \markvarof{l} ~\mstar~ {}(\keyvarof{l},\infty] \subseteq \flowvarof{t}
	\ .
	\end{aligned}
\end{gather*}
This future plays the role of $\FUT{\apred}{\acom_1}{\apredp}$ in our application of rule~\ruleref{f-seq}.
To obtain the future playing the role of $\FUT{\apredp}{\acom_2}{\apredpp}$, we proceed in multiple steps.
First, we use \ruleref{f-intro} to derive
\begin{gather*}
	\FUTp{\futharrisprep(l,t,\tnext,t) ~\mstar~ \nextvarof{t} = \tnext \mstar \markvarof{t}}
	{~\assignof{l.\nextfld}{\tnext}~}
	{\futharrispostp(l,t,\tnext,\tnext)}
	\ .
\end{gather*}
To satisfy the premise of \ruleref{f-intro}, we need to show that (i) the update $\assignof{l.\nextfld}{\tnext}$ is frame-preserving, i.e., the interface of the flow graph consisting of the nodes $\{l,t,\tnext\}$ does not change, and (ii) the invariants $\invNode^i(x)$ are preserved for all $x \in \{l,t,\tnext\}$.
First observe that $\neg\markvarof{l}$ and $\invNode^1(l)$ imply $\insetof{l} \neq \emptyset$.
Using $\keyvarof{l}\prall{<}\infty$,\, $\invNode^4(t)$, and the fixpoint equation defining insets, we obtain $\insetof{t}=\insetof{l} \cap (\keyvarof{l},\infty] \neq \emptyset$.
From $\invNode^3(t)$ and $\markvarof{t}$ we obtain $\keyvarof{t} \neq \infty$.
Thus, using similar reasoning as above, we conclude $\insetof{\tnext}=\insetof{t} \cap (\keyvarof{t},\infty] \neq \emptyset$.
The inset and inflow of $l$ are unaffected by the update, so its invariant is trivially preserved.
For $t$, let $\newinsetof{t} = \emptyset$ denote the new inset.
Since $t$ is marked, this means that all its invariants are preserved and its content is empty.
The new inset of $\tnext$ is $\newinsetof{\tnext} = \insetof{l} \cap (\keyvarof{l},\infty]$.
Observe that we have $\insetof{\tnext} \subseteq \newinsetof{\tnext} \neq \emptyset$, so the invariants for $\tnext$ are also maintained.
Finally, to show that the interface of the modified region remains the same, it suffices to prove $\insetof{\tnext} \cap (\keyvarof{\tnext},\infty] = \newinsetof{\tnext} \cap (\keyvarof{\tnext},\infty]$.
This holds true if $\keyvarof{t} < \keyvarof{\tnext}$, which follows from $\invNode^2(\tnext)$ and $\emptyset \neq \insetof{\tnext} \subseteq (\keyvarof{t},\infty]$.

Next, we apply \ruleref{f-account} for $\nextvarof{t} = \tnext \mstar \markvarof{t}$ from the proof context and use \ruleref{f-frame} to add the remaining part $\inv(\emptyset,M \setminus \set{l,t,\tnext})$ of the segment $M$ as a frame.
This yields \[
	\FUTp{\futharrisprep(l,t,\tnext,t) \mstar \inv(\emptyset,M\setminus\set{l,t,\tnext})}
	{~l.\nextfld \assign \tnext~}
	{\futharrispostp(l,t,\tnext,\tnext) \mstar \inv(\emptyset,M\setminus\set{l,t,\tnext})}
	.
\]
We can now use \ruleref{f-seq} with this future and the one derived above.
Note that the premise of the rule follows easily because $\acom_1$ and $\acom_2$ update the same memory location and $\acom_2 = \acom$.
We obtain \[
	\FUTp{\futharrisprep(l,t,\tnext,\lnext) \mstar \inv(\emptyset,M\setminus\set{l,t,\tnext})}
	{~l.\nextfld \assign \tnext~}
	{\futharrispostp(l,t,\tnext,\tnext) \mstar \inv(\emptyset,M\setminus\set{l,t,\tnext})}
	.
\]
Applying \ruleref{f-infer} and introducing a fresh existential $M'$, we rewrite this into the form
\begin{align*}
	M' = M \cup \set{l,t,\tnext} \mstar
	\FUTp{\futharrispre(M'\mkern-2mu,l,\lnext,\tnext,\lnext) \mstar t \prall{\in} M'}
	{\;l.\nextfld \assign \tnext\;}
	{\futharrispost(M'\mkern-2mu,l,\lnext,\tnext,\tnext)}
	.
\end{align*}
Now, we apply \ruleref{f-account} one more time for $t \in M'$ and use $\tnext \in N$, $\lnext \in M$, and $M \subseteq N$ from the proof context, to obtain \[
	M' \subseteq N {~}\mstar{~} \set{l,\lnext,\tnext} \subseteq M' {~}\mstar{~} \futharris(M',l,\lnext,\tnext)
	\enspace.
\]
This allows us to reestablish $\tinv(\abscontent,N,M',l,\lnext,\lmark,\tnext)\mstar\futharris(M',l,\lnext,\tnext)$.
As $M'$ and $M$ are existentially quantified, we can finally rename $M'$ to $M$, which yields the assertion on \cref{line-harris-post-f-seq}.

%-------------------------------------------------------------------------------------------------%
%-------------------------------------------------------------------------------------------------%
%-------------------------------------------------------------------------------------------------%
\smartparagraph{Checking Interference Freedom}
We briefly discuss why the proof is interference-free relative to other threads performing set operations.
First, all commands maintain $\inv(\abscontent,N)$ and $N$ can only grow larger.
Next, assertions depending on the field \code{key} are interference-free since a node's $\keyfld$ is never changed after initialization.
Similarly, $\markfld$ is only changed monotonically from $\false$ to $\true$.
Moreover, $\nextfld$ is only changed for unmarked nodes (e.g., the proof guarantees $\neg \markvarof{l}$ in the successful case of the \code{CAS} on \cref{line:search-cas} and the insert operation provides a similar guarantee).
This is why assertions such as $\markvarof{t}$ on \cref{line-harris-pre-f-seq} are interference-free.
Because of that, the contents $\contentsof{M \setminus \{l,t\}}$ cannot change.
Finally, futures constructed using rule~\ruleref{f-intro} are always interference-free.
All remaining futures are constructed by accounting interference-free facts or by composing interference-free futures via \ruleref{f-seq}.

%!TEX root = ../main.tex

\newcommand{\iset}{\mathit{I}}
\newcommand{\isetp}{\mathit{J}}
\newcommand{\anindex}{\mathit{i}}
\newcommand{\anindexp}{\mathit{j}}

\section{Histories}
\label{Section:History}

We next present an extension of our developed theory that allows us to reason about \emph{separated computation histories}.
We integrate a form of hindsight reasoning for propagating knowledge between current and past states---hindsight is a key technique to handle non-fixed linearization points~\cite{DBLP:conf/podc/OHearnRVYY10,DBLP:conf/wdag/Lev-AriCK15,DBLP:conf/wdag/FeldmanE0RS18,DBLP:journals/pacmpl/FeldmanKE0NRS20}.
We develop the new theory again in the general setting of \cref{Section:Preliminaries} and \Cref{Section:OG} and then apply it to our running~example.

%-------------------------------------------------------------------------------------------------%
%-------------------------------------------------------------------------------------------------%
%-------------------------------------------------------------------------------------------------%
\smartparagraph{Read-and-validate Pattern}
Optimistic implementations commonly have future-dependent linearization points: whether or not a thread's next action is its linearization point depends on future interferences from other threads.
This issue often arises in uses of the \emph{read-and-validate pattern} where threads
\begin{inparaenum}[(i)]
	\item read out some shared heap region,
	\item later on validate the read value, and
	\item succeed with their operation if the validation succeeds or roll-back otherwise.
\end{inparaenum}
The read and validation step are neither executed atomically nor within a critical/lock-protected section.
Hence, the read heap region is subject to interference and may change.

The Harris set employs the read-and-validate pattern.
Method \code{find}, for instance, validates (some of) the values read by \code{traverse} by re-reading them with the \code{CAS} on \cref{line:search-cas}.
If the \code{CAS} fails, so does the validation and \code{find} rolls back by restarting.
Method \code{traverse} implements the pattern with a more intricate validation and roll-back mechanism:
the \code{next} field of node $t$ read on \cref{line:traverse-read-next} is validated by inspecting $t$'s mark bit, \cref{line:traverse-check-mark}, and if the mark bit is set then the validation of \code{traverse}'s search for the right node fails and continues with the subsequent nodes.

%!TEX root = ../main.tex

\begin{figure}
	\newcommand\MSTAR{\,\mstar\,}
	\begin{minipage}[t]{.25\textwidth}
	\begin{lstlisting}[language=SPL,belowskip=0pt]
val $\vlen$ = /* last index */
val $\varr$ = new Int[$\vlen$+1] { 0, ..., 0 }
\end{lstlisting}
	\begin{lstlisting}[language=SPL,belowskip=0pt]
procedure copy() {
	val $\vres$ = new Int[$\vlen$] { 0, ..., 0 }
	$\annot{
		\varr \mapsto i_0,\dots,i_{\vlen}
		\MSTAR
		\vres \mapsto j_0,\dots,j_{\vlen}
		\MSTAR
		\smash{\textstyle \bigwedge_{n=0}^{\vlen} j_n \leq i_n}
	}$ @\label{code:array:pre-copy}@
	for (val $k$ in [0, $\vlen$]) $\vres$[$k$] = $\varr$[$k$] @\label{code:array:the-copy}@
	$\annot{
		\varr \mapsto i_0,\dots,i_{\vlen}
		\MSTAR
		\vres \mapsto j_0,\dots,j_{\vlen}
		\MSTAR
		\smash{\textstyle \bigwedge_{n=0}^{\vlen} j_n \leq i_n}
	}$@\label{code:array:post-copy}@ @\smash{\textcolor{teal}{skip}}@ @\label{code:array:skip}@
	$\annot{
		\varr \mapsto i_0,\dots,i_{\vlen}
		\MSTAR
		\vres \mapsto j_0,\dots,j_{\vlen}
		\MSTAR
		\pastOf{(\varr \mapsto \old{i}_0,\dots,\old{i}_{\vlen})}
		\MSTAR
		\smash{\textstyle \bigwedge_{n=0}^{\vlen} j_n \leq \old{i}_n \leq i_n}
	}$ @\label{code:array:pre-validate}@
	for (val $k$ in [0, $\vlen$]) { @\label{code:array:validate-loop}@
		$\annot{
			\varr \mapsto i_0,\dots,i_{\vlen}
			\MSTAR
			\vres \mapsto j_0,\dots,j_{\vlen}
			\MSTAR
			\pastOf{(\varr \mapsto \old{i}_0,\dots,\old{i}_{\vlen})}
			\MSTAR
			\smash{\textstyle \bigwedge_{n=0}^{\vlen} j_n \leq \old{i}_n \leq i_n}
			\MSTAR
			\smash{\textstyle \bigwedge_{n=0}^{k-1} j_n=\old{i}_n}
		}$ @\label{code:array:validate-inv}@
		if ($\vres$[$k$] != $\varr$[$k$]) restart @\label{code:array:validate-check}@
		$\annot{
			\varr \mapsto i_0,\dots,i_{\vlen}
			\MSTAR
			\vres \mapsto j_0,\dots,j_{\vlen}
			\MSTAR
			\pastOf{(\varr \mapsto \old{i}_0,\dots,\old{i}_{\vlen})}
			\MSTAR
			\smash{\textstyle \bigwedge_{n=0}^{\vlen} j_n \leq \old{i}_n \leq i_n}
			\MSTAR
			\smash{\textstyle \bigwedge_{n=0}^{k\phantom{-1}} j_n=\old{i}_n}
		}$ @\label{code:array:validate-post-check}@
	}
	$\annot{
		\varr \mapsto i_0,\dots,i_{\vlen}
		\MSTAR
		\vres \mapsto \old{i}_0,\dots,\old{i}_{\vlen}
		\MSTAR
		\pastOf{(\varr \mapsto \old{i}_0,\dots,\old{i}_{\vlen})}
	}$ @\label{code:array:post-validate}@
	return $\vres$
}
\end{lstlisting}
	\end{minipage}
	\hfill
	\begin{minipage}[t]{.40\textwidth}
	\begin{lstlisting}[language=SPL,belowskip=0pt]
procedure inc($k$: Int) {
	if ($k$ > $\vlen$) return false @\label{code:array:bounds-check}@
	FAA($\varr$[$k$], 1) @\label{code:array:faa}@
	return true
}
\end{lstlisting}
\end{minipage}
	\caption{%
		A simple array of counters.
		The counters can be incremented individually.
		The entire array can be snapshot in an optimistic fashion, resulting in a non-fixed linearization point.
		\label{fig:optimistic-array}
	}
\end{figure}

\looseness=-1
The pattern is not restricted to search structures.
For an example, consider the counter array from \Cref{fig:optimistic-array}.
The implementation maintains an array $\varr$ of $\vlen+1$ integer counters.
Counters are individually increased by $1$ using \code{inc}.
A snapshot of the counter array is created by \code{copy}.
As a first stage, a simple copy $\vres$ of the array is created by reading out the individual counter values non-atomically, \cref{code:array:the-copy}.
In a second stage, the copy is validated against the current counter values, \crefrange{code:array:validate-loop}{code:array:validate-check}.
The procedure is restarted if there are any discrepancies.
Otherwise, the copy is guaranteed to be a consistent snapshot of the counter array as of the moment immediately after the last counter was read on \cref{code:array:the-copy}.
It is worth noting that, while the copy is being validated, some counters that have been validated already may be changed.
Nevertheless, the validation succeeds (rightfully so).

The read-and-validate pattern in \code{copy} results in a future-dependent linearization point.
As alluded to above, the linearization point is the moment the last counter is read on \cref{code:array:the-copy}.
However, this moment depends on whether the subsequent validation \emph{will succeed}.
This, in turn, is unpredictable and not under the control of the executing thread.

To handle this in a proof, we suggest the following strategy which mimics closely the behavior of the implementation.
During the first stage, we track the interference-free fact that the entries of the copy array $\vres$ are less than or equal to the current value of the corresponding counter, $j_n \leq i_n$ for all $n$.
This fact follows easily from the counters increasing monotonically.
Then, we snapshot the current counter array into a past predicate, \cref{code:array:pre-validate}.
(Technically, this requires the \code{skip} on \cref{code:array:skip}, cf. \Cref{Lemma:CopyToPast} below.)
Note that we rename the counter values $i_n$ under the past predicate to $\old{i}_n$.
Because $i_n = \old{i}_n$ at the moment of the snapshot, we obtain $j_n \leq \old{i}_n$ for all $n$.
However, the equality $i_n = \old{i}_n$ is not interference-free as the counters may change, but the point in time and thus the values the past predicate refers to are fixed.
This justifies our renaming to $\old{i}_n$.
As before, we record that the counter values under the past predicate are less than or equal to the current counter values, $\old{i}_n \leq i_n$ for all $n$.
During the second stage, a successful validation implies that the $k$-th copy is equal to the current counter, $j_k=i_k$.
Together with the estimate $j_k \leq \old{i}_k \leq i_k$, we obtain $j_k = \old{i}_k$.
Overall, this means that the copy $\vres$ corresponds to the counter array snapshot in the past predicate, \cref{code:array:post-validate}.
Hence, $\vres$ is a consistent snapshot of the counters in the sense that there was a point in time where the counter array was equal to $\vres$---the operation is linearizable as desired.
In the following, we formalize past predicates and show their usefulness for linearizability~proofs.

%-------------------------------------------------------------------------------------------------%
%-------------------------------------------------------------------------------------------------%
%-------------------------------------------------------------------------------------------------%
\smartparagraph{History Separation Algebras}
\looseness=-1
Recall that our states are taken from a separation algebra $(\setstates, \statemult, \emp)$.
We refer to a non-empty sequence of states $\astateseq \in \setstates^+$ as a \emph{computation history}.
Computation histories also form a separation algebra by lifting the composition on states as follows. 
First, for sequences $\astateseq=\astate_1 \ccc \astate_n$ and $\astateseqp=\astatep_1 \ccc \astatep_m$ with $\cdot$ denoting sequence concatenation, the composition $\astateseq * \astateseqp$ is defined, written $\astateseq \statemultdef \astateseqp$, iff $n=m$ and for all $i$ we have $\astate_i\statemultdef\astatep_i$.
In this case, we let $\astateseq \statemult \astateseqp \defeq (\astate_1\statemult\astatep_1) \ccc (\astate_n\statemult \astatep_n)$.
The set of units is given by $\emp^+$.
\begin{lemma}
  $(\setstates^+, \mstar, \emp^+)$ is a separation algebra.
\end{lemma}

Predicates $\acpred, \acpredp, \acpredpp\subseteq\setstates^+$ now refer to sets of computations.
We lift the semantics of commands to computation predicates in the expected way.
\begin{definition}
	$\csemOf{\acom}(\astateseq.\astate_1)\defeq\setcond{\astateseq\cc\astate_1\cc\astate_2}{\astate_2\in\semOf{\acom}(\astate_1)}$.
\end{definition}

However, the locality assumption~\eqref{cond-loccom} on the semantics of commands that is needed for the soundness of framing does not necessarily carry over from state predicates to computation predicates.
If we want to frame computation predicates, we have to make an additional assumption.
\begin{definition}\label{Definition:Frameable}
	A predicate $\acpred\subseteq\setstates^+$ is \emph{frameable}, if it satisfies $\; \forall \astateseq.\forall\astate.\;\;\astateseq\cc\astate\in\acpred\;\Rightarrow\;\astateseq\cc\astate\cc\astate\in \acpred.$ 
\end{definition}
\begin{lemma}\label{Lemma:Frameable}
If $\acpredpp$ is frameable, $\csemOf{\acom}(\acpred\prall{\mstar}\acpredpp)\prall{\subseteq}\csemOf{\acom}(\acpred)\prall{\mstar}\acpredpp$. 
\end{lemma} 

We lift the semantics of concurrency libraries to history separation algebras $(\setshared\times\setlocal)^+$.
The notions of initial and accepting configurations as well as soundness remain unchanged except that they now range over computation predicates instead of state predicates.
The technical details of this lifting are straightforward\techreport{, we defer them to Appendix~\ref{Section:HistoryProofs}}.
The soundness guarantee in \Cref{Theorem:Soundness} continues to hold for history separation algebras modulo a subtlety.
We can only apply rule~\ruleref{frame} if the predicate to be added is frameable in the sense of \Cref{Definition:Frameable}.
\begin{theorem}[Soundness]
	\label{Theorem:SoundnessComput}
	$\thePredicates, \theInterference\semCalc\hoareOf{\acpred}{\astmt}{\acpredp}$
	and $\,\isInterferenceFreeOf[\theInterference]{\thePredicates}$
	and $\acpred\in\thePredicates$
	imply $\,\subModels\hoareOf{\acpred}{\astmt}{\acpredp}$.
\end{theorem}

%-------------------------------------------------------------------------------------------------%
%-------------------------------------------------------------------------------------------------%
%-------------------------------------------------------------------------------------------------%
\smartparagraph{Frameable Computation Predicates}
We next discuss general principles for constructing frameable computation predicates from state predicates.
\begin{definition}
	A state predicate $\apred\subseteq \setstates$ yields the following predicates over computations histories:
	\begin{compactenum}[(i)]
		\item The \emph{now predicate} $\nowOf{\apred}\;\defeq\;\setstates^*\cc\apred$.
		\item The \emph{past predicate} $\pastOf{\apred}\;\defeq\; \setstates^*\cc\apred\cc\setstates^*$. 
	\end{compactenum}
\end{definition}

The now predicate refers to the current state. 
The past predicate allows us to track auxiliary information about the computation. 
These predicates work well in our setting in that they are~frameable. 
\begin{lemma}\label{Lemma:FrameablePredicates}
	\begin{inparaenum}[(i)]
		\item $\nowOf{\apred}$ and $\pastOf{\apred}$ are frameable. 
		\item If $\acpred$ and $\acpredp$ are frameable, so are $\acpred\mstar\acpredp, \acpred\cap\acpredp$, and $\acpred\cup\acpredp$.
	\end{inparaenum}
\end{lemma} 

Frameability is not preserved under complementation and separating implication.
However, the now operator is compatible with the SL operators in a strong sense. 
\begin{lemma}\label{Lemma:NowSLOperators}
	$\nowOf{(\apred\anop\apredp)} = \nowOf{\apred}\anop\nowOf{\apredp}$ for all $\anop\in\setnd{\cap, \cup, \mstar,\sepimp}$,
	$\nowOf{(\overline{\apred})} {\:=\:} \overline{\nowOf{\apred}}$,\,
	$\false {\:=\:} \nowOf{\false}$,\,
	$\true {\:=\:} \nowOf{\true}$, and
	$\nowOf{\apred}{\:\subseteq\:}\nowOf{\apredp}$ iff $\apred{\:\subseteq\:}\apredp$. 
\end{lemma}

For the past operator, we rely on the properties stated by the following lemma.
In particular, the last equivalence justifies rule \ruleref{h-infer} used in \Cref{sec:overview}.
\begin{lemma}\label{Lemma:PastSLOperators}
	$\nowOf{\apred} {\,\subseteq\,} \pastOf{\apred}$,\,
	$\true = \pastOf{\true}$,\,
	$\true\mstar\pastOf{\apred} = \pastOf{(\apred\mstar\true)}$,\,
	$\false = \pastOf{\false}$,\,
	$\pastOf{(\apred\mstar\apredp)} \subseteq \pastOf{\apred}\mstar\pastOf{\apredp}$,\,
	$\pastOf{\apred}\sepimp\pastOf{\apredp} \subseteq \pastOf{(\apred\sepimp\apredp)}$,\,
	$\pastOf{(\apred\cap\apredp)} \subseteq \pastOf{\apred}\cap\pastOf{\apredp}$,\,
	$\pastOf{(\apred\cup\apredp)} = \pastOf{\apred}\cup\pastOf{\apredp}$, and
	$\pastOf{\apred}\subseteq\pastOf{\apredp}$ iff $\apred\subseteq\apredp$.
\end{lemma}

The interplay between computation predicates and commands is stated in the following lemma.
Recall that we defined $\wpreOf{\acom}{\acpred}\defeq\setcond{\astateseq}{\csemOf{\acom}(\astateseq)\subseteq\acpred}$.
\begin{lemma}\label{Lemma:WPNowPast}
	We have
	\begin{inparaenum}[(i)]
		\item $\wpreOf{\acom}{\nowOf{\apred}} = \nowOf{\wpreOf{\acom}{\apred}}$, and
		\item $\wpreOf{\acom}{\pastOf{\apred}}=\pastOf{\apred}\cup \wpreOf{\acom}{\nowOf{\apred}}$.
	\end{inparaenum}
\end{lemma}

The first identity of \Cref{Lemma:WPNowPast} implies that interference checking for a now predicate reduces to inference checking for the underlying state predicate. The second identity implies that past predicates are interference-free for all commands.

Next we justify rule \ruleref{h-intro} used in \Cref{sec:overview}.
Recall that this rule provides a way to record information about the current state in a past predicate so that we can use this information later in the proof. This involves a stuttering step.
\begin{lemma}
	\label{Lemma:CopyToPast}
	$\nowOf{\apred}\subseteq\wpreOf{\cskip}{\nowOf{\apred}\mstar\pastOf{\apred}}$.
\end{lemma}

%-------------------------------------------------------------------------------------------------%
%-------------------------------------------------------------------------------------------------%
%-------------------------------------------------------------------------------------------------%
\smartparagraph{Hindsight Reasoning}
We now use history separation algebras to justify the hindsight reasoning principle introduced in \Cref{sec:overview}.
The key idea is \emph{state-independent quantification}, and best explained with reference to an assertion language.
An assertion language over computations will support quantified logical variables.
As those quantifiers live on the level of computations, the resulting valuation of the logical variables will be independent of (the same for all) the states inside the computation.
This means facts that we learn about the variables in one state will also be true in all other states.
In particular, if we learn facts about a quantified variable now, we can draw conclusions in hindsight.
We illustrate this on an example. 
In the assertion $\exists v.(\pastOf{(x\pointsto{} v)}\mstar \nowOf{(v=0)})$, the logical variable $v$ is quantified on the level of computations, meaning its value is independent of the actual state in the computation.
We learn that now $v$ is zero, and since the valuation is state independent, $v$ has also been zero when $x$ pointed to it. 
Hence, from the now state we can conclude, in hindsight, that also $\pastOf{(x\pointsto{} 0)}$ holds. 
This is indeed a consequence of the previous assertion (entailment holds). 
Rather than moving to an assertion language, we formalize this reasoning on the semantic level.

For hindsight reasoning, we construct a product separation algebra $\setstates\times \iset$. 
The first component is the above state separation algebra $(\setstates, \mstar,\emp)$.
We refer to the second component as the \emph{valuation} separation algebra $(\iset, \mstar, \iset)$, because its elements can be understood as variable valuations.
There are no requirements on $\iset$, it is just an arbitrary set, but we note that it is also the set of units.
The multiplication $\mstar$ between valuations $i,j\in \iset$ is defined if and only if $i=j$, in which case $i\mstar i = i$.

The semantics of commands $\semCom{\acom}$ is lifted to predicates of the product separation algebra by leaving the valuation component untouched.
We then lift $\setstates \times \iset$ to a separation algebra of computation histories $(\setstates \times \iset)^+$ as before.
We adapt the definition of $\nowOf{\apred}$ and $\pastOf{\apred}$ for $\apred\subseteq(\setstates \times \iset)^+$ so that the valuation component is kept constant over the whole computation, in accordance with the lifted semantics of commands:
\[
	\nowOf{\apred}\:\defeq\: \setcond{(\setstates \times \set{\anindex})^*\cdot (\astate,\anindex)}{(\astate,\anindex) \in \apred}
	\quad~
	\pastOf{\apred}\:\defeq\: \setcond{(\setstates \times \set{\anindex})^*\cdot (\astate,\anindex) \cdot (\setstates \times \set{\anindex})^*}{(\astate,\anindex) \in \apred}
	\ .
\]
Separation logic assertions are called pure, if they are independent of the heap and only refer to the valuation of logical variables.
In the semantic setting, we define a predicate $\apred$ to be \emph{pure}, if it leaves the heap unconstrained, $\apred = \setstates \times \isetp$ for some $\isetp \subseteq \iset$.
The following \namecref{thm:soundness-hindsight-rules} then justifies rule \ruleref{h-hindsight} used in \Cref{sec:overview}.
\begin{lemma}
	\label{thm:soundness-hindsight-rules}
  If $\apred$ is pure, then $\nowOf{\apred} \mstar \pastOf{\apredp} \;=\; \pastOf{(\apred \mstar \apredp)}$.
\end{lemma}

%-------------------------------------------------------------------------------------------------%
%-------------------------------------------------------------------------------------------------%
%-------------------------------------------------------------------------------------------------%
\smartparagraph{Linearizability in Hindsight}
Past predicates together with the hindsight principle have an appealing application in linearizability proofs: they allow for retrospective linearization.
That is, the proof does not need to observe the very moment when a thread executes the linearization point.
It suffices to show that the linearization point must have occurred in the past, after the remainder of the thread's execution has been observed.
This is inspired by concurrent data structure practice, like the read-and-validate pattern from before.
To support this kind of linearizability argument, we extend the proof system $\semcalclin$ from \Cref{Section:OG} with rule~\ruleref{lin-past} from \Cref{fig:linearizability-rules}, repeated here for~convenience: \[
	\inferhref{lin-past}{lin-past-repeat}{
		\acpred\subseteq\pastOf{\bigl(\acss(\abscontent)\cap\acssup(\abscontent,\abscontent,k,v)\bigr)}
	}{
		\thePredicates, \theInterference
		\semcalclin
		\hoareOf{\anobl{\asspec}\mstar\acpred}{\cskip}{\aful{\asspec}{v}\mstar\acpred}
	}\ .
\]
The rule formalizes our intuition. It trades the update token $\anobl{\asspec}$ for $\aful{\asspec}{v}$ if a past predicate can certify that the sequential specification was satisfied at some point during the computation, i.e., a linearization point definitely occurred.
It is worth stressing that the way past predicates are introduced (cf. \Cref{Lemma:PastSLOperators,Lemma:WPNowPast}) guarantees that they refer to a moment during the execution of the corresponding operation, as required for linearizability.
Also observe that the precondition of the rule requires the linearization point to be pure.
With the restriction to pure linearization points, we avoid the complexity of ensuring that there is a one-to-one correspondence between updates of the logical contents of the structure and threads \emph{claiming} an update as their linearization point.
Put differently, our rule exploits the fact that arbitrarily many threads may linearize in a single state that satisfies a pure case of the sequential specification.

%-------------------------------------------------------------------------------------------------%
%-------------------------------------------------------------------------------------------------%
%-------------------------------------------------------------------------------------------------%
\smartparagraph{Proving Linearizability of the Harris Set}
By using the now predicate, we obtain a conservative extension of separation logic.
In the following, the application of the now operator is kept implicit: a state predicate that occurs in a context expecting a computation predicate is interpreted as a now predicate.
This is justified by \Cref{Lemma:NowSLOperators}.

We demonstrate the reasoning power of the resulting logic by proving linearizability of the Harris set \code{search} operation.
The proof outline is in \Cref{fig-harris-history-reasoning}, reusing our earlier proof of \code{traverse}.
The code of \code{find} makes the semantics of the \code{CAS} explicit.
We have also eliminated the case $\lnext == r$ before the \code{CAS}.
This focuses the discussion on a single linearization point.
\techreport{Appendix~\ref{Section:HarrisHistoryFull} gives the proof for the full version.}

As in \Cref{sec:overview}, we decorate logical variables occurring below a past operator with a prime to consistently rename existentially quantified variables in order to avoid clashes with variables describing the current state.
For example, $\old{\nextvarof{x}}$ will refer to the old value of $x$'s $\nextfld$ field in some past state.

%!TEX root = ../main.tex

\begin{figure}
  \newcommand\MSTAR{\,\mstar\,}
  \newcommand{\lePast}[1]{\makeHistCol{#1}}
  \newcommand{\lePastAnnot}[2][]{\lePast{\ifthenelse{\equal{#1}{}}{}{\givename{#1}}\left\{\,\begin{aligned}#2\end{aligned}\,\right\}}}
  \begin{minipage}[b]{.55\textwidth}
\begin{lstlisting}[language=SPL,escapechar=@, belowskip=.1em]
$\lePast{\pastharris(r) \defeq \keyvarof{r} = \old{\keyvarof{r}} \MSTAR (\old{\markvarof{r}} \Rightarrow \markvarof{r}) \MSTAR \past(\, \inv(\old{\abscontent},\old{\somenodes}) \MSTAR r\in\old{\somenodes} \MSTAR \old{(k \in \keysetof{r})} \,)}$
\end{lstlisting}
\begin{lstlisting}[language=SPL,escapechar=@, belowskip=0pt]
$\annot{
  \exists \abscontent\,\somenodes.~
  \inv(\abscontent,\somenodes) \MSTAR -\infty < k < \infty
}$
procedure find($k$: K) : N * N {
  val $\hnext$, $\hmark$ = atomic {head.next, head.mark}
  $\annot{
    \tinv(\abscontent,\somenodes,\{\head,\hnext\},\head,\hnext,\hmark,\hnext) \MSTAR \futharris(\{\head,\hnext\},\head,\hnext,\hnext)
  }$
  val $l$, $\lnext$, $\lmark$, $r$ = traverse($k$, head, $\hnext$, $\hmark$, $\hnext$) @\label[line]{line-search-prepare-traverse}@
  $\annot{
    \tinv(\abscontent,\somenodes,M,l,\lnext,\lmark,r) \MSTAR \futharris(M, l, \lnext, r)  \MSTAR r\!\neq\!\head \MSTAR r\!\neq\!l \MSTAR \keyvarof{l}\!<\!k\!\leq\!\keyvarof{r}
  }$
  val $\vsucc$ = atomic { // CAS
    $l$.next == $\lnext$ && $l$.mark == $\lmark$ ? {
      $\annot{
        &\inv(\abscontent,\somenodes) \MSTAR \futharris(M, l, \lnext, r) \MSTAR \contentsof{M}\subseteq\contentsof{l,r} \MSTAR
        \nextvarof{l} = \lnext \\& \MSTAR \neg \markvarof{l}\MSTAR \set{l,\lnext,r} \subseteq M \subseteq \somenodes \MSTAR \keyvarof{l} < k \leq \keyvarof{r}
      }$ @\label[line]{line-search-pre-invoke-future}@  @\label[line]{line-search-pre-case-history}@
      $l$.next := $r$ @\label[line]{line-search-invoke-future}@
      $\annot{
        &\inv(\abscontent,\somenodes) \MSTAR \set{l,r} \subseteq \somenodes \MSTAR (\keyvarof{l},\infty] \subseteq \flowvarof{r} \\& \MSTAR {} \keyvarof{l} < k \leq \keyvarof{r}
      }$ @\label[line]{line-search-post-cas-history}@
      $\lePast{\mathtt{skip}}$
      $\lePastAnnot{
        &\inv(\abscontent,\somenodes) \MSTAR \set{l,r} \subseteq \somenodes \MSTAR (\keyvarof{l},\infty] \subseteq \flowvarof{r} \\& \MSTAR {}
        \keyvarof{l} < k \leq \keyvarof{r}
        \MSTAR \lePast{ \pastharris(r)}
      }$ @\label[line]{line-search-cas-history}@
    } true : false
  } 
  $\annot{
    \inv(\abscontent,\somenodes) \MSTAR \set{l,r} \subseteq \somenodes \MSTAR \lePast{ (\vsucc \,\Rightarrow\, \pastharris(r))}
  }$
  if ($\vsucc$ && !$r$.mark) { @\label[line]{line-search-condition}@
    $\annot{
      \inv(\abscontent,\somenodes) \MSTAR \set{l,r} \subseteq \somenodes \MSTAR \lePast{ \neg\old{\markvarof{r}} \MSTAR \pastharris(r)}
    }$  @\label[line]{line-search-success}@
    return ($l$, $r$)
  } else find($k$)
}
$\annot{
  (l,r).\; \exists \abscontent\,\somenodes.~
  \inv(\abscontent,\somenodes) {\:\mstar\:} \set{l,r} \prall{\subseteq} \somenodes {\:\mstar\:} \lePast{ \neg\old{\markvarof{r}} {\:\mstar\:} \pastharris(r)}
}$ @\label[line]{line-serach-post}@
\end{lstlisting}
  \end{minipage}
  \hfill
  \begin{minipage}[b]{.425\textwidth}
\newcommand{\myObl}{\OBL{search($k$)}}
\newcommand{\myFul}{\FUL{search($k$)}{\vres}}
\begin{lstlisting}[language=SPL,escapechar=@, belowskip=0pt]
$\annot{
  \exists \abscontent\,\somenodes.~
  &\inv(\abscontent,\somenodes) \MSTAR -\infty < k < \infty \\& \MSTAR {} \lePast{\myObl}
}$ @\label[line]{line-find-obl}@
procedure search($k$: K) : Bool {
  val _, $r$ = find($k$) @\label[line]{line-search-invokefind}@
  $\annot{
    &\inv(\abscontent,\somenodes) \MSTAR r \in \somenodes \MSTAR \lePast{ \neg\old{\markvarof{r}} } \\& \lePast{ \MSTAR {} \pastharris(r) \MSTAR \myObl}
  }$ @\label[line]{line-search-called}@
  val $\vres$ = $r$.key == $k$
  $\annot{
    &\inv(\abscontent,\somenodes) \MSTAR r \in \somenodes \MSTAR \lePast{ \myObl } \\&\strut \lePast{\smash{\MSTAR \past( \inv(\old{\abscontent},\old{\somenodes}) \mstar \vres\Leftrightarrow k \prall{\in} \old{\abscontent} )}}
  }$ @\label[line]{line-find-flowreasoning}@
  $\annot{
    \inv(\abscontent,\somenodes) \MSTAR \lePast{\myFul}
  }$ @\label[line]{line-find-ful}@
  return $\vres$
}
$\annot{
  \vres.\; \exists \abscontent\,\somenodes.~ \inv(\abscontent,\somenodes) \MSTAR \lePast{\myFul}
}$
\end{lstlisting}
  \end{minipage}
\caption{%
  Proof outline showing that Harris set \code{search} preserves the \makeTeal{\text{data structure invariant}} and is \lePast{\text{linearizable}}.
  \label{fig-harris-history-reasoning}
}
\end{figure}

We proceed with the proof.
First, we focus on the overall linearizability argument in \code{search}.
To that end, observe the history assertion resulting from a call to \code{find}, \cref{line-search-called}:
\[
	\pastOf{\bigl(\, \inv(\old{\abscontent},\old{\somenodes}) \:\mstar\: r\in\old{\somenodes} \:\mstar\: k \in \old{\keysetof{r}} \bigr)}
	\:\mstar\:
	\neg\old{\markvarof{r}}
	\:\mstar\:
	\keyvarof{r} \prall{=} \old{\keyvarof{r}}
	\enspace .
\]
Following our discussion from~\Cref{sec:overview}, the keysets guide the proof in the sense that $k \in \old{\keysetof{r}}$ means that $r$ is the decisive node for \code{search($k$)}.
We localize the reasoning to this decisive node by rewriting the invariant $\inv(\old{\abscontent},\old{\somenodes})$ under the past operator along \Cref{thm:invariant-locality}:\footnote{%
	Compared to \Cref{sec:overview}, we record here the full invariant $\inv(\old{\abscontent},\old{\somenodes})$ under the past operator, not just the predicate $\nodeof{r}$.
	This is needed to be compatible with the sequential specification and rule~\ruleref{lin-past}.
}
\[
	\inv(\,\old{\abscontent}\setminus\set{k},\:\old{\somenodes}\setminus\set{r}\,) ~\mstar~ \inv(\,\old{\contentsof{r}},\:\set{r}\,)
	\enspace .
\]
The \namecref{thm:invariant-locality} is applicable as its precondition $\contentsof{\old{\somenodes}\prall{\setminus}\set{r}}\cap\set{k}=\emptyset$ follows from contraposition: if there was a non-$r$ node with $k$ in its contents, then $k$ was also in its keyset by the invariant, which contradicts the disjointness of keysets because $k\in\old{\keysetof{r}}$.
So we get $k\in\old{\contentsof{r}}$ iff $k=\old{\keyvarof{r}}$.
By the above localization, this means $k\in\old{\abscontent}$ iff $k=\old{\keyvarof{r}}$.
Because the key of $r$ has not changed, $\keyvarof{r} \prall{=} \old{\keyvarof{r}}$, we arrive at $k\in\old{\abscontent}$ iff $k=\keyvarof{r}$.
That is, the current value of \code{$r$.key} reveals whether or not $k$ has been in the contents of the data structure at some point in the past.
With this, it is immediate that the past predicate certifies the existence of a linearization point for the usual sequential specification of \code{search}: \[
	\hoareOf{ \abscontent.~ \acss(\abscontent) }{ \mymathtt{search}(k) }{ \vres.~ \acss(\abscontent) \:\mstar\: \vres \Leftrightarrow k\in\abscontent }
	\enspace .
\]
We know that $\vres$ is equal to the truth of the equality $k=\keyvarof{r}$.
Hence, we can linearize in hindsight using rule~\ruleref{lin-past} from above.

To prove the postcondition of \code{find}, we need to establish that at some point during the execution of \code{find}, $k \in \keysetof{r}$ and $\neg\markvarof{r}$ were satisfied.
This is on \cref{line-search-post-cas-history}, but we determine $r$'s mark bit only later as execution continues past the condition on \cref{line-search-condition}.
The high-level proof idea is, thus, to record $r$'s keyset and mark bit on \cref{line-search-post-cas-history} in the past assertion shown on \cref{line-search-cas-history} (we explain this step in more detail below).
We then propagate the assertion on \cref{line-search-cas-history} into the \emph{then} branch of the conditional (using \Cref{Lemma:WPNowPast}). The logical variables link the past and current state, which allows us to apply hindsight reasoning (rule \ruleref{h-hindsight}).
Specifically, on \cref{line-search-success} we know that $\markvarof{r}$ is false in the current state.
Then, using $\old{\markvarof{r}} \prall{\Rightarrow} \markvarof{r}$, we can conclude that $\old{\markvarof{r}}$ is also false, thus learning retrospectively the crucial fact about the past state at \cref{line-search-post-cas-history}.
We then transfer this pure fact into the past predicate to derive $\contentsof{r}=\set{\keyvarof{r}}$ and thus $\vres\Leftrightarrow k\in\abscontentp$ on \cref{line-find-flowreasoning}.

We briefly discuss the mechanical aspects of deriving the past predicate on \cref{line-search-cas-history}.
We start with the intermediate assertion established on \cref{line-search-post-cas-history} in the earlier proof.
First, using $\keyvarof{l} < k$ and ${(\keyvarof{l},\infty] \subseteq \flowvarof{r}}$ we derive $k \in \insetof{r}$.
Then, using $k \leq \keyvarof{r}$ we obtain $k \in \keysetof{r}$.
We then perform a stuttering step to record a copy of this assertion below a past operator using rule~\ruleref{h-intro}.
The resulting assertion is not interference-free since the current state of $r$ referred to inside the past assertion can be changed by concurrent threads.
So we perform a series of weakening steps
by introducing fresh logical variables to arrive at the assertion on \cref{line-search-cas-history}.

It is worth pointing out that the proof of \code{find} relies on the future constructed in \Cref{Section:Futures} and the fact that its update is pure, \cref{line-search-invoke-future}.
To see this, we rewrite the assertion from \cref{line-search-pre-invoke-future} along \Cref{thm:invariant-locality}:
\[
	\inv(\abscontent\setminus\contentsof{\set{l,r}},\somenodes\setminus M)
	\:\mstar\:
	\inv(\contentsof{\set{l,r}},M)
	\:\mstar\:
	\futharris(M, l, \lnext, r)
	\ .
\]
Then, we frame out $\inv(\abscontent\setminus\contentsof{\set{l,r}},\somenodes\setminus M)$ and use the remaining $\inv(\contentsof{\set{l,r}},M)$ to invoke the future $\futharris(M, l, \lnext, r)$.
This gives $\inv(\contentsof{\set{l,r}},M)$ which, combined with the frame, results in $\inv(\abscontent,\somenodes)$.
Consequently, the update does not alter the logical contents of the structure so that rule~\ruleref{lin-none} is applicable.
This also justifies framing out the update token before invoking \code{find} from within \code{search} on \cref{line-search-invokefind}---the update token is not needed for the proof of \code{find}.

We note that throughout the entire proof, all explicit inductive reasoning was carried out at the level of the program logic in lock-step with the program execution, using only local facts about the nodes in the data structure captured by the resource invariant, futures, and history predicates.
In particular, we did not need explicit inductive reasoning over heap graph predicates or computation histories.
All such reasoning is carried out \emph{for free} by our developed meta-theory.

%!TEX root = ../main.tex

\section{Prototype Implementation}
\label{Section:Evaluation}

We substantiate our claim that the presented techniques aid automation and are useful in practice.
To that end, we implemented a \cpp prototype called \plankton \cite{plankton}.
\plankton takes as input the program under scrutiny and a candidate node invariant.
It then fully automatically generates a proof within our novel program logic, establishing that the given program is linearizable and does adhere to the given invariant.
We give a brief overview of \plankton's proof generation and report on our findings.
We stress that the present paper focuses on the theoretical foundations and as such does not give a detailed discussion of \plankton's implementation.

%-------------------------------------------------------------------------------------------------%
%-------------------------------------------------------------------------------------------------%
%-------------------------------------------------------------------------------------------------%
\smartparagraph{Implementation}
The proof generation in \plankton is implemented as a fixpoint computation that saturates an increasing sequence $\theInterference_0 \subseteq \theInterference_1 \subseteq \cdots$ of interference sets \cite{DBLP:conf/cav/HenzingerJMQ03}.
Initially, the interference set is empty, $\theInterference_0\defeq\emptyset$.
Once $\theInterference_{k}$ has been obtained, a proof of the input program with respect to $\theInterference_k$ is constructed and the interferences $\theInterference_\mathit{new}$ discovered during this proof are recorded, yielding $\theInterference_{k+1}\defeq\theInterference_{k}\cup\theInterference_\mathit{new}$.
A fixpoint $\theInterference_\mathit{lfp} \defeq \theInterference_k$ is reached if no new interference is found, $\theInterference_k = \theInterference_{k+1}$.
The proof generated for $\theInterference_k$ is then the overall proof for the input program.

For efficiency reasons, it is crucial to reduce the size of the computed interference sets \cite{DBLP:conf/vmcai/Vafeiadis10}.
We reduce an interference set $\theInterference$ by dropping any interference $(\apredpp, \acom)$ that is already covered, that is, there is $(\apredppp, \acom)\in\theInterference$ with $\apredpp\subseteq\apredppp$.

Given an interference set $\theInterference$, \plankton constructs a proof $\theInterference\semCalc\hoareOf{\apred_0}{\afun}{\apredp}$ for each function $\afun$ of the input program.
The proof construction starts from the precondition $\apred_0$, which captures just the invariant, as done, e.g., in \Cref{fig-harris-history-reasoning}.
From there, the rules of the program logic (\Cref{Figure:ProgramLogic}) are applied to inductively construct the postcondition.
As $\theInterference$ is fixed, plankton does not track the predicates $\thePredicates$.
We elaborate on the interesting ingredients of the proof construction.

Rule~\ruleref{com-sem} for atomic commands $\acom$ requires \plankton to compute $\semComOf{\acom}{\apred}$ for some precondition $\apred$.
The behavior of $\semComOf{\acom}{\apred}$ is prescribe by the standard axioms of separation logic~\cite{DBLP:conf/csl/OHearnRY01}.
If $\acom$ updates the heap, however, we have to additionally infer a flow footprint such that, roughly,
(i) all nodes updated by $\acom$ are contained in the footprint,
(ii) the interface of the footprint remains unchanged by the update, and
(iii) after the update all nodes inside the footprint still satisfy their node-local invariant.
As discussed in \cref{Section:Keysets}, this localizes the update to the footprint: the nodes outside the footprint continue to satisfy the invariant.
\plankton chooses the footprint by collecting all nodes whose (non-flow) fields are updated by $\acom$ and adds those nodes that are reachable in a small, constant number of steps.
If this choice does not satisfy (i), verification fails.
The restriction to finite footprints is essential for automating (ii) and (iii).
Yet, the restriction does not limit our approach: unbounded footprints are handled with futures, as seen in \cref{Section:Futures}.
Conditions (ii) and (iii) are then encoded into SMT and discharged using \atoolname{Z3}~\cite{DBLP:conf/tacas/MouraB08}.
Lastly, we apply the interferences $\theInterference$ to $\semComOf{\acom}{\apred}$.
The result is $\apredp\defeq\semOf{\theInterference}(\semComOf{\acom}{\apred})$ whose computation is inspired by~\cite{DBLP:conf/vmcai/Vafeiadis10}.
Overall, we obtain $\theInterference\semCalc\hoareOf{\apred}{\acom}{\apredp}$.

Rule~\ruleref{loop} requires a loop invariant $I$ for program $\loopof{\astmt}$ and precondition $\apred$ such that $\theInterference\semCalc\hoareOf{I}{\astmt}{I}$ and $\apred \subseteq I$.
To find one, \plankton generates a sequence $I_0, I_1, \dots$ of candidates.
The first candidate is $I_0\defeq\apred$.
Candidate $I_{n+1}$ is obtained from a sub-proof $\theInterference\semCalc\hoareOf{I_n}{\astmt}{I_n'}$ whose pre- and postcondition are joined, i.e., $I_{n+1}\defeq I_n \join I_n'$.
Intuitively, this join corresponds to the disjunction $I_n \cup I_n'$.
For performance reasons, however, \plankton uses a disjunction-free domain \cite{DBLP:conf/cav/YangLBCCDO08,DBLP:journals/toplas/RivalM07}, which means the join is actually weaker than union.
A loop invariant $I\defeq I_{n}$ is found, if the implication $I_{n+1} \subseteq I_{n}$ holds.

A core aspect of our novel program logic are history and future predicates.
\plankton tries to construct a strongest proof for the input program. Hence, new history and future predicates are added eagerly to an assertion $\apred$ whenever it participates in a join or interference is applied to it.
The rational behind this strategy is to \emph{save} information from $\apred$ in a history/future before it is lost.
More specifically, all boxed points-to predicates from $\apred$ that are subject to interference are added to a new history predicate.
New futures are introduced either from scratch with rule~\ruleref{f-intro} followed by rule~\ruleref{f-account} or from existing futures with rules~\ruleref{f-seq} and~\ruleref{f-frame}.
It is worth pointing out that \plankton uses rule~\ruleref{f-account} only to account duplicable facts, as in the proof from \Cref{fig-harris-future-reasoning}.
The introduction of futures is guided by a set of candidates.
These candidates are computed upfront by collecting all \code{CAS} commands in the input program.
A \code{CAS} may be dropped from the candidates if its footprint is statically known to be finite, e.g., because it only updates the mark bit of a pointer or inserts a new (and thus owned) node.
The approach discovers the necessary futures needed for our experiments.
Avoiding unnecessary futures produced by this method is considered future work.

%-------------------------------------------------------------------------------------------------%
%-------------------------------------------------------------------------------------------------%
%-------------------------------------------------------------------------------------------------%
\smartparagraph{Evaluation}
We used \plankton to automatically verify linearizability of fine-grained state-of-the-art set implementations from the literature:
\begin{inparaenum}[]
	\item a lock-coupling set~\cite[Chapter~9.5]{DBLP:books/daglib/0020056},
	\item the Lazy set~\cite{DBLP:conf/opodis/HellerHLMSS05},
	\item FEMRS tree~\cite{DBLP:conf/wdag/FeldmanE0RS18} which is a variation of the contention-friendly binary tree~\cite{DBLP:conf/europar/CrainGR13,DBLP:journals/ppl/CrainGR16},
	\item Vechev\&Yahav~2CAS set~\cite[Figures~8 and 9]{DBLP:conf/pldi/VechevY08},
	\item Vechev\&Yahav~CAS set~\cite[Figure~2]{DBLP:conf/pldi/VechevY08},
	\item ORVYY set~\cite{DBLP:conf/podc/OHearnRVYY10},
	\item Michael set~\cite{DBLP:conf/spaa/Michael02},
	\item Harris set~\cite{DBLP:conf/wdag/Harris01},
        \item and a variation with wait-free \code{search} of the Michael and Harris set algorithms.
\end{inparaenum}
For the FEMRS tree, \plankton cannot handle the maintenance operations because they have updates with an unbounded footprint that is not traversed. This is a limitation of our current future reasoning.
However, we are not aware of any other tool that can automatically verify even this simplified version of FEMRS trees. Also, \plankton is the first tool to automate hindsight reasoning for the Harris set.

The results are summarized in \Cref{table:benchmarks}.
The first three columns of the \namecref{table:benchmarks} list
\begin{inparaenum}[(i)]
	\item the number of iterations until the fixpoint $\theInterference_\mathit{lfp}$ is reached,
	\item the size of $\theInterference_\mathit{lfp}$, and
	\item the number of future candidates.
The next five columns list the percentage of runtime spent on
	\item rule~\ruleref{com-sem},
	\item future reasoning,
	\item history reasoning,
	\item joins, and
	\item applying interferences.
The last column gives
	\item the overall runtime, averaged across $10$ runs, and the linearizability verdict~(success is marked with~\symbolYes).
\end{inparaenum}
Across all benchmarks we observe that two iterations are sufficient to reach the fixpoint $\theInterference_\mathit{lfp}$: the first iteration discovers all interferences, the second iteration confirms that none are missing.
This is remarkable because the first iteration uses $\theInterference_0\defeq\emptyset$, i.e., considers the sequential setting.
Further, we observe that most benchmarks spend significantly more time reasoning about the past than the future.
The reason is twofold.
(1) Introducing new futures either succeeds, meaning that a future candidate is resolved and can be ignored going forward, or it \emph{fails fast}, which we attribute to Z3 finding counterexamples much faster than proving the validity of our SMT encoding of heap updates.
(2) For histories, we do not have a heuristic identifying candidates.
Instead, we eagerly introduce histories upon interference.
We also apply hindsight reasoning eagerly.
Lastly, we observe that the overall runtime tends to increase with the nesting depth and complexity of loops, as \plankton requires several loop iterations (often between $3$ and $5$) to find an invariant.
A proper investigation of how finding loop invariants affects the overall runtime is future work.

%!TEX root = ../main.tex

\begin{table*}%
	\caption{Experimental results for verifying set implementations with \plankton, conducted an Apple M1 Pro.}%
	\label{table:benchmarks}%
		\center%
		\newcommand{\cellCountTime}[2]{\(#2\)}
		\newcommand{\cellYes}[1]{\makebox[0.85cm][r]{\(#1\)}\hspace{1.5mm}\makebox[.3cm][l]{\rawsymbolYes}}
		\setlength{\tabcolsep}{4pt}
		\begin{tabularx}{\textwidth}{Xccccccccc}%
			\toprule
			Benchmark
				& $\#\textrm{Iter}$
				& $\#\theInterference_\mathit{lfp}$
				& $\#\textrm{Cand}$
				& Com.
				& Fut.
				& Hist.
				& Join
				& Inter.
				& Lineariz.
				\\
			\midrule
			Fine-Grained set
				& \( 2 \)
				& \( 5 \)
				& \( 2 \)
				& \cellCountTime{??}{11\%}
				& \cellCountTime{??}{15\%}
				& \cellCountTime{??}{43\%}
				& \cellCountTime{??}{15\%}
				& \cellCountTime{??}{8\%}
				& \cellYes{46s}
				\\
			Lazy set
				& \( 2 \)
				& \( 6 \)
				& \( 2 \)
				& \cellCountTime{??}{10\%}
				& \cellCountTime{??}{13\%}
				& \cellCountTime{??}{54\%}
				& \cellCountTime{??}{11\%}
				& \cellCountTime{??}{5\%}
				& \cellYes{77s}
				\\
			FEMRS tree {(no maintenance)}
				& \( 2 \)
				& \( 5 \)
				& \( 2 \)
				& \cellCountTime{??}{19\%}
				& \cellCountTime{??}{0\%}
				& \cellCountTime{??}{49\%}
				& \cellCountTime{??}{1\%}
				& \cellCountTime{??}{9\%}
				& \cellYes{130s}
				\\
			Vechev\&Yahav~2CAS set
				& \( 2 \)
				& \( 3 \)
				& \( 1 \)
				& \cellCountTime{??}{14\%}
				& \cellCountTime{??}{0\%}
				& \cellCountTime{??}{33\%}
				& \cellCountTime{??}{31\%}
				& \cellCountTime{??}{9\%}
				& \cellYes{125s}
				\\
			Vechev\&Yahav~CAS set
				& \( 2 \)
				& \( 4 \)
				& \( 1 \)
				& \cellCountTime{??}{15\%}
				& \cellCountTime{??}{7\%}
				& \cellCountTime{??}{39\%}
				& \cellCountTime{??}{23\%}
				& \cellCountTime{??}{6\%}
				& \cellYes{54s}
				\\
			ORVYY set
				& \( 2 \)
				& \( 3 \)
				& \( 0 \)
				& \cellCountTime{??}{17\%}
				& \cellCountTime{??}{0\%}
				& \cellCountTime{??}{40\%}
				& \cellCountTime{??}{26\%}
				& \cellCountTime{??}{6\%}
				& \cellYes{47s}
				\\
			Michael set
				& \( 2 \)
				& \( 4 \)
				& \( 2 \)
				& \cellCountTime{??}{11\%}
				& \cellCountTime{??}{29\%}
				& \cellCountTime{??}{30\%}
				& \cellCountTime{??}{15\%}
				& \cellCountTime{??}{6\%}
				& \cellYes{306s} 
				\\
			Michael set {(wait-free search)}
				& \( 2 \)
				& \( 4 \)
				& \( 2 \)
				& \cellCountTime{??}{11\%}
				& \cellCountTime{??}{28\%}
				& \cellCountTime{??}{30\%}
				& \cellCountTime{??}{15\%}
				& \cellCountTime{??}{6\%}
				& \cellYes{246s}
				\\
			Harris set
				& \( 2 \)
				& \( 4 \)
				& \( 2 \)
				& \cellCountTime{??}{7\%}
				& \cellCountTime{??}{8\%}
				& \cellCountTime{??}{19\%}
				& \cellCountTime{??}{32\%}
				& \cellCountTime{??}{4\%}
				& \cellYes{1378s}
				\\
			Harris set {(wait-free search)}
				& \( 2 \)
				& \( 4 \)
				& \( 2 \)
				& \cellCountTime{??}{8\%}
				& \cellCountTime{??}{10\%}
				& \cellCountTime{??}{17\%}
				& \cellCountTime{??}{34\%}
				& \cellCountTime{??}{3\%}
				& \cellYes{1066s}
				\\
			\bottomrule
		\end{tabularx}
\end{table*}

We also stress-tested \plankton with faulty variants of the benchmarks.
All buggy benchmarks failed verification.
Note that \plankton does not implement error explanation techniques, which are beyond the scope of the present paper.

%!TEX root = ../main.tex

\section{Related Work}
\label{Section:Related}

%!TEX root = ../main.tex

\smartparagraph{Program Logics with History and Prophecy}
Program logics have been extended by mechanisms for temporal reasoning in various ways~\cite{DBLP:conf/concur/FuLFSZ10,DBLP:conf/esop/GotsmanRY13,DBLP:conf/esop/SergeyNB15,DBLP:conf/ecoop/DelbiancoSNB17,DBLP:conf/sas/BellAW10,DBLP:conf/popl/ParkinsonBO07,DBLP:conf/wdag/HemedRV15,DBLP:conf/pldi/LiangF13,DBLP:journals/tcs/AbadiL91,DBLP:books/daglib/0080029,DBLP:series/txcs/Schneider97}.

The work closest to ours is HLRG~\cite{DBLP:conf/concur/FuLFSZ10}, a separation logic based on local rely-guarantee~\cite{DBLP:conf/popl/Feng09} that tracks and reasons about history information, and its variation~\cite{DBLP:conf/esop/GotsmanRY13}.
The separation algebra behind HLRG is constructed like ours.
The focus of~\cite{DBLP:conf/concur/FuLFSZ10,DBLP:conf/esop/GotsmanRY13}, however, are temporal operators in the assertion language and means of reasoning about them in the program logic.
We only have now and past, but add the ability to propagate information between them.
The simplicity of our approach enables automation (\cite{DBLP:conf/concur/FuLFSZ10,DBLP:conf/esop/GotsmanRY13} has not been implemented in any automated or interactive tool, as far as we know).
A minor difference is that we work over general separation algebras to integrate flows~\cite{DBLP:journals/pacmpl/KrishnaSW18,DBLP:conf/esop/KrishnaSW20} easily and make the requirement of frameability explicit.

A program logic with temporal information based on different principles appears in~\cite{DBLP:conf/esop/SergeyNB15,DBLP:conf/ecoop/DelbiancoSNB17}. 
There, histories are sub-computations represented by timestamped sets of events.
The product of histories is disjoint union.
While highly expressive, we are not aware of implementations of the approach.
Since our goal is automated proof construction, we strive for assertions that are simple to prove, instead.

A separation logic for proving producer-consumer applications is proposed in~\cite{DBLP:conf/sas/BellAW10}. 
The logic uses history but is domain-specific and provides no mechanism to reason about the temporal development.
A non-blocking stack with memory reclamation is verified in \cite{DBLP:conf/popl/ParkinsonBO07}.
The proof relies on history information stored in auxiliary variables and manipulated by ghost code.
The ghost code is justified by informal arguments (outside the program logic).
We do not consider memory reclamation as it can be verified separately~\cite{DBLP:journals/pacmpl/MeyerW19,DBLP:journals/pacmpl/MeyerW20}.
Beyond linearizability, history variables have recently been used to give specs to non-linearizable objects~\cite{DBLP:conf/wdag/HemedRV15}. 

Several separation logics have been extended with prophecy variables~\cite{DBLP:conf/pldi/LiangF13,DBLP:journals/pacmpl/JungLPRTDJ20}, which complement history-based reasoning with a mechanism to speculate about future events.
However, prophecies are not well-suited for automatic proofs because they rely on backward reasoning~\cite{DBLP:conf/cav/BouajjaniEEM17}.

History reasoning has been used early on in program verification \cite{DBLP:journals/tcs/AbadiL91}.
In program logics \cite{DBLP:books/daglib/0080029,DBLP:series/txcs/Schneider97}, the focus has been on causality formulas which, in our notation, take the form ${\nowOf{\apred}\Rightarrow \pastOf{\apredp}}$.
Our history reasoning is more flexible, in particular incorporates hindsight reasoning (see below), and inherits the benefits of modern separation logics.

Overall, the existing program logics with history are heavier than ours while missing the important trick of communicating information from the current state to the past by means of logical variables shared between the two.

%!TEX root = ../main.tex

\smartparagraph{Hindsight Reasoning}
The idea of propagating information from the current state into the past is inspired by the recent hindsight theory~\cite{DBLP:conf/podc/OHearnRVYY10,DBLP:conf/wdag/Lev-AriCK15,DBLP:conf/wdag/FeldmanE0RS18,DBLP:journals/pacmpl/FeldmanKE0NRS20}.
Hindsight lemmas ensure that information about a data structure obtained by sequential reasoning (typically the reachability of keys) remains valid for concurrent executions. 
The argument behind such results is that the existence of a sequentially-reachable state implies the existence of a related concurrently-reachable state in the past. 
The implication requires that updates to the structure do not interfere with the reachability condition one tries to establish~(\emph{forepassing condition} in~\cite{DBLP:journals/pacmpl/FeldmanKE0NRS20}).

So far, hindsight reasoning has been limited to pencil-and-paper proofs, with the exception of the \code{poling} tool~\cite{DBLP:conf/cav/ZhuPJ15}.
\atoolname{poling} automates the specific hindsight lemma of \citet{DBLP:conf/podc/OHearnRVYY10}.
Unlike histories, it does not immediately generalize to other forms of retrospective reasoning, like \cite{DBLP:conf/wdag/FeldmanE0RS18,DBLP:journals/pacmpl/FeldmanKE0NRS20,DBLP:conf/wdag/Lev-AriCK15}.

Our program logic makes past states explicit and can be understood as a formal framework in which to execute hindsight reasoning.
Indeed, the sequential-to-concurrent lifting of hindsight matches the thread-modular nature of our logic.
Executing hindsight arguments in our framework brings several benefits.
For our engine, hindsight arguments provide a \emph{strategy} for finding history information.
For hindsight theory, it not only gives the classical benefits of program logics like precision and mechanization resp. automation.
One also inherits the other features of our logic.
We found futures indispensable to prove the Harris set. 
As an interesting remark, our work solves a limitation that has been criticized in~\cite{DBLP:journals/pacmpl/FeldmanKE0NRS20}, namely that the local-to-global lifting of the keyset theorem would not apply to optimistic algorithms with future-dependent linearization points.
We simply invoke the theorem below past predicates.

A limitation of the existing hindsight theory as well as our work is that it does not apply to algorithms with impure future-dependent linearization points like the Herlihy-Wing queue~\cite{DBLP:journals/toplas/HerlihyW90}.
We leave the extension of the theory to such algorithms as future work.
However, pure future-dependent linearization points are more common in concurrent data structures.

Atomic triples in TADA~\cite{DBLP:conf/ecoop/PintoDG14}, CAP~\cite{DBLP:conf/ecoop/Dinsdale-YoungDGPV10}, and Iris~\cite{DBLP:conf/popl/JungSSSTBD15} are specifications that justify a logical notion of atomicity for operations whose execution may take more than one physical step.
This makes them suitable for compositional reasoning about nested modules with logically atomic specifications.
The lifting of our program logic to prove linearizability is inspired by the reasoning principles underlying atomic triples.

When it comes to pure future-dependent linearization points, our retrospective linearization with rule~\ruleref{lin-past} compares favorably to atomic triples. Proofs relying on atomic triples must typically implement intricate helping protocols that transfer ownership of the update tokens $\anobl{\asspec}$ and $\aful{\asspec}{v}$ between the thread that is linearized and the thread where the linearization point occurs.
This is necessary because the update token trade must happen at the very moment when the linearization point occurs~\cite{DBLP:journals/pacmpl/JungLPRTDJ20,DBLP:journals/pacmpl/PatelKSW21}.
Our technique avoids such helping protocols altogether, which aids proof automation.

%!TEX root = ../main.tex

\smartparagraph{Futures}
Our futures are nothing but Hoare triples in separation logic (with separating implication), and their use as assertions is well-known from program logics like Iris~\cite{DBLP:journals/jfp/JungKJBBD18}.
What we add is the observation that futures capture complex heap updates by iteratively combining futures of small updates found during the traversal preparing the complex update.
This iterative combination is the key novelty of our development.
It allows us to reason about updates of unbounded heap regions by means of updates of bounded regions.

Futures can be thought of as the opposite of atomic triples in that they prove the specification of a single physically atomic command like a \code{CAS} using a sequence of logical ghost steps.

A method for automatically handling updates affecting unbounded heap regions is proposed in \cite{DBLP:journals/pacmpl/Ter-GabrielyanS19}, however, their method is tailored towards reachability.
Being Hoare triples, our futures are not restricted to a specific class of properties.

%!TEX root = ../main.tex

\smartparagraph{Automation}
There is a considerable body of work on the automated verification of concurrent data structures.
For static linearization points, there are tools~\cite{DBLP:conf/tacas/AbdullaHHJR13} and well-chosen abstract domains~\cite{DBLP:conf/esop/AbdullaJT18}.
For dynamic linearization points, there are reductions to safety verification~\cite{DBLP:conf/esop/BouajjaniEEH13,DBLP:conf/icalp/BouajjaniEEH15,DBLP:conf/cav/BouajjaniEEM17}.
Common to these works is that, in the end, they rely on a state-space search whereas our approach reasons in a program logic.
Notably, the \atoolname{poling} tool \cite{DBLP:conf/cav/ZhuPJ15} extends \atoolname{cave} \cite{DBLP:conf/vmcai/Vafeiadis10,DBLP:conf/cav/Vafeiadis10,DBLP:conf/vmcai/Vafeiadis09} to support dynamic linearization points, e.g., to verify intricate stacks and queues (which our tool \plankton does not support because they are not search structures).
Related is also~\cite{DBLP:conf/popl/ItzhakyBILNS14} in the sense that flows in particular can express heap paths. 
But we are not interested in verification condition generation and complete reductions to SMT, but rather proof generation, including invariant synthesis. 

Other promising tools automating program logics include \atoolname{Starling} \cite{DBLP:conf/cav/WindsorDSP17}, \atoolname{Caper} \cite{DBLP:conf/esop/Dinsdale-YoungP17}, \atoolname{Voila} \cite{DBLP:conf/fm/WolfSM21}, and \atoolname{Diaframe} \cite{DBLP:conf/pldi/MulderKG22}. However, these are closer to proof-outline checkers when compared to our tool. In particular, they do not perform loop invariant and interference inference or try to identify linearization points. Instead, they target more complex logics that are not designed for ease of automation.

	%!TEX root = ../main.tex

\begin{acks}
This work is funded in parts by the \grantsponsor{GS100000001}{National Science Foundation}{http://dx.doi.org/10.13039/100000001} under grant~\grantnum{GS100000001}{1815633}.
The third author is supported by a Junior Fellowship from the Simons Foundation (855328, SW).
\end{acks}

	\bibliography{bib}

	\techreport{%!TEX root = ../main.tex

\clearpage
\appendix
%!TEX root = ../main.tex

\section{Full History Proof for Harris' Set}
\label{Section:HarrisHistoryFull}

Figure~\ref{harris:proof_search_full} shows the full proof outline for the linearizability of \code{find} from the Harris' set.

\begin{figure}[!h]
\newcommand{\MSTAR}{\mkern+1mu\mstar\mkern+1mu}
% $\makeTeal{\tinv(N,M,l,\lnext,\lmark,t) \defeq \inv(N) \MSTAR \neg \lmark \MSTAR \keyvarof{l} < k < \infty
%  \MSTAR \set{l,\lnext,t} \subseteq M \subseteq N}$
% $\makeTeal{\futharris(M, l, \lnext, t) \defeq \FUT{\futharrispre(M,l,\lnext,t,\lnext)}{l.\nextfld \assign t}{\futharrispost(M,l, \lnext,t,t)}}$
% $\makeTeal{\futharrispre(M, l, \lnext, t, u) \defeq \inv(M) \MSTAR \set{l,t} \subseteq M \MSTAR \nextvarof{l} = u \MSTAR \neg \markvarof{l}}$
% $\makeTeal{\futharrispost(M, l, \lnext, t, u) \defeq \inv(M) \MSTAR \set{l,t} \subseteq M \MSTAR \nextvarof{l} = u \MSTAR \neg \markvarof{l} \MSTAR (\keyvarof{l},\infty] \subseteq \flowvarof{t}}$
\begin{lstlisting}[language=SPL,belowskip=.2em]
$\makeTeal{\travinv(N,M,l,\lnext,\lmark,t) \,\defeq\, \tinv(N,M,l,\lnext,\lmark,t) \MSTAR \futharris(M,l,\lnext,t)}$
$\makeTeal{\hist(P, t) \,\defeq\, \pastOf{\left(\,\old{\nodeof{t}} \MSTAR (P \Rightarrow \old{(k \in \keysetof{t})})\,\right)} \MSTAR (\old{\markvarof{t}} \Rightarrow \markvarof{t}) \MSTAR \keyvarof{t} = \old{\keyvarof{t}}}$
\end{lstlisting}
\begin{lstlisting}[language=SPL,escapechar=@, belowskip=.3em]
$\annot{
  \exists N\,M.~\travinv(N,M,l,\lnext,\lmark,t) \MSTAR \hist(k \leq \old{\keyvarof{\lnext}}, \lnext)
}$
procedure traverse($k$: K,$\,$$l$: N,$\,$$\lnext$: N,$\,$$\lmark$: Bool,$\,$$t$: N) {
  val $\tnext$, $\tmark$ = atomic {$t$.next, $t$.mark}
  $\annot{
    \travinv(N,M,l,\lnext,\lmark,t) \MSTAR \hist(k \leq \old{\keyvarof{\lnext}}, \lnext)
  }$
  skip
  $\annot{
    \travinv(N,M,l,\lnext,\lmark,t) \MSTAR \hist(k \leq \old{\keyvarof{\lnext}}, \lnext)
    \MSTAR \bigl(\keyvarof{t}<k \Rightarrow \hist((k \leq \old{\keyvarof{\tnext}} \land \neg\tmark), \tnext)\bigr)
  }$
  if ($\tmark$) {
    $\annot{
      \travinv(N,M,l,\lnext,\lmark,\tnext) \MSTAR \hist(k \leq \old{\keyvarof{\lnext}}, \lnext)
    }$
    return traverse($k$, $l$, $\lnext$, $\lmark$, $\tnext$)
  } else if ($t$.key < $k$) {
    $\annot{
      \travinv(N,M,t,\tnext,\tmark,\tnext) \MSTAR \hist(k \leq \old{\keyvarof{\tnext}}, \tnext)
    }$
    return traverse($k$, $t$, $\tnext$, $\tmark$, $\tnext$)
  } else {
    $\annot{
      \travinv(N,M,l,\lnext,\lmark,t) \MSTAR t \neq \head \MSTAR t \neq l \MSTAR
      \keyvarof{l} < k \leq \keyvarof{t} \MSTAR \hist(k \leq \old{\keyvarof{\lnext}}, \lnext)
    }$ 
    return ($l$, $\lnext$, $\lmark$, $t$)
} }
$\annot{
  (l,\lnext,\lmark,r).\; \exists N\,M.\;
  \travinv(N,M,l,\lnext,\lmark,r) \MSTAR r \neq \head
  \MSTAR \keyvarof{l} < k \leq \keyvarof{r} \MSTAR \hist(k \leq \old{\keyvarof{\lnext}}, \lnext)
}$
\end{lstlisting}
\begin{lstlisting}[language=SPL,escapechar=@]
$\annot{\exists N.\; \inv(N) * -\infty < k < \infty}$
procedure find($k$: K) : N * N {
  val $\hnext$, $\hmark$ = atomic {head.next, head.mark}
  $\annot{
    \travinv(N,\{\head,\hnext\},\head,\hnext,\hmark,\hnext) \MSTAR \hnext = \nextvarof{head}
  }$
  skip
  $\annot{
    \travinv(N,\{\head,\hnext\},\head,\hnext,\hmark,\hnext)
    \MSTAR \hist(k \prall{\leq} \old{\keyvarof{\hnext}}, \hnext)\!
  }$ 
  val $l$, $\lnext$, $\lmark$, $r$ = traverse($k$, head, $\hnext$, $\hmark$, $\hnext$)
  $\annot{
    \travinv(N,M,l,\lnext,\lmark,r) \MSTAR r \neq \head \MSTAR
    \keyvarof{l} < k \leq \keyvarof{r} \MSTAR \hist(k \leq \old{\keyvarof{\lnext}}, \lnext)
  }$
  val $\vsucc$ = $\lnext$ == $r$ || atomic { // CAS
    $l$.next == $\lnext$ && $l$.mark == $\lmark$ ? {
      $\annot{
        \nextvarof{l} = \lnext \MSTAR  \neg \markvarof{l} \MSTAR \set{l,\lnext,r} \subseteq M \subseteq N \MSTAR
        \futharris(M, l, \lnext, r)  \MSTAR \keyvarof{l} < k \leq \keyvarof{r} \MSTAR \inv(N) 
      }$ 
      $l$.next := $r$; true 
      $\annot{
        \inv(N) \MSTAR \set{l,r} \subseteq N \MSTAR (\keyvarof{l},\infty] \subseteq \flowvarof{r} \MSTAR
        \keyvarof{l} < k \leq \keyvarof{r}
      }$ 
      skip
      $\annot{
        \inv(N) \MSTAR \set{l,r} \subseteq N \MSTAR \hist(\true, r)
      }$ 
      } : false
  }
  $\annot{
    \inv(N) \MSTAR \set{l,r} \subseteq N \MSTAR \left(\vsucc \; \Rightarrow \hist(\true, r) \right)
  }$ 
  if ($\vsucc$ && !$r$.mark) { 
    $\annot{
      \inv(N) \MSTAR \set{l,r} \subseteq N \MSTAR \neg\old{\markvarof{r} \MSTAR} \hist(\true, r)
    }$ 
    return ($l$, $r$)
  } else find($k$)
}
$\annot{(l,r).~ \exists N.\; \inv(N) \MSTAR \set{l,r} \subseteq N \MSTAR \neg\old{\markvarof{r} \MSTAR} \hist(\true, r)}$
$\annot{
  (l,r).~ \exists N.~
  \inv(N) \MSTAR \set{l,r} \subseteq N \MSTAR \keyvarof{r}=\old{\keyvarof{r}} \MSTAR
  \neg\old{\markvarof{r}} \MSTAR \pastOf{\left(\,\old{\nodeof{r}} \MSTAR \old{(k \in \keysetof{r})}\,\right)}
}$
\end{lstlisting}
\caption{Linearizability proof outline for Harris' set \code{find}.\label{harris:proof_search_full}}
\end{figure}

\section{Additional Material for Section~\ref{Section:Preliminaries}}

\begin{example}
\label{ex-heap-semantics}
We use a standard heap semantics of separation logic as a running example throughout this section. We consider heaps consisting of objects mapping fields to values. Let $\setsel$ be a set of \emph{field selectors}. Let further $\setvalues$ be a set of values, $\setaddrs \subseteq \setvalues$ an infinite set of \emph{addresses}, and $\setsvars$ a set of \emph{program variables}. A \emph{stack} $\astack \in \setstacks \defeq \setsvars \to \setvalues$ is an assignment from program variables to values.
% --------------------------------------------------------------------------------
% --------------------------------------------------------------------------------
A \emph{heap graph} is a tuple $\aheapgraph = (\setnodes, \sval)$ consisting of a finite set of nodes $\setnodes\subseteq\setaddrs$ and a total valuation of the selectors $\sval:\setnodes\times\setsel\rightarrow \setvalues$. Note that selectors may evaluate to nodes outside the heap graph. Let $\setheapgraphs$ be the set of all heap graphs.

The composition of heap graphs is the disjoint union, $\aheap_1 \heapmult \aheap_2 \defeq (\aheap_1.\setnodes \uplus \aheap_2.\setnodes, \aheap_1.\sval \uplus \aheap_2.\sval)$, and the empty heap is the only unit, $\heapemp \defeq \set{(\emptyset, \emptyset)}$. 
We extend the composition to pairs of stack and heap by defining $(\astack_1, \aheap_1) \stackheapmult (\astack_2, \aheap_2) \defeq (\astack_1, \aheap_1 \heapmult \aheap_2)$ if $\astack_1 = \astack_2$.
We then define the global states as $(\setshared, \sharedmult, \sharedemp) \defeq (\setheaps, \heapmult, \heapemp)$ and the local states as $(\setlocal, \localmult, \localemp) \defeq (\setstacks \times \setheaps, \stackheapmult, \setstacks \times \heapemp)$. A state is a pair of global and local state whose heaps are disjoint, $\setstates \defeq \setcond{(\aheap_\sharedindex, (\astack_\localindex,\aheap_\localindex))}{\aheap_\sharedindex \heapmultdef \aheap_\localindex}$. %\rfm{Did minimal wording things in this para, please check.}
\end{example}

\begin{example}
  Continuing with Example~\ref{ex-heap-semantics}, we define the semantics of a field update command $\assignof{l.\nextsel}{r}$ that assigns the field $\nextsel$ of the object at address $l$ to value $r$ as follows:
  \begin{align*}
    \semComOf{\assignof{l.\nextsel}{r}}{\aheap_\sharedindex, (\astack, \aheap_\localindex)}
    \defeq
    \begin{cases}
      \set{(\aheap_\sharedindex, (\astack, \aheap_\localindex[(\astack(l),\nextsel) \mapsto \astack(r)]))} & \text{if } \astack(l) \in \aheap_\localindex.\setnodes\\
      \set{(\aheap_\sharedindex[(\astack(l),\nextsel) \mapsto \astack(r)], (\astack, \aheap_\localindex))} & \text{if } \astack(l) \in \aheap_\sharedindex.\setnodes\\
      \abort & \text{otherwise}\ .
    \end{cases}
  \end{align*}
  We lift $\semCom{\mathtt{\mathit{l.\nextsel} := \mathit{r}}}$ to predicates in the expected way and also define $\semComOf{\mathtt{\mathit{l.\nextsel} := \mathit{r}}}{\abort} = \abort$.
\end{example}

\begin{proof}[Proof of Lemma~\ref{lem-composed-state-algebra}]
  Assume that $(\setshared, \sharedmult, \sharedemp)$ and $(\setlocal, \localmult, \localemp)$ are separation algebras and assume that $\setstates$ is such that $\sharedemp \times \localemp \subseteq \setstates \subseteq \setshared \times \setlocal$, and $\setstates$ is closed under decomposition.

  The laws of commutativity and units follow directly from the fact that $\statemult$ is the component-wise lifting of $\sharedmult$ and $\localmult$ on $\setshared$ and $\setlocal$, and that $\sharedemp \times \localemp \subseteq \setstates$.

  To show associativity of $\statemult$, let $\astate_1,\astate_2,\astate_3 \in \setstates$ such that $\astate_2 \statemultdef \astate_3$ and $\astate_1 \statemultdef (\astate_2 \statemult \astate_3)$. We have:
  \begin{align*}
\astate_1 \statemult (\astate_2 \statemult \astate_3) = {} &  (\astate_1.\ashared \sharedmult (\astate_2.\ashared \sharedmult \astate_3.\ashared), \astate_1.\alocal \localmult (\astate_2.\alocal \localmult \astate_3.\alocal))\\
= {} & ((\astate_1.\ashared \sharedmult \astate_2.\ashared) \sharedmult \astate_3.\ashared, (\astate_1.\alocal \localmult \astate_2.\alocal) \localmult \astate_3.\alocal)
 = {} (\astate_1 \statemult \astate_2) \statemult \astate_3
\end{align*}
where the second equality follows by associativity of $\sharedmult$ and $\localmult$, and the third equality holds because $\setstates$ is closed under decomposition.

The other direction follows by commutativity.
\end{proof}

%!TEX root = ../main.tex

% ------------------------------------------------------------------------------
% ------------------------------------------------------------------------------
% ------------------------------------------------------------------------------
\section{Flows}\label{Section:Flows}

The flow framework~\cite{DBLP:conf/esop/KrishnaSW20,DBLP:journals/pacmpl/KrishnaSW18} enables local reasoning about inductive properties of graphs. Intuitively, the framework views a heap as a (data) flow graph that induces \emph{flow constraints} prescribing how \emph{flow values} are propagated along the pointer edges of the heap. A \emph{flow} is a solution to these flow constraints and assigns a flow value to every node in the graph. For instance, one can define a \emph{path-counting flow} that computes for every node $\anode$ the number of paths that reach $\anode$ from some dedicated root node. By imposing an appropriate invariant on the flow at every node, one can then express inductive properties about the heap. For example, using the path-counting flow one can capture structural invariants such as that a heap region forms a list, a tree, or a DAG.

Flow graphs are endowed with a separation algebra that extends disjoint composition of heap graphs $\aheap_1 \mstar \aheap_2$ with a stronger definedness condition. This stronger condition ensures that the flow of $\aheap_1 \mstar \aheap_2$ is also the union of the flows of $\aheap_1$ and $\aheap_2$. This separation algebra admits frame rules that capture whether a local update of the heap graph $h_1$ preserves the flow of the context $\aheap_2$ (and hence the flow-based inductive invariant of $\aheap_1 \mstar \aheap_2$).

In the following, we develop a new meta theory of flows that, unlike~\cite{DBLP:conf/esop/KrishnaSW20,DBLP:journals/pacmpl/KrishnaSW18}, is geared towards automatic inference of flow invariants, yielding a flow-based abstract domain for program analysis.
% We defer the detailed comparison with~\cite{DBLP:conf/esop/KrishnaSW20,DBLP:journals/pacmpl/KrishnaSW18} to \cref{Section:Related}.

\smartparagraph{Separation Algebra}
For a commutative monoid $(\amonoid, \monadd, \monunit)$ we define the binary relation $\leq$ on $\amonoid$ as $\amonvalp \leq \amonval$ if there is $\amonvalpp\in\amonoid$ with $\amonval = \amonvalp+\amonvalpp$.
Flow values are drawn from a \emph{flow monoid}, which is a commutative monoid for which the relation $\leq$ is a complete partial order. In the following, we fix a flow monoid $(\amonoid, \monadd, \monunit)$.

Let $\monfun{\amonoid\rightarrow \amonoid}$ be the monotonic functions in $\amonoid\rightarrow\amonoid$. We lift $\monadd$ and $\leq$ to functions $\amonoid \to \amonoid$ in the expected way.
An iterated sum over an empty index set $\sum_{i\in \emptyset} \amonval_i$ is defined to be $\monunit$.
% ------------------------------------------------------------------------------
% ------------------------------------------------------------------------------
% ------------------------------------------------------------------------------
% ------------------------------------------------------------------------------

A \emph{flow constraint} is a tuple $\aflowconstraint=(\setnodes, \edges, \inflow)$ consisting of a finite set of variables $\setnodes\subseteq \setaddrs$, a set of edges $\edges:\setnodes\times\setaddrs\rightarrow\monfun{\amonoid\rightarrow \amonoid}$ labeled by monotonic functions, and an \emph{inflow} $\inflow: (\setaddrs\setminus\setnodes)\times\setnodes\rightarrow\amonoid$. We use $\setflowconstraints$ for the set of all flow constraints and denote the empty flow constraint by $\aflowconstraint_{\emptyset}\defeq (\emptyset, \emptyset, \emptyset)$.

We define two derived functions for flow constraints.
The \emph{flow} is the least function $\fval:\setnodes\rightarrow\amonoid$ satisfying \[\fvalof{\anode} = \inflowof{\anode}+\rhsatof{\anode}{\fval} \quad\text{for all } \anode\in\setnodes\ .\]
Here, $\inflowof{\anode}\defeq\sum_{\anodep\in(\setaddrs\setminus\setnodes)}\inflowof{\anodep, \anode}$ is a monoid value and $\rhs_{\anode} \defeq \sum_{\anodep\in \setnodes}\edges_{(\anodep, \anode)}$ is a monotone function from $\monfun{(\setnodes\rightarrow \amonoid)\rightarrow \amonoid}$.
We also define the \emph{outflow} $\outflow:\setnodes\times(\setaddrs\setminus\setnodes)\rightarrow\amonoid$ by $\outflowof{\anode, \anodep}\defeq\edges_{(\anode, \anodep)}(\fvalof{\anode})$.

Intuitively, a flow constraint $\aflowconstraint$ abstracts a heap graph $\aheap$ with domain $\aflowconstraint.\setnodes$ whose contents induce the edge functions $\aflowconstraint.\edges$. The inflow $\aflowconstraint.\inflow$ captures the contribution of $\aheap$'s context to $\aflowconstraint.\fval$ and the outflow $\aflowconstraint.\outflow$ is the contribution of $\aflowconstraint$ to the context's flow. In fact, if we abstract from the specific inflow of $\aflowconstraint$, we can view $\aflowconstraint$ as a transformer from inflows to outflows. We make this intuition formally precise later. However, let us first discuss some examples.

\begin{example}
  Path-counting flows are defined over the flow monoid of addition on natural numbers extended with the first limit ordinal, $\nat \cup \set\omega$. Here, $\monadd$ is defined to be absorbing on $\omega$. \Cref{fig-path-counting-flow} shows two flow constraints $\aflowconstraint_1$ and $\aflowconstraint_2$ for this flow monoid together with their flows and outflows. For example, we have $\aflowconstraint_1.\fval(y)=\aflowconstraint_1.\inflow(r,y) + \lambda_{\mathit{id}}(\aflowconstraint_1.\fval(x))=1+1=2$. The flow of each node corresponds to the number of paths starting from $r$ that reach the node in the combined graph (shown on the right).
  \begin{figure}
    \centering
  \begin{tikzpicture}[>=stealth, scale=0.8, every node/.style={scale=0.8}]

    \def\xsep{2}
    \def\ysep{-1.8}
    \def\xshift{2.5}
    \def\sd{-.25}
    \def\s{.4}
    \def\t{.8}

    \node[unode] (x) {$x$};
    \node[lbl] (flowx) at ($(x) + (-.65, 0)$) {$ 1$};
    \node[unode, blue, dashed] (r) at ($(x) + (\xsep, 0)$) {$r$};
    \node[unode] (y) at ($(x) + (0, \ysep)$) {$y$};
    \node[lbl] (flowy) at ($(y) + (.65, 0)$) {$ 2$};
    \node[unode, blue, dashed] (u) at ($(r) + (0, \ysep)$) {$u$};
    \node[lbl,blue] (flowu) at ($(u) + (.65, 0)$) {$ \omega$};
    \node[unode] (z) at ($(y) + (0, \ysep)$) {$z$};
    \node[lbl] (flowz) at ($(z) + (-.65, 0)$) {$ \omega$};
    \node[unode, blue, dashed] (v) at ($(u) + (0, \ysep)$) {$v$};
    %\node[unode, color=white] (mn) at ($(m) + (0, \ysep)$) {};

    \draw[edge, blue, dashed] (r) to node[above]{$ 1$} (x);
    \draw[edge, blue, dashed] (r) to node[above]{$ 1$} (y);
    \draw[edge] (x) to node[left] {$\lambda_{\mathit{id}}$} (y);
    \draw[edge] (x) to[bend right=40] node[left] {$\lambda_{\mathit{id}}$} (z);
    \draw[edge] (z) to node[above] {$\lambda_{\mathit{id}}$} (u);
    \draw[edge, blue, dashed] (v) to node[above]{$ \omega$} (z);

    \begin{scope}[on background layer] 
      \draw[draw=red!10, rounded corners, thick, fill=red!10]
      ($(x.north west) + (-\t-.4, \s)$)
      -- ($(x.north east) + (\t, \s)$)
      -- ($(z.south east) + (\t, \sd)$)
      -- ($(z.south west) + (-\t-.4, \sd)$) -- cycle;
    \end{scope}
  \end{tikzpicture}\quad\quad
  \begin{tikzpicture}[>=stealth, scale=0.8, every node/.style={scale=0.8}]

    \def\xsep{2}
    \def\ysep{-1.8}
    \def\xshift{2.5}
    \def\sd{-.25}
    \def\s{.4}
    \def\t{.8}

    \node[unode, red, dashed] (x) {$x$};
    \node[lbl,red] (flowx) at ($(x) + (-.65, 0)$) {$ 1$};
    \node[unode] (r) at ($(x) + (\xsep, 0)$) {$r$};
    \node[lbl] (flowr) at ($(r) + (.65, 0)$) {$ 0$};
    \node[unode, red, dashed] (y) at ($(x) + (0, \ysep)$) {$y$};
    \node[lbl,red] (flowy) at ($(y) + (.65, 0)$) {$ 1$};
    \node[unode] (u) at ($(r) + (0, \ysep)$) {$u$};
    \node[lbl] (flowu) at ($(u) + (.65, 0)$) {$ \omega$};
    \node[unode, red, dashed] (z) at ($(y) + (0, \ysep)$) {$z$};
    \node[lbl,red] (flowz) at ($(z) + (-.65, 0)$) {$ \omega$};
    \node[unode] (v) at ($(u) + (0, \ysep)$) {$v$};
    \node[lbl] (flowv) at ($(v) + (.65, 0)$) {$ \omega$};
    %\node[unode, color=white] (mn) at ($(m) + (0, \ysep)$) {};

    \draw[edge] (r) to node[above]{$\lambda_1$} (x);
    \draw[edge] (r) to node[above]{$\lambda_1$} (y);
    \draw[edge] (u) to node[right] {$\lambda_{\mathit{id}}$} (v);
    \draw[edge, red, dashed] (z) to node[above] {$ \omega$} (u);
    \draw[edge] (v) to node[above right] {$\lambda_{\mathit{id}}$} (z);

    \begin{scope}[on background layer] 
      \draw[draw=blue!10, rounded corners, thick, fill=blue!10]
      ($(r.north west) + (-\t, \s)$)
      -- ($(r.north east) + (\t, \s)$)
      -- ($(v.south east) + (\t, \sd)$)
      -- ($(v.south west) + (-\t, \sd)$) -- cycle;
    \end{scope}
  \end{tikzpicture}\hfill
  \begin{tikzpicture}[>=stealth, scale=0.8, every node/.style={scale=0.8}]

    \def\xsep{2}
    \def\ysep{-1.8}
    \def\xshift{2.5}
    \def\sd{-.25}
    \def\s{.4}
    \def\t{.8}

    \node[unode] (x) {$x$};
    \node[lbl] (flowx) at ($(x) + (-.65, 0)$) {$ 1$};
    \node[unode] (r) at ($(x) + (\xsep, 0)$) {$r$};
    \node[lbl] (flowr) at ($(r) + (.65, 0)$) {$ 0$};
    \node[unode] (y) at ($(x) + (0, \ysep)$) {$y$};
    \node[lbl] (flowy) at ($(y) + (.65, 0)$) {$ 2$};
    \node[unode] (u) at ($(r) + (0, \ysep)$) {$u$};
    \node[lbl] (flowu) at ($(u) + (.65, 0)$) {$\omega$};
    \node[unode] (z) at ($(y) + (0, \ysep)$) {$z$};
    \node[lbl] (flowz) at ($(z) + (-.65, 0)$) {$\omega$};
    \node[unode] (v) at ($(u) + (0, \ysep)$) {$v$};
    \node[lbl] (flowv) at ($(v) + (.65, 0)$) {$\omega$};
    %\node[unode, color=white] (mn) at ($(m) + (0, \ysep)$) {};

    \draw[edge] (r) to node[above]{$\lambda_1$} (x);
    \draw[edge] (r) to node[above]{$\lambda_1$} (y);
    \draw[edge] (u) to node[right] {$\lambda_{\mathit{id}}$} (v);
    \draw[edge] (v) to node[above right] {$\lambda_{\mathit{id}}$} (z);
    \draw[edge] (x) to node[left] {$\lambda_{\mathit{id}}$} (y);
    \draw[edge] (x) to[bend right=40] node[left] {$\lambda_{\mathit{id}}$} (z);
    \draw[edge] (z) to node[above] {$\lambda_{\mathit{id}}$} (u);

    \begin{scope}[on background layer] 
      \draw[draw=violet!10, rounded corners, thick, fill=violet!10]
      ($(x.north west) + (-\t-.4, \s)$)
      -- ($(r.north east) + (\t, \s)$)
      -- ($(v.south east) + (\t, \sd)$)
      -- ($(z.south west) + (-\t-.4, \sd)$) -- cycle;
    \end{scope}
  \end{tikzpicture}
  \caption{Two flow constraints $\aflowconstraint_1$ with $\aflowconstraint_1.\setnodes = \set{x,y,z}$ (left) and $\aflowconstraint_2$ with $\aflowconstraint_2.\setnodes = \set{r,u,v}$ (center) for the path-counting flow monoid $\nat \cup \set\omega$. The edge label $\lambda_{\mathit{id}}$ stands for the identity function and $\lambda_1$ for the constant $1$ function. Omitted edges are labeled by the constant $0$ function. Dashed edges represent the inflows. Nodes are labeled by their flow, respectively, outflow.
    The right side shows the composition $\aflowconstraint_1 \mstar \aflowconstraint_2$.\label{fig-path-counting-flow}}
  \end{figure}
\end{example}

\begin{example}
  \label{ex-keyset-flow}
  For our linearizability proofs of concurrent search structures we use a flow that labels every data structure node $x$ with its \emph{inset}, the set of keys~$k'$ such that a thread searching for~$k'$ traverses the node~$x$. That is, the flow monoid is sets of keys, $\powerset{\ZZ \cup \set{-\infty,\infty}}$, with set union as addition. \Cref{fig-keyset-flow} shows two flow constraints that abstract potential states of the Harris list. The idea is that an edge leaving a node~$x$ that stores a key~$k$ is labeled by the function~$\lambda_k$. This is because a search for $k' \in \ZZ$ will traverse the edge leaving~$x$ iff $k < k'$ or~$x$ is marked. In the figure,~$l$ and $r$ are assumed to be unmarked, storing keys~$6$ and~$8$, respectively. Node~$t$ is assumed to be marked. Flow constraint $\aflowconstraint_2$ can be thought as being obtained from $\aflowconstraint_1$ by physically unlinking the marked node $t$.

  \begin{figure}
    \centering
  \begin{tikzpicture}[>=stealth, scale=0.8, every node/.style={scale=0.8}]

    \def\xsep{1.8}
    \def\ysep{-1.8}
    \def\xshift{2.5}
    \def\sd{-.55}
    \def\s{.4}
    \def\t{.7}

    \node[unode,blue,dashed] (u) {$u$};
    \node[unode] (l) at ($(u) + (\xsep, 0)$) {$l$};
    \node[lbl] (flowl) at ($(l) + (0, -.65)$) {$(3,\infty)$};
    %\node[unode, blue, dashed] (r) at ($(x) + (\xsep, 0)$) {$r$};
    \node[unode] (t) at ($(l) + (\xsep, 0)$) {$t$};
    \node[lbl] (flowt) at ($(t) + (0, -.65)$) {$(6,\infty)$};
    \node[unode] (r) at ($(t) + (\xsep, 0)$) {$r$};
    \node[lbl] (flowr) at ($(r) + (0, -.65)$) {$(6,\infty)$};
    \node[unode,blue,dashed] (v) at ($(r) + (\xsep, 0)$) {$v$};
    \node[lbl,blue] (flowv) at ($(v) + (0, -.65)$) {$(8,\infty)$};
    %\node[lbl] (flowy) at ($(y) + (.65, 0)$) {$\bf 2$};
    %\node[unode, blue, dashed] (u) at ($(r) + (0, \ysep)$) {$u$};
    %\node[lbl,blue] (flowu) at ($(u) + (.65, 0)$) {$\bf \omega$};
    %\node[unode] (z) at ($(y) + (0, \ysep)$) {$z$};
    %\node[lbl] (flowz) at ($(z) + (-.65, 0)$) {$\bf \omega$};
    %\node[unode, blue, dashed] (v) at ($(u) + (0, \ysep)$) {$v$};
    %\node[unode, color=white] (mn) at ($(m) + (0, \ysep)$) {};

    \draw[edge, blue, dashed] (u) to node[above]{$(3,\infty)$} (l);
    \draw[edge] (l) to node[above]{$\lambda_6$} (t);
    \draw[edge] (t) to node[above]{$\lambda_{-\infty}$} (r);
    \draw[edge] (r) to node[above]{$\lambda_8$} (v);
    %\draw[edge, blue, dashed] (r) to node[above]{$\bf 1$} (y);
    %\draw[edge] (x) to node[left] {$\lambda_{\mathit{id}}$} (y);
    %\draw[edge] (x) to[bend right=40] node[left] {$\lambda_{\mathit{id}}$} (z);
    %\draw[edge] (z) to node[above] {$\lambda_{\mathit{id}}$} (u);
    %\draw[edge, blue, dashed] (v) to node[above]{$\bf \omega$} (z);

    \begin{scope}[on background layer] 
      \draw[draw=red!10, rounded corners, thick, fill=red!10]
      ($(l.north west) + (-\t, \s)$)
      -- ($(r.north east) + (\t, \s)$)
      -- ($(r.south east) + (\t, \sd)$)
      -- ($(l.south west) + (-\t, \sd)$) -- cycle;
    \end{scope}
  \end{tikzpicture}\hfill
    \begin{tikzpicture}[>=stealth, scale=0.8, every node/.style={scale=0.8}]

    \def\xsep{1.8}
    \def\ysep{-1.8}
    \def\xshift{2.5}
    \def\sd{-.55}
    \def\s{.4}
    \def\t{.7}

    \node[unode,blue,dashed] (u) {$u$};
    \node[unode] (l) at ($(u) + (\xsep, 0)$) {$l$};
    \node[lbl] (flowl) at ($(l) + (0, -.65)$) {$(3,\infty)$};
    %\node[unode, blue, dashed] (r) at ($(x) + (\xsep, 0)$) {$r$};
    \node[unode] (t) at ($(l) + (\xsep, 0)$) {$t$};
    \node[lbl] (flowt) at ($(t) + (0, -.65)$) {$\emptyset$};
    \node[unode] (r) at ($(t) + (\xsep, 0)$) {$r$};
    \node[lbl] (flowr) at ($(r) + (0, -.65)$) {$(6,\infty)$};
    \node[unode,blue,dashed] (v) at ($(r) + (\xsep, 0)$) {$v$};
    \node[lbl,blue] (flowv) at ($(v) + (0, -.65)$) {$(8,\infty)$};
    %\node[lbl] (flowy) at ($(y) + (.65, 0)$) {$\bf 2$};
    %\node[unode, blue, dashed] (u) at ($(r) + (0, \ysep)$) {$u$};
    %\node[lbl,blue] (flowu) at ($(u) + (.65, 0)$) {$\bf \omega$};
    %\node[unode] (z) at ($(y) + (0, \ysep)$) {$z$};
    %\node[lbl] (flowz) at ($(z) + (-.65, 0)$) {$\bf \omega$};
    %\node[unode, blue, dashed] (v) at ($(u) + (0, \ysep)$) {$v$};
    %\node[unode, color=white] (mn) at ($(m) + (0, \ysep)$) {};

    \draw[edge, blue, dashed] (u) to node[above]{$(3,\infty)$} (l);
    \draw[edge] (l) to[bend left=30] node[above]{$\lambda_6$} (r);
    \draw[edge] (t) to node[above]{$\lambda_{-\infty}$} (r);
    \draw[edge] (r) to node[above]{$\lambda_8$} (v);
    %\draw[edge, blue, dashed] (r) to node[above]{$\bf 1$} (y);
    %\draw[edge] (x) to node[left] {$\lambda_{\mathit{id}}$} (y);
    %\draw[edge] (x) to[bend right=40] node[left] {$\lambda_{\mathit{id}}$} (z);
    %\draw[edge] (z) to node[above] {$\lambda_{\mathit{id}}$} (u);
    %\draw[edge, blue, dashed] (v) to node[above]{$\bf \omega$} (z);

    \begin{scope}[on background layer] 
      \draw[draw=red!10, rounded corners, thick, fill=red!10]
      ($(l.north west) + (-\t, \s)$)
      -- ($(r.north east) + (\t, \s)$)
      -- ($(r.south east) + (\t, \sd)$)
      -- ($(l.south west) + (-\t, \sd)$) -- cycle;
    \end{scope}
  \end{tikzpicture}
  \caption{Two flow constraints $\aflowconstraint_1$ (left) and $\aflowconstraint_2$ (right) with $\aflowconstraint_1.\setnodes = \aflowconstraint_2.\setnodes = \set{l,t,r}$ for the keyset flow monoid $\powerset{\ZZ \cup \set{-\infty,\infty}}$. The edge label $\lambda_k$ for a key $k$ denotes the function $\lambda \amonval.\, (m \setminus [-\infty, k])$.\label{fig-keyset-flow}} 
  \end{figure}
\end{example}
% 

% ----------------------------------------------------------------------------
% ----------------------------------------------------------------------------
% ----------------------------------------------------------------------------
% ----------------------------------------------------------------------------

% ----------------------------------------------------------------------------
% ----------------------------------------------------------------------------
% ----------------------------------------------------------------------------
% ----------------------------------------------------------------------------
To define the composition of flow constraints, $\aflowconstraint_1 \mstar \aflowconstraint_2$, we proceed in two steps.
We first define an auxiliary composition that may suffer from \emph{fake} flows, local flows that disappear in the composition.
The actual composition then restricts the auxiliary composition to rule out such fake flows.
Definedness of the auxiliary composition requires disjointness of the variables in $\aflowconstraint_1$ and $\aflowconstraint_2$.
Moreover, the outflow of one constraint has to match the inflow expectations of the other:
\begin{align*}
\aflowconstraint_1\statemultdef\statemultdef\aflowconstraint_2~~\text{if}~~&  \aflowconstraint_1.\setnodes\cap\aflowconstraint_2.\setnodes=\emptyset\\\text{and}~~&\forall \anode\in\aflowconstraint_1.\setnodes.\;\forall\anodep\in\aflowconstraint_2.\setnodes.\!\!\!
\begin{aligned}[t]
&\aflowconstraint_1.\outflowof{\anode, \anodep}=\aflowconstraint_2.\inflowof{\anode, \anodep}
\wedge \aflowconstraint_2.\outflowof{\anodep, \anode}=\aflowconstraint_1.\inflowof{\anodep, \anode}\,.
\end{aligned}
\end{align*}
The auxiliary composition removes the inflow that is now provided by the other component, denoted by $\aflowconstraint_1.\inflow'$ and $\aflowconstraint_2.\inflow'$:
\begin{align*}
\aflowconstraint_1\discup\aflowconstraint_2 \defeq &\, (\aflowconstraint_1.\setnodes\discup \aflowconstraint_2.\setnodes, \aflowconstraint_1.\edges\discup \aflowconstraint_2.\edges, \aflowconstraint_1.\inflow'\discup\aflowconstraint_2.\inflow')\\
\aflowconstraint_i.\inflow'(x,y) \defeq &\, \ite{(x \prall{\in} \aflowconstraint_{3-i}.\setnodes \land \aflowconstraint_{3-i}.\outflow(x,y)\prall{\neq}\monunit)}{\monunit\!}{\!\aflowconstraint_i.\inflow(x,y)} .
\end{align*}

To rule out fake flows, we incorporate a suitable equality on the flows into the definedness for the composition:
\begin{align*}
\aflowconstraint_1\statemultdef\aflowconstraint_2\quad\text{if}\quad  \aflowconstraint_1\statemultdef\statemultdef\aflowconstraint_2\;\;\wedge\;\;\aflowconstraint_1.\fval\discup\aflowconstraint_2.\fval = (\aflowconstraint_1\discup\aflowconstraint_2).\fval.
\end{align*}
Only if the latter equality holds, do we have the composition $\aflowconstraint_1\mstar\aflowconstraint_2\defeq\aflowconstraint_1\discup\aflowconstraint_2$.
It is worth noting that $\aflowconstraint_1.\fval\discup\aflowconstraint_2.\fval\geq (\aflowconstraint_1\discup\aflowconstraint_2).\fval$ always holds.
What definedness really asks for is the reverse inequality.

\begin{example}
  The right side of \Cref{fig-path-counting-flow} shows the composition $\aflowconstraint_1 \mstar \aflowconstraint_2$ of the flow constraints $\aflowconstraint_1$ (left) and $\aflowconstraint_2$ (center). Note that if the edge $(x,z)$ is removed from $\aflowconstraint_1$, then the composition is no longer defined since the flow of $\omega$ in the cycle formed by $z$, $u$, and $v$ becomes a fake flow that vanishes in the flow constraint $\aflowconstraint_1 \discup \aflowconstraint_2$ (the flow for $z$, $u$, and $v$ becomes~$0$).
\end{example}
\begin{lemma}\label{Lemma:FlowAlgebra}
$(\setflowconstraints, \mstar, \set{\aflowconstraint_{\emptyset}})$ is a separation algebra.
\end{lemma}
% ------------------------------------------------------------------------------
% ------------------------------------------------------------------------------
% ------------------------------------------------------------------------------
% ------------------------------------------------------------------------------

\smartparagraph{Flow Graphs}
Our heap model is a product between heap graphs and flow constraints.
The flow component is meant to provide ghost information about the heap graph component.
As we are interested in automating the reasoning, we will not modify the ghost information by ghost code but provide a mechanism to generate the flow constraint from the heap graph.
This allows us to keep the flow constraint largely implicit.

Heap graphs $\aheap \in \setheapgraphs$ are as defined in Example~\ref{ex-heap-semantics}. 

\begin{lemma}
$(\setheapgraphs, \mstar, \set{\aheapgraph_{\emptyset}})$ is a separation algebra.
\end{lemma}
% ------------------------------------------------------------------------------
% ------------------------------------------------------------------------------
% ------------------------------------------------------------------------------
% ------------------------------------------------------------------------------

To relate heap graphs and flow constraints, we classify selectors as data and pointer selectors, $\setsel=\setdsel\discup\setpsel$, and also write a heap graph $(\setnodes, \sval)$ as $(\setnodes, \dval, \pval)$.
Only pointer selectors will induce edges in a flow constraint via which we forward flow values.
The data selectors will be used to determine the edge function.
To this end, we assume to be given a generator for the edge functions:
\begin{align*}
\genrhs:\setpsel\rightarrow(\setdsel\rightarrow\setvalues)\rightarrow\monfun{\amonoid\rightarrow \amonoid} \enspace. 
\end{align*}
It takes a pointer selector and a valuation of the data selectors and returns a monotonic function over the monoid of flow values.
The generator is a parameter to our development, like the flow monoid and the sets of data and pointer selectors.

\begin{example}
  For the Harris list we have a single pointer selector $\mathtt{next}$ and data selectors $\mathtt{mark}$ and $\mathtt{key}$ storing the mark bit and key of a node, respectively.\footnote{Technically, the mark bit is encoded in the value of $\mathtt{next}$. We here treat them separately for the sake of exposition.}
  In order to obtain the edge functions that abstract the Harris list as described in Example~\ref{ex-keyset-flow}, we use the following generator:
  \[\genrhsof{\nextsel, \dval} \defeq
  \ite{\dvalof{\keysel}\!=\!-\infty}{\lambda \amonval.\, (-\infty,\infty]}
  {\lambda_{\dvalof{\keysel}}}\]
  Recall that the condition $\dvalof{\keysel} = -\infty$ only holds for the $\head$ node. We use it here to fix the flow leaving $\head$ to the constant interval $(-\infty,\infty]$. 
\end{example}

% ------------------------------------------------------------------------------
% ------------------------------------------------------------------------------
% ------------------------------------------------------------------------------
% ------------------------------------------------------------------------------
A \emph{flow graph} is a pair $\aflowgraph=(\aheapgraph, \aflowconstraint)$ consisting of a heap graph $\aheapgraph$ and a flow constraint $\aflowconstraint$ so that the nodes of the heap graph and flow constraint coincide, $\aheapgraph.\setnodes=\aflowconstraint.\setnodes$, and so that $\aflowconstraint.\edges$ is induced by $\aheapgraph.\pval$ and $\aheapgraph.\dval$ via $\genrhs$.
Being induced means that for all $\anode\in\aflowconstraint.\setnodes$ and $\anodep\in\setaddrs$:
\begin{align*}
\aflowconstraint.\edges(\anode, \anodep)\quad =\quad \sum_{\substack{\apsel\in\setpsel\\\anodep=\aheapgraph.\pvalof{\anode, \apsel}}} \genrhsof{\apsel, \aheapgraph.\dvalof{\anode}}. 
\end{align*}
Recall that empty sums are defined to be $\monunit$.
The edge function iterates over all pointer selectors that lead from the source to the target node.
It sums up the monoid functions given by the generator.
Note that the edge functions are independent of the target node.

\begin{example}
  \label{ex-flow-graph}
  Assuming the generator $\genrhs$ from the previous example, $(\aheap_1,\aflowconstraint_1)$ is a flow graph where $\aflowconstraint_1$ is as depicted in 
  % \Cref{fig-keyset-flow}, 
  % $\aheap_1.\pval = \set{(l,\mathtt{next}) \mapsto t, (t,\mathtt{next}) \mapsto r, (r,\mathtt{next}) \mapsto v}$, and $\aheap_1.\dval = \set{(l,\mathit{mark}) \mapsto \false, (t,\mathit{mark}) \mapsto \true, (r,\mathit{mark}) \mapsto \false, (l,\mathit{key}) \mapsto 6, (t,\mathit{key}) \mapsto 7, (r,\mathit{key}) \mapsto 8}$.
  \Cref{fig-keyset-flow}:
  \begin{align*}
    \aheap_1.\pval ~=~& \set{(l,\mathtt{next}) \mapsto t, (t,\mathtt{next}) \mapsto r, (r,\mathtt{next}) \mapsto v}
    \\
    \aheap_1.\dval ~=~& \set{(l,\mathit{mark}) \mapsto \false, (t,\mathit{mark}) \mapsto \true}
    \\
    \uplus~&\set{(r,\mathit{mark}) \mapsto \false}
    \\
    \uplus~&\set{(l,\mathit{key}) \mapsto 6, (t,\mathit{key}) \mapsto 7, (r,\mathit{key}) \mapsto 8}
    \enspace.
  \end{align*}\
\end{example}

We use $\setflowgraphs$ for the set of all flow graphs. 
%
% ------------------------------------------------------------------------------
% ------------------------------------------------------------------------------
% ------------------------------------------------------------------------------
% ------------------------------------------------------------------------------
It is a subset of the product separation algebra between the heap graphs and the flow constraints, $\setflowgraphs\subseteq\setheapgraphs\times\setflowconstraints$.
% ------------------------------------------------------------------------------
% ------------------------------------------------------------------------------
This gives flow graphs a notion of definedness: $\aflowgraph_1\statemultdef\aflowgraph_2$,  if $\aflowgraph_1.\aheapgraph\statemultdef\aflowgraph_2.\aheapgraph$ and $\aflowgraph_1.\aflowconstraint\statemultdef\aflowgraph_2.\aflowconstraint$.
% ------------------------------------------------------------------------------
% ------------------------------------------------------------------------------
They also inherit the component-wise composition in case definedness holds: $\aflowgraph_1\mstar\aflowgraph_2\defeq(\aflowgraph_1.\aheapgraph\mstar\aflowgraph_2.\aheapgraph, \aflowgraph_1.\aflowconstraint\mstar\aflowgraph_2.\aflowconstraint)$.
% ------------------------------------------------------------------------------
% ------------------------------------------------------------------------------
With this definition, we obtain a submonoid. 
% ------------------------------------------------------------------------------
% ------------------------------------------------------------------------------
\begin{lemma}\label{Lemma:FlowGraphs}
$(\setflowgraphs, \mstar, \set{(\aheapgraph_{\emptyset}, \aflowconstraint_{\emptyset})})$ is a separation algebra.
\end{lemma}
% ------------------------------------------------------------------------------
% ------------------------------------------------------------------------------
% ------------------------------------------------------------------------------
% ------------------------------------------------------------------------------

\smartparagraph{Commands and Frame-Preserving Updates}
Let $(\aheap_1, \aflowconstraint_1)$ be a flow graph and suppose that a command $\acom$ updates the heap graph $\aheap_1$ to a new heap graph $\aheap_2$. Then the heap update induces a corresponding update of the flow constraint $\aflowconstraint_1$ to the flow constraint $\aflowconstraint_2$ generated from $\aheap_2$. For instance, consider the flow graph $(\aheap_1, \aflowconstraint_1)$ from Example~\ref{ex-flow-graph}. The heap $\aheap_2$ obtained by executing the command $l.\mathtt{next} \,\mathtt{:=}\, r$ generates the flow constraint $\aflowconstraint_2$ shown in \Cref{fig-keyset-flow}.

In order to ensure that the semantics of commands satisfies the locality condition~(\ref{cond-loccom}) assumed in \cref{Section:Preliminaries}, which is critical for the soundness of the \ruleref{frame} rule, we need to ensure that the induced updates of the flow constraints are \emph{frame-preserving} with respect to flow constraint composition.
That is, if we have $\aflowconstraint_1 \statemultdef \aflowconstraint$ and we exchange $\aflowconstraint_1$ by $\aflowconstraint_2$, does $\aflowconstraint_2 \statemultdef \aflowconstraint$ still hold? Intuitively, $\aflowconstraint_2 \statemultdef \aflowconstraint$ still holds if $\aflowconstraint_1$ and $\aflowconstraint_2$ transform inflows to outflows in the same way.

Formally, for a flow constraint $\aflowconstraint$ we define its \emph{transfer function} \(\transformerof{\aflowconstraint}\) mapping inflows to outflows, \[\transformerof{\aflowconstraint}:((\setaddrs\setminus\setnodes)\times\setnodes\rightarrow \amonoid)\rightarrow\setnodes\times(\setaddrs\setminus\setnodes)\rightarrow \amonoid\ ,\] by $\transformerof{\aflowconstraint}(\inflow') \defeq \aflowconstraint[\inflow \mapsto \inflow'].\outflow$.
% Formally, for a flow constraint $\aflowconstraint$ we define its \emph{transfer function} \(\transformerof{\aflowconstraint}:((\nat\setminus\setnodes)\times\setnodes\rightarrow \amonoid)\rightarrow\setnodes\times(\nat\setminus\setnodes)\rightarrow \amonoid\) mapping inflows to outflows by $\transformerof{\aflowconstraint}(\inflow') \defeq \aflowconstraint[\inflow \mapsto \inflow'].\outflow$.
For a given inflow $\inflow$, we may also write $\transformerof{\aflowconstraint_1} =_{\inflow} \transformerof{\aflowconstraint_2}$ to mean that for all inflows $\inflow' \leq \inflow$, we have $\transformerof{\aflowconstraint_1}(\inflow') = \transformerof{\aflowconstraint_2}(\inflow')$.

For example, we have $\transformerof{\aflowconstraint_1} =_{\aflowconstraint_1.\inflow} \transformerof{\aflowconstraint_2}$ for the flow constraints $\aflowconstraint_1$ and $\aflowconstraint_2$ in \Cref{fig-keyset-flow}. This is due to $\aflowconstraint_1.\inflow(x)=\emptyset$ for all $x \neq l$ and $\lambda_6 \circ \lambda_{-\infty}=\lambda_6$. Observe that $\aflowconstraint_2$ composes with any frame $\aflowconstraint$ that $\aflowconstraint_1$ composes with. This captures the fact that physically unlinking the marked node $t$ does not affect the remaining state of the data structure. The following lemma generalizes this observation. Its proof relies on Bekic's lemma~\cite{DBLP:conf/ibm/Bekic84} and can be found in Appendix~\ref{Section:FlowsProofs}.

%Flow constraints that induce the same transformer can be exchanged during composition. 
% ------------------------------------------------------------------------------
% ------------------------------------------------------------------------------
% ------------------------------------------------------------------------------
% ------------------------------------------------------------------------------
\begin{lemma}\label{Lemma:Transformers}
Assume $\aflowconstraint_1\statemultdef\aflowconstraint$, $\aflowconstraint_1.\setnodes=\aflowconstraint_2.\setnodes$, $\aflowconstraint_1.\inflow=\aflowconstraint_2.\inflow$, and $\transformerof{\aflowconstraint_1}=_{\aflowconstraint_1.\inflow}\transformerof{\aflowconstraint_2}$. 
Then (i)~$\aflowconstraint_2\statemultdef\aflowconstraint$ and (ii)~$\transformerof{\aflowconstraint_1\mstar\aflowconstraint}=_{(\aflowconstraint_1\mstar\aflowconstraint).\inflow}\transformerof{\aflowconstraint_2\mstar\aflowconstraint}$.  
\end{lemma}
%\todo{Update proof in appendix.}
% ------------------------------------------------------------------------------
% ------------------------------------------------------------------------------
% ------------------------------------------------------------------------------
% ------------------------------------------------------------------------------

% ------------------------------------------------------------------------------
% ------------------------------------------------------------------------------
% ------------------------------------------------------------------------------
% ------------------------------------------------------------------------------
In Appendix~\ref{Section:Instantiation} we develop a semantics of a simple imperative programming language that satisfies the locality condition assumed in \cref{Section:Preliminaries}. Global states $\sharedstates$ and local states $\localstates$ are obtained by taking the product between flow graphs and variable valuations. The variable valuations track the values of both (immutable) program variables as well as logical variables. The construction of the state monoid $\setstates \subseteq \sharedstates \times \localstates$ is mostly standard and similar to the one discussed in Example~\ref{ex-heap-semantics}.  The semantics aborts a heap update $\acom$ if its induced flow constraint update does not satisfy the frame preservation condition identified in Lemma~\ref{Lemma:Transformers}.

\section{Instantiation}\label{Section:Instantiation}

In this section, we instantiate the developed program logic to a simple imperative programming language whose states consist of flow graphs.

In the following, we write flow graphs as $\aflowgraph = (\setnodes, \sval, \inflow)$.
This goes without loss of generality since the nodes coincide for the heap graph and the flow constraint, and the edge function of the flow constraint is induced by the data and pointer selector valuations.
For notational convenience, we also access the flow value by $\aflowgraph.\fval$ rather than $\aflowgraph.\aflowconstraint.\fval$.

\subsection{States}
Valuations of program variables $\setvars$, or simply 
\emph{stacks}, stem from $\setvalsof{\setvars}\defeq\setvars\rightarrow\nat$.
Stacks are never split, $\aval_1\statemultdef\aval_2$ is defined as $\aval_1=\aval_2$.
\begin{lemma}
$(\setvalsof{\setvars}, \mstar, \setvalsof{\setvars})$ is a separation algebra. 
\end{lemma}
% ------------------------------------------------------------------------------
% ------------------------------------------------------------------------------
% ------------------------------------------------------------------------------
% ------------------------------------------------------------------------------
As we divide states into a global and a local component, we distinguish between global program variables from the finite set $\setgpvars$ and local variables from the disjoint and also finite set  $\setlpvars$.
We use $\apvar, \apvarp\in\setpvars=\setgpvars\discup\setlpvars$ to refer to variables that are either global or local.
% ------------------------------------------------------------------------------
% ------------------------------------------------------------------------------
The separation algebra of global states is the product $\setshared \defeq \setflowgraphs\times\setvalsof{\setgpvars}$. Similarly, the separation algebra of local states is the product $\setlocal\defeq \setflowgraphs\times\setvalsof{\setlpvars}$. 
The set of states is a subset of the product separation algebra:
\begin{align*}
\setstates\;\defeq\;\setcond{(\ashared, \alocal)\in \setshared\times\setlocal}{\wellsplitof{\ashared,\alocal}}.
\end{align*}
Here, the predicate $\wellsplitof{\ashared, \alocal}$ says that it is possible to combine the global and the local heap.
Moreover, the global state does not point into the local heap. Formally,
\begin{align*}
\wellsplitof{\ashared, \alocal}\quad\text{if}
&\quad\begin{aligned}[t]
&\ashared.\aflowgraph\statemultdef\alocal.\aflowgraph
\;\;\wedge\;\;\forall \apvar\in\setgpvars.\;\;\ashared.\avalof{\apvar}\notin\alocal.\aflowgraph.\setnodes
\;\;\wedge\;\;\\&\forall \anaval\in\ashared.\aflowgraph.\setnodes.\;\;\forall\asel\in\setsel.\;\;\ashared.\aflowgraph.\svalof{\anaval, \asel}\notin\alocal.\aflowgraph.\setnodes
\ .
\end{aligned}
\end{align*}

The assumptions made by Lemma~\ref{lem-composed-state-algebra} follow from the next lemma.
\begin{lemma}
  For all $(\ashared, \alocal) \in \sharedemp \times \localemp$, $\wellsplitof{\ashared, \alocal}$ holds.
  Furthermore, for all $\ashared_1, \ashared_2 \in \setshared$ and $\alocal_1, \alocal_2 \in \setlocal$, if we have $\wellsplitof{\ashared_1 \sharedmult \ashared_2, \alocal_1 \localmult \alocal_2}$, then also $\wellsplitof{\ashared_1, \alocal_1}$.
\end{lemma}

\begin{corollary}
 $(\setstates, \statemult, \sharedemp \times \localemp)$ is a separation algebra.
\end{corollary}

To fix the notation, we write $\astate.\aflowgraph$ for $\astate.\ashared.\aflowgraph\mstar\astate.\alocal.\aflowgraph$ and similarly $\astate.\aval$ for $\astate.\ashared.\aval\discup\astate.\alocal.\aval$.
Moreover, we denote $\astate.\aflowgraph.\aheapgraph$ by $\astate.\aheapgraph$ and introduce a similar short-hand for $\astate.\aflowgraph.\aflowconstraint$.
%

% --------------------------------------------------------------------------------
% --------------------------------------------------------------------------------
% --------------------------------------------------------------------------------
% --------------------------------------------------------------------------------
\begin{comment}
\todo{Here is something to fix.}\\
In Section~\ref{Subsection:States}, we required the set of states to be $\setstates=\setshared\times\setlocal$.\\
This is not satisfied in two ways.\\
First, we only have a subset of the Cartesian product.\\
We will have to revise the meta theory to work with $\setstates\subseteq\setshared\times\setlocal$.\\\\
Another point is that $\abort$ does not respect the product shape.\\
This is a minor issue as we could split the element into a fault on the global and a fault on the local side to recover the desired shape.\\
Keeping this splitting implicit will make the development cleaner. 
\end{comment}
% --------------------------------------------------------------------------------
% --------------------------------------------------------------------------------
% --------------------------------------------------------------------------------
% --------------------------------------------------------------------------------
When we modify the valuations of program variables and selectors, the result is not necessarily well-split.
We introduce the following function that turns a pair of global and local state into a state that is well-split:  
% --------------------------------------------------------------------------------
% --------------------------------------------------------------------------------
% --------------------------------------------------------------------------------
% --------------------------------------------------------------------------------
\begin{align*}
\splitfg:\setshared\times \setlocal\pfun \setstates. 
\end{align*}
% --------------------------------------------------------------------------------
% --------------------------------------------------------------------------------
% --------------------------------------------------------------------------------
% --------------------------------------------------------------------------------
The idea is to move the nodes from the local heap that are reachable from the shared state to the shared heap.
The function is partial because the nodes from the local heap may in turn point outside the state, in which case the effect would not be limited to the state.
Definedness of the function is as follows:
\begin{align*}
(\ashared, \alocal)\statemultdef\splitfg\quad\text{if}\quad
\ashared.\aflowgraph\statemultdef\alocal.\aflowgraph
\quad\text{and}\quad
&\forall \anaval\in\alocal.\aflowgraph.\setnodes~
\forall \apvar\in\setgpvars~~
\forall \anavalp\in\ashared.\aflowgraph.\setnodes~~
\forall \asel\in\setsel.~~\\
&\qquad\quad\begin{aligned}
&\apvar\text{ reaches }\anaval\;\;\vee\;\;\anavalp\text{ reaches }\anaval
\\\Rightarrow~&\alocal.\aflowgraph.\svalof{\anaval, \asel}\in\alocal.\aflowgraph.\setnodes\discup\ashared.\aflowgraph.\setnodes
\end{aligned}
\end{align*}
The reachability predicate $\apvar$ \emph{reaches} $\anaval$ resp. $\anavalp$ \emph{reaches} $\anaval$ is defined by either directly referencing address~$\anaval$ (in the case of $\anavalp$ through a selector) or by referencing an address that in turn reaches $\anaval$.
The formal definition is via a least fixed point.
If defined, we set $\splitOf{\ashared, \alocal}\defeq \astate$ with 
$\astate.\aflowgraph=\ashared.\aflowgraph\mstar\alocal.\aflowgraph$, $\astate.\aval=\ashared.\aval\discup\alocal.\aval$, $\ashared.\aflowgraph.\setnodes\subseteq\astate.\ashared.\aflowgraph.\setnodes$.
Moreover, $\astate.\alocal.\aflowgraph$ is the maximal (wrt. inclusion on $\astate.\aflowgraph.\setnodes$) flowgraph that satisfies these constraints.
The function is well-defined in that it uniquely determines a state and the state does exist.
To see this, note that the state can be obtained by a least fixed point iteration that moves local nodes into the shared~heap.
\begin{lemma}
$\splitOf{\ashared, \alocal}$ is well-defined. 
\end{lemma}

We will need that $\splitfg$ is a local function.
To formulate this, we define state modifications of the form $\updOf{\apvar}{\anaval}$ and $\updOf{(\anaval, \asel)}{\anavalp}$.
This will modify the corresponding component in the state.
The component will be clear from the type of the modification. 
For example, $\astate\updOf{(\anaval, \asel)}{\anavalp}$ will modify the valuation $\sval$ at the selector $\asel$ of address $\anaval$.
Since the flow graphs in the global and in the local state can be composed, only one of them will hold the address.
% ------------------------------------------------------------------------------------------------------
% ------------------------------------------------------------------------------------------------------
% ------------------------------------------------------------------------------------------------------
% ------------------------------------------------------------------------------------------------------

The modification itself should also have a local effect.
To formulate this, we introduce updates \[\updfun:\nat\times \nat\pfun \monfun{\amonoid\rightarrow\amonoid}\] as commands for changing flow constraints $\aflowconstraint$, which we then subsequently relate to the updates of heap graphs.
The application of the update yields the new flow constraint \[\aflowconstraint\applyupdfun\defeq(\aflowconstraint.\setnodes, \edges, \aflowconstraint.\inflow)\] where the edges are defined by $\edges(\anode, \anodep)\defeq \updfun(\anode, \anodep)$ if $\updfun$ is defined on $(\anode, \anodep)$ and $\edges(\anode, \anodep)\defeq \aflowconstraint.\edges(\anode, \anodep)$ otherwise.
Note that neither the nodes nor the inflow change.
In Figure~\ref{fig-keyset-flow}, e.g., we have $\aflowconstraint_2=\aflowconstraint_1\applyupdfun$ for $\updfun = \set{(l,t) \mapsto \lambda_\infty, (l,r) \mapsto \lambda_6}$.
% ------------------------------------------------------------------------------
% ------------------------------------------------------------------------------
% ------------------------------------------------------------------------------
% ------------------------------------------------------------------------------

When considering updates as commands, we let the semantics abort should the update not be defined on the flow constraint or should it change the transformer.
Using Lemma~\ref{Lemma:Transformers}, this makes the semantics suitable for framing:
\begin{align*}
\semComOf{\updfun}{\aflowconstraint}\defeq
\begin{cases}
\aflowconstraint\applyupdfun&\quad\text{if }
\begin{aligned}[t]
  &\domof{\updfun}\subseteq\aflowconstraint.\setnodes\times\nat \wedge \transformerof{\aflowconstraint}=_{\aflowconstraint.\inflow}\transformerof{\aflowconstraint\applyupdfun}
\end{aligned}\\
\abort&\quad\text{otherwise.}
\end{cases}
\end{align*}
%Here, $\abort$ is a distinguished symbol for a computation that \emph{aborts}.\\\\
%We generalize flow graph composition to abort by $\aflowconstraint\mstar\abort\defeq\abort$. \\\\
The interplay between the semantics of updates and composition is captured by the following lemma. It ensures that the locality condition assumed in \S~\ref{Section:Preliminaries} is satisfied, once we tie flow constraint updates to updates on the heap graphs from which the flow constraints are derived.
% ------------------------------------------------------------------------------
% ------------------------------------------------------------------------------
% ------------------------------------------------------------------------------
% ------------------------------------------------------------------------------
% ------------------------------------------------------------------------------
% ------------------------------------------------------------------------------
\begin{lemma}\label{Lemma:UpdatesConjunction}
Let $\semComOf{\updfun}{\aflowconstraint}\neq\abort$ and $\aflowconstraint\statemultdef\aflowconstraint'$. 
Then we have:
\begin{compactenum}[(i)]
\item $\semComOf{\updfun}{\aflowconstraint}\statemultdef \aflowconstraint'$ and
\item $\semComOf{\updfun}{\aflowconstraint}\mstar\aflowconstraint'=(\aflowconstraint\mstar\aflowconstraint')\applyupdfun$ and 
\item $\semComOf{\updfun}{\aflowconstraint\mstar\aflowconstraint'}\neq\abort$ and 
\item $\semComOf{\updfun}{\aflowconstraint\mstar\aflowconstraint'}=\semComOf{\updfun}{\aflowconstraint}\mstar\aflowconstraint'$.
\end{compactenum}
\end{lemma}

We can now associate with a modification $\updOf{(\anaval, \asel)}{\anavalp}$ an update $\updfun$ of the state's flow constraint $\astate.\aflowconstraint$.
Here, we use $\astate\statemultdef\updOf{(\anaval, \asel)}{\anavalp}$ to denote the fact that the update induced by the modification does not abort on the flow constraint, $\semOf{\updfun}(\astate.\aflowconstraint)\neq\abort$.
% ------------------------------------------------------------------------------------------------------
% ------------------------------------------------------------------------------------------------------
% ------------------------------------------------------------------------------------------------------
% ------------------------------------------------------------------------------------------------------
\begin{lemma}\label{Lemma:LocalitySplitSel}
If we have
\begin{asparaitem}
  \item $\astate\statemultdef\astate'$,
  \item $\anaval\in\astate.\ashared.\aflowgraph.\setnodes\Rightarrow\anavalp\in\astate.\aflowgraph.\setnodes$,
  \item $\astate\statemultdef\updOf{(\anaval, \asel)}{\anavalp}$, and
  \item $\astate\updOf{(\anaval, \asel)}{\anavalp}\statemultdef\splitfg$,
\end{asparaitem}
then we get
\begin{asparaitem}
  \item $(\astate\mstar\astate')\statemultdef\updOf{(\anaval, \asel)}{\anavalp}$,
  \item $(\astate\mstar\astate')\updOf{(\anaval, \apsel)}{\anavalp}\statemultdef\splitfg$, and
  \item $\splitOf{(\astate\mstar\astate')\updOf{(\anaval, \asel)}{\anavalp}}=\splitOf{\astate\updOf{(\anaval, \asel}{\anavalp}}\mstar \astate'$.
\end{asparaitem}
% $\astate\statemultdef\astate'$,
% $\anaval\in\astate.\ashared.\aflowgraph.\setnodes\Rightarrow\anavalp\in\astate.\aflowgraph.\setnodes$,
% $\astate\statemultdef\updOf{(\anaval, \asel)}{\anavalp}$, and
% $\astate\updOf{(\anaval, \asel)}{\anavalp}\statemultdef\splitfg$.
% Then $(\astate\mstar\astate')\statemultdef\updOf{(\anaval, \asel)}{\anavalp}$,
% $(\astate\mstar\astate')\updOf{(\anaval, \apsel)}{\anavalp}\statemultdef\splitfg$, and
% $\splitOf{(\astate\mstar\astate')\updOf{(\anaval, \asel)}{\anavalp}}=\splitOf{\astate\updOf{(\anaval, \asel}{\anavalp}}\mstar \astate'$.
The latter separating conjunction is defined. 
\end{lemma}

\begin{proof}
We have $(\astate\mstar\astate')\statemultdef\updOf{(\anaval, \asel)}{\anavalp}$ by Lemma~\ref{Lemma:UpdatesConjunction}(iii).\\
Moreover, $(\astate\mstar\astate')\updOf{(\anaval, \asel)}{\anavalp}=\astate\updOf{(\anaval, \asel)}{\anavalp}\times\astate'$ by Lemma~\ref{Lemma:UpdatesConjunction}(iv).
The product on the right hand side of the previous equality denotes the separating conjunction applied componentwise, on the global and on the local state.
It does not take into account well-splitting.
If $\anaval\in\astate.\alocal.\aflowgraph.\setnodes$, there is no splitting to be done  and we obtain
\begin{align*}
&~~\splitOf{(\astate\mstar\astate')\updOf{(\anaval, \asel)}{\anavalp}}\\
\commentl{No splitting}=&~~(\astate\mstar\astate')\updOf{(\anaval, \asel)}{\anavalp}\\
\commentl{Above}=&~~\astate\updOf{(\anaval, \asel)}{\anavalp}\times\astate'\\
\commentl{$\astate\statemultdef\astate'$, $\anaval\in\astate.\alocal.\aflowgraph$}=&~~\astate\updOf{(\anaval, \asel)}{\anavalp}\mstar\astate'\\
\commentl{No splitting}=&~~\splitOf{\astate\updOf{(\anaval, \asel)}{\anavalp}}\mstar\astate'\ . 
\end{align*}
If $\anaval\in\astate.\ashared.\aflowgraph.\setnodes$, we have $\anavalp\in\astate.\aflowgraph.\setnodes$.
Since $\astate, \astate'$ are states, they are well-split and so the global part of each state does not point into the local heap of that state.
Moreover, since $\astate\statemultdef\astate'$ by the assumption, the global part of state $\astate'$ does not point into the local heap of $\astate$.
With the same argument, the global part of $\astate$ does not point into the local heap of $\astate'$.\\
If $\anavalp\in\astate.\ashared.\aflowgraph.\setnodes$, this continues to hold for $\astate\updOf{(\anaval, \asel)}{\anavalp}$.\\
If $\anavalp\in\astate.\alocal.\aflowgraph.\setnodes$, this continues to hold for $\astate\updOf{(\anaval, \asel)}{\anavalp}$ as well, because $\astate\updOf{(\anaval, \asel)}{\anavalp}\statemultdef\splitfg$.\\ 
The definedness says that no local node from $\astate\updOf{(\anaval, \asel)}{\anavalp}$ reachable from the global part points outside $\astate\updOf{(\anaval, \asel)}{\anavalp}$, in particular $\anavalp$.
Note that traversing $\anavalp$ is the only way to reach a local node from the global part in $\astate\updOf{(\anaval, \asel)}{\anavalp}$.
With the above discussion, splitting the state \[(\astate\mstar\astate')\updOf{(\anaval, \asel)}{\anavalp}=\astate\updOf{(\anaval, \asel)}{\anavalp}\times\astate'\] will only move reachable nodes from $\astate.\alocal.\aflowgraph$ through $\anavalp$ to the global part.
Since $\astate\updOf{(\anaval, \asel)}{\anavalp}\statemultdef\splitfg$, no such local node from $\astate\updOf{(\anaval, \asel)}{\anavalp}$ that is reachable through $\anavalp$ points outside $\astate\updOf{(\anaval, \asel)}{\anavalp}$.
Hence $(\astate\mstar\astate')\updOf{(\anaval, \apsel)}{\anavalp}\statemultdef\splitfg$.

Splitting will add a new flow graph $\aflowgraph'$ to the shared heap of $(\astate\mstar\astate')\updOf{(\anaval, \apsel)}{\anavalp}$.
The result is 
\begin{align*}
&\quad(\astate\mstar\astate')\updOf{(\anaval, \apsel)}{\anavalp}.\ashared.\aflowgraph\mstar\aflowgraph'\\
% \commentl{Discussed above)}
\smash{\stackrel{(1)}{=}}&~~(\astate\updOf{(\anaval, \apsel)}{\anavalp}\times\astate').\ashared.\aflowgraph\mstar\aflowgraph'\\
% \commentl{Componentwise separating conjunction)}
\smash{\stackrel{(2)}{=}}&~~\astate\updOf{(\anaval, \apsel)}{\anavalp}.\ashared.\aflowgraph\mstar\astate'.\ashared.\aflowgraph\mstar\aflowgraph'\\
% \commentl{Commutativity)}
\smash{\stackrel{(3)}{=}}&~~ (\astate\updOf{(\anaval, \apsel)}{\anavalp}.\ashared.\aflowgraph\mstar\aflowgraph')\mstar\astate'.\ashared.\aflowgraph.
\end{align*}
where the qualities follow from (1) the above discussion, (2) componentwise separating conjunction, and (3) commutativity.
This is the shared heap of $\splitOf{\astate\updOf{(\anaval, \asel}{\anavalp}}$ composed with the shared heap of $\astate'$.
For the local heap, we already discussed above that we only move nodes from $\astate.\alocal.\aflowgraph$.
Hence, $\astate'$ remains unchanged.
\end{proof}

% --------------------------------------------------------------------------------
% --------------------------------------------------------------------------------
% --------------------------------------------------------------------------------
% --------------------------------------------------------------------------------
For modifications of program variables, there is a related result.
Since stacks are not split, however, the modification applies to both states of the separating conjunction.
% --------------------------------------------------------------------------------
% --------------------------------------------------------------------------------
% --------------------------------------------------------------------------------
% --------------------------------------------------------------------------------
\begin{lemma}\label{Lemma:LocalitySplitPvar}
Consider $\astate\statemultdef\astate'$ and $\apvar\in\setgpvars\Rightarrow \anavalp\in\astate.\aflowgraph.\setnodes$ and $\astate\updOf{\apvar}{\anavalp}\statemultdef\splitfg$. 
Then we have $(\astate\mstar\astate')\updOf{\apvar}{\anavalp}\statemultdef\splitfg$ and
${\splitOf{(\astate\mstar\astate')\updOf{\apvar\mkern-1mu}{\mkern-1mu\anavalp}}=\splitOf{\astate\updOf{\apvar\mkern-1mu}{\mkern-1mu\anavalp}}\mstar \astate'\updOf{\apvar\mkern-1mu}{\mkern-1mu\anavalp}}$.
The latter separating conjunction is defined. 
\end{lemma}

\begin{proof}
Assume $\apvar\in\setgpvars$, otherwise there is no splitting.
We have \[(\astate\mstar\astate')\updOf{\apvar}{\anavalp}=\astate\updOf{\apvar}{\anavalp}\mstar\astate'\updOf{\apvar}{\anavalp}\] since stacks are not split.
The claim then follows with an argumentation similar to the one in the previous lemma.
%\todo{Do the details.}
\end{proof}

% --------------------------------------------------------------------------------
% --------------------------------------------------------------------------------
% --------------------------------------------------------------------------------
% --------------------------------------------------------------------------------
We also have a modification that corresponds to an allocation. For defining allocations, we introduce the \emph{initial flow graph} for adress $\anaval\in\nat$, which is
\begin{align*}
\initfgof{\anaval}\;\;\defeq\;\;(\setnd{\anaval}, \setcond{(\anaval, \asel, \anaval)}{\asel\in\setsel}, \setcond{(\anavalp, \lambda_\monunit, \anaval)}{\anavalp\in\nat\setminus\setnd{\anaval}}) \ . 
\end{align*}
In the latter set, $\lambda_\monunit$ is the constant $\monunit$ function.
Every selector points to the fresh node and the inflow is $\monunit$ on all incoming edges.
Allocation is then captured by the function
\begin{align*}
\extfg:\setstates\times\nat\pfun\setstates.
\end{align*}
It extends the local flow graph by an initial flow graph for the given address.
The function yields \[\extOf{\astate}{\anaval}\defeq (\astate.\ashared, (\astate.\alocal.\aflowgraph\mstar\initfgof{\anaval}, \astate.\alocal.\aval))\ .\]
It is defined only if the address is fresh, $\anaval\notin\astate.\aflowgraph.\setnodes$, and not pointed to from a global variable or the global heap of $\astate$.
We denote this definedness by $(\astate, \anaval)\statemultdef\extfg$.
The result is well-split.
% --------------------------------------------------------------------------------
% --------------------------------------------------------------------------------
% --------------------------------------------------------------------------------
% --------------------------------------------------------------------------------
\begin{lemma}
If $(\astate, \anaval)\statemultdef\extfg$ then $\extOf{\astate}{\anaval}\in\setstates$.
\end{lemma}
%\todo{Proof}
% ------------------------------------------------------------------------------
% ------------------------------------------------------------------------------
% ------------------------------------------------------------------------------
% ------------------------------------------------------------------------------
\subsection{Commands}
The set of \myemph{commands} is 
\begin{align*}
\acom~::=~~ &\cskip
\bnf \assignof{\apvar}{\anaopof{\apvarp_1,\ldots, \apvarp_n}}
\bnf \assignof{\apvar}{\selof{\apvarp}{\asel}}
\bnf \assignof{\selof{\apvar}{\asel}}{\apvarp} \\
\bnf &~\mallocof{\apvar}
\bnf \assumeof{\anafpredof{\apvar_1,\ldots, \apvar_n}}
\bnf \assertof{\anafpredof{\apvar_1,\ldots, \apvar_n}}
\ . 
\end{align*}
Here, $\apvar, \apvarp\in\setpvars$ are program variables, $\asel\in\setsel$ is a selector, and $\anaop:\nat^n\rightarrow\nat$ and $\anafpred:\nat^n\rightarrow\bool$ are an operation resp. a predicate over the address domain.

% ------------------------------------------------------------------------------------------------------
% ------------------------------------------------------------------------------------------------------
% ------------------------------------------------------------------------------------------------------
% ------------------------------------------------------------------------------------------------------
% ------------------------------------------------------------------------------------------------------
% ------------------------------------------------------------------------------------------------------
% ------------------------------------------------------------------------------------------------------
% ------------------------------------------------------------------------------------------------------
The semantics is strict in $\abort$:
\begin{align*}
\semComOf{\acom}{\abort}\;\defeq\;\set{\abort}\; .
\end{align*}
For a state $\astate\neq \abort$, we define the following preconditions:
\begin{align*}
\assignfromop~\defeq~
&~~\anavalp = \semOf{\anaopof{\apvar_1,\ldots, \apvar_n}}\ \astate.\aval
\\&~\,\wedge (\apvar\in\setgpvars\Rightarrow \anavalp\in\astate.\aflowgraph.\setnodes)\wedge \astate\updOf{\apvar}{\anavalp}\statemultdef\splitfg\\
\assignfromsel~\defeq~
&~~\anaval= \semOf{\apvarp}\  \astate.\aval\wedge 
\anaval\in\astate.\aflowgraph.\setnodes
\wedge\anavalp =  \astate.\avalof{\anaval, \asel}
\\&~\,\wedge
(\apvar\in\setgpvars\Rightarrow \anavalp\in\astate.\aflowgraph.\setnodes) \wedge \astate\updOf{\apvar}{\anavalp}\statemultdef\splitfg\\
\assigntosel~\defeq~
&~~\anaval = \semOf{\apvar}\ \astate.\aval \wedge 
\anavalp = \semOf{\apvarp}\ \astate.\aval\wedge  \astate\statemultdef\updOf{(\anaval, \asel)}{\anavalp}
\\&~\,\wedge
(\anaval\in\astate.\aflowgraph.\ashared.\setnodes\Rightarrow \anavalp\in\astate.\aflowgraph.\setnodes)\wedge
\astate\updOf{(\anaval, \asel)}{\anavalp}\statemultdef\splitfg.
\end{align*}
We use the notation $\condval{\mathit{condition}}{\mathit{value}}$ as a short hand for the expression $\mathit{condition}\;\:?\;\:\mathit{value}\:\;:\;\setcompact{\abort}$. 
The semantics of commands is given in Figure~\ref{sem-commands-instantiation-appendix}.

\begin{figure*}
\begin{align*}
\semComOf{\cskip}{\astate}\quad &\defeq\quad\set{\astate}\\
\semComOf{\mallocof{\apvar}}{\astate}\quad &\defeq\quad\setcond{\splitOf{\astate'\updOf{\apvar}{\anaval}}}{\anaval\notin\astate.\aflowgraph.\setnodes\wedge \astate'=\extOf{\astate}{\anaval}}\\
\semComOf{\assumeof{\anafpredof{\apvar_1,\ldots, \apvar_n}}}{\astate}\quad &\defeq\quad\semOf{\anafpredof{\apvar_1,\ldots, \apvar_n}}\ \astate.\aval\;?\;\set{\astate}\;:\;\emptyset\\
\semComOf{\assertof{\anafpredof{\apvar_1,\ldots, \apvar_n}}}{\astate}\quad &\defeq\quad
\condval{\semOf{\anafpredof{\apvar_1,\ldots, \apvar_n}}\ \astate.\aval}{\set{\astate}}\\
\semComOf{\assignof{\apvar}{\anaopof{\apvar_1,\ldots, \apvar_n}}}{\astate}\quad &\defeq\quad
\condval{\assignfromop}{\set{\splitOf{\astate\updOf{\apvar}{\anavalp}}}}\\
\semComOf{\assignof{\apvar}{\selof{\apvarp}{\asel}}}{\astate}\quad &\defeq\quad
\condval{\assignfromsel}{\set{
\splitOf{\astate\updOf{\apvar}{\anavalp}}}}\\
\semComOf{\assignof{\selof{\apvar}{\asel}}{\apvarp}}{\astate}\quad &\defeq\quad
\condval{\assigntosel}{\set{
\splitOf{\astate\updOf{(\anaval, \asel)}{\anavalp}}}}.
\end{align*}
\vspace{-6mm}
\caption{Semantics of commands.\label{sem-commands-instantiation-appendix}}
\end{figure*}
% ------------------------------------------------------------------------------
% ------------------------------------------------------------------------------
\begin{comment}
\begin{lemma}
\todo{Split does not abort in the case of malloc}.
\end{lemma}
% ------------------------------------------------------------------------------
% ------------------------------------------------------------------------------
\begin{remark}
$\nil$ als spezielle Adresse, die nie alloziert wird.
\end{remark}
\end{comment}
% ------------------------------------------------------------------------------
% ------------------------------------------------------------------------------
% ------------------------------------------------------------------------------
% ------------------------------------------------------------------------------
\subsection{Assertions}
We use countable sets $\setavars$ and $\setfvars$ of logical variables with typical elements $\anavar$ resp. $\anfvar$.
The first set holds variables that store addresses, the second holds variables for flow values.
The logical variables are assumed to be disjoint and disjoint from the program variables.
We use \[\setlogval \defeq (\setavars\rightarrow\nat)\discup (\setfvars\rightarrow\amonoid)\] for the set of all valuations of the logical variables that respect the typing.
% ------------------------------------------------------------------------------
% ------------------------------------------------------------------------------
% ------------------------------------------------------------------------------
% ------------------------------------------------------------------------------
We write \[\setproglogval\defeq\setvalsof{\setlpvars}\discup\setvalsof{\setgpvars}\discup\setlogval\] for all valuations of the program and logical variables.

% ------------------------------------------------------------------------------
% ------------------------------------------------------------------------------
% ------------------------------------------------------------------------------
% ------------------------------------------------------------------------------
Terms are either \emph{address terms} of type $\nat$ or \emph{flow terms} of type $\amonoid$.
We write $\adom$ to mean either $\nat$ or $\amonoid$.
To form terms, we assume the address domain and the flow domain each come with a set of operations $\setaops$ resp. $\setfops$ defined on them.
Each operation has a type.
For address operations, this type has the form $\adom_1\rightarrow\ldots\rightarrow\adom_n\rightarrow \nat$.
For flow operations, the result is a monoid value.
To ease the notation, we simplify the types to $(\natmon)^n\rightarrow \adom$ and write $\afterm$ for an address or a flow term.
In applications, we will respect the precise types.
We also assume to have a set of predicates $\setafpreds$ that follow a similar typing scheme.

% ------------------------------------------------------------------------------
% ------------------------------------------------------------------------------
% ------------------------------------------------------------------------------
% ------------------------------------------------------------------------------
The sets of \myemph{address} and \myemph{flow terms} are defined by mutual induction:
\begin{align*}
\aterm\quad &::=\quad \anaval\bnf \anavar\bnf\apvar\bnf \anaopof{\afterm_1,\ldots, \afterm_n}\\
\fterm\quad &::=\quad\anfval\bnf \anfvar \bnf \anfopof{\afterm_1,\ldots, \afterm_n}.   
\end{align*}
An address term is an address $\anaval\in\nat$, a logical address variable $\anavar\in\setavars$, a program variable $\apvar\in\setpvars$, and an operation $\anaop\in\setaops$ applied to $n$ terms, with $n$ is the operation's arity.
A flow term is a flow value $\anfval\in\amonoid$, a logical flow variable $\anfvar\in\setfvars$, and a flow operation $\anfop\in\setfops$ applied to terms.

% ------------------------------------------------------------------------------
% ------------------------------------------------------------------------------
The semantics expects a valuation $\apvlvval$ of the program and logical variables and returns an address or flow value: 
\begin{align*}
\semOf{\afterm}: \setproglogval\rightarrow \natmon.
\end{align*} 
The definition is as follows:
\begin{gather*}
\semOf{\mathit{d}}\ \apvlvval~\defeq~ \mathit{d}
\qquad\qquad
\semOf{\mathit{x}}\ \apvlvval~\defeq~ \apvlvvalof{\mathit{x}}
\\
\begin{aligned}
  \semOf{\mathsf{op}(\afterm_1,\ldots, \afterm_n)}\ \apvlvval ~\defeq~ \mathsf{op}(\semOf{\afterm_1}\ \apvlvval,\ldots, \semOf{\afterm_n}\ \apvlvval)
  \ . 
\end{aligned}
\end{gather*}
Here, $\mathit{d}$ is an address $\anaval\in\nat$ or a flow value $\anfval\in\amonoid$, $\mathit{x}$ is a logical flow variable $\anfvar\in\setfvars$, a logical address variable $\anavar\in\setavars$, or a program variable $\apvar\in\setpvars$, and $\mathsf{op}$ is a flow operation $\anfop\in\setfops$ or an address operation~${\anaop\in\setaops}$.

% ----------------------------------------------------------------------------
% ----------------------------------------------------------------------------
% ----------------------------------------------------------------------------
% ----------------------------------------------------------------------------
% ----------------------------------------------------------------------------
% ----------------------------------------------------------------------------
% ----------------------------------------------------------------------------
% ----------------------------------------------------------------------------
The set of \myemph{atomic predicates} is defined as follows, with  $\asel\in\setsel$:
\begin{align*}
\atomicpred~::=~~~&\emp\bnf\anafpredof{\afterm_1, \ldots, \afterm_n}\\\bnf&\aterm_1\pointsto{\asel}\aterm_2
\bnf \aterm\pointsto{\insel}\fterm
\bnf \aterm\pointsto{\fsel}\fterm.
\end{align*}
% ----------------------------------------------------------------------------
% ----------------------------------------------------------------------------
A flow graph $\aflowgraph$ and a valuation $\apvlvval$ of the program and logical variables satisfy an atomic predicate as prescribed by Figure~\ref{satisfaction-atomic-predicates-appendix}. %as follows:
% ----------------------------------------------------------------------------
% ----------------------------------------------------------------------------
\begin{figure*}
\begin{alignat*}{5}
\hspace{.3cm}
\aflowgraph, \apvlvval&\models \emp\quad&\text{if}\quad&\aflowgraph.\setnodes=\emptyset\\
\aflowgraph, \apvlvval&\models \anafpredof{\afterm_1, \ldots, \afterm_n}\quad&\text{if}\quad&\anafpredof{
\semOf{\afterm_1}\ \apvlvval, \ldots, \semOf{\afterm_n}\ \apvlvval} \\
\aflowgraph, \apvlvval&\models \aterm_1\pointsto{\asel}\aterm_2\quad&\text{if}\quad&\exists\anaval, \anavalp.\;\anaval=\semOf{\aterm_1}\ \apvlvval\;\wedge\; \anavalp=\semOf{\aterm_2}\ \apvlvval\\&&&\qquad~\wedge\;\aflowgraph.\setnodes=\setnd{\anaval}\;\wedge\;\aflowgraph.\svalof{\anaval, \asel}=\anavalp\\
\aflowgraph, \apvlvval&\models \aterm\pointsto{\insel}\fterm\quad&\text{if}\quad&
\exists\anaval, \anfval.\;\anaval=\semOf{\aterm}\ \apvlvval\;\wedge\; \anfval=\semOf{\fterm}\ \apvlvval\\&&&\qquad~\;\wedge\;\aflowgraph.\setnodes=\setnd{\anaval}\;\wedge\;\aflowgraph.\inflowof{\anaval}=\anfval\\
\aflowgraph, \apvlvval&\models \aterm\pointsto{\fsel}\fterm\quad&\text{if}\quad&
\exists\anaval, \anfval.\;\anaval=\semOf{\aterm}\ \apvlvval\;\wedge\; \anfval=\semOf{\fterm}\ \apvlvval\\&&&\qquad~\;\wedge\;\aflowgraph.\setnodes=\setnd{\anaval}\;\wedge\;\aflowgraph.\fvalof{\anaval}=\anfval. 
\end{alignat*}
\vspace{-3mm}
\caption{Satisfaction relation for atomic predicates.\label{satisfaction-atomic-predicates-appendix}}
\end{figure*}
% ----------------------------------------------------------------------------
% ----------------------------------------------------------------------------
% ----------------------------------------------------------------------------
% ----------------------------------------------------------------------------

The set of \myemph{state predicates} is defined by
\begin{align*}
\statepred\quad::=\quad &\atomicpred\bnf \neg \statepred\bnf\statepred_1\wedge\statepred_2\bnf \exists \anavar.\statepred\bnf \exists \anfvar.\statepred
\bnf \statepred_1\mstar\statepred_2\bnf\statepred_1\sepimp\statepred_2\bnf \fbox{$\statepred$}\; . 
\end{align*}
It will be convenient to have the box predicate \fbox{$\statepred$} as introduced in RGSep~\cite{DBLP:conf/concur/VafeiadisP07,DBLP:phd/ethos/Vafeiadis08}.
The predicate is meant to describe constraints over the shared state.
Therefore, we expect $\statepred$ to contain neither local program variables nor further boxes.

% ----------------------------------------------------------------------------
% ----------------------------------------------------------------------------
The $\abort$ state never satisfies a state predicate.\\
For states $\astate=(\ashared, \alocal)\in\setstates$ and valuations $\anlvval$ of the logical variables that occur free in the formula the satisfaction relation is defined as follows:
% ----------------------------------------------------------------------------
% ----------------------------------------------------------------------------
\begin{alignat*}{5}
\astate, \anlvval&\models \atomicpred\quad&\text{if}\quad&\astate.\alocal.\aflowgraph,\anlvval\discup\astate.\aval\models\atomicpred\\
\astate, \anlvval&\models \neg\statepred\quad&\text{if}\quad&\astate, \anlvval\not\models \statepred\\
\astate, \anlvval&\models \statepred_1\wedge \statepred_2\quad&\text{if}\quad&\astate, \anlvval\models \statepred_1\;\;\wedge\;\;\astate, \anlvval\models \statepred_2\\
\astate, \anlvval&\models \exists \anavar.\statepred\quad&\text{if}\quad&\exists \anaval\in\nat.\;\;\astate, \anlvval[\anavar\mapsto\anaval]\models \statepred\\
\astate, \anlvval&\models \exists \anfvar.\statepred\quad&\text{if}\quad&\exists \anfval\in\amonoid.\;\;\astate, \anlvval[\anfvar\mapsto\anfval]\models \statepred\\
\astate, \anlvval&\models \statepred_1\mstar \statepred_2\quad&\text{if}\quad&\exists \alocal_1, \alocal_2.\;\;\alocal_1\statemultdef\alocal_2\;\wedge\;\alocal_1\mstar\alocal_2=\astate.\alocal\;\wedge
(\astate.\ashared, \alocal_1), \anlvval\models \statepred_1\;\wedge\;(\astate.\ashared, \alocal_2), \anlvval\models \statepred_2\\
\astate, \anlvval&\models \statepred_1\sepimp \statepred_2\quad&\text{if}\quad&\forall \astate'.
\begin{aligned}[t]
  \;\;(\astate\statemultdef\astate'\;\wedge\;\astate', \anlvval\models\statepred_1)\Rightarrow\; \astate\mstar\astate', \anlvval\models \statepred_2
\end{aligned}\\
\astate, \anlvval&\models \fbox{$\statepred$}\quad&\text{if}\quad&\astate.\alocal\in\emp_\localindex\;\;\wedge\;\;\astate.\ashared, \anlvval\modelsp\statepred.
\end{alignat*}
% ----------------------------------------------------------------------------
% ----------------------------------------------------------------------------
The definition of $\modelsp$ is similar.
The key difference is that we split the global flow graph when defining the separating conjunction, and similar for the separating implication.
Note that by the assumption on the shape of formulas, we do not need a valuation of the local variables to determine the truth of atomic predicates that occur inside boxes.

% ----------------------------------------------------------------------------
% ----------------------------------------------------------------------------
% ----------------------------------------------------------------------------
% ----------------------------------------------------------------------------
The set of \myemph{computation predicates} is
% ----------------------------------------------------------------------------
% ----------------------------------------------------------------------------
\begin{align*}
\comppred\quad::=\quad&\nowOf{\statepred}\bnf\pastOf{\statepred}\bnf \comppred\wedge\comppred\bnf \comppred\vee\comppred
\bnf \comppred\mstar\comppred\bnf\exists \anavar.\comppred\bnf\exists \anfvar.\comppred.
\end{align*}
% ----------------------------------------------------------------------------
% ----------------------------------------------------------------------------

The semantics is defined over non-empty sequences $\sigma$ with $\astateseq\in(\setstates\discup\setnd{\abort})^+$ and valuations of logical variables $\anlvval$ as follows.
A sequence containing $\abort$ never satisfies a predicate.
For a sequence $\astateseq\in\setstates^+$, we define:
% ----------------------------------------------------------------------------
% ----------------------------------------------------------------------------
\begin{alignat*}{5}
\astateseq.\astate, \anlvval&\models \nowOf{\statepred}\quad&\text{if}\quad&\astate, \anlvval\models \statepred\\
\astateseq, \anlvval&\models \pastOf{\statepred}\quad&\text{if}\quad&\exists \astate\in\setstates, \astateseq_1, \astateseq_2\in\setstates^*.
\begin{aligned}[t]
  \astateseq=\astateseq_1.\astate.\astateseq_2\wedge\;\astate, \anlvval\models \statepred
\end{aligned}\\
\astateseq, \anlvval&\models \comppred_1\wedge\comppred_2\quad&\text{if}\quad&\astateseq, \anlvval\models\comppred_1\;\;\wedge\;\;\astateseq, \anlvval\models\comppred_2\\
\astateseq, \anlvval&\models \comppred_1\vee\comppred_2\quad&\text{if}\quad&\astateseq, \anlvval\models\comppred_1\;\;\vee\;\;\astateseq, \anlvval\models\comppred_2\\
\astateseq, \anlvval&\models \comppred_1\mstar\comppred_2\quad&\text{if}\quad&\exists \astateseq_1, \astateseq_2\in\setstates^+.\;\;\astateseq_1\statemultdef\astateseq_2\;\wedge\;\astateseq=\astateseq_1\mstar\astateseq_2
\;\wedge\;\astateseq_1, \anlvval\models\comppred_1\;\wedge\;\astateseq_2, \anlvval\models\comppred_2\\
\astateseq, \anlvval&\models \exists \anavar.\comppred\quad&\text{if}\quad&\exists \anaval\in\nat.\;\;\astateseq, \anlvval[\anavar\mapsto\anaval]\models\comppred\\
\astateseq, \anlvval&\models \exists \anfvar.\comppred\quad&\text{if}\quad&\exists \anfval\in\amonoid.\;\;\astateseq, \anlvval[\anfvar\mapsto\anfval]\models\comppred.
\end{alignat*}
% ----------------------------------------------------------------------------
% ----------------------------------------------------------------------------
A computation predicate is \emph{closed} if it does not contain free logical variables.
In this case, the satisfaction relation only depends on the computation but is independent of the valuation.
We write $\semOf{\comppred}$ for the set of computations that satisfy a predicate.
We need the following lemma, which is a consequence of Lemma~\ref{Lemma:FrameablePredicates} with the consideration of quantifiers added. 
% ----------------------------------------------------------------------------
% ----------------------------------------------------------------------------
% ----------------------------------------------------------------------------
% ----------------------------------------------------------------------------
\begin{lemma}\label{Lemma:SyntPredFrameable}
 $\semOf{\comppred}$ is frameable for every $\comppred$.
\end{lemma}
% ----------------------------------------------------------------------------
% ----------------------------------------------------------------------------
% ----------------------------------------------------------------------------
% ----------------------------------------------------------------------------
%\begin{proof}
%\todo{Do this.}
%\end{proof}
% ----------------------------------------------------------------------------
% ----------------------------------------------------------------------------
% ----------------------------------------------------------------------------
% ----------------------------------------------------------------------------
\begin{lemma}\label{Lemma:NotFree}
If we have $\astateseq.\astate\in \semOf{\comppred}$ and $\apvar\notin\freeVariablesOf{\comppred}$, then we also have $\astateseq.\astate.(\astate\updOf{\apvar}{\anaval})\in \semOf{\comppred}$. 
\end{lemma}
% ----------------------------------------------------------------------------
% ----------------------------------------------------------------------------
% ----------------------------------------------------------------------------
% ----------------------------------------------------------------------------
%\begin{proof}
%\todo{Do this.}
%\end{proof}
% ----------------------------------------------------------------------------
% ----------------------------------------------------------------------------
In addition to the equivalences listed in Lemma~\ref{Lemma:NowSLOperators}, we have the following for quantifiers over logical variables.
% ----------------------------------------------------------------------------
% ----------------------------------------------------------------------------
\begin{lemma}
We have: \[
  \astateseq, \anlvval\models \exists \anavar.\nowOf{\statepred} \iff \astateseq, \anlvval\models \nowOf{(\exists \anavar.\statepred)}
  \qquad\text{and}\qquad
  \astateseq, \anlvval\models \exists \anfvar.\nowOf{\statepred} \iff \astateseq, \anlvval\models \nowOf{(\exists \anfvar.\statepred)}
  \ .
\]
\end{lemma}
%\todo{Don't forget the proof.}
% ----------------------------------------------------------------------------
% ----------------------------------------------------------------------------
% ----------------------------------------------------------------------------
% ----------------------------------------------------------------------------
\subsection{Semantics}
% -----------------------------------------------------------------------------
% -----------------------------------------------------------------------------
% -----------------------------------------------------------------------------
% -----------------------------------------------------------------------------
With the development in Section~\ref{Section:History}, the semantics generalizes to computations.
The key property required in Section~\ref{Subsection:States} is as follows.
%
% -----------------------------------------------------------------------------
% -----------------------------------------------------------------------------
% -----------------------------------------------------------------------------
% -----------------------------------------------------------------------------
\begin{lemma}\label{Lemma:SemanticsNoPvar}
Let $\acom$ be different from $\assignof{\apvar}{\anaopof{\apvarp_1,\ldots, \apvarp_n}}$, $\assignof{\apvar}{\selof{\apvarp}{\asel}}$, and $\mallocof{\apvar}$. 
Then $\csemOf{\acom}(\semOf{\comppred_1})\subseteq\semOf{\comppred_2}$ implies $\csemOf{\acom}(\semOf{\comppred_1\mstar\comppred})\subseteq\semOf{\comppred_2\mstar\comppred}$.
\end{lemma} 

\begin{proof}
Recall: $\semOf{\comppred_1\mstar\comppred}$ contains sequences of states $\astateseq.\astate\in \setstates^+$.
It is readily checked that $\semOf{\comppred_1\mstar\comppred} = \semOf{\comppred_1}\mstar\semOf{\comppred}$.
So there are $\astateseq_1.\astate_1\in\semOf{\comppred_1}$ and $\astateseq_2.\astate_2\in\semOf{\comppred_2}$ with $\astateseq.\astate=\astateseq_1.\astate_1\mstar\astateseq_2.\astate_2$.
We have that $\csemOf{\acom}(\astateseq.\astate)$ extends the sequence $\astateseq.\astate$ by elements from $\semComOf{\acom}{\astate}$.
Such elements are either $\abort$ or states.\\\\
% -----------------------------------------------------------------------------
% -----------------------------------------------------------------------------
% -----------------------------------------------------------------------------
% -----------------------------------------------------------------------------
\textit{Case $\abort$.}\quad
We  show that $\semComOf{\acom}{\astate}=\abort$ implies $\semComOf{\acom}{\astate_1}=\abort$.
Then $\csemOf{\acom}(\comppred_1)$ contains a sequence with $\abort$, contradicting $\csemOf{\acom}(\semOf{\comppred_1})\subseteq\semOf{\comppred_2}$.
We distinguish the following cases:
\begin{compactitem}
\item Skip and assume do not abort.
\item An assertion only depends on the valuation of the program variables.
Valuations are never split.
This means $\astate_1$ will have the same valuation, and lead to the same abort.
\item We have excluded assignments to program variables and malloc.
\item An assignment to a selector aborts.
\end{compactitem}
It remains to argue for the last case.
An assignment to a selector aborts if the condition $\assigntosel$ fails.
% \begin{align*}
% \assigntosel\qquad=\qquad&
% \anaval = \semOf{\apvar}\ \astate.\aval \wedge 
% \anavalp = \semOf{\apvarp}\ \astate.\aval\wedge  \astate\statemultdef\updOf{(\anaval, \asel)}{\anavalp} \\
% \wedge\ & 
% (\anaval\in\astate.\aflowgraph.\ashared.\setnodes\Rightarrow \anavalp\in\astate.\aflowgraph.\setnodes)\wedge
% \astate\updOf{(\anaval, \asel)}{\anavalp}\statemultdef\splitfg.
% \end{align*}
%
We consider each case and argue that the abort also happens on $\astate_1$.
Note that valuations are never split, so the addresses $\anaval$ and $\anavalp$ are the same in $\astate_1$.
If $\astate\statemultdef\updOf{(\anaval, \asel)}{\anavalp}$ fails, we have $\semComOf{\updfun}{\astate.\aflowconstraint}=\abort$, with $\updfun$ induced by $\updOf{(\anaval, \asel)}{\anavalp}$.
By Lemma~\ref{Lemma:UpdatesConjunction} applied in contraposition, $\semOf{\updfun}(\astate_1.\aflowconstraint)=\abort$.
The contraposition applies because $\astate_1\statemultdef\astate_2$ implies $\astate_1.\aflowconstraint\statemultdef\astate_2.\aflowconstraint$ and $\astate_1.\aflowconstraint\mstar\astate_2.\aflowconstraint = \astate.\aflowconstraint$.
Hence, the command will abort on $\astate_1$.

If the implication $\anaval\in\astate.\aflowgraph.\ashared.\setnodes\Rightarrow \anavalp\in\astate.\aflowgraph.\setnodes$ fails, we have $\anaval\in\astate.\aflowgraph.\ashared.\setnodes$ and $\anavalp\notin\astate.\aflowgraph.\setnodes$.
To see that the implication also fails on $\astate_1$, consider first the case that the implication becomes trivial, $\anaval\notin\astate_1.\aflowgraph.\ashared.\setnodes$.
Then $\semOf{\updfun}(\astate_1.\aflowconstraint)=\abort$, because the domain of the update induced by the modification is not included in the state.
If the implication is non-trivial, $\anavalp\notin\astate.\aflowgraph.\setnodes$ implies $\anavalp\notin\astate_1.\aflowgraph.\setnodes$.
This is because of $\astate.\aflowgraph.\setnodes = \astate_1.\aflowgraph.\setnodes\discup\astate_2.\aflowgraph.\setnodes$.
Hence the implication, and so the command, also fail on $\astate_1$.

If $\astate\updOf{(\anaval, \asel)}{\anavalp}\statemultdef\splitfg$ does not hold, then Lemma~\ref{Lemma:LocalitySplitSel} applies in contraposition and yields a failure of definedness, $\astate_1\updOf{(\anaval, \asel)}{\anavalp}\statemultdef\splitfg$.
A failure of $\astate_1\updOf{(\anaval, \asel)}{\anavalp}\statemultdef\splitfg$ in turn  means that the command aborts on $\astate_1$.
To see that the contraposition applies, note that $\astate_1\mstar\astate_2=\astate$.
Moreover, we can rely on the implication $\anaval\in\astate_1.\aflowgraph.\ashared.\setnodes\Rightarrow \anavalp\in\astate_1.\aflowgraph.\setnodes$, for otherwise the command would abort on $\astate_1$.
Similarly, if $\astate_1\statemultdef\updOf{(\anaval, \asel)}{\anavalp}$ failed, we would also abort on $\astate_1$.\\\\
% \end{compactitem}
% ------------------------------------------------------------------------------
% ------------------------------------------------------------------------------
% ------------------------------------------------------------------------------
% ------------------------------------------------------------------------------
\textit{Case $\astate'\in \semComOf{\acom}{\astate}$.}\quad
We show that there is $\astate_1'\in\semComOf{\acom}{\astate_1}$ with $\astate_1'\statemultdef\astate_2$ and $\astate'=\astate_1'\mstar\astate_2$.
Then $\csemOf{\acom}(\semOf{\comppred_1})\subseteq\semOf{\comppred_2}$, $\astateseq_1.\astate_1\in\semOf{\comppred_1}$, $\astate_1'\in\semComOf{\acom}{\astate_1}$ imply $\astateseq_1.\astate_1.\astate_1'\in \semOf{\comppred_2}$.
Since $\semOf{\comppred}$ is frameable by Lemma~\ref{Lemma:SyntPredFrameable}, $\astateseq_2.\astate_2.\astate_2\in\semOf{\comppred}$.
So, $\astateseq.\astate.\astate'=\astateseq_1.\astate_1.\astate_1'\mstar\astateseq_2.\astate_2.\astate_2 \in \semOf{\comppred_2}\mstar\semOf{\comppred}=\semOf{\comppred_2\mstar\comppred}$.
Case analysis:
\begin{compactitem}
	\item For skip, we have $\astate'=\astate$ and $\astate_1'=\astate_1\in\semComOf{\cskip}{\astate_1}$.
	Then $\astate'=\astate_1'\mstar\astate_2$. 
	\item Assume and assert are similarly simple, as the semantics only depends on the  program variable valuation and that is not split.
	\item We have excluded assignments and mallocs to program variables.
	\item For an assignment to a selector that results in the state $\astate' = \splitOf{\astate\updOf{(\anaval, \asel)}{\anavalp}}$, we have condition $\assigntosel$ for $\astate_1$ as follows:
	\begin{align*}
		\qquad&\anaval = \semOf{\apvar}\ \astate_1.\aval ~\wedge~ 
		\anavalp = \semOf{\apvarp}\ \astate_1.\aval~\wedge 
			\astate_1\statemultdef\updOf{(\anaval, \asel)}{\anavalp}
			\\\wedge~& \astate_1\updOf{(\anaval, \asel)}{\anavalp}\statemultdef\splitfg ~\wedge
 (\anaval\in\astate_1.\aflowgraph.\ashared.\setnodes\Rightarrow \anavalp\in\astate_1.\aflowgraph.\setnodes)
\ .
\end{align*}
	Otherwise, the assignment would abort on $\astate_1$,contradicting $\csemOf{\acom}(\semOf{\comppred_1})\subseteq\semOf{\comppred_2}$.
	When applied to $\astate_1$, we obtain $\astate_1'=\splitOf{\astate_1\updOf{(\anaval, \asel)}{\anavalp}}$.
	Then 
	\begin{align*}
	\astate_1'\mstar\astate_2&=\splitOf{\astate_1\updOf{(\anaval, \asel)}{\anavalp}}
\mstar\astate_2\\
& = \splitOf{(\astate_1\mstar\astate_2)\updOf{(\anaval, \asel)}{\anavalp}}=	\splitOf{\astate\updOf{(\anaval, \asel)}{\anavalp}} = \astate'
	\end{align*}
	The first equality is the definition of $\astate_1'$, the next is Lemma~\ref{Lemma:LocalitySplitSel}, the third is the fact that $\astate_1\mstar\astate_2=\astate$, and the last is the definition of $\astate'$.
	The lemma applies by $\assigntosel$ stated above.
\end{compactitem}
This concludes the proof.
\end{proof}

%---------------------------------------------------------------------
% ---------------------------------------------------------------------
% ---------------------------------------------------------------------
% ---------------------------------------------------------------------
Since we do not use program variables as resources, the above lemma will not hold without prerequisites for assignments to program variables and malloc.
The prerequisite will be that the program variable modified by the command does not occur free in the predicate that is added with the separating conjunction.
The prerequisite will then form a side-condition to applying the frame rule. 
% ---------------------------------------------------------------------
% ---------------------------------------------------------------------
% ---------------------------------------------------------------------
% ---------------------------------------------------------------------
\begin{lemma}\label{Lemma:SemanticsPvar}
Let $\acom$ be one of: $\assignof{\apvar}{\anaopof{\apvarp_1,\ldots, \apvarp_n}}$ or $\assignof{\apvar}{\selof{\apvarp}{\asel}}$ or $\mallocof{\apvar}$. 
% Let $\acom$ be $\assignof{\apvar}{\anaopof{\apvarp_1,\ldots, \apvarp_n}}$, $\assignof{\apvar}{\selof{\apvarp}{\asel}}$, or $\mallocof{\apvar}$. 
Let $\apvar\notin\freeVariablesOf{\comppred}$. 
Then we have: $\csemOf{\acom}(\semOf{\comppred_1})\subseteq\semOf{\comppred_2}$ implies $\csemOf{\acom}(\semOf{\comppred_1\mstar\comppred})\subseteq\semOf{\comppred_2\mstar\comppred}$.
\end{lemma}
% ---------------------------------------------------------------------
% ---------------------------------------------------------------------
% ---------------------------------------------------------------------
% ---------------------------------------------------------------------
\begin{proof}
  The proof follows the same strategy as the proof for the previous lemma.
Recall that $\semOf{\comppred_1\mstar\comppred}$ contains sequences of states $\astateseq.\astate\in \setstates^+$.
Hence, there are some ${\astateseq_1.\astate_1\in\semOf{\comppred_1}}$ and $\astateseq_2.\astate_2\in\semOf{\comppred_2}$ with $\astateseq.\astate=\astateseq_1.\astate_1\mstar\astateseq_2.\astate_2$.
Moreover, we have that $\csemOf{\acom}(\astateseq.\astate)$ extends the sequence 
$\astateseq.\astate$ by elements from $\semComOf{\acom}{\astate}$.
Such elements are either $\abort$ or states.\\\\
% -----------------------------------------------------------------------------
% -----------------------------------------------------------------------------
% -----------------------------------------------------------------------------
% -----------------------------------------------------------------------------
\textit{Case $\abort$.}\quad
We are going to show that $\semComOf{\acom}{\astate}=\abort$ implies $\semComOf{\acom}{\astate_1}=\abort$.
Then $\csemOf{\acom}(\comppred_1)$ contains a sequence with $\abort$, contradicting $\csemOf{\acom}(\semOf{\comppred_1})\subseteq\semOf{\comppred_2}$.
\begin{compactitem}
\item Consider $\assignof{\apvar}{\anaopof{\apvarp_1,\ldots, \apvarp_n}}$.
It aborts due to a failure of the precondition $\assignfromop$.
% \begin{align*}
% \assignfromop\qquad\defeq\qquad& 
% \anavalp = \semOf{\anaopof{\apvar_1,\ldots, \apvar_n}}\ \astate.\aval \\
% \wedge\ & (\apvar\in\setgpvars\Rightarrow \anavalp\in\astate.\aflowgraph.\setnodes)\wedge \astate\updOf{\apvar}{\anavalp}\statemultdef\splitfg.
% \end{align*}
The expression will evaluate to the same address $\anavalp$ in $\astate_1$, because stacks are never split.
Assume the implication fails, meaning $\apvar\in\setgpvars$ and $\anavalp\notin\astate.\aflowgraph.\setnodes$.
Then $\anavalp\notin\astate_1.\aflowgraph.\setnodes$ as $\astate.\aflowgraph.\setnodes=\astate_1.\aflowgraph.\setnodes\discup\astate_2.\aflowgraph.\setnodes$.
Hence, the implication will also fail on $\astate_1$ and we have the desired abort.

\enspace\enspace If $\astate\updOf{\apvar}{\anavalp}\statemultdef\splitfg$ fails, then either the local and the global flow graph of $\astate\updOf{\apvar}{\anavalp}$ do not compose or we reach a local node from the global state that points to a node outside the state.
In the former case, we note that $\astate\updOf{\apvar}{\anavalp}.\ashared.\aflowgraph = \astate.\ashared.\aflowgraph$ and $\astate\updOf{\apvar}{\anavalp}.\alocal.\aflowgraph = \astate.\alocal.\aflowgraph$.\\
Moreover, we know that $\astate.\ashared.\aflowgraph\statemultdef\astate\alocal.\aflowgraph$ as states are well-split.
The case does not occur.
% ---------------------------------------------------------------------------
% ---------------------------------------------------------------------------
% ---------------------------------------------------------------------------
% ---------------------------------------------------------------------------

\enspace\enspace In the latter case, the node that points outside $\astate$ is reachable through $\anavalp$.
To see this, note that $\astate$ is well-split and hence there is no way for the global state to reach the local heap.
In this case, split is defined (and the identity).
Hence, the only way to make split fail is by following the pointer from $\apvar$ to $\anavalp$ and potentially further to the mentioned node.
In $\astate_1\updOf{\apvar}{\anavalp}$, we may be able to follow the same path, in which case split would also fail on $\astate_1\updOf{\apvar}{\anavalp}$ due to the same node.
Alternatively, the path may lead outside $\astate_1\updOf{\apvar}{\anavalp}$.
In that case, there is a first node on the path that leaves $\astate_1\updOf{\apvar}{\anavalp}$.
This first node points outside $\astate_1\updOf{\apvar}{\anavalp}$, namely to the next node on the path.
As a result, split also fails on $\astate_1\updOf{\apvar}{\anavalp}$, as required.
\item Consider $\assignof{\apvar}{\selof{\apvarp}{\asel}}$.
The command aborts due to a failure of the precondition $\assignfromsel$
% \begin{align*}
% \assignfromsel\qquad\defeq\qquad& 
% \anaval= \semOf{\apvarp}\ \astate.\aval \wedge 
% \anaval\in\astate.\aflowgraph.\setnodes\wedge\anavalp = \astate.\avalof{\anaval, \asel} \\
% \wedge\ & 
% (\apvar\in\setgpvars\Rightarrow \anavalp\in\astate.\aflowgraph.\setnodes) \wedge \astate\updOf{\apvar}{\anavalp}\statemultdef\splitfg. 
% \end{align*}
The command may abort as $\anaval\notin\astate.\aflowgraph.\setnodes$.
Then $\anaval\notin\astate_1.\aflowgraph.\setnodes$ by $\astate_1.\aflowgraph.\setnodes\discup\astate_2.\aflowgraph.\setnodes=\astate.\aflowgraph.\setnodes$.
Hence, the command will also abort on $\astate_1$.

\enspace\enspace For the implication and splitting, the argumentation is like in the previous case. 
\item A malloc does not abort, so the case does not occur.
\end{compactitem}~\\
% -------------------------------------------------------------------------------
% -------------------------------------------------------------------------------
% -------------------------------------------------------------------------------
% -------------------------------------------------------------------------------
\textit{Case $\astate'\in \semComOf{\acom}{\astate}$.}\quad
We show there are $\astate_1'\in\semComOf{\acom}{\astate_1}$ and $\astateseq_2.\astate_2.\astate_2'\in\semOf{\comppred}$ with $\astate_1'\statemultdef\astate_2'$ and $\astate'=\astate_1'\mstar\astate_2'$.
From the former condition and the assumption, we get $\astateseq_1.\astate_1.\astate_1'\in \semOf{\comppred_2}$.
Hence, $\astateseq.\astate.\astate'=\astateseq_1.\astate_1.\astate_1'\mstar\astateseq_2.\astate_2.\astate_2' \in \semOf{\comppred_2}\mstar\semOf{\comppred}=\semOf{\comppred_2\mstar\comppred}$.
We distinguish the cases.
\begin{compactitem}
	\item Consider $\assignof{\apvar}{\anaopof{\apvarp_1,\ldots, \apvarp_n}}$.
	The resulting state is 	$\astate'=\splitOf{\astate\updOf{\apvar}{\anavalp}}$ with \[\anavalp=\semOf{\anaopof{\apvar_1,\ldots, \apvar_n}}\ \astate.\aval \ .\]
	The command does not abort on $\astate_1$, as this would contradict $\csemOf{\acom}(\semOf{\comppred_1})\subseteq\semOf{\comppred_2}$.
	Moreover, the valuations of the program variables will coincide for $\astate$ and $\astate_1$.
	Hence, address $\anavalp$ will not change and we get the precondition $\assignfromop$ for $\astate_1$:
\begin{align*}
\qquad
\anavalp = \semOf{\anaopof{\apvar_1,\ldots, \apvar_n}}\ \astate_1.\aval ~\wedge~
(\apvar\in\setgpvars\Rightarrow \anavalp\in\astate_1.\aflowgraph.\setnodes) ~\wedge~
\astate_1\updOf{\apvar}{\anavalp}\statemultdef\splitfg
\ .
\end{align*}
	The state resulting from the command is \[\astate_1'=\splitOf{\astate_1\updOf{\apvar}{\anavalp}} \ .\]
	Moreover, the precondition is strong enough to apply Lemma~\ref{Lemma:LocalitySplitPvar}.
	With $\astate_2'=\astate_2\updOf{\apvar}{\anavalp}$, the lemma shows $\astate_1'\statemultdef \astate_2'$ and $\astate'=\astate_1'\mstar\astate_2'$.
    For $\astate_2'$, Lemma~\ref{Lemma:NotFree} yields $\astateseq_2.\astate_2.\astate_2'\in \semOf{\comppred}$.
    This concludes the case.
    % ------------------------------------------------------------------------
    % ------------------------------------------------------------------------    
    % ------------------------------------------------------------------------    
    % ------------------------------------------------------------------------
   	\item Consider $\assignof{\apvar}{\selof{\apvarp}{\asel}}$.
	The resulting state is this one:	$\astate'=\splitOf{\astate\updOf{\apvar}{\anavalp}}$ with $\anavalp=\astate.\svalof{\anaval, \asel}$ and $\anaval=\semOf{\apvarp}\ \astate.\aval$. 
	The command does not abort on $\astate_1$, as this would contradict $\csemOf{\acom}(\semOf{\comppred_1})\subseteq\semOf{\comppred_2}$.
	Moreover, the valuations of the program variables will coincide for $\astate$ and $\astate_1$.
	Hence, address $\anaval$ will not change and we get the precondition $\assignfromsel$ for $\astate_1$:
\begin{align*}
\quad& 
\anaval= \semOf{\apvarp}\ \astate_1.\aval~\wedge~\anavalp' = \astate_1.\avalof{\anaval, \asel} ~\wedge
\anaval\in\astate_1.\aflowgraph.\setnodes \\\wedge~& \astate_1\updOf{\apvar}{\anavalp'}\statemultdef\splitfg ~\wedge~
(\apvar\in\setgpvars\Rightarrow \anavalp'\in\astate_1.\aflowgraph.\setnodes)
\ .
\end{align*}
Because $\astate=\astate_1\mstar\astate_2$ and $\astate_1.\avalof{-}$ is a function, we have $\anavalp'=\anavalp$.
Hence, the state resulting from the command is $\astate_1'=\splitOf{\astate_1\updOf{\apvar}{\anavalp}}$.
The remainder of the reasoning is as in the previous case.
% ------------------------------------------------------------------------
% ------------------------------------------------------------------------    
% ------------------------------------------------------------------------    
% ------------------------------------------------------------------------
\item Consider $\mallocof{\apvar}$.
The resulting state is $\astate'=\splitOf{\tilde\astate\updOf{\apvar}{\anavalp}}$ with $\tilde\astate=\extOf{\astate}{\anavalp}$.
Since $(\astate, \anavalp)\statemultdef\extfg$, we have $(\astate_1, \anavalp)\statemultdef\extfg$.
The address is fresh for $\astate_1$ as $\astate_1.\aflowgraph.\setnodes\discup\astate_2.\aflowgraph.\setnodes=\astate.\aflowgraph.\setnodes$ and the address is fresh for $\astate$.
To see that the global part of $\astate_1$ does not point to the address, note that the global part of $\astate$ does not point to it and we have $\astate_1.\ashared\mstar\astate_2.\ashared=\astate.\ashared$.
Hence, $(\astate_1, \anavalp)\statemultdef\extfg$ and we let $\tilde \astate_1=\extOf{\astate_1}{\anavalp}$.
With the same argument, also the global part of $\astate_2$ does not point to the address.
The global part of $\astate_2$ does not point to the remaining local heap of $\astate_1$ as $\astate_1\statemultdef\astate_2$.
With the same argument, the global part of $\astate_1$ does not point to the local heap of $\astate_2$.
Hence, $\tilde\astate_1\statemultdef\astate_2$.
Since $\astate_1\mstar\astate_2=\astate$, we moreover have  $\tilde \astate_1\mstar\astate_2 = \tilde \astate$.

\enspace\enspace Since $\astate_1$ is a state, it is well-split.
By definition, $\anavalp$ only points to itself.
Hence, we have $\tilde\astate_1\updOf{\apvar}{\anavalp}\statemultdef\splitfg$ and set $\astate_1'=\splitOf{\tilde \astate_1\updOf{\apvar}{\anavalp}}$.
The splitting will move $\anavalp$ to the shared heap should $\apvar$ be a shared variable.
Lemma~\ref{Lemma:LocalitySplitPvar} applies and yields $\astate' = \astate_1'\mstar\astate_2'$ where we have $\astate_2' = \astate_2\updOf{\apvar}{\anavalp}$. 
It is readily checked that $\astate_1'\in\semComOf{\mallocof{\apvar}}{\astate_1}$.
Moreover, $\astateseq_2.\astate_2.\astate_2'\in\semOf{\comppred}$ by Lemma~\ref{Lemma:NotFree}.
% This concludes the proof.
\end{compactitem}
This concludes the proof.
\end{proof}

%!TEX root = ../main.tex

% \section{Soundness Proof of Program Logic (Theorem~\ref{Theorem:Soundness})}
\section{Proof of Theorem~\ref{Theorem:Soundness}: Soundness of the Program Logic}
\label{Section:SoundnessProgramLogic}

The guarantee given by the thread-local proofs constructed in our program logic is captured by the notion of safety. 
Safety guarantees that upon termination the postcondition holds, $\accept{\astmt}{\apredppp}{\apred}$ is defined as $\astmt=\cskip\Rightarrow\apredppp\subseteq\apred$. 
Moreover, for every command to be executed we are sure to have captured the interference and to execute safely for another $k$ steps from the resulting predicate. 
% --------------------------------------------------------------------------------------
% --------------------------------------------------------------------------------------
% --------------------------------------------------------------------------------------
% --------------------------------------------------------------------------------------
\begin{definition}
We define $\locsafedef{0}{\astmt}{\apredppp}{\apred}\defeq\true$ ~and~
$\locsafedef{k+1}{\astmt}{\apredppp}{\apred}\defeq (1)\wedge (2)$ with
\begin{itemize}
\item[(1)] $\accept{\astmt}{\apredppp}{\apred}$
\item[(2)] $\forall \acom.\forall \astmt'.\;\astmt\xrightarrow{\acom}\astmt'\;\Rightarrow\;\inter{\apredppp}{\acom}\subseteq \theInterference\;\wedge\; \exists \apredpppp\in\thePredicates.\;\semOf{\acom}(\apredppp)\subseteq\apredpppp\;\wedge\; \locsafek{\astmt'}{\apredpppp}{\apred}$.
\end{itemize}
\end{definition}
% --------------------------------------------------------------------------------------
% --------------------------------------------------------------------------------------
% --------------------------------------------------------------------------------------
% --------------------------------------------------------------------------------------
The predicate is monotonic in the various arguments as follows.
% --------------------------------------------------------------------------------------
% --------------------------------------------------------------------------------------
\begin{lemma}\label{Lemma:Monotonicity}
Consider $\thePredicates_1\subseteq\thePredicates_2$, $\theInterference_1\subseteq\theInterference_2$, $\apred_1\supseteq \apred_2$, $\apredp_1\subseteq \apredp_2$, and $k_1\geq k_2$. Then 
$\locsafe{\thePredicates_1}{\theInterference_1}{k_1}{\astmt}{\apred_1}{\apredp_1}$ implies $\locsafe{\thePredicates_2}{\theInterference_2}{k_2}{\astmt}{\apred_2}{\apredp_2}$. 
\end{lemma}

\begin{proof}
Consider $\thePredicates_1\subseteq\thePredicates_2$, $\theInterference_1\subseteq\theInterference_2$, and $\apredp_1\subseteq \apredp_2$.\\
We proceed by induction on $k_1$.\\
$k_1=0$: There is nothing to do.\\
$k_1+1$: The induction hypothesis is
\begin{align*}
\forall k_2.
\forall\astmt.\forall \apred_1.\forall \apred_2.\quad
k_2\leq k_1\wedge \apred_1\supseteq \apred_2\wedge \locsafe{\thePredicates_1}{\theInterference_1}{k_1}{\astmt}{\apred_1}{\apredp_1}
\implies\locsafe{\thePredicates_2}{\theInterference_2}{k_2}{\astmt}{\apred_2}{\apredp_2} \ . 
\end{align*}
Further, consider $\astmt$ and $\apred_1$ so that \( \locsafe{\thePredicates_1}{\theInterference_1}{k_1+1}{\astmt}{\apred_1}{\apredp_1}. \)
Consider $k_2\leq k_1+1$ and $\apred_2\subseteq\apred_1$.
We show 
\begin{align*}
\locsafe{\thePredicates_2}{\theInterference_2}{k_2}{\astmt}{\apred_2}{\apredp_2}\ . 
\end{align*}
We can assume $0<k_2$, otherwise there is nothing to do.\\
(1) We have $\locsafe{\thePredicates_1}{\theInterference_1}{k_1+1}{\astmt}{\apred_1}{\apredp_1}$.
Hence, if $\astmt=\cskip$, then $\apred_1\subseteq\apredp_1$.
We have $\apred_2\subseteq\apred_1$ and $\apredp_1\subseteq\apredp_2$ by assumption.
Hence, $\apred_2\subseteq\apredp_2$ as required. \\
(2a) Consider $\astmt\xxrightarrow{\acom}\astmt'$.
We have $\locsafe{\thePredicates_1}{\theInterference_1}{k_1+1}{\astmt}{\apred_1}{\apredp_1}$.
Hence, $\inter{\apred_1}{\acom}\subseteq\theInterference_1$.
This means there is $(\apredpp, \acom)\in\theInterference_1$ with $\apred_1\subseteq\apredpp$.
We have $\apred_2\subseteq\apred_1$ and so $\apred_2\subseteq\apredpp$.
Hence, $\inter{\apred_2}{\acom}\subseteq\theInterference_1$.
We have $\theInterference_1\subseteq\theInterference_2$ by assumption.
Hence, $\inter{\apred_2}{\acom}\subseteq\theInterference_2$ as required.\\
(2b) Because of $\locsafe{\thePredicates_1}{\theInterference_1}{k_1+1}{\astmt}{\apred_1}{\apredp_1}$, 
there is $\apredpppp\in\thePredicates_1$ as follows.
We have $\semOf{\acom}(\apred_1)\subseteq\apredpppp$ as well as $\locsafe{\thePredicates_1}{\theInterference_1}{k_1}{\astmt'}{\apredpppp}{\apredp_1}$.
As $\thePredicates_1\subseteq\thePredicates_2$, we also have $\apredpppp\in\thePredicates_2$.
We moreover have $\apred_2\subseteq\apred_1$ and hence $\semOf{\acom}(\apred_2)\subseteq\semOf{\acom}(\apred_1)$.
We conclude $\semOf{\acom}(\apred_2)\subseteq\apredpppp$ as required.\\ 
(2c) It remains to argue for $\locsafe{\thePredicates_2}{\theInterference_2}{k_2-1}{\astmt'}{\apredpppp}{\apredp_2}$.
We have $0<k_2\leq k_1+1$. 
Hence, $0\leq k_2-1\leq k_1$.
We have $\locsafe{\thePredicates_1}{\theInterference_1}{k_1}{\astmt'}{\apredpppp}{\apredp_1}$.
The induction hypothesis yields $\locsafe{\thePredicates_2}{\theInterference_2}{k_2-1}{\astmt'}{\apredpppp}{\apredp_2}$.
\end{proof}
% --------------------------------------------------------------------------------------
% --------------------------------------------------------------------------------------
% --------------------------------------------------------------------------------------
% --------------------------------------------------------------------------------------
Soundness of the thread-local derivation is stated in the next proposition.
% --------------------------------------------------------------------------------------
% --------------------------------------------------------------------------------------
% --------------------------------------------------------------------------------------
% --------------------------------------------------------------------------------------
\begin{proposition}\label{Proposition:LocalSound}
Consider $\thePredicates, \theInterference\semCalc\hoareOf{\apred}{\astmt}{\apredp}$ and $k\in\nat$. 
We have $\locsafek{\astmt}{\apred}{\apredp}$. 
\end{proposition}

The proof of Proposition~\ref{Proposition:LocalSound} proceeds by rule induction over the derivation rules of the program logic (Figure~\ref{Figure:ProgramLogic}). We break the proof down into individual lemmas based on the case analysis of the last derivation rule used in the proof.

\subsection{Soundness of \ruleref{seq}}
\begin{lemma}
If $\locsafe{\thePredicates_1}{\theInterference_1}{k}{\astmt_1}{\apred}{\apredp}$ and $\locsafe{\thePredicates_2}{\theInterference_2}{k}{\astmt_2}{\apredp}{\apredpp}$ and $\thePredicates=\thePredicates_1\cup\thePredicates_2$ and $\theInterference=\theInterference_1\cup\theInterference_2$, then $\locsafe{\setcompact{\apredp}\cup\thePredicates}{\theInterference}{k}{\astmt_1;\astmt_2}{\apred}{\apredpp}$. 
\end{lemma}
\begin{proof}
Consider $\astmt_2$ and $\apredp$.
We proceed by induction on $k$.\\
$k=0$: Trivial.\\
$k+1$: The induction hypothesis is
\begin{align*}
\forall \astmt_1.\forall \apred.\quad \locsafe{\thePredicates_1}{\theInterference_1}{k}{\astmt_1}{\apred}{\apredp}\wedge \locsafe{\thePredicates_2}{\theInterference_2}{k}{\astmt_2}{\apredp}{\apredpp} \implies \locsafe{\setcompact{\apredp}\cup\thePredicates}{\theInterference}{k}{\astmt_1;\astmt_2}{\apred}{\apredpp} \ . 
\end{align*}
Consider $\astmt_1$ and $\apred$ so that
\begin{align*}
\locsafe{\thePredicates_1}{\theInterference_1}{k+1}{\astmt_1}{\apred}{\apredp}\wedge \locsafe{\thePredicates_2}{\theInterference_2}{k+1}{\astmt_2}{\apredp}{\apredpp}\ .
\end{align*}
Our goal is to show 
\begin{align*}
\locsafe{\setcompact{\apredp}\cup\thePredicates}{\theInterference}{k+1}{\astmt_1;\astmt_2}{\apred}{\apredpp}\ . 
\end{align*}
(1) As $\astmt_1;\astmt_2\neq \cskip$, there is nothing to show. \\
(2) Consider $\acom$ and $\astmt'$ with $\astmt_1;\astmt_2\xrightarrow{\acom}\astmt'$.\\
There are two cases for transitions from $\astmt_1;\astmt_2$.\\\\
\textit{Case 1:} $\astmt_1=\cskip$ and $\cskip;\astmt_2\xrightarrow{\cskip}\astmt_2$. \\
(2a) We have $\neg\effectful{\apred}{\cskip}$, hence $\inter{\apred}{\cskip}=\emptyset\subseteq\theInterference$.\\
(2b) We argue that $\apredp$ is the right choice for a predicate.
It is in $\setcompact{\apredp}\cup\thePredicates$.
We have $\locsafe{\thePredicates_1}{\theInterference_1}{k+1}{\astmt_1}{\apred}{\apredp}$ by assumption and $\astmt_1=\cskip$.
Hence, $\apred\subseteq\apredp$ holds.
With this, $\semOf{\cskip}(\apred)=\apred\subseteq\apredp$ follows.\\
(2c) By assumption, $\locsafe{\thePredicates_2}{\theInterference_2}{k+1}{\astmt_2}{\apredp}{\apredpp}$.
Monotonicity in Lemma~\ref{Lemma:Monotonicity} then yields $\locsafe{\setcompact{\apredp}\cup\thePredicates}{\theInterference}{k}{\astmt_2}{\apredp}{\apredpp}$.\\\\
% Then, $\locsafe{\setcompact{\apredp}\cup\thePredicates}{\theInterference}{k}{\astmt_2}{\apredp}{\apredpp}$ follows from monotonicity in Lemma~\ref{Lemma:Monotonicity}.\\\\
\textit{Case 2:} $\astmt_1\xrightarrow{\acom}\astmt_1'$ and $\astmt_1;\astmt_2\xrightarrow{\acom}\astmt_1';\astmt_2$. \\
(2a) By the assumption of the induction step, we have the following: $\locsafe{\thePredicates_1}{\theInterference_1}{k+1}{\astmt_1}{\apred}{\apredp}$.
Hence, $\inter{\apred}{\acom}\subseteq\theInterference_1\subseteq\theInterference$.\\
(2b) Furthermore, there is some $\apredpppp\in\thePredicates_1$ with $\semOf{\acom}(\apred)\subseteq\apredpppp$ and  $\locsafe{\thePredicates_1}{\theInterference_1}{k}{\astmt_1'}{\apredpppp}{\apredp}$.
Since $\thePredicates_1\subseteq \setcompact{\apred}\cup\thePredicates$, we can again pick $\apredpppp$.\\
(2c) We have $\locsafe{\thePredicates_2}{\theInterference_2}{k+1}{\astmt_2}{\apredp}{\apredpp}$ by assumption.
By monotonicity in Lemma~\ref{Lemma:Monotonicity}, $\locsafe{\thePredicates_2}{\theInterference_2}{k}{\astmt_2}{\apredp}{\apredpp}$.
We already noticed $\locsafe{\thePredicates_1}{\theInterference_1}{k}{\astmt_1'}{\apredpppp}{\apredp}$.
The induction hypothesis yields $\locsafe{\setcompact{\apredp}\cup\thePredicates}{\theInterference}{k}{\astmt_1';\astmt_2}{\apredpppp}{\apredpp}$. 
\end{proof}
%-----------------------------------------------------------------------------
%-----------------------------------------------------------------------------
%-----------------------------------------------------------------------------
%-----------------------------------------------------------------------------
\subsection{Soundness of \ruleref{com-sem}}
We prove a fact that will be helpful. 
%-----------------------------------------------------------------------------
%-----------------------------------------------------------------------------
%-----------------------------------------------------------------------------
%-----------------------------------------------------------------------------
\begin{lemma}\label{Lemma:Skip}
$\locsafe{\setcompact{\apred}}{\emptyset}{k}{\cskip}{\apred}{\apred}$.  
\end{lemma}
\begin{proof}
Consider $\apred$.
We proceed by induction on $k$.\\
$k = 0$: Done.\\
$k+1$: The induction hypothesis is 
\begin{align*}
\locsafe{\setcompact{\apred}}{\emptyset}{k}{\cskip}{\apred}{\apred}\ . 
\end{align*}
Our goal is to prove
\begin{align*}
\locsafe{\setcompact{\apred}}{\emptyset}{k+1}{\cskip}{\apred}{\apred}\ . 
\end{align*}
(1) We have $\apred\subseteq \apred$.\\
(2) The only transition is $\cskip\xrightarrow{\cskip}\cskip$.
As $\neg\effectful{\apred}{\cskip}$, we have $\inter{\apred}{\cskip}=\emptyset$.\\
(2b) As predicate $q$ we pick $\apred$, which is in $\setcompact{\apred}$.
Then $\semOf{\cskip}(\apred)=\apred$.\\
(2c) We have $\locsafek{\cskip}{\apred}{\apred}$ by induction.
\end{proof}
%-----------------------------------------------------------------------------
%-----------------------------------------------------------------------------
%-----------------------------------------------------------------------------
%-----------------------------------------------------------------------------
Soundness of \ruleref{com-sem} is formalized by the following lemma. 
%-----------------------------------------------------------------------------
%-----------------------------------------------------------------------------
%-----------------------------------------------------------------------------
%-----------------------------------------------------------------------------
\begin{lemma}\label{Lemma:Com}
$\thePredicates=\setcompact{\apredp}$ and $\theInterference=\inter{\apred}{\acom}$ and $\semOf{\acom}(\apred)\subseteq\apredp$ imply 
$\locsafek{\acom}{\apred}{\apredp}$. 
\end{lemma}
\begin{proof}
We consider $k+1$, for $k=0$ there is nothing to do.
We show
\begin{align*}
\locsafedef{k+1}{\acom}{\apred}{\apredp}\ .
\end{align*}
(1) If $\acom=\cskip$, we have $\apred=\semOf{\acom}(\apred)$.
By the assumption, $\semOf{\acom}(\apred)\subseteq\apredp$.
Hence, $\apred\subseteq\apredp$ as required.\\
(2) Consider $\acom\xrightarrow{\acom}{\cskip}$. \\
(2a) We have $\inter{\apred}{\acom}=\theInterference$ by assumption.\\
(2b) We pick $\apredpppp=\apredp$, which is in $\thePredicates=\setcompact{\apredp}$.
Then $\semOf{\acom}(\apred)\subseteq\apredp$ by the assumption.\\
(2c) We have $\setcompact{\apredp}=\thePredicates$.
Then $\locsafe{\thePredicates}{\emptyset}{k}{\cskip}{\apredp}{\apredp}$ by Lemma~\ref{Lemma:Skip}.
Monotonicity in Lemma~\ref{Lemma:Monotonicity} yields $\locsafek{\cskip}{\apredp}{\apredp}$.  
\end{proof}
%-----------------------------------------------------------------------------
%-----------------------------------------------------------------------------
%-----------------------------------------------------------------------------
%-----------------------------------------------------------------------------
\subsection{Soundness of \ruleref{choice}} 
\begin{lemma}\label{Lemma:Choice}
If $\locsafe{\thePredicates_1}{\theInterference_1}{k}{\astmt_1}{\apred}{\apredp}$ and $\locsafe{\thePredicates_2}{\theInterference_2}{k}{\astmt_2}{\apred}{\apredp}$ and $\thePredicates=\thePredicates_1\cup\thePredicates_2$ and $\theInterference=\theInterference_1\cup\theInterference_2$, then $\locsafek{\choiceof{\astmt_1}{\astmt_2}}{\apred}{\apredp}$. 
\end{lemma}
\begin{proof}
Consider $\astmt_1$, $\astmt_2$, and $\apredp$.
We proceed by induction on $k$.\\
$k=0$: Done.\\
$k+1$: The induction hypothesis is 
\begin{align*}
\forall \apred.\quad \locsafe{\thePredicates_1}{\theInterference_1}{k}{\astmt_1}{\apred}{\apredp}\wedge\locsafe{\thePredicates_2}{\theInterference_2}{k}{\astmt_2}{\apred}{\apredp} \implies \locsafek{\choiceof{\astmt_1}{\astmt_2}}{\apred}{\apredp}\ .
\end{align*}
Consider $\apred$ with 
\begin{align*}
\locsafe{\thePredicates_1}{\theInterference_1}{k+1}{\astmt_1}{\apred}{\apredp}\wedge\locsafe{\thePredicates_2}{\theInterference_2}{k+1}{\astmt_2}{\apred}{\apredp}\ . 
\end{align*}
Our goal is to show 
\begin{align*}
\locsafedef{k+1}{\choiceof{\astmt_1}{\astmt_2}}{\apred}{\apredp}\ .
\end{align*}
(1) We have $\choiceof{\astmt_1}{\astmt_2}\neq \cskip$, hence there is nothing to show.\\
(2a) We have $\choiceof{\astmt_1}{\astmt_2}\xrightarrow{\cskip}{\astmt_i}$ with $i=1, 2$.
As $\neg\effectful{\apred}{\cskip}$, we have $\inter{\apred}{\acom}=\emptyset\subseteq\theInterference$.\\
(2b) As $\locsafe{\thePredicates_1}{\theInterference_i}{k+1}{\astmt_i}{\apred}{\apredp}$ by assumption, there is $\apredpppp\in\thePredicates_i$ as follows.
We have $\semOf{\acom}(\apred)\subseteq\apredpppp$ and $\locsafe{\thePredicates_i}{\theInterference_i}{k}{\astmt_i}{\apredpppp}{\apredp}$.
As $\thePredicates_i\subseteq\thePredicates$, we can again pick $\apredpppp$.\\
(2c) Moreover, $\locsafe{\thePredicates_i}{\theInterference_i}{k}{\astmt_i}{\apredpppp}{\apredp}$ implies $\locsafek{\astmt_i}{\apredpppp}{\apredp}$ by monotonicity in Lemma~\ref{Lemma:Monotonicity}. 
\end{proof}
%-----------------------------------------------------------------------------
%-----------------------------------------------------------------------------
%-----------------------------------------------------------------------------
%-----------------------------------------------------------------------------
\subsection{Soundness of \ruleref{infer-sem}} 
Follows directly from monotonicity in Lemma~\ref{Lemma:Monotonicity}. 
%-----------------------------------------------------------------------------
%-----------------------------------------------------------------------------
%-----------------------------------------------------------------------------
%-----------------------------------------------------------------------------
\subsection{Soundness of \ruleref{loop}}
We first prove some helpful auxiliary lemmas.

\begin{lemma}\label{Lemma:LoopNext}
If we have $\apredppp\in\thePredicates$ and $\locsafek{\astmt;\loopof{\astmt}}{\apredppp}{\apred}$ and $\locsafek{\cskip}{\apredppp}{\apred}$, then $\locsafedef{k+1}{\loopof{\astmt}}{\apredppp}{\apred}$. 
\end{lemma}
\begin{proof}
Consider $\astmt$ and $\apred$. 
We proceed by induction on $k$.\\
$k=0$: Done.\\
$k+1$: The induction hypothesis is: 
\begin{align*}
\forall \apredppp\in\thePredicates.\quad \locsafek{\astmt;\loopof{\astmt}}{\apredppp}{\apred}\wedge\locsafek{\cskip}{\apredppp}{\apred}\implies\locsafedef{k+1}{\loopof{\astmt}}{\apredppp}{\apred}\ .
\end{align*} 
Consider $\apredppp\in\thePredicates$ so that 
\begin{align*}
\locsafedef{k+1}{\astmt;\loopof{\astmt}}{\apredppp}{\apred}\wedge\locsafedef{k+1}{\cskip}{\apredppp}{\apred}\ .
\end{align*}
We show 
\begin{align*}
\locsafedef{k+2}{\loopof{\astmt}}{\apredppp}{\apred}\ .
\end{align*} 
(1) As $\loopof{\astmt}\neq\cskip$, there is nothing to show.\\
(2) The only transition is $\loopof{\astmt}\xrightarrow{\cskip}\choiceof{\cskip}{\astmt;\loopof{\astmt}}$.\\
(2a) As $\neg\effectful{\apredppp}{\cskip}$, we have $\inter{\apredppp}{\acom}=\emptyset$ and hence $\inter{\apredppp}{\acom}\subseteq \theInterference$. \\
(2b) As predicate $\apredpppp$ we choose $\apredppp$, of which we know it is in $\thePredicates$.
We have $\semOf{\cskip}(\apredppp)=\apredppp$.\\
(2c) By assumption, we have $\locsafedef{k+1}{\astmt;\loopof{\astmt}}{\apredppp}{\apred}$ as well as $\locsafedef{k+1}{\cskip}{\apredppp}{\apred}$.
With Lemma~\ref{Lemma:Choice}, we get $\locsafedef{k+1}{\choiceof{\cskip}{\astmt;\loopof{\astmt}}}{\apredppp}{\apred}$.\\
This concludes the proof of $\locsafedef{k+2}{\loopof{\astmt}}{\apredppp}{\apred}$.
\end{proof}
%-----------------------------------------------------------------------------
%-----------------------------------------------------------------------------
%-----------------------------------------------------------------------------
%-----------------------------------------------------------------------------
\begin{lemma}\label{Lemma:LoopUnroll}
If $\locsafek{\astmt_1}{\apredppp}{\apred}$ and 
$\locsafe{\setcompact{\apred}\cup\thePredicates}{\theInterference}{k}{\loopof{\astmt}}{\apred}{\apred}$ hold, then 
$\locsafe{\setcompact{\apred}\cup\thePredicates}{\theInterference}{k}{\astmt_1;\loopof{\astmt}}{\apredppp}{\apred}$. 
\end{lemma}
\begin{proof}
Consider $\astmt$ and $\apred$.
We proceed by induction on $k$.\\
$k=0$: There is nothing to do.\\
$k+1$: The induction hypothesis is
\begin{align*}
\forall\astmt_1.\forall \apredppp.\quad \locsafek{\astmt_1}{\apredppp}{\apred}\wedge\locsafe{\setcompact{\apred}\cup\thePredicates}{\theInterference}{k}{\loopof{\astmt}}{\apred}{\apred}\implies\locsafe{\setcompact{\apred}\cup\thePredicates}{\theInterference}{k}{\astmt_1;\loopof{\astmt}}{\apredppp}{\apred}\ . 
\end{align*}
Consider $\astmt_1$ and $\apredppp$ so that
\begin{align*}
\locsafedef{k+1}{\astmt_1}{\apredppp}{\apred}\wedge\locsafe{\setcompact{\apred}\cup\thePredicates}{\theInterference}{k+1}{\loopof{\astmt}}{\apred}{\apred}.
\end{align*}
We show
\begin{align*}
\locsafe{\setcompact{\apred}\cup\thePredicates}{\theInterference}{k+1}{\astmt_1;\loopof{\astmt}}{\apredppp}{\apred}.
\end{align*}
(1) As $\astmt_1;\loopof{\astmt}\neq \cskip$, there is nothing to show.\\
(2) Consider $\astmt_1;\loopof{\astmt}\xrightarrow{\acom}\astmt'$.
There are two cases for transitions from $\astmt_1;\loopof{\astmt}$.\\\\
Case 1: We have $\astmt_1=\cskip$ and $\astmt_1;\loopof{\astmt}\xrightarrow{\cskip}\loopof{\astmt}$.\\
(2a) Since $\neg\effectful{\apredppp}{\cskip}$, we have $\inter{\apredppp}{\cskip}=\emptyset\subseteq\theInterference$.\\
(2b) We argue that $\apred$ is the right predicate to pick as $\apredpppp$.
It is in $\setcompact{\apred}\cup\thePredicates$.
We have $\semOf{\cskip}(\apredppp)=\apredppp$.
By $\locsafedef{k+1}{\astmt_1}{\apredppp}{\apred}=\locsafedef{k+1}{\cskip}{\apredppp}{\apred}$, we have $\apredppp\subseteq\apred$.
Hence, $\semOf{\cskip}(\apredppp)\subseteq \apred$. \\
(2c) By assumption, $\locsafe{\setcompact{\apred}\cup\thePredicates}{\theInterference}{k+1}{\loopof{\astmt}}{\apred}{\apred}$.
By Lemma~\ref{Lemma:Monotonicity}, we get the desired $\locsafe{\setcompact{\apred}\cup\thePredicates}{\theInterference}{k}{\loopof{\astmt}}{\apred}{\apred}$.\\\\
Case 2: We have $\astmt_1;\loopof{\astmt}\xrightarrow{\acom}\astmt'$ due to $\astmt_1\xrightarrow{\acom}\astmt_1'$ and $\astmt'=\astmt_1';\loopof{\astmt}$.\\
(2a) Since $\locsafedef{k+1}{\astmt_1}{\apredppp}{\apred}$, we have $\inter{\apredppp}{\acom}\subseteq \theInterference$.\\
(2b) By the same assumption, there is some $\apredpppp\in\thePredicates$ such that $\semOf{\acom}(\apredppp)\subseteq \apredpppp$ and $\locsafek{\astmt_1'}{\apredpppp}{\apred}$.\\
(2c) Since $\locsafe{\setcompact{\apred}\cup\thePredicates}{\theInterference}{k+1}{\loopof{\astmt}}{\apred}{\apred}$, we have $\locsafe{\setcompact{\apred}\cup\thePredicates}{\theInterference}{k}{\loopof{\astmt}}{\apred}{\apred}$ with Lemma~\ref{Lemma:Monotonicity}.
We apply the induction hypothesis to $\locsafek{\astmt_1'}{\apredpppp}{\apred}$ and $\locsafe{\setcompact{\apred}\cup\thePredicates}{\theInterference}{k}{\loopof{\astmt}}{\apred}{\apred}$.
It yields $\locsafe{\setcompact{\apred}\cup\thePredicates}{\theInterference}{k}{\astmt_1;\loopof{\astmt}}{\apredpppp}{\apred}$.\\
This concludes the proof of $\locsafe{\setcompact{\apred}\cup\thePredicates}{\theInterference}{k+1}{\astmt_1;\loopof{\astmt}}{\apredppp}{\apred}$. 
\end{proof}
%-----------------------------------------------------------------------------
%-----------------------------------------------------------------------------
%-----------------------------------------------------------------------------
%-----------------------------------------------------------------------------
Soundness of \ruleref{loop} is the following.
%-----------------------------------------------------------------------------
%-----------------------------------------------------------------------------
%-----------------------------------------------------------------------------
%-----------------------------------------------------------------------------
\begin{lemma}
If $\locsafek{\astmt}{\apred}{\apred}$, then $\locsafe{\setcompact{\apred}\cup\thePredicates}{\theInterference}{k}{\loopof{\astmt}}{\apred}{\apred}$.  
\end{lemma}
\begin{proof}
Consider $\thePredicates$, $\theInterference$, $\apred$, and $\astmt$.\
We proceed by induction on $k$.\\
$k=0$: Done.\\
$k+1$: The induction hypothesis is
\begin{align*}
\locsafek{\astmt}{\apred}{\apred}\quad \Rightarrow\quad\locsafe{\setcompact{\apred}\cup\thePredicates}{\theInterference}{k}{\loopof{\astmt}}{\apred}{\apred}. 
\end{align*}
We assume 
\begin{align*}
\locsafedef{k+1}{\astmt}{\apred}{\apred}
\end{align*}
Our goal is to show 
\begin{align*}
\locsafe{\setcompact{\apred}\cup\thePredicates}{\theInterference}{k+1}{\loopof{\astmt}}{\apred}{\apred}. 
\end{align*}
By assumption, we have $\locsafedef{k+1}{\astmt}{\apred}{\apred}$.
This yields $\locsafek{\astmt}{\apred}{\apred}$ by Lemma~\ref{Lemma:Monotonicity}.
The induction hypothesis yields $\locsafe{\setcompact{\apred}\cup\thePredicates}{\theInterference}{k}{\loopof{\astmt}}{\apred}{\apred}$.
Lemma~\ref{Lemma:LoopUnroll} yields $\locsafe{\setcompact{\apred}\cup\thePredicates}{\theInterference}{k}{\astmt;\loopof{\astmt}}{\apred}{\apred}$.
Lemma~\ref{Lemma:Skip} yields $\locsafe{\setcompact{\apred}}{\emptyset}{k}{\cskip}{\apred}{\apred}$.
Monotonicity in Lemma~\ref{Lemma:Monotonicity} yields $\locsafe{\setcompact{\apred}\cup\thePredicates}{\theInterference}{k}{\cskip}{\apred}{\apred}$.
We can thus invoke Lemma~\ref{Lemma:LoopNext}.
It yields the required $\locsafe{\setcompact{\apred}\cup\thePredicates}{\theInterference}{k+1}{\loopof{\astmt}}{\apred}{\apred}$.  
\end{proof}
%-----------------------------------------------------------------------------
%-----------------------------------------------------------------------------
%-----------------------------------------------------------------------------
%-----------------------------------------------------------------------------
\subsection{Soundness of \ruleref{frame}} $\phantom{Test}$\newline
We first prove an auxiliary lemma stating that due to the locality of commands, effectfullness of commands is compatible with framing.
% ---------------------------------------------------------------------------------------------
% ---------------------------------------------------------------------------------------------
% ---------------------------------------------------------------------------------------------
% ---------------------------------------------------------------------------------------------
\begin{lemma}\label{Lemma:EffectfulFrame}
$\effectful{\apredpp\mstar\apred}{\acom}$ implies $\effectful{\apredpp}{\acom}$. 
\end{lemma}

\begin{proof}
We prove the contrapositive: \[ \neg \effectful{\apredpp}{\acom} \implies \neg\effectful{\apredpp\mstar\apred}{\acom} \ .  \]
Consider some $(\ashared_1, \alocal_1)\in\apredpp$ and some $(\ashared_2, \alocal_2)\in\apred$ such that $(\ashared_1\mstar\ashared_2, \alocal_1\mstar\alocal_2)=(\ashared_1, \alocal_1)\mstar(\ashared_2, \alocal_2)\in\apredpp\mstar\apred$.  
By the locality of commands, we have \[\semOf{\acom}((\ashared_1, \alocal_1)\mstar(\ashared_2, \alocal_2))\subseteq\semOf{\acom}(\ashared_1, \alocal_1)\mstar\setcompact{(\ashared_2, \alocal_2)}\ .\]
Since $\neg \effectful{\apredpp}{\acom}$ and $(\ashared_1, \alocal_1)\in\apredpp$, every state in $\semOf{\acom}(\ashared_1, \alocal_1)$ takes the form $(\ashared_1, \alocal)$ for some $\alocal$.
Consequently, every state in $\semOf{\acom}((\ashared_1, \alocal_1)\mstar(\ashared_2, \alocal_2))$ must be of the following form: $(\ashared_1, \alocal)\mstar(\ashared_2, \alocal_2)=(\ashared_1\mstar\ashared_2, \alocal\mstar\alocal_2)$, and is thus as required.
\end{proof}
% ---------------------------------------------------------------------------------------------
% ---------------------------------------------------------------------------------------------
% ---------------------------------------------------------------------------------------------
% ---------------------------------------------------------------------------------------------

Soundness of \ruleref{frame} is this.
%-----------------------------------------------------------------------------
%-----------------------------------------------------------------------------
%-----------------------------------------------------------------------------
%-----------------------------------------------------------------------------
\begin{lemma}\label{Lemma:Frame}
If $\locsafek{\astmt}{\apredppp}{\apredp}$, then $\locsafe{\thePredicates\mstar\apredpp}{\theInterference\mstar\apredpp}{k}{\astmt}{\apredppp\mstar\apredpp}{\apredp\mstar\apredpp}$. 
\end{lemma}
\begin{proof}
Consider $\apredp$, $\theInterference$, $\thePredicates$, and $\apredpp$.
We proceed by induction on $k$. \\
$k=0$: There is nothing to do.\\
$k+1$: The induction hypothesis is
\begin{align*}
&\forall \astmt.\forall \apredppp.~~\locsafek{\astmt}{\apredppp}{\apredp}~~ \Rightarrow~~\locsafe{\thePredicates\mstar\apredpp}{\theInterference\mstar\apredpp}{k}{\astmt}{\apredppp\mstar\apredpp}{\apredp\mstar\apredpp}.
\end{align*}
Consider $\astmt$ and $\apredppp$ with $\locsafedef{k+1}{\astmt}{\apredppp}{\apredp}$.
We show 
\begin{align*}
\locsafedef{k+1}{\astmt}{\apredppp\mstar\apredpp}{\apredp\mstar\apredpp}. 
\end{align*}
(1) Assume $\astmt=\cskip$, otherwise there is nothing to do.
Since $\locsafedef{k+1}{\astmt}{\apredppp}{\apredp}$, we have $\apredppp\subseteq\apredp$.
Hence, $\apredppp\mstar\apredpp\subseteq \apredp\mstar\apredpp$, as required.\\ 
(2) Consider $\acom$ and $\astmt'$ so that $\astmt\xrightarrow{\acom}\astmt'$.
By assumption, we have $\locsafedef{k+1}{\astmt}{\apredppp}{\apredp}$.\\
(2a) If $\inter{\apredppp\mstar\apredpp}{\acom}=\emptyset$, then $\inter{\apredppp\mstar\apredpp}{\acom}\subseteq\theInterference\mstar\apredpp$ follows trivially.
So assume that we have $\inter{\apredppp\mstar\apredpp}{\acom}=\setcompact{(\apredppp\mstar\apredpp, \acom)}$.
Then $\effectful{\apredppp\mstar\apredpp}{\acom}$.
By Lemma~\ref{Lemma:EffectfulFrame}, we get $\effectful{\apredppp}{\acom}$.
Hence, $\inter{\apredppp}{\acom}=\setcompact{(\apredppp, \acom)}$.
Since $\locsafedef{k+1}{\astmt}{\apredppp}{\apredp}$, we have $\inter{\apredppp}{\acom}\subseteq \theInterference$. 
This means there is $(\apredpp', \acom)\in\theInterference$ with $\apredppp\subseteq \apredpp'$.
Then $\apredppp\mstar\apredpp\subseteq\apredpp'\mstar\apredpp$.
Moreover, $(\apredpp'\mstar\apredpp, \acom)\in\theInterference\mstar\apredpp$, as required. \\
(2b) Because of $\locsafedef{k+1}{\astmt}{\apredppp}{\apredp}$, there is $\apredpppp\in\thePredicates$ with $\semOf{\acom}(\apredppp)\subseteq\apredpppp$ and $\locsafek{\astmt'}{\apredpppp}{\apredp}$.
We have $\apredpppp\mstar\apredpp\in\thePredicates\mstar\apredpp$.
We argue that this is the right choice for a predicate. 
By $\semOf{\acom}(\apredppp)\subseteq\apredpppp$ and the locality of commands, $\semOf{\acom}(\apredppp\mstar\apredpp)\subseteq\apredpppp\mstar\apredpp$, as required.\\
(2c) We have $\locsafek{\astmt'}{\apredpppp}{\apredp}$.
Induction yields the desired $\locsafek{\astmt'}{\apredpppp\mstar\apredpp}{\apredp\mstar\apredpp}$.
\end{proof}
%-----------------------------------------------------------------------------
%-----------------------------------------------------------------------------
%-----------------------------------------------------------------------------
%-----------------------------------------------------------------------------
% This concludes the proof of Proposition~\ref{Proposition:LocalSound}.

%%%

% --------------------------------------------------------------------------------------
% --------------------------------------------------------------------------------------
% --------------------------------------------------------------------------------------
% --------------------------------------------------------------------------------------
\subsection{Interference-Freedom}
Interference freedom lifts the safety guarantee from an isolated thread to the concurrency library. 
%-----------------------------------------------------------------------------
%-----------------------------------------------------------------------------
%-----------------------------------------------------------------------------
%-----------------------------------------------------------------------------
%-----------------------------------------------------------------------------
\begin{definition}
We define $\cfsafedef{0}{\aconfig}{\apredp}\defeq\true$ as well as
$\cfsafedef{k+1}{\aconfig}{\apredp} \defeq (1)\wedge (2)$ with
\begin{itemize}
\item[(1)] $\aconfig\in \acceptset{\apredp}$\;
\item[(2)] $\forall \aconfig'.\;\aconfig\rightarrow\aconfig'\;\Rightarrow\;\cfsafek{\aconfig'}{\apredp}$.
\end{itemize}
\end{definition}
% --------------------------------------------------------------------------------------
% --------------------------------------------------------------------------------------
% --------------------------------------------------------------------------------------
% --------------------------------------------------------------------------------------
The key lemma for lifting the thread-local safety guarantee to configurations is the following. 
% --------------------------------------------------------------------------------------
% --------------------------------------------------------------------------------------
% --------------------------------------------------------------------------------------
% --------------------------------------------------------------------------------------
\begin{lemma}\label{Lemma:LocalGlobal}
If $\aconfig = (\ashared, \apc)$ and $[\forall i.\forall\alocal.\forall \astmt.\apc(i)=(\alocal, \astmt)\Rightarrow
\exists \apredppp\in\thePredicates.\; (\ashared, \alocal)\in\apredppp\wedge 
\locsafek{\astmt}{\apredppp}{\apredp}]$, then $\cfsafek{\aconfig}{\apredp}$. 
% \swfootnote{interference-freedom premise is missing here}
\end{lemma}
% --------------------------------------------------------------------------------------
% --------------------------------------------------------------------------------------
% --------------------------------------------------------------------------------------
% --------------------------------------------------------------------------------------
\begin{proof}
Consider $\thePredicates$ and $\theInterference$ with $\isInterferenceFreeOf[\theInterference]{\thePredicates}$.
Consider $\apredp$.
We proceed by induction on $k$.\\
$k=0$: Done.\\
$k+1$: The induction hypothesis is
\begin{align*}
\forall \ashared.\forall \apc.~~~ &
\left(\begin{aligned}
	\forall i.\forall\alocal.\forall \astmt.\apc(i)=(\alocal, \astmt)~\Rightarrow~
\exists \apredppp\in\thePredicates.\; (\ashared, \alocal)\in\apredppp\wedge 
\locsafek{\astmt}{\apredppp}{\apredp}
\end{aligned}\right)
\\ \Rightarrow~ &\cfsafek{(\ashared, \apc)}{\apredp}. 
\end{align*}
Consider $\aconfig = (\ashared, \apc)$ so that for all $i, \alocal, \astmt$ we have
\begin{align*}
\apc(i)=(\alocal, \astmt)\Rightarrow
\exists \apredppp\in\thePredicates.\; (\ashared, \alocal)\in\apredppp\wedge 
\locsafedef{k+1}{\astmt}{\apredppp}{\apredp}.
\end{align*}
We show $\cfsafedef{k+1}{\aconfig}{\apredp}$.\\
(1) To show $\aconfig\in\acceptset{\apredp}$, let thread $i$ be with $\apc(i)=(\alocal, \cskip)$.
By assumption, there is a predicate $\apredppp\in\thePredicates$ so that $(\ashared, \alocal)\in\apredppp$ and $\locsafedef{k+1}{\cskip}{\apredppp}{\apredp}$.
By definition of $\locsafedef{k+1}{\cskip}{\apredppp}{\apredp}$, we have 
$\apredppp\subseteq\apredp$.
Hence, $(\ashared, \alocal)\in\apredp$, as required.\\
(2) Consider a configuration $\aconfig'$ with $\aconfig\rightarrow\aconfig'=(\ashared', \apc')$.
To establish $\cfsafek{\aconfig'}{\apredp}$, let thread $i$ be with $\apc(i)=(\alocal, \astmt)$ and $\apc'(i)=(\alocal', \astmt')$. 
We show that there is a predicate $\apredpppp\in\thePredicates$ with $(\ashared', \alocal')\in\apredpppp$ and $\locsafek{\astmt'}{\apredpppp}{\apredp}$.
The induction hypothesis then yields $\cfsafek{\aconfig'}{\apredp}$ and concludes the proof.
There are two cases.\\
\textit{Case 1:} Thread~$i$ executes the command $\astmt\xrightarrow{\acom}\astmt'$ that leads to the transition $\aconfig\rightarrow\aconfig'$.
Then  $(\ashared', \alocal')\in\semOf{\acom}(\ashared, \alocal)$.
By assumption, there is a predicate $\apredppp\in\thePredicates$ so that $(\ashared, \alocal)\in\apredppp$ and $\locsafedef{k+1}{\astmt}{\apredppp}{\apredp}$.
The definition of safety gives some $\apredpppp\in\thePredicates$ with $\semOf{\acom}(\apredppp)\subseteq\apredpppp$.
Then $(\ashared', \alocal')\in\semOf{\acom}(\ashared, \alocal)$, $(\ashared, \alocal)\in\apredppp$, and $\semOf{\acom}(\apredppp)\subseteq\apredpppp$ together entail $(\ashared', \alocal')\in\apredpppp$.
Moreover, the predicate $\apredpppp$ satisfies $\locsafedef{k}{\astmt'}{\apredpppp}{\apredp}$, as required.
In the second case, we use that safety of the executing thread also yields $\inter{\apredppp}{\acom}\subseteq \theInterference$.
Hence, if $\effectful{\apredppp}{\acom}$, then there is $(\apredpp, \acom)\in\theInterference$ with $\apredppp\subseteq\apredpp$. \\
\textit{Case 2:} Thread~$i$ experiences the command as an interference.
Hence, $\apc'(i)=\apc(i)=(\alocal, \astmt)$ but potentially $\ashared'\neq \ashared$.
By assumption, there is a predicate $\apredppp\in\thePredicates$ so that $(\ashared, \alocal)\in\apredppp$ and $\locsafedef{k+1}{\astmt}{\apredppp}{\apredp}$.
We show that $\apredpppp=\apredppp$ is the right choice. 
We already argued that the interference is covered by some $(\apredpp, \acom)\in\theInterference$.
Hence, we have $(\ashared', \alocal)\in\semOf{(\apredpp, \acom)}(\ashared, \alocal)$.
Moreover $\semOf{(\apredpp, \acom)}(\ashared, \alocal)\subseteq \semOf{(\apredpp, \acom)}(\apredppp)$. 
By interference freedom, we have $\semOf{(\apredpp, \acom)}(\apredppp)\subseteq \apredppp$. 
Hence, $(\ashared', \alocal)\in\apredppp$. 
Moreover, $\locsafedef{k+1}{\astmt}{\apredppp}{\apredp}$ entails the desired $\locsafek{\astmt}{\apredppp}{\apredp}$ by Lemma~\ref{Lemma:Monotonicity}. 
\end{proof}

\subsection{Soundness}
Finally, we can prove the overall soundness theorem.

\begin{proof}[Proof of Theorem~\ref{Theorem:Soundness}]
Assume $\thePredicates, \theInterference\semCalc\hoareOf{\apred}{\astmt}{\apredp}$ and $\isInterferenceFreeOf[\theInterference]{\thePredicates}$ and $\apred\in\thePredicates$. %\\
In order to show that $\subModels\hoareOf{\apred}{\astmt}{\apredp}$ holds, consider a configuration $(\ashared, \apc)\in\initset{\apred}{\astmt}$. %\\
By definition of $\initset{\apred}{\astmt}$, every thread $i$ satisfies $\apc(i)=(\alocal, \astmt)$ with $(\ashared, \alocal)\in\apred$. %\\
By $\thePredicates, \theInterference\semCalc\hoareOf{\apred}{\astmt}{\apredp}$ and Proposition~\ref{Proposition:LocalSound}, we have $\locsafek{\astmt}{\apred}{\apredp}$ for every $k$. %\\
With $\isInterferenceFreeOf[\theInterference]{\thePredicates}$, $\apred\in\thePredicates$, and Lemma~\ref{Lemma:LocalGlobal}, we get $\cfsafek{(\ashared, \apc)}{\apredp}$ for all $k$. %\\
This shows that every reachable configuration is accepting for $\apredp$.
\end{proof}

%%% Local Variables:
%%% mode: latex
%%% TeX-master: "../main"
%%% End:

%!TEX root = ../main.tex

\section{Proofs of Section~\ref{Section:Futures}}
% ---------------------------------------------------------------------------------
% ---------------------------------------------------------------------------------
% ---------------------------------------------------------------------------------
% ---------------------------------------------------------------------------------

We will first prove an auxiliary lemma, which we will use to show the soundness of \ruleref{f-infer}.
% ---------------------------------------------------------------------------------
% ---------------------------------------------------------------------------------
% ---------------------------------------------------------------------------------
% ---------------------------------------------------------------------------------
\begin{lemma}\label{Lemma:MonotonicitySepImp}
$\apred_2\subseteq\apred_1$ and $\apredp_1\subseteq\apredp_2$ imply $\apred_1\sepimp\apredp_1\subseteq \apred_2\sepimp\apredp_2$. 
\end{lemma}
% ---------------------------------------------------------------------------------
% ---------------------------------------------------------------------------------
% ---------------------------------------------------------------------------------
% ---------------------------------------------------------------------------------
\begin{proof}
Consider $\astate\in\apred_1\sepimp\apredp_1$.
Then $\setcompact{\astate}\mstar\apred_1\subseteq \apredp_1$.
As $\apred_2\subseteq\apred_1$, we get $\setcompact{\astate}\mstar\apred_2\subseteq \apredp_1$.
As $\apredp_1\subseteq\apredp_2$, we get $\setcompact{\astate}\mstar\apred_2\subseteq \apredp_2$.
This means $\astate\in\apred_2\sepimp\apredp_2$.
\end{proof}

\begin{proof}[Proof of Lemma~\ref{Lemma:FutureSoundness}]
  We prove validity of each of the rules in succession.
\begin{asparaenum}[(i)]
\item \textit{Soundness of \ruleref{f-intro}.} 
We have $\emp * \apred = \apred$.
Moreover, we have:
\begin{align*}
  \FUT{\apred}{\acom}{\apredp}
  =\apred\sepimp \wpreOf{\acom}{\apredp}
  =\setcond{\astate}{\setcompact{\astate}\mstar \apred \subseteq \wpreOf{\acom}{\apredp}}
\end{align*}
We thus have to show $\emp \mstar \apred=\apred\subseteq \wpreOf{\acom}{\apredp}$.
This holds by assumption.

% ---------------------------------------------------------------------------------
% ---------------------------------------------------------------------------------
% ---------------------------------------------------------------------------------
% ---------------------------------------------------------------------------------

\item \textit{Soundness of \ruleref{f-seq}.} 
Consider $\astate\in \FUT{\apred}{\acom_1}{\apredp}\statemult\FUT{\apredp}{\acom_2}{\apredpp}$.
Then we have $\astate=\astate_1\mstar\astate_2$ with $\astate_1\in \FUT{\apred}{\acom_1}{\apredp}$ and $\astate_2 \in \FUT{\apredp}{\acom_2}{\apredpp}$.
So we arrive at $\setcompact{\astate}\mstar\apred=\setcompact{\astate_2}\mstar(\setcompact{\astate_1}\mstar\apred)$.
Then, by the locality of commands, we obtain $\semOf{\acom_1}(\setcompact{\astate}\mstar\apred)\subseteq\setcompact{\astate_2}\mstar\semOf{\acom_1}(\setcompact{\astate_1}\mstar\apred)$.
As $\astate_1\in\FUT{\apred}{\acom_1}{\apredp}$, we have $\semOf{\acom_1}(\setcompact{\astate_1}\mstar\apred)\subseteq\apredp$. 
Hence, we get $\semOf{\acom_1}(\setcompact{\astate}\mstar\apred)\subseteq \setcompact{\astate_2}\mstar\apredp$.
This means we arrive at $\setcompact{\astate}\mstar\apred\subseteq \wpreOf{\acom_1}{\setcompact{\astate_2}\mstar\apredp}$. 
As $\astate_2\in\FUT{\apredp}{\acom_2}{\apredpp}$, we have $\setcompact{\astate_2}\mstar\apredp\subseteq\wpreOf{\acom_2}{\apredpp}$.
Together:
\begin{align*}
  \setcompact{\astate}\mstar\apred
  \subseteq \wpreOf{\acom_1}{\setcompact{\astate_2}\mstar\apredp}
  \subseteq \wpreOf{\acom_1}{\wpreOf{\acom_2}{\apredpp}}
  =\wpreOf{\acom_1;\acom_2}{\apredpp}
\end{align*}
By assumption, $\wpreOf{\acom_1;\acom_2}{\apredpp}\subseteq \wpreOf{\acom}{\apredpp}$.
Hence, we arrive at $\setcompact{\astate}\mstar\apred\subseteq \wpreOf{\acom}{\apredpp}$.
This shows $\astate\in \FUT{\apred}{\acom}{\apredpp}$, as required.

% ---------------------------------------------------------------------------------
% ---------------------------------------------------------------------------------
% ---------------------------------------------------------------------------------
% ---------------------------------------------------------------------------------

\item \textit{Soundness of \ruleref{f-account}.} 
Consider $\astate\in \apred\statemult\FUT{\apred\mstar\apredp}{\acom}{\apredpp}$.
Then $\astate = \astate_1\mstar\astate_2$ with $\astate_1\in\apred$ and $\astate_2\in \FUT{\apred\mstar\apredp}{\acom}{\apredpp}$.
The latter means $\semOf{\acom}(\setcompact{\astate_2}\mstar\apred\mstar\apredp)\subseteq\apredpp$.
As $\astate_1\in\apred$, we in particular have $\semOf{\acom}(\setcompact{\astate_2}\mstar\setcompact{\astate_1}\mstar\apredp)\subseteq\apredpp$.
Since $\semOf{\acom}(\setcompact{\astate_2}\mstar\setcompact{\astate_1}\mstar\apredp)=\semOf{\acom}(\setcompact{\astate}\mstar\apredp)$, we obtain the required $\astate\in\FUT{\apredp}{\acom}{\apredpp}$.
%\end{proof}
% ---------------------------------------------------------------------------------
% ---------------------------------------------------------------------------------
% ---------------------------------------------------------------------------------
% ---------------------------------------------------------------------------------
% ---------------------------------------------------------------------------------
% ---------------------------------------------------------------------------------
% ---------------------------------------------------------------------------------
% ---------------------------------------------------------------------------------
\item \textit{Soundness of \ruleref{f-infer}.} 
If $\apredp_1\subseteq \apredp_2$, then $\wpreOf{\acom}{\apredp_1}\subseteq \wpreOf{\acom}{\apredp_2}$.
We conclude as follows:
\begin{align*}
\;\;&\FUT{\apred_1}{\acom}{\apredp_1}\\
\commentl{\small{Definition~\ref{Definition:Future}}}=\;\;&\apred_1\sepimp\wpreOf{\acom}{\apredp_1}\\
\commentl{\small{Lemma~\ref{Lemma:MonotonicitySepImp} and assumption}}\subseteq\;\;&\apred_2\sepimp\wpreOf{\acom}{\apredp_1}\\
\commentl{\small{Lemma~\ref{Lemma:MonotonicitySepImp} and remark above}}\subseteq\;\;&\apred_2\sepimp\wpreOf{\acom}{\apredp_2}\\
\commentl{\small{Definition~\ref{Definition:Future}}}=\;\;&\FUT{\apred_2}{\acom}{\apredp_2}.
\end{align*}
% ---------------------------------------------------------------------------------
% ---------------------------------------------------------------------------------
% ---------------------------------------------------------------------------------
% ---------------------------------------------------------------------------------

\item \textit{Soundness of \ruleref{f-frame}.} 
Consider $\astate\in\FUT{\apred}{\acom}{\apredp}$.
Then $\setcompact{\astate}\mstar\apred\subseteq\wpreOf{\acom}{\apredp}$.
This in turn means $\semOf{\acom}(\setcompact{\astate}\mstar\apred)\subseteq \apredp$.
By the locality of commands, we have $\semOf{\acom}(\setcompact{\astate}\mstar\apred\mstar\apredpp)\subseteq \apredp\mstar\apredpp$.
This means $\setcompact{\astate}\mstar\apred\mstar\apredpp\subseteq\wpreOf{\acom}{\apredp\mstar\apredpp}$.
Hence, $\astate\in\FUT{\apred\mstar\apredpp}{\acom}{\apredp\mstar\apredpp}$.
% ---------------------------------------------------------------------------------
% ---------------------------------------------------------------------------------
% ---------------------------------------------------------------------------------
% ---------------------------------------------------------------------------------

\item \textit{Soundness of \ruleref{f-invoke}.} 
We conclude by:
\begin{align*}
\;\;&\apred\statemult\FUT{\apred}{\acom}{\apredp}\\
\commentl{\small{$\emp$ neutral}}=\;\;&\apred\statemult\FUT{\apred\mstar\emp}{\acom}{\apredp}\\
\commentl{\small{Soundness of \ruleref{f-account}}}\subseteq\;\;& \FUT{\emp}{\acom}{\apredp}\\
\commentl{\small{Definition~\ref{Definition:Future}}}=\;\;& \emp\sepimp\wpreOf{\acom}{\apredp}\\
\commentl{\small{$\emp$ neutral}}=\;\;& \wpreOf{\acom}{\apredp}.
\end{align*}
\end{asparaenum}
\end{proof}

%%% Local Variables:
%%% mode: latex
%%% TeX-master: "../main"
%%% End:

%!TEX root = ../main.tex

\newcommand{\stateunitseq}[1][\astateseq]{\mathbf{\stateunit}_{#1}}

\section{Proofs of Section~\ref{Section:History}}
\label{Section:HistoryProofs}
% --------------------------------------------------------------------------------
% --------------------------------------------------------------------------------
% --------------------------------------------------------------------------------
% --------------------------------------------------------------------------------

\begin{proof}[Proof of Lemma~\ref{Lemma:Frameable}]
Consider $\astateseq'\in\csemOf{\acom}(\acpred\mstar\acpredpp)$.
Then there is $\astateseq\in\acpred\mstar\acpredpp$ with $\astateseq'\in\csemOf{\acom}(\astateseq)$. 
Then $\astateseq=\astateseq_1\cc\astate_1\mstar\astateseq_2\cc\astate_2$ with $\astateseq_1\cc\astate_1\in\acpred$ and $\astateseq_2\cc\astate_2\in\acpredpp$.
By definition, $\astateseq'=\astateseq\cc\astate'$ with $\astate'\in\semOf{\acom}(\astate_1\mstar\astate_2)$.
By locality of commands, $\semOf{\acom}(\astate_1\mstar\astate_2)\subseteq \semOf{\acom}(\astate_1)\mstar\setcompact{\astate_2}$.
Hence, $\astate'=\astate_1'\mstar\astate_2$ with $\astate_1'\in\semOf{\acom}(\astate_1)$. 
By definition, $\astateseq_1\cc\astate_1\cc\astate_1'\in \csemOf{\acom}(\astateseq_1\cc\astate_1)$.
Since $\astateseq_1\cc\astate_1\in\acpred$, we have $\astateseq_1\cc\astate_1\cc\astate_1'\in \csemOf{\acom}(\acpred)$.
By frameability, we have $\astateseq_2\cc\astate_2\cc\astate_2\in\acpredpp$.
Hence, $\astateseq_1\cc\astate_1\cc\astate_1'\mstar\astateseq_2\cc\astate_2\cc\astate_2\in\csemOf{\acom}(\acpred)\mstar\acpredpp$.
Altogether, we arrive at $\astateseq_1\cc\astate_1\cc\astate_1'\mstar\astateseq_2\cc\astate_2\cc\astate_2=\astateseq\cc(\astate_1'\mstar\astate_2)=\astateseq'.$   
\end{proof}
% --------------------------------------------------------------------------------------
% --------------------------------------------------------------------------------------
% --------------------------------------------------------------------------------------
% --------------------------------------------------------------------------------------

We adapt the semantics of concurrency libraries to track history and show that the thread-modular reasoning principle remains sound. 
The first step is to lift the product separation algebra to a history separation algebra containing sequences from $(\setshared\times\setlocal)^+$. 
In a configuration of the concurrency library, we store the two components of such a sequence separately. 
So we store $(\ashared_1, \alocal_1)\ldots (\ashared_n, \alocal_n)$ as the pair $(\asharedseq, \alocalseq)$ with $\asharedseq=\ashared_1\ldots\ashared_n$ and $\alocalseq=\alocal_1\ldots\alocal_n$. 
Note that a predicate $\acpred\subseteq(\setshared\times\setlocal)^+$ only contains pairs of sequences of the same length.
A configuration of the concurrency library has the shape $(\asharedseq, \apc)$ with $\asharedseq\in\setshared^+$ a sequence of global states and for all threads $\apc(i)=(\alocalseq, \astmt)$ with $\alocalseq\in\setlocal^+$ an equally long sequence of local states. 
When we execute a transition, we not only update the local state of the thread being active.
We also store a copy of the current local state in all other threads.
This ensures the resulting sequences of global and local states are again of the same length. 
The modified transition relation is this.
\begin{mathpar}
	\infrule{
		\apc_1(i)=(\alocalseq_1, \astmt_1)\\
		\pcStepOf{\astmt_1}{\acom}{\astmt_2}\\
		(\asharedseq_2, \alocalseq_2)\in\csemOf{\acom}(\asharedseq_1, \alocalseq_1)\\
				\apc_2(i)=(\alocalseq_2, \astmt_2)\\\\
		\forall j\neq i.\;\; \apc_1(j)=(\alocalseq_j, \astmt_j)\\
		\apc_2(j)=(\alocalseq_j\cc\lastOf{\alocalseq_j}, \astmt_j)
}{
		(\asharedseq_1, \apc_1)
		\rightarrow	(\asharedseq_2, \apc_2)
	}
\end{mathpar}
%
% We also adapt the notion of interference to copy the last local state:
We also adapt interference to copy the last local state:
\begin{gather*}
	\csemOf{(\apredpp, \acom)}(\asharedseq, \alocalseq)
	\defeq
	\setcond{
		(\asharedseq', \alocalseq\cc\lastOf{\alocalseq})
	}{
		\exists \alocalseq_1, \alocalseq_2.\;\; 
		(\asharedseq, \alocalseq_1)\in \apredpp
		~\wedge~
		(\asharedseq', \alocalseq_2)\in \csemOf{\acom}(\asharedseq, \alocalseq_1)
	}
	\ .
\end{gather*}

Proposition~\ref{Proposition:LocalSound} refers to a general separation algebra and therefore continues to hold in the present setting modulo the restriction of the \ruleref{frame} rule to frameable predicates.

We now show that the key lemma for lifting the thread-local safety guarantee to configurations carries over to history separation algebras. 

\begin{lemma}\label{Lemma:LocalGlobalComput}
If\, $\aconfig = (\asharedseq, \apc)$ and\, $[\forall i.\forall\alocalseq.\forall \astmt.\apc(i)=(\alocalseq, \astmt).\Rightarrow
\exists \acpredppp\in\thePredicates.\; (\asharedseq, \alocalseq)\in\acpredppp\wedge 
\locsafek{\astmt}{\acpredppp}{\acpredp}]$ and\, $\isInterferenceFreeOf[\theInterference]{\thePredicates}$ hold, then we have\, $\cfsafek{\aconfig}{\acpredp}$. 
\end{lemma}

\begin{proof}[Proof of Lemma~\ref{Lemma:LocalGlobalComput}]
Consider $\thePredicates$ and $\theInterference$ with $\isInterferenceFreeOf[\theInterference]{\thePredicates}$.
Consider $\acpredp$.
We proceed by induction on $k$.\\
$k=0$: Done.\\
$k+1$: The induction hypothesis is
\begin{align*}
\forall \asharedseq.\forall \apc.~~&
\left(
\begin{aligned}
	\forall i.\forall\alocalseq.\forall \astmt.~~ \apc(i)=(\alocalseq, \astmt)
	~\Rightarrow~
	\exists \acpredppp\in\thePredicates.\; (\asharedseq, \alocalseq)\in\acpredppp\wedge 
	\locsafek{\astmt}{\acpredppp}{\acpredp}
\end{aligned}
\right)
\\\Rightarrow~& \cfsafek{(\asharedseq, \apc)}{\acpredp}. 
\end{align*}
Consider $\aconfig = (\asharedseq, \apc)$ so that for all $i, \alocalseq, \astmt$ we have
\begin{align*}
\apc(i)=(\alocalseq, \astmt)\Rightarrow
\exists \acpredppp\in\thePredicates.\; (\asharedseq, \alocalseq)\in\acpredppp\wedge 
\locsafedef{k+1}{\astmt}{\acpredppp}{\acpredp}.
\end{align*}
We show $\cfsafedef{k+1}{\aconfig}{\acpredp}$.\\
(1) To show $\aconfig\in\acceptset{\acpredp}$, let thread $i$ be with $\apc(i)=(\alocalseq, \cskip)$.
By assumption, there is a predicate $\acpredppp\in\thePredicates$ so that $(\asharedseq, \alocalseq)\in\acpredppp$ and $\locsafedef{k+1}{\cskip}{\acpredppp}{\acpredp}$.
By definition of $\locsafedef{k+1}{\cskip}{\acpredppp}{\acpredp}$, we have 
$\acpredppp\subseteq\acpredp$.
Hence, $(\asharedseq, \alocalseq)\in\acpredp$, as required.\\
(2) Consider a configuration $\aconfig'$ with $\aconfig\rightarrow\aconfig'=(\asharedseq', \apc')$.
We establish $\cfsafek{\aconfig'}{\acpredp}$. 
Consider a thread $i$ with $\apc(i)=(\alocalseq, \astmt)$ and $\apc'(i)=(\alocalseq', \astmt')$. 
We show that there is a predicate $\acpredpppp\in\thePredicates$ with $(\asharedseq', \alocalseq')\in\acpredpppp$ and $\locsafek{\astmt'}{\acpredpppp}{\acpredp}$.
The induction hypothesis then yields $\cfsafek{\aconfig'}{\acpredp}$ and concludes the proof.
There are two cases.\\
\textit{Case 1:} Thread~$i$ executes the command $\astmt\xrightarrow{\acom}\astmt'$ that leads to the transition $\aconfig\rightarrow\aconfig'$.
Then  $(\asharedseq', \alocalseq')\in\csemOf{\acom}(\asharedseq, \alocalseq)$.
By assumption, there is a predicate $\acpredppp\in\thePredicates$ with $(\asharedseq, \alocalseq)\in\acpredppp$ and $\locsafedef{k+1}{\astmt}{\acpredppp}{\acpredp}$.
The definition of safety gives a predicate $\acpredpppp\in\thePredicates$ with $\csemOf{\acom}(\acpredppp)\subseteq\acpredpppp$.
Then $(\asharedseq', \alocalseq')\in\csemOf{\acom}(\asharedseq, \alocalseq)$, $(\asharedseq, \alocalseq)\in\acpredppp$, and $\csemOf{\acom}(\acpredppp)\subseteq\acpredpppp$ together entail $(\asharedseq', \alocalseq')\in\acpredpppp$.
Moreover, the predicate $\acpredpppp$ satisfies $\locsafedef{k}{\astmt'}{\acpredpppp}{\acpredp}$, as required.
In the second case, we will use that safety of the executing thread also yields $\inter{\acpredppp}{\acom}\subseteq \theInterference$.
Hence, if $\effectful{\acpredppp}{\acom}$, then there is $(\acpredpp, \acom)\in\theInterference$ with $\acpredppp\subseteq\acpredpp$. \\
\textit{Case 2:} Thread~$i$ experiences the command as an interference.
Then we have $\apc(i)=(\alocalseq, \astmt)$ and $\apc'(i)=(\alocalseq\cc\lastOf{\alocalseq}, \astmt)$.
By assumption, there is a predicate $\acpredppp\in\thePredicates$ so that $(\asharedseq, \alocalseq)\in\acpredppp$ and $\locsafedef{k+1}{\astmt}{\acpredppp}{\acpredp}$.
We show that $\acpredpppp=\acpredppp$ is the right choice. 
We already argued that the interference is covered by some $(\acpredpp, \acom)\in\theInterference$.
Hence, we have $(\asharedseq', \alocalseq\cc\lastOf{\alocalseq})\in\csemOf{(\acpredpp, \acom)}(\asharedseq, \alocalseq)$.
Moreover $\csemOf{(\acpredpp, \acom)}(\asharedseq, \alocalseq)\subseteq \csemOf{(\acpredpp, \acom)}(\acpredppp)$. 
By interference freedom, we have $\csemOf{(\acpredpp, \acom)}(\acpredppp)\subseteq \acpredppp$. 
Hence, $(\asharedseq', \alocalseq\cc\lastOf{\alocalseq})\in\acpredppp$. 
Moreover, $\locsafedef{k+1}{\astmt}{\acpredppp}{\acpredp}$ entails $\locsafek{\astmt}{\acpredppp}{\acpredp}$ by Lemma~\ref{Lemma:Monotonicity}. 
\end{proof}
% --------------------------------------------------------------------------------------
% --------------------------------------------------------------------------------------
% --------------------------------------------------------------------------------------
% --------------------------------------------------------------------------------------

%
%The notions of initial and accepting configurations, soundness, and configuration safety remain unchanged except that now we use letters $\acpred, \acpredp$ to make clear that predicates are sets of sequences of states rather than sets of states, and letters $\asharedseq, \alocalseq$ to indicate that also the shared and local states are sequences.

The proof of Theorem~\ref{Theorem:SoundnessComput} is the same as the one for Theorem~\ref{Theorem:Soundness}, except that we replace Lemma~\ref{Lemma:LocalGlobal} by the above Lemma~\ref{Lemma:LocalGlobalComput}.

\begin{proof}[Proof of Theorem~\ref{Theorem:SoundnessComput}]
Assume that $\thePredicates, \theInterference\semCalc\hoareOf{\acpred}{\astmt}{\acpredp}$ and $\isInterferenceFreeOf[\theInterference]{\thePredicates}$ and $\acpred\in\thePredicates$ hold.
To show $\subModels\hoareOf{\acpred}{\astmt}{\acpredp}$, consider a configuration $(\asharedseq, \apc)\in\initset{\acpred}{\astmt}$.
By definition of $\initset{\acpred}{\astmt}$, every thread $i$ satisfies $\apc(i)=(\alocalseq, \astmt)$ with $(\asharedseq, \alocalseq)\in\acpred$.
By $\thePredicates, \theInterference\semCalc\hoareOf{\acpred}{\astmt}{\acpredp}$ together with Proposition~\ref{Proposition:LocalSound}, we have $\locsafek{\astmt}{\acpred}{\acpredp}$ for every $k$.
With $\isInterferenceFreeOf[\theInterference]{\thePredicates}$, $\acpred\in\thePredicates$, and Lemma~\ref{Lemma:LocalGlobalComput}, we get $\cfsafek{(\asharedseq, \apc)}{\acpredp}$ for all $k$.
This shows that every reachable configuration is accepting for $\acpredp$.
\end{proof}
% --------------------------------------------------------------------------------------
% --------------------------------------------------------------------------------------
% --------------------------------------------------------------------------------------
% --------------------------------------------------------------------------------------
\begin{proof}[Proof of Lemma~\ref{Lemma:FrameablePredicates}]
(1) Let $\astateseq\cc\astate\in\nowOf{\apred}$, which means $\astate\in\apred$.
Then also for $\astateseq\cc\astate\cc\astate$ the last state belongs to $\apred$.
Thus, $\astateseq\cc\astate\cc\astate\in\nowOf{\apred}$.\\
%----------------------------------------------------------------------------------
%----------------------------------------------------------------------------------
Consider $\astateseq\cc\astate\in\pastOf{\apred}$.
This means there is a state from $\apred$ in $\astateseq\cc\astate$.
That state is still present in $\astateseq\cc\astate\cc\astate$.
Hence $\astateseq\cc\astate\cc\astate\in\pastOf{\apred}$, as required.\\
%----------------------------------------------------------------------------------
%----------------------------------------------------------------------------------
%----------------------------------------------------------------------------------
%----------------------------------------------------------------------------------
(2) Consider $\astateseq\cc\astate\in\acpred\mstar\acpredp$.
This means $\astateseq\cc\astate =\astateseq_1\cc\astate_1\mstar\astateseq_2\cc\astate_2$ with 
$\astateseq_1\cc\astate_1\in\acpred$ and $\astateseq_2\cc\astate_2\in\acpredp$.
As $\acpred$ and $\acpredp$ are frameable, we have $\astateseq_1\cc\astate_1\cc\astate_1\in\acpred$ and $\astateseq_2\cc\astate_2\cc\astate_2\in\acpredp$.
Hence, we obtain \[\astateseq\cc\astate\cc\astate=\astateseq_1\cc\astate_1\cc\astate_1\mstar \astateseq_2\cc\astate_2\cc\astate_2\in\acpred\mstar\acpredp \ . \]
%----------------------------------------------------------------------------------
%----------------------------------------------------------------------------------
Consider $\astateseq\cc\astate\in\acpred\cap\acpredp$.
Since $\acpred$ is frameable, we have $\astateseq\cc\astate\cc\astate\in\acpred$.
With the same argument, $\astateseq\cc\astate\cc\astate\in\acpredp$.
Hence, $\astateseq\cc\astate\cc\astate\in\acpred\cap\acpredp$.\\
%----------------------------------------------------------------------------------
%----------------------------------------------------------------------------------
%----------------------------------------------------------------------------------
%----------------------------------------------------------------------------------
The proof for union is the same.
\end{proof}
% --------------------------------------------------------------------------------------
% --------------------------------------------------------------------------------------
% --------------------------------------------------------------------------------------
% --------------------------------------------------------------------------------------
% --------------------------------------------------------------------------------------
% --------------------------------------------------------------------------------------
% --------------------------------------------------------------------------------------
% --------------------------------------------------------------------------------------
\begin{proof}[Proof of Lemma~\ref{Lemma:NowSLOperators}]
We proceed by case analysis.\\[0.2cm]
\textit{Case $\sepimp$:}
``$\subseteq$''\quad
Let $\astateseq\cc\astate\in \nowOf{(\apred\sepimp\apredp)}$.
To show $\astateseq\cc\astate\in \nowOf{\apred}\sepimp\nowOf{\apredp}$, consider a computation $\astateseqp\cc\astatep\in\nowOf{\apred}$ with $\astateseq\cc\astate\statemultdef\astateseqp\cc\astatep$.
Note that definedness of the multiplication guarantees this shape.
We have to show $\astateseq\cc\astate\mstar \astateseqp\cc\astatep = (\astateseq\mstar\astateseqp)\cc(\astate\mstar\astatep)\in \nowOf{\apredp}$.
Since $\astateseq\cc\astate\in \nowOf{(\apred\sepimp\apredp)}$, we have $\astate\in \apred\sepimp\apredp$.
Since $\astateseqp\cc\astatep\in\nowOf{\apred}$, we have $\astatep\in\apred$.
Hence, $\astate\mstar\astatep\in \apredp$.
This yields $(\astateseq\mstar\astateseqp)\cc(\astate\mstar\astatep)\in \nowOf{\apredp}$.\\[0.2cm]
% --------------------------------------------------------------------------------------
% --------------------------------------------------------------------------------------
``$\supseteq$''\quad 
Let $\astateseq\cc\astate\in\nowOf{\apred}\sepimp\nowOf{\apredp}$.
To show $\astateseq\cc\astate\in \nowOf{(\apred\sepimp\apredp)}$, we have to show $\astate\in\apred\sepimp\apredp$.
Consider $\astatep\in\apred$ with $\astatep\statemultdef\astate$.
Let $\stateunitseq$ be the unique computation in $\emp^{\cardof{\astateseq}}$ such that $\stateunitseq \statemultdef \astateseq$.
The computation $\mathbf{\stateunit}_\sigma\cc\astatep$ is in $\nowOf{\apred}$. 
Moreover, since $\stateunitseq \statemultdef \astateseq$ and $\astatep\statemultdef\astate$, we have $\stateunitseq\cc\astatep\statemultdef\astateseq\cc\astate$. 
Hence, by the assumption, $\stateunitseq\cc\astatep\mstar\astateseq\cc\astate = \astateseq\cc(\astatep\mstar\astate)\in\nowOf{\apredp}$.
This shows $\astatep\mstar\astate\in\apredp$ as required.\\[0.2cm]
% --------------------------------------------------------------------------------------
% --------------------------------------------------------------------------------------
\textit{Case $\mstar$:}
``$\subseteq$''\quad
Consider $\astateseq\cc\astate\in\nowOf{(\apred\mstar\apredp)}$.
Then $\astate\in\apred\mstar\apredp$.
This means $\astate=\astate_1\mstar\astate_2$ with $\astate_1\in\apred$ and $\astate_2\in\apredp$.
Let $\stateunitseq$ be the unique computation in $\emp^{\cardof{\astateseq}}$ such that $\stateunitseq \statemultdef \astateseq$.
Then $\astateseq\cc\astate_1\in\nowOf{\apred}$ and $\stateunitseq\cc\astate_2\in\nowOf{\apredp}$. 
Moreover, $\astateseq\cc\astate_1\statemultdef\stateunitseq\cc\astate_2$ with $\astateseq\cc\astate_1\mstar\stateunitseq\cc\astate_2=\astateseq\cc\astate\in\nowOf{\apred}\mstar\nowOf{\apredp}$.\\
``$\supseteq$''\quad Consider $\astateseq\cc\astate\in\nowOf{\apred}\mstar\nowOf{\apredp}$.
By definition, there are $\astateseq_1\cc\astate_1\in\nowOf{\apred}$ as well as $\astateseq_2\cc\astate_2\in\nowOf{\apredp}$ such that $\astateseq_1\cc\astate_1\mstar\astateseq_2\cc\astate_2=\astateseq\cc(\astate_1\mstar\astate_2)=\astateseq\cc\astate$.
We have $\astate_1\in\apred$ and $\astate_2\in\apredp$.
Hence, $\astate\in\apred\mstar\apredp$.
Hence, $\astateseq\cc\astate\in\nowOf{(\apred\mstar\apredp)}$.\\[0.2cm]
% --------------------------------------------------------------------------------------
% --------------------------------------------------------------------------------------
\textit{Case $\cap$:}
``$\subseteq$''\quad Consider $\astateseq\cc\astate\in\nowOf{(\apred\cap\apredp)}$.
We get $\astate\in \apred\cap\apredp$, which means $\astate\in\apred$ and $\astate\in\apredp$.
Then $\astateseq\cc\astate\in\nowOf{\apred}$ and $\astateseq\cc\astate\in \nowOf{\apredp}$.
This means $\astateseq\cc\astate\in\nowOf{\apred}\cap \nowOf{\apredp}$. \\[0.2cm]
% --------------------------------------------------------------------------------------
% --------------------------------------------------------------------------------------
``$\supseteq$''\quad Consider $\astateseq\cc\astate\in\nowOf{\apred}\cap \nowOf{\apredp}$. 
Then $\astateseq\cc\astate\in\nowOf{\apred}$ and $\astateseq\cc\astate\in \nowOf{\apredp}$.
Then $\astate\in\apred$ and $\astate\in\apredp$, hence $\astate\in\apred\cap\apredp$.
Hence, $\astateseq\cc\astate\in\nowOf{(\apred\cap\apredp)}$.\\[0.2cm]
% --------------------------------------------------------------------------------------
% --------------------------------------------------------------------------------------
% --------------------------------------------------------------------------------------
% --------------------------------------------------------------------------------------
\textit{Case $\overline{\bullet}$:}
Consider $\astateseq\cc\astate$. 
We have $\astateseq\cc\astate\in\nowOf{(\overline{\apred})}$ iff $\astate\in\overline{\apred}$ iff $\astate\notin\apred$ iff $\astateseq\cc\astate\notin\nowOf{\apred}$ iff $\astateseq\cc\astate\in\overline{\nowOf{\apred}}$.\\[0.2cm]
% --------------------------------------------------------------------------------------
% --------------------------------------------------------------------------------------
% --------------------------------------------------------------------------------------
% --------------------------------------------------------------------------------------
\textit{Case $\cup$:} Follows from $\cap$ and $\overline{\bullet}$. \\[0.2cm]
% --------------------------------------------------------------------------------------
% --------------------------------------------------------------------------------------
% --------------------------------------------------------------------------------------
% --------------------------------------------------------------------------------------
\textit{Case $\false$:}\quad We have $\false = \emptyset = \Sigma^*\cc\emptyset = \nowOf{\false}$.\\[0.2cm]
% --------------------------------------------------------------------------------------
% --------------------------------------------------------------------------------------
% --------------------------------------------------------------------------------------
% --------------------------------------------------------------------------------------
\textit{Case $\true$:}\quad We have $\true = \Sigma^+ = \Sigma^*\cc\Sigma = \nowOf{\true}$.\\[0.2cm]
% --------------------------------------------------------------------------------------
% --------------------------------------------------------------------------------------
% --------------------------------------------------------------------------------------
% --------------------------------------------------------------------------------------
\textit{Case inclusion:}
``$\Rightarrow$''\quad
Assume $\nowOf{\apred}\subseteq \nowOf{\apredp}$.
Let $\astate\in\apred$ be some state.
The state is also a computation in $\nowOf{\apred}$.
As $\nowOf{\apred}\subseteq\nowOf{\apredp}$, we have $\astate\in\nowOf{\apredp}$, which means $\astate\in\apredp$.\\[0.2cm]
``$\Leftarrow$''\quad Assume $\apred\subseteq\apredp$.
Consider a computation $\astateseq\cc\astate\in\nowOf{\apred}$.
Then we have $\astate\in\apred$, and hence $\astate\in\apredp$.
Hence, $\astateseq\cc\astate\in\nowOf{\apredp}$. 
\end{proof}
% --------------------------------------------------------------------------------------
% --------------------------------------------------------------------------------------
% --------------------------------------------------------------------------------------
% --------------------------------------------------------------------------------------
\begin{proof}[Proof of Lemma~\ref{Lemma:PastSLOperators}]
We proceed by case analysis. \\[0.2cm]
\textit{Case $\nowOf{\apred}\subseteq \pastOf{\apred}$:}\quad We have $\nowOf{\apred}=\Sigma^*\cc\apred\subseteq\Sigma^*\cc\apred\cc\Sigma^* = \pastOf{\apred}$.\\[0.2cm]
% --------------------------------------------------------------------------------------
% --------------------------------------------------------------------------------------
% --------------------------------------------------------------------------------------
% --------------------------------------------------------------------------------------
\textit{Case $\pastOf{(\apred\mstar\apredp)}\subseteq \pastOf{\apred}\mstar\pastOf{\apredp}$:}\quad
Consider $\astateseq\in\pastOf{(\apred\mstar\apredp)}$.
Then there is a decomposition $\astateseq_1\cc\astate\cc\astateseq_2$ with $\astate\in\apred\mstar\apredp$.
Then we have $\astate=\astate_1\mstar\astate_2$ with $\astate_1\in\apred$ and $\astate_2\in\apredp$.
We define $\astateseqp_1 = \astateseq_1\cc\astate_1\cc\astateseq_2\in\pastOf{\apred}$.
Let $\stateunitseq[\astateseq_1]$ be the unique computation in $\emp^{\cardof{\astateseq_1}}$ such that $\stateunitseq[\astateseq_1] \statemultdef \astateseq_1$ and define $\stateunitseq[\astateseq_2]$ similarly for $\astateseq_2$.
We set $\astateseqp_2=\stateunitseq[\astateseq_1]\cc\astate_2\cc\stateunitseq[\astateseq_2]\in\pastOf{\apredp}$.
Then $\astateseq=\astateseqp_1\mstar\astateseqp_2\in\pastOf{\apred}\mstar\pastOf{\apredp}$.\\[0.2cm]
% --------------------------------------------------------------------------------------
% --------------------------------------------------------------------------------------
% --------------------------------------------------------------------------------------
% --------------------------------------------------------------------------------------
\textit{Case $\pastOf{\apred}\sepimp\pastOf{\apredp}\subseteq\pastOf{(\apred\sepimp\apredp)}$:}\quad
Consider $\astateseq\in \pastOf{\apred}\sepimp\pastOf{\apredp}$.
Towards a contradiction, assume $\astateseq\notin \pastOf{(\apred\sepimp\apredp)}$.
Then for every decomposition $\astateseq = \astateseq_1\cc\astate\cc\astateseq_2$ we have $\astate\notin \apred\sepimp\apredp$.
This means there is $\astatep_{\astate}\in \apred$ with $\astatep_{\astate}\statemultdef\astate$ but $\astate\mstar\astatep_{\astate}\notin \apredp$.
Let $\astateseqp$ be the computation consisting of all such $\astatep_{\astate}$ (in the right order).
Then $\astateseqp\in\pastOf{\apred}$ and $\astateseqp\statemultdef\astateseq$.
Hence, $\astateseqp\statemult\astateseq \in\pastOf{\apredp}$, because $\astateseq\in \pastOf{\apred}\sepimp\pastOf{\apredp}$.
This, however, contradicts the construction of $\astateseqp$.\\[0.2cm]
% --------------------------------------------------------------------------------------
% --------------------------------------------------------------------------------------
% --------------------------------------------------------------------------------------
% --------------------------------------------------------------------------------------
\textit{Case $\pastOf{(\apred\cap\apredp)}\subseteq \pastOf{\apred}\cap\pastOf{\apredp}$:}\quad We have $\apred\cap\apredp\subseteq \apred$ and $\apred\cap\apredp\subseteq\apredp$.
Hence, $\pastOf{(\apred\cap\apredp)}\subseteq \pastOf{\apred}$ and $\pastOf{(\apred\cap\apredp)}\subseteq \pastOf{\apredp}$ by the equivalence for inclusion.
Hence, $\pastOf{(\apred\cap\apredp)}\subseteq \pastOf{\apred}\cap\pastOf{\apredp}$.\\[0.2cm]
% --------------------------------------------------------------------------------------
% --------------------------------------------------------------------------------------
% --------------------------------------------------------------------------------------
% --------------------------------------------------------------------------------------
\textit{Case $\pastOf{(\apred\cup\apredp)}= \pastOf{\apred}\cup\pastOf{\apredp}$:}
``$\subseteq$''\quad Consider $\astateseq\in\pastOf{(\apred\cup\apredp)}$.
Then $\astateseq=\astateseq_1\cc\astate\cc\astateseq_2$ with $\astate\in\apred\cup\apredp$, say $\astate\in\apred$.
Then $\astateseq=\astateseq_1\cc\astate\cc\astateseq_2\in\pastOf{\apred}$.\\
% --------------------------------------------------------------------------------------
% --------------------------------------------------------------------------------------
``$\supseteq$''\quad We get $\pastOf{(\apred\cup\apredp)}\supseteq \pastOf{\apred}$ from $\apred\subseteq\apred\cup\apredp$ and the equivalence for inclusion.\\[0.2cm]
% --------------------------------------------------------------------------------------
% --------------------------------------------------------------------------------------
% --------------------------------------------------------------------------------------
% --------------------------------------------------------------------------------------
\textit{Case $\true\mstar\pastOf{\apred}= \pastOf{(\apred\mstar\true)}$:}
``$\subseteq$''\quad Consider $\astateseqp\mstar\astateseq\in\true\mstar\pastOf{\apred}$.
Then, $\astateseqp = \astateseqp_1\cc\astatep\cc\astateseqp_2$ and $\astateseq = \astateseq_1\cc\astate\cc\astateseq_2$ with $\cardof{\astateseqp_1}=\cardof{\astateseq_1}$, $\cardof{\astateseqp_2}=\cardof{\astateseq_2}$, and $\astate\in\apred$.
Consequently, we have $\astatep\mstar\astate\in\true\mstar\apred$.
Hence, we arrive at $\astateseqp\mstar\astateseq = (\astateseqp_1\mstar\astateseq_1)\cc(\astatep\mstar\astate)\cc(\astateseqp_2\mstar\astateseq_2)\in\pastOf{(\true\mstar\apred)}$.\\
% --------------------------------------------------------------------------------------
% --------------------------------------------------------------------------------------
``$\supseteq$''\quad Consider $\astateseq\in\pastOf{(\apred\mstar\true)}$.
Then, $\astateseq = \astateseq_1\cc\astate\cc\astateseq_2$ with $\astate\in\apred\mstar\true$.
Then, we get $\astate = \astate_1\mstar\astate_2$ with $\astate_1\in\apred$. 
Then, $\astateseq_1\cc\astate_1\cc\astateseq_2\in\pastOf{\apred}$. 
Let $\stateunitseq[\astateseq_1]$ be the unique computation in $\emp^{\cardof{\astateseq_1}}$ such that $\stateunitseq[\astateseq_1] \statemultdef \astateseq_1$ and define $\stateunitseq[\astateseq_2]$ similarly for $\astateseq_2$.
Then $\stateunitseq[\astateseq_1]\cc\astate_2\cc\stateunitseq[\astateseq_2]\in\true$ is such that $\astateseq_1\cc\astate_1\cc\astateseq_2\statemultdef \stateunitseq[\astateseq_1]\cc\astate_2\cc\stateunitseq[\astateseq_2]$. 
We have $\astateseq_1\cc\astate_1\cc\astateseq_2\mstar \stateunitseq[\astateseq_1]\cc\astate_2\cc\stateunitseq[\astateseq_2]=\astateseq\in\pastOf{\apred}\mstar\true$. \\[0.2cm]
% --------------------------------------------------------------------------------------
% --------------------------------------------------------------------------------------
The remaining cases are similar to the proof of Lemma~\ref{Lemma:NowSLOperators}.
\end{proof}
% --------------------------------------------------------------------------------------
% --------------------------------------------------------------------------------------
% --------------------------------------------------------------------------------------
% --------------------------------------------------------------------------------------
\begin{proof}[Proof of Lemma~\ref{Lemma:WPNowPast}]
\textit{Now:}
``$\subseteq$''\quad Consider some $\astateseq\cc\astate_1\in\wpreOf{\acom}{\nowOf{\apred}}$, which means that we have $\csemOf{\acom}(\astateseq\cc\astate_1)\subseteq\nowOf{\apred}$.
Hence, for all $\astateseq\cc\astate_1\cc\astate_2\in \csemOf{\acom}(\astateseq\cc\astate_1)$ we have $\astate_2\in\apred$.
By definition, $\astate\in\semOf{\acom}(\astate_1)$ implies $\astateseq\cc\astate_1\cc\astate\in \csemOf{\acom}(\astateseq\cc\astate_1)$.
Hence, $\semOf{\acom}(\astate_1)\subseteq\apred$, which means $\astate_1\in\wpreOf{\acom}{\apred}$.
Hence, $\astateseq\cc\astate_1\in\nowOf{\wpreOf{\acom}{\apred}}$.\\[0.2cm]
% ------
``$\supseteq$''\quad For the reverse inclusion, consider $\astateseq\cc\astate_1\in\nowOf{\wpreOf{\acom}{\apred}}$.
Then $\astate_1\in\wpreOf{\acom}{\apred}$, which means $\semOf{\acom}(\astate_1)\subseteq \apred$.
Consider a computation $\astateseq\cc\astate_1\cc\astate_2\in \csemOf{\acom}(\astateseq\cc\astate_1)$.
By definition, $\astate_2\in \semOf{\acom}(\astate_1)$.
Since $\semOf{\acom}(\astate_1)\subseteq \apred$, we get $\astateseq\cc\astate_1\cc\astate_2\in \nowOf{\apred}$.
This shows $\csemOf{\acom}(\astateseq\cc\astate_1)\subseteq \nowOf{\apred}$, which means $\astateseq\cc\astate_1\in \wpreOf{\acom}{\nowOf{\apred}}$.\\[0.2cm]
% --------------------------------------------------------------------------------------
% --------------------------------------------------------------------------------------
% --------------------------------------------------------------------------------------
% --------------------------------------------------------------------------------------
\textit{Past:}
``$\subseteq$''\quad Consider $\astateseq\cc\astate_1\in\wpreOf{\acom}{\pastOf{\apred}}$, which means that we have $\csemOf{\acom}(\astateseq\cc\astate_1)\subseteq\pastOf{\apred}$.
Hence, for all $\astateseq\cc\astate_1\cc\astate_2\in \csemOf{\acom}(\astateseq\cc\astate_1)$ we have $\astateseq\cc\astate_1\cc\astate_2\in\pastOf{\apred}$.
There are two cases.\\
\textit{Case 1}: A state from $\astateseq\cc\astate_1$ belongs to $\apred$.
Then $\astateseq\cc\astate_1\in\pastOf{\apred}$.\\
\textit{Case 2}: No state from $\astateseq\cc\astate_1$ belongs to $\apred$.
Then for all $\astateseq\cc\astate_1\cc\astate_2\in \csemOf{\acom}(\astateseq\cc\astate_1)$ we have $\astate_2\in\apred$.
This means $\csemOf{\acom}(\astateseq\cc\astate_1)\subseteq\nowOf{\apred}$.
Hence, $\astateseq\cc\astate_1\in\wpreOf{\acom}{\nowOf{\apred}}$.\\[0.2cm]
%----------------------------------------------------------------------------------
%----------------------------------------------------------------------------------
``$\supseteq$''\quad For the reverse inclusion, consider $\astateseq\cc\astate_1\in\pastOf{\apred}$.
Then for all $\astateseq\cc\astate_1\cc\astate_2\in \csemOf{\acom}(\astateseq\cc\astate_1)$ we have $\astateseq\cc\astate_1\cc\astate_2\in\pastOf{\apred}$.
Hence, $\astateseq\cc\astate_1\in\wpreOf{\acom}{\pastOf{\apred}}$. \\
Consider $\astateseq\cc\astate_1\in\wpreOf{\acom}{\nowOf{\apred}}$.
Then for all $\astateseq\cc\astate_1\cc\astate_2\in \csemOf{\acom}(\astateseq\cc\astate_1)$ we have $\astateseq\cc\astate_1\cc\astate_2\in\nowOf{\apred}$.
Since $\nowOf{\apred}\subseteq\pastOf{\apred}$, we get $\astateseq\cc\astate_1\in\wpreOf{\acom}{\pastOf{\apred}}$. 
\end{proof}

%----------------------------------------------------------------------------------
%----------------------------------------------------------------------------------
%----------------------------------------------------------------------------------
%----------------------------------------------------------------------------------
\begin{proof}[Proof of Lemma~\ref{thm:soundness-hindsight-rules}]
\label{proof:proof-hindsight-rules}
%\twout{Rule~\ruleref{h-intro} is trivial.}
%\twout{Rule~\ruleref{h-infer} is a consequence of \Cref{Lemma:PastSLOperators}.}
%\twout{So, consider Rule~\ruleref{h-hindsight}.}
%
Let $\astateseq \in \nowOf{\apred} \mstar \pastOf{\apredp}$. 
This means there are $\astateseq_1 \in \nowOf{\apred}$ and $\astateseq_2 \in \pastOf{\apredp}$ so that $\astateseq=\astateseq_1\mstar\astateseq_2$. By the definition of $\pastOf{q}$ we must have $\astateseq_1=\astateseqp_1\cdot(\astate_1,\anindex)\cdot\astateseqp'_1$ and $\astateseq_2=\astateseqp_2\cdot(\astate_2,\anindex)\cdot\astateseqp'_2$ such that $\astateseqp_1 \# \astateseqp_2$, $\astate_1 \# \astate_2$, $\astateseqp'_1 \# \astateseqp'_2$, and $(\astate_2,\anindex) \in \apredp$. Likewise, by definition of $\nowOf{p}$ we must also have $\astateseqp_1\cdot(\astate_1,\anindex)=(\astate_1',\anindex)\cdot\astateseqp_1'$ for some $(\astate_1',\anindex) \in p$.
Since $\apred$ is pure, it follows that $(\astate_1,\anindex) \in \apred$.
Hence, we obtain $(\astate_1,\anindex) \mstar (\astate_2,\anindex) \in (\apred \mstar \apredp)$ and therefore $\astateseq \in \pastOf{(\apred \mstar \apredp)}$.

For proving the reverse inclusion, let $\astateseq \in \pastOf{(\apred \mstar \apredp)}$. By the definition of $\pastOf{(p \mstar q)}$ we must have $\astateseq = \astateseq_1 \cdot (\astate,\anindex) \cdot \astateseq_2$ with $(\astate, \anindex) \in p \mstar q$ and $\astateseq_1,\astateseq_2 \in (\setstates \times \set{\anindex})^*$. Moreover, $\astate = \astate_\apred \mstar \astate_\apredp$ for some $(\astate_\apred, \anindex) \in \apred$ and $(\astate_\apredp, \anindex) \in \apredp$. Now choose $\astateseqp_1,\astateseqp_2 \in (\emp \times \{\anindex\})^*$ such that $\astateseqp_1 \# \astateseq_1$ and $\astateseqp_2 \# \astateseq_2$. 
Define $\astateseq_\apred \defeq \astateseqp_1 \cdot (\astate_\apredp,\anindex) \cdot \astateseqp_2$ and $\astateseq_\apredp \defeq \astateseq_1 \cdot (\astate_\apredp,\anindex) \cdot \astateseq_2$. Then we have $\astateseq_\apred \# \astateseq_\apredp$. We also have $\astateseq_\apredp \in \pastOf{\apredp}$ by definition of $\pastOf{\apredp}$. Next, we know that we must have $\astateseqp_1 \cdot (\astate_\apredp, \anindex) = (\astate_\apred', \anindex) \cdot \astateseqp_1'$ for some $\astate_\apred' \in \setstates$ and $\astateseqp_1' \in (\setstates \times \set{\anindex})^*$. Since $\apred$ is pure, we know that $(\astate_\apred', \anindex) \in \apred$. Hence, $\astateseq_\apred \in \nowOf{\apred}$ and therefore $(\astateseq_\apred \mstar \astateseq_\apredp) \in \nowOf{\apred} \mstar \pastOf{\apredp}$. To conclude the proof, observe that by construction we have $\astateseq_\apred \mstar \astateseq_\apredp=\astateseq $.

%
%
%Again since $\nowOf{\apred}$ is pure, we have $\apred\subseteq \emp$. 
%
%Hence, we obtain $\apredp\subseteq \apred\mstar \apredp$ and thus $\pastOf{\apredp}\subseteq\pastOf{(\apred\mstar \apredp)}$ by \mbox{\Cref{Lemma:PastSLOperators}}.
%
% \rfm{This is the reverse inclusion and has to be commented out.}
% Consider $(\astateseq, \anindex)\in \pastOf{(\apred\mstar\apredp)}$. 
% %
% This means $\astateseq=\astateseq_1.\astate.\astateseq_2$ with $\astate\in \apred\mstar\apredp$. 
% Hence, there are $\astate_1\in\apred$  and $\astate_2\in\apredp$ with $\astate=\astate_1\mstar\astate_2$. 
% %
% As $\nowOf{\apred}$ is pure, we have $\apred\subseteq\emp$ and so $\astate_1\mstar\astate_2=\astate_2$. 
% %
% By the definition of separation algebras, there is a unit for every state so that the multiplication is defined.  
% %
% Hence, there are $\astateseqp_1\in\emp^+$ and $\astateseqp_2\in\emp^+$ with $\astateseqp_1.\astate_1.\astateseqp_2\statemultdef\astateseq_1.\astate_2.\astateseq_2$. 
% %
% This yields \ldots\quad \rfm{Hier geht es nicht weiter.}
\end{proof}
%----------------------------------------------------------------------------------
%----------------------------------------------------------------------------------
%----------------------------------------------------------------------------------
%----------------------------------------------------------------------------------
Recall that a predicate $\anipred$ is \emph{intuitionistic}, if $\anipred\mstar\true=\anipred$.  
Intuitionism propagates to the computation predicates, where $\true$ is defined to be the set of all computations.

\begin{lemma}\label{Lemma:Intuitionism}
	If $\anipred$ is an intuitionistic state predicate, then $\nowOf{\anipred}$ and $\pastOf{\anipred}$ are intuitionistic computation predicates. 
\end{lemma}
\begin{proof}[Proof of Lemma~\ref{Lemma:Intuitionism}]
% We consider $\nowOf{\anipred}$, the case of $\pastOf{\anipred}$ is similar.
% We have $\nowOf{\anipred}\mstar\true=\nowOf{\anipred}\mstar\nowOf{\true}=\nowOf{(\anipred\mstar\true)}=\nowOf{\anipred}$.
% The first two equalities are by Lemma~\ref{Lemma:NowSLOperators}, the last additionally uses the fact that $\anipred$ is intuitionistic.
Consider $\nowOf{\anipred}$.
We have $\nowOf{\anipred}\mstar\true=\nowOf{\anipred}\mstar\nowOf{\true}=\nowOf{(\anipred\mstar\true)}=\nowOf{\anipred}$.
The first two equalities are by Lemma~\ref{Lemma:NowSLOperators}, the last additionally uses the fact that $\anipred$ is intuitionistic.
The case of $\pastOf{\anipred}$ is similar.
\end{proof}

%%% Local Variables:
%%% mode: latex
%%% TeX-master: "../main"
%%% End:

%!TEX root = ../main.tex

\section{Proofs of Section~\ref{Section:Flows}}
\label{Section:FlowsProofs}
\begin{proof}[Proof of Lemma~\ref{Lemma:FlowAlgebra}]
The laws of commutativity and units follow immediately from the definition of composition.
For associativity we need to show that $\aflowconstraint_2\statemultdef \aflowconstraint_3$ and $\aflowconstraint_1\statemultdef(\aflowconstraint_2\mstar\aflowconstraint_3)$ if and only if $\aflowconstraint_1\statemultdef\aflowconstraint_2$ and $(\aflowconstraint_1\mstar\aflowconstraint_2)\statemultdef\aflowconstraint_3$.
We prove the direction from left to right.
The reverse direction holds by symmetry of definedness.

Since $\aflowconstraint_2\statemultdef\aflowconstraint_3$, we have $(\aflowconstraint_2\discup\aflowconstraint_3).\fval\geq\aflowconstraint_2.\fval\discup\aflowconstraint_3.\fval$.
Furthermore, by $\aflowconstraint_1\statemultdef(\aflowconstraint_2\mstar\aflowconstraint_3)$ we have \[(\aflowconstraint_1\discup(\aflowconstraint_2\mstar\aflowconstraint_3)).\fval\geq\aflowconstraint_1.\fval\discup(\aflowconstraint_2\mstar\aflowconstraint_3).\fval \ .\]
Note that $(\aflowconstraint_1\discup(\aflowconstraint_2\mstar\aflowconstraint_3)).\fval=(\aflowconstraint_1\discup\aflowconstraint_2\discup\aflowconstraint_3).\fval$.

To show $\aflowconstraint_1\statemultdef\aflowconstraint_2$, we have to argue that \[(\aflowconstraint_1\discup\aflowconstraint_2).\fval\geq \aflowconstraint_1.\fval\discup\aflowconstraint_2.\fval\ .\]
To see this, note that
\begin{align*}
  (\aflowconstraint_1\discup\aflowconstraint_2).\fval \geq (\aflowconstraint_1\discup \aflowconstraint_2\discup\aflowconstraint_3).\fval|_{(\aflowconstraint_1\discup\aflowconstraint_2).\setnodes}\geq\aflowconstraint_1.\fval\discup\aflowconstraint_2.\fval\ .
\end{align*}
The latter inequality is by the assumptions.
For the former inequality, we note that the fixed point iteration for $(\aflowconstraint_1\discup\aflowconstraint_2).\fval$ starts with a contribution from $\aflowconstraint_3$ (given as inflow) that the iteration for $(\aflowconstraint_1\discup \aflowconstraint_2\discup\aflowconstraint_3).\fval$ only receives when reaching the fixed point.
By monotonicity, every fixed point approximant to the left is then larger than the corresponding approximant to the right, and so is the fixed point.
For  $(\aflowconstraint_1\mstar\aflowconstraint_2)\statemultdef\aflowconstraint_3$, we note that
\begin{align*}
  ((\aflowconstraint_1\mstar\aflowconstraint_2)\discup\aflowconstraint_3).\fval
  =&~(\aflowconstraint_1\discup\aflowconstraint_2\discup\aflowconstraint_3).\fval\geq \aflowconstraint_1.\fval\discup\aflowconstraint_2.\fval\discup\aflowconstraint_3.\fval
  \\\geq&~ (\aflowconstraint_1\discup\aflowconstraint_2).\fval\discup\aflowconstraint_3.\fval \ .
\end{align*}
The first inequality is by the above assumptions.
The second always holds, as remarked above.
\end{proof}
% ------------------------------------------------------------------------------
% ------------------------------------------------------------------------------
% ------------------------------------------------------------------------------
% ------------------------------------------------------------------------------
\begin{proof}[Proof of Lemma~\ref{Lemma:Transformers}]
(i): For $\aflowconstraint_2.\setnodes\cap\aflowconstraint.\setnodes=\emptyset$, we use the fact that $\aflowconstraint_1.\setnodes\cap\aflowconstraint.\setnodes=\emptyset$ follows from $\aflowconstraint_1\statemultdef\aflowconstraint$ and $\aflowconstraint_1.\setnodes=\aflowconstraint_2.\setnodes$.

To see that the outflow of $\aflowconstraint_2$ matches the inflow of $\aflowconstraint$, we consider $\anode\in\aflowconstraint_2.\setnodes$ and $\anodep\in\aflowconstraint.\setnodes$ and reason as follows:
\begin{align*}
\aflowconstraint_2.\outflowof{\anode, \anodep}=\aflowconstraint_1.\outflowof{\anode, \anodep}=\aflowconstraint.\inflowof{\anode, \anodep}. 
\end{align*}
The former equality follows from $\aflowconstraint_2.\inflow=\aflowconstraint_1.\inflow$ together with $\transformerof{\aflowconstraint_1}=_{\aflowconstraint_1.\inflow}\transformerof{\aflowconstraint_2}$.
The second equality is by $\aflowconstraint_1\statemultdef\aflowconstraint$.

The inflow is preserved by the assumption, hence we have the following: \[\aflowconstraint.\outflowof{\anodep, \anode}=\aflowconstraint_1.\inflowof{\anodep, \anode}=\aflowconstraint_2.\inflowof{\anodep, \anode}\ . \]

% ------------------------------------------------------------------------------
% ------------------------------------------------------------------------------
It remains to show $(\aflowconstraint_2\discup\aflowconstraint).\fval=\aflowconstraint_2.\fval\discup\aflowconstraint.\fval$.
We use Bekic's lemma.
Define the target pairing of two functions $f:A\rightarrow B$ and $g:A\rightarrow C$ over the same domain $A$ as the function $\pairingof{f}{g}:A\rightarrow B\times C$ with $\pairingof{f}{g}(a) \defeq (f(a), g(a))$.
We compute the flow of $\aflowconstraint_2\discup \aflowconstraint$ as the least fixed point of a target pairing $\pairingof{f}{g}$ with
\begin{align*}
f:&\; ((\aflowconstraint_2.\setnodes\discup\aflowconstraint.\setnodes) \rightarrow \amonoid)\rightarrow \aflowconstraint_2.\setnodes\rightarrow \amonoid\\
g:&\; ((\aflowconstraint_2.\setnodes\discup\aflowconstraint.\setnodes) \rightarrow \amonoid)\rightarrow \aflowconstraint.\setnodes\rightarrow \amonoid\; .  
\end{align*}
Function $f$ updates the flow of the nodes in $\aflowconstraint_2$ depending on the flow in/inflow from $\aflowconstraint$.
Function $g$ is responsable for the flow of the nodes in $\aflowconstraint$.
The inflow from the nodes outside $\aflowconstraint_2\discup\aflowconstraint$ is constant.
The definition guarantees $(\aflowconstraint_2\discup\aflowconstraint).\fval=\lfpof{\pairingof{f}{g}}$.

We curry the former function, 
\begin{align*}
f:\; (\aflowconstraint.\setnodes\rightarrow \amonoid)\rightarrow (\aflowconstraint_2.\setnodes\rightarrow \amonoid)\rightarrow \aflowconstraint_2.\setnodes\rightarrow \amonoid\; ,
\end{align*}
and obtain, for every $\cval:\aflowconstraint.\setnodes\rightarrow \amonoid$, the function
\begin{align*}
f(\cval):\; (\aflowconstraint_2.\setnodes\rightarrow \amonoid)\rightarrow \aflowconstraint_2.\setnodes\rightarrow \amonoid\; .
\end{align*} 
This function is still monotonic and therefore has a least fixed point.
Hence, the function 
\begin{align*}
f^{\dagger}: (\aflowconstraint.\setnodes\rightarrow \amonoid)\rightarrow \aflowconstraint_2.\setnodes\rightarrow \amonoid\; 
\end{align*}
mapping valuation $\cval:\aflowconstraint.\setnodes\rightarrow \amonoid$ to the least fixed point $\lfpof{f(\cval)}$ is well-defined. 

Beki\'c's lemma~\cite{DBLP:conf/ibm/Bekic84} tells us how to compute least fixed points of target pairings like $\pairingof{f}{g}$ above by successive elimination of the variables.
We first determine $f^{\dagger}$, which is a function in $\aflowconstraint.\setnodes\rightarrow\amonoid$.
We plug this function into $g$ to obtain a function solely in $\aflowconstraint.\setnodes\rightarrow\amonoid$.
To be precise, since $g$ expects a function from $(\aflowconstraint_2.\setnodes\discup\aflowconstraint.\setnodes)\rightarrow\amonoid$, we pair $f^{\dagger}$ with $\myid=\myidof{\aflowconstraint.\setnodes\rightarrow\amonoid}$ and obtain
\begin{align*}
\pairingof{f^{\dagger}}{\myid}:(\aflowconstraint.\setnodes\rightarrow\amonoid)\rightarrow (\aflowconstraint_2.\setnodes\discup\aflowconstraint.\setnodes)\rightarrow\amonoid\; . 
\end{align*}
We compose this function with $g$ and get 
\begin{align*}
g\circ \pairingof{f^{\dagger}}{\myid}:(\aflowconstraint.\setnodes\rightarrow\amonoid)\rightarrow \aflowconstraint.\setnodes\rightarrow\amonoid\; .
\end{align*}
We compute the least fixed point of this composition to obtain the values of the least fixed point of interest on $\aflowconstraint.\setnodes$. 
For the values on $\aflowconstraint_2.\setnodes$, we reinsert the $\aflowconstraint.\setnodes$-values into $f^{\dagger}$. 
Beki\'c's lemma guarantees the correctness of this successive elimination procedure: 
\begin{align*}
  &\lfpof{\pairingof{f}{g}} ~=~ (f^{\dagger}(\cval), \cval)
  \\\text{with}\qquad& \cval ~=~ \lfpof{g\circ \pairingof{f^{\dagger}}{\myid}}
  \ .
\end{align*}

To conclude the proof, we recall that $\transformerof{\aflowconstraint_2}=_{\aflowconstraint_1.\inflow}\transformerof{\aflowconstraint_1}$.
Moreover, for all $\anodep \in \aflowconstraint.\setnodes$ and $\anode \in \aflowconstraint_1.\setnodes$ we have \[
  \aflowconstraint_1.\inflow(\anodep,\anode) = c.\outflow(\anodep,\anode) = \aflowconstraint.\edges_{(\anodep, \anode)}(\aflowconstraint.\fvalof{\anode})\ .
\]
Together with monotonicity of the edge functions, this implies $g\circ \pairingof{f^{\dagger}}{\myid}=_{\aflowconstraint.\fval} g\circ \pairingof{h^{\dagger}}{\myid}$.
Here, $h \prall{:} ((\aflowconstraint_2.\setnodes\discup\aflowconstraint.\setnodes) \prall{\rightarrow} \amonoid)\prall{\rightarrow} \aflowconstraint_2.\setnodes\prall{\rightarrow} \amonoid$ is the transformer derived from $\aflowconstraint_1$ in the same way $f$ was derived from~$\aflowconstraint_2$.
We thus have for all $\anode\in\aflowconstraint.\setnodes$: 
\begin{align*}
&\quad(\aflowconstraint_2\discup\aflowconstraint).\fvalof{\anode}\\
\commentl{\small{Definition $\fval$, $f$, $g$}} =&\quad \lfpof{\pairingof{f}{g}}(\anode)\\
\commentl{\small{Beki\'c's lemma, $\anode\in\aflowconstraint.\setnodes$}} =&\quad \lfpof{g\circ \pairingof{f^{\dagger}}{\myid}}(\anode)\\
\commentl{\small{$\transformerof{\aflowconstraint_2}=_{\aflowconstraint_1.\inflow}\transformerof{\aflowconstraint_1}$, see above}}=&\quad \lfpof{g\circ \pairingof{h^{\dagger}}{\myid}}(\anode)\\
\commentl{\small{Beki\'c's lemma, $\anode\in\aflowconstraint.\setnodes$}} =&\quad \lfpof{\pairingof{h}{g}}(\anode)\\
\commentl{\small{Definition flow, $f$, $g$}} =&\quad (\aflowconstraint_1\discup\aflowconstraint).\fvalof{\anode}\\
\commentl{\small{$\aflowconstraint_1\statemultdef\aflowconstraint$, $\anode\in\aflowconstraint.\setnodes$}}=&\quad\aflowconstraint.\fvalof{\anode}.
\end{align*}
We argue that also for nodes $\anode\in\aflowconstraint_2.\setnodes$ that we have the equality $(\aflowconstraint_2\discup\aflowconstraint).\fvalof{\anode}=\aflowconstraint_2.\fvalof{\anode}$, as follows:
\begin{align*}
&\quad(\aflowconstraint_2\discup\aflowconstraint).\fvalof{\anode}\\
\commentl{\small{Definition $\fval$, $f$, $g$}} =&\quad \lfpof{\pairingof{f}{g}}(\anode)\\
\commentl{\small{Beki\'c's lemma, $\anode\in\aflowconstraint_2.\setnodes$}} =&\quad [f^{\dagger}(\cval)](\anode)\\
\commentl{\small{See above}} =&\quad [f^{\dagger}(\aflowconstraint.\fval)](\anode)\\
\commentl{\small{See below}}=&\quad \aflowconstraint_2.\fvalof{\anode}.
\end{align*}
To see  the last equality, note that we have $\aflowconstraint_2.\inflowof{\anodep, \anode}=\aflowconstraint.\outflowof{\anodep, \anode}=\aflowconstraint.\edges_{(\anodep, \anode)}(\aflowconstraint.\fvalof{\anode})$, for all $\anodep\in\aflowconstraint.\setnodes$.
Hence, $\aflowconstraint_2.\fval$, which we compute from the inflow, is 
the least fixed point of $f$ computed with $\aflowconstraint.\fval$ fixed.

(ii). Follows with a similar but simpler application of Beki\'c's lemma.
\end{proof}
% ------------------------------------------------------------------------------
% ------------------------------------------------------------------------------
% ------------------------------------------------------------------------------
% ------------------------------------------------------------------------------
\begin{proof}[Proof of Lemma~\ref{Lemma:UpdatesConjunction}]
We show the individual cases.

(i):\quad We apply Lemma~\ref{Lemma:Transformers} to derive that $\semComOf{\updfun}{\aflowconstraint}\statemultdef\aflowconstraint'$ holds.
To see that the preconditions in the lemma are met, note that we have $\aflowconstraint\statemultdef\aflowconstraint'$ by assumption, $\semComOf{\updfun}{\aflowconstraint}.\setnodes=\aflowconstraint.\setnodes$ and $\semComOf{\updfun}{\aflowconstraint}.\inflow=\aflowconstraint.\inflow$ because the semantics does not abort, and the application of updates neither changes the set of nodes nor the inflow, and $\transformerof{\aflowconstraint}=\transformerof{\semComOf{\updfun}{\aflowconstraint}}$ because the update does not abort.

% ------------------------------------------------------------------------------
% ------------------------------------------------------------------------------
(ii):\quad Since the separating conjunctions on both sides of the equality are defined by (i) resp. the assumption, the sets of nodes and the inflows coincide.
The sets of edges coincide, because $\domof{\updfun}\subseteq\aflowconstraint.\setnodes\times\nat$ as the update does not abort on~$\aflowconstraint$.

(iii):\quad With the previous argument, we arrive at: $\domof{\updfun}\subseteq(\aflowconstraint.\setnodes\discup\aflowconstraint'.\setnodes)\times\nat=(\aflowconstraint\mstar\aflowconstraint').\setnodes\times\nat$.
Here, disjointness and the separating conjunction rely on $\aflowconstraint\statemultdef\aflowconstraint'$.
For the transformer, we use \[\transformerof{(\aflowconstraint\mstar\aflowconstraint')\applyupdfun}=\transformerof{\semComOf{\updfun}{\aflowconstraint}\mstar\aflowconstraint'}=\transformerof{\aflowconstraint\mstar\aflowconstraint'} \ .\]
The former equality is by (ii).
The latter is Lemma~\ref{Lemma:Transformers} and we already argued that the conditions are met.

(iv)\quad: By (iii), we have $\semComOf{\updfun}{\aflowconstraint\mstar\aflowconstraint'}=(\aflowconstraint\mstar\aflowconstraint')\applyupdfun$.
The desired equality is by (ii).
\end{proof}
% ------------------------------------------------------------------------------
% ------------------------------------------------------------------------------
% ------------------------------------------------------------------------------
% ------------------------------------------------------------------------------
\begin{proof}[Proof of Lemma~\ref{Lemma:FlowGraphs}]
Since separation algebras are closed under Cartesian products, the product operation on flow graphs is associative and commutative and $(\aheapgraph_{\emptyset}, \aflowconstraint_{\emptyset})$ is the neutral element.

% ------------------------------------------------------------------------------
% ------------------------------------------------------------------------------
We show $(\aheapgraph_{\emptyset}, \aflowconstraint_{\emptyset})\in\setflowgraphs$.
For the nodes, we have \[\aheapgraph_{\emptyset}.\setnodes=\emptyset=\aflowconstraint_{\emptyset}.\setnodes \ .\]
The empty edge function $\aflowconstraint.\edges=\emptyset$ is trivially induced (the universal quantifier is over an empty set).
% ------------------------------------------------------------------------------
% ------------------------------------------------------------------------------

We show closedness: if $\aflowgraph_1, \aflowgraph_2\in\setflowgraphs$ and $\aflowgraph_1\statemultdef\aflowgraph_2$, then we also have $\aflowgraph_1\mstar\aflowgraph_2\in\setflowgraphs$.
Due to the Cartesian product construction, we have $(\aflowgraph_1.\aheapgraph\mstar\aflowgraph_2.\aheapgraph, \aflowgraph_1.\aflowconstraint\mstar\aflowgraph_2.\aflowconstraint)\in\setheapgraphs\times\setflowconstraints$.
Note that this uses the fact that $\aflowgraph_1\statemultdef\aflowgraph_2$ implies $\aflowgraph_1.\aheapgraph\statemultdef\aflowgraph_2.\aheapgraph$ and $\aflowgraph_1.\aflowconstraint\statemultdef\aflowgraph_2.\aflowconstraint$.

It remains to show that the constraints on nodes and edges hold.
For the nodes, we have
\begin{align*}
  (\aflowgraph_1.\aheapgraph\mstar\aflowgraph_2.\aheapgraph).\setnodes
  = \aflowgraph_1.\aheapgraph.\setnodes\discup\aflowgraph_2.\aheapgraph.\setnodes
  = \aflowgraph_1.\aflowconstraint.\setnodes\discup\aflowgraph_2.\aflowconstraint.\setnodes = (\aflowgraph_1.\aflowconstraint\mstar\aflowgraph_2.\aflowconstraint).\setnodes
  \ .
\end{align*}
The first equality is by the definition of products on heap graphs, the second is the assumption $\aflowgraph_1, \aflowgraph_2\in\setflowgraphs$, and the last is the definition of products on flow constraints.
% ------------------------------------------------------------------------------
% ------------------------------------------------------------------------------
% ------------------------------------------------------------------------------
% ------------------------------------------------------------------------------
For the edge function, we consider $\anode\in\aflowgraph_1.\aheapgraph.\setnodes$ and $\anodep\in\nat$. 
For $\anode\in\aflowgraph_2.\aheapgraph.\setnodes$, the reasoning is similar:
\begin{align*}
  &\sum_{\substack{\apsel\in\setpsel\\\anodep=(\aflowgraph_1.\aheapgraph\mstar\aflowgraph_2.\aheapgraph).\pvalof{\anode, \apsel}}} \hspace{-.8cm} \genrhsof{\apsel, (\aflowgraph_1.\aheapgraph\mstar\aflowgraph_2.\aheapgraph).\dvalof{\anode}}
  \\
  % \text{\small(Product of heap graphs)}\quad
  \stackrel{(1)}{=}&~\sum_{\substack{\apsel\in\setpsel\\\anodep=(\aflowgraph_1.\aheapgraph.\pval\discup\aflowgraph_2.\aheapgraph.\pval)(\anode, \apsel)}} \hspace{-1cm} \genrhsof{\apsel, (\aflowgraph_1.\aheapgraph.\dval\discup\aflowgraph_2.\aheapgraph.\dval)(\anode)}
  \\
  % \text{\small($\anode\in\aflowgraph_1.\aheapgraph.\setnodes$, disjointness)}\quad
  \stackrel{(2)}{=}&~\sum_{\substack{\apsel\in\setpsel\\\anodep=\aflowgraph_1.\aheapgraph.\pvalof{\anode, \apsel}}} \hspace{-3mm} \genrhsof{\apsel, \aflowgraph_1.\aheapgraph.\dvalof{\anode}}
  \\
  % \text{\small($\aflowgraph_1, \aflowgraph_2$ flow graphs)}\quad
  \stackrel{(3)}{=}&~~ \aflowgraph_1.\aflowconstraint.\edges(\anode, \anodep)
  \\
  % \text{\small($\anode\in\aflowgraph_1.\aheapgraph.\setnodes$, union of functions)}\quad
  \stackrel{(4)}{=}&~~ (\aflowgraph_1.\aflowconstraint.\edges \discup \aflowgraph_2.\aflowconstraint.\edges)(\anode, \anodep)
  \\
  % \text{\small(Product of flow constraints)}\quad
  \stackrel{(5)}{=}&~~ (\aflowgraph_1.\aflowconstraint \mstar \aflowgraph_2.\aflowconstraint).\edges(\anode, \anodep)
  \ .
\end{align*}
where the equalties are due to:
(1) the definition of products of heap graphs,
(2) $\anode\in\aflowgraph_1.\aheapgraph.\setnodes$ together with disjointness,
(3) the fact that $\aflowgraph_1, \aflowgraph_2$ are flow graphs,
(4) $\anode\in\aflowgraph_1.\aheapgraph.\setnodes$ together with union of functions, and
(5) the definition of products of flow constraints.
\end{proof}

}
\end{document}